\documentclass[prd,aps,amsfonts,eqsecnum,superscriptaddress,nofootinbib,notitlepage,longbibliography,final]{revtex4-1}
\usepackage[a4paper, left=1.8cm, right=1.8cm, top=2cm, bottom=2cm]{geometry}

\usepackage[english]{babel}
\addto\extrasenglish{%
	\def\appendixautorefname{Appendix}%
}
\usepackage[utf8]{inputenc}
\usepackage[T1]{fontenc}
\usepackage{lmodern}
\usepackage{amsmath,amsthm,mathdots,amssymb,bm,bbm,mathrsfs,mathtools}
\usepackage{tikz,pgfplots}\pgfplotsset{compat=1.16}
\usepackage[compat=0.6]{yquant}

\usepackage{thm-restate}
\newtheorem{thm}{Theorem}[section]
\newtheorem{cor}{Corollary}[section]
\newtheorem{lem}{Lemma}[section]

\newtheorem{fact}{Fact}[section]
\newtheorem{prop}{Proposition}[section]

\newtheorem{defn}{Definition}[section]

\usepackage{comment}

\usepackage[colorlinks=true, urlcolor=violet, linkcolor=blue, citecolor=red, hyperindex=true, linktocpage=true, pagebackref=true]{hyperref}
\usepackage{mleftright}\mleftright 

\mathtoolsset{showonlyrefs=true}
\usepackage{microtype,color,graphicx} 
\usepackage{tikz-cd}
\usepackage{xcolor}
\usepackage{multirow}
\usepackage{tablefootnote}

\renewcommand{\vec}{\bm}
\newcommand{\CA}{\mathcal{A}}

\newcommand{\CC}{\mathcal{C}}
\newcommand{\BC}{\mathbb{C}}
\newcommand{\CD}{\mathcal{D}}

\newcommand{\CF}{\mathcal{F}}

\newcommand{\CH}{\mathcal{H}}

\newcommand{\CL}{\mathcal{L}}

\newcommand{\CO}{\mathcal{O}}

\newcommand{\CR}{\mathcal{R}}
\newcommand{\BR}{\mathbb{R}}

\newcommand{\BZ}{\mathbb{Z}}

\newcommand{\vA}{\bm{A}}

\newcommand{\vB}{\bm{B}}

\newcommand{\vC}{\bm{C}}
\newcommand{\vD}{\bm{D}}

\newcommand{\vF}{\bm{F}}

\newcommand{\vH}{\bm{H}}
\newcommand{\vI}{\bm{I}}

\newcommand{\vL}{\bm{L}}

\newcommand{\vM}{\bm{M}}
\newcommand{\vN}{\bm{N}}
\newcommand{\vO}{\bm{O}}
\newcommand{\vP}{\bm{P}}
\newcommand{\vPi}{\bm{\Pi}}

\newcommand{\vR}{\bm{R}}
\newcommand{\vS}{\bm{S}}
\newcommand{\vT}{\bm{T}}
\newcommand{\bt}{\bar{t}}
\newcommand{\vU}{\bm{U}}
\newcommand{\vV}{\bm{V}}
\newcommand{\vW}{\bm{W}}
\newcommand{\vX}{\bm{X}}
\newcommand{\vY }{\bm{Y }}
\newcommand{\vZ}{\bm{Z}}

\newcommand{\vrho}{\bm{ \rho}}
\renewcommand{\L}{\left}
\newcommand{\R}{\right}

\newcommand{\bnu}{\bar{\nu}}

\newcommand{\bmu}{\bar{\mu}}

\newcommand{\bvH}{\bar{\vH}}
\newcommand{\bE}{\bar{E}}
\newcommand{\bomega}{\bar{\omega}}
\newcommand{\tOmega}{\tilde{\Omega}}
\newcommand{\tCO}{\tilde{\CO}}

\newcommand{\dagg}{\dagger}

\newcommand{\vertiii}[1]{{\left\vert\kern-0.25ex\left\vert\kern-0.25ex\left\vert #1 \right\vert\kern-0.25ex\right\vert\kern-0.25ex\right\vert}}
\newcommand{\norm}[1]{\Vert {#1} \Vert}

\newcommand{\normp}[2]{\norm{#1}_{#2}}
\newcommand{\lnormp}[2]{\lnorm{#1}_{#2}}

\newcommand{\labs}[1]{\left\vert {#1} \right\vert}
\newcommand{\lnorm}[1]{\left\Vert {#1} \right\Vert}
\newcommand{\e}{\mathrm{e}}

\newcommand{\rd}{\mathrm{d}}
\newcommand{\ri}{\mathrm{i}}
\newcommand*{\tr}{\mathrm{Tr}}
\newcommand*{\poly}{\mathrm{Poly}}

\newcommand{\indicator}{\mathbbm{1}}
\newcommand{\Lword}[1]{Lindbladian}

\newcommand{\ceil}[1]{\left\lceil{#1}\right\rceil}

\newcommand{\nrm}[1]{\left\| #1 \right\|}
\newcommand{\ipc}[2]{\left\langle#1,#2\right\rangle}

\DeclareMathOperator{\sinc}{sinc}
\DeclarePairedDelimiterX{\braket}[1]{\langle}{\rangle}{#1}
\DeclarePairedDelimiterX\ketbra[2]{| }{|}{#1 \delimsize\rangle\!\delimsize\langle #2}	
\DeclarePairedDelimiterX\dotp[2]{\langle}{\rangle}{#1, #2}
\newcommand{\bigO}[1]{\mathcal{O}\left( #1 \right)}
\newcommand{\bigOt}[1]{\widetilde{\mathcal{O}}\left( #1 \right)}
\newcommand{\bigOm}[1]{\Omega\left( #1 \right)}
\newcommand{\bigTh}[1]{\Theta\left( #1 \right)}

\makeatletter
\DeclareRobustCommand*{\pmzerodot}{%
	\nfss@text{%
		\sbox0{$\vcenter{}$}
		\sbox2{0}%
		\sbox4{0\/}%
		\ooalign{%
			0\cr
			\hidewidth
			\kern\dimexpr\wd4-\wd2\relax 
			\raise\dimexpr(\ht2-\dp2)/2-\ht0\relax\hbox{%
				\if b\expandafter\@car\f@series\@nil\relax
				\mathversion{bold}%
				\fi
				$\cdot\m@th$%
			}%
			\hidewidth
			\cr
			\vphantom{0}
		}%
	}%
}

\usepackage{ifdraft}
\ifdraft{
	\usepackage{environ}
	\newcounter{tikzfigcntr}
	\RenewEnviron{tikzpicture}[1][]{%
		\par
		\stepcounter{tikzfigcntr}
		
		\par
		\begin{align}
			\phantom{a} \\
			\text{This is \texttt{tikzpicture}~\thetikzfigcntr -- turn off draft mode to view it!} \\
			\phantom{c} \\
		\end{align}
	}
	\newcommand{\authnote}[3]{{\color{#3} {\bf  #1:} #2}}
}{
    \newcommand{\authnote}[3]{}
}

\usetikzlibrary{quantikz2}
\makeatletter
\def\l@subsection#1#2{}
\def\l@subsubsection#1#2{}
\makeatother
\begin{document}
\renewcommand{\appendixautorefname}{Appendix}
\renewcommand{\chapterautorefname}{Chapter}
\renewcommand{\sectionautorefname}{Section}
\renewcommand{\subsectionautorefname}{Section}
\renewcommand{\subsubsectionautorefname}{Section}	

	\title{Quantum Thermal State Preparation}
	\author{Chi-Fang Chen}
	\email{chifang@caltech.edu}
	\affiliation{Institute for Quantum Information and Matter,
	California Institute of Technology, Pasadena, CA, USA}
        \affiliation{AWS Center for Quantum Computing, Pasadena, CA}
        \author{Michael J. Kastoryano}
        \affiliation{AWS Center for Quantum Computing, Pasadena, CA}
        \affiliation{IT University of Copenhagen, Denmark}
	\author{Fernando G.S.L. Brand\~ao}
	\affiliation{Institute for Quantum Information and Matter,
	California Institute of Technology, Pasadena, CA, USA}
	\affiliation{AWS Center for Quantum Computing, Pasadena, CA}
        \author{András Gilyén}
        \affiliation{Alfréd Rényi Institute of Mathematics, HUN-REN, Budapest, Hungary}
 
	\begin{abstract}
Preparing ground states and thermal states is essential for simulating quantum systems on quantum computers. Despite the hope for practical quantum advantage in quantum simulation, popular state preparation approaches have been challenged. Monte Carlo-style quantum Gibbs samplers have emerged as an alternative, but prior proposals have been unsatisfactory due to technical obstacles rooted in energy-time uncertainty. 
We introduce simple continuous-time quantum Gibbs samplers that overcome these obstacles by efficiently simulating Nature-inspired quantum master equations (\Lword{}s).
In addition, we construct the first provably accurate and efficient algorithm for preparing certain purified Gibbs states (called thermal field double states in high-energy physics) of rapidly thermalizing systems; this algorithm also benefits from a quantum walk speedup.
Our algorithms' costs have a provable dependence on temperature, accuracy, and the mixing time (or spectral gap) of the relevant \Lword{}. We complete the first rigorous proof of finite-time thermalization for physically derived \Lword{}s by developing a general analytic framework for nonasymptotic secular approximation and approximate detailed balance.
Given the success of classical Markov chain Monte Carlo (MCMC) algorithms and the ubiquity of thermodynamics, we anticipate that quantum Gibbs sampling will become indispensable in quantum computing.
\end{abstract}
\maketitle

\tableofcontents\newpage

\section{Introduction}
How do we prepare quantum Gibbs states or ground states on a quantum computer? This initial state preparation problem appears as the obstacle for simulating quantum systems~\cite{feynman1982SimQPhysWithComputers,lloyd1996UnivQSim}-- a popular candidate for practical quantum advantage. This mystery has its roots in the seemingly contradictory teachings of computer science and physics: computational complexity theory tells us that few-body Hamiltonian ground states are generally QMA-hard to prepare~\cite{kitaev2002classical,aharonov2009power,gottesman2009quantum} and thus are likely intractable in general even for quantum computers; on the contrary, thermodynamics asserts that physical systems interacting with a thermal bath are naturally in the thermal states or ground states. How do we draw an appropriate boundary between the two cases?

Practically, recent end-to-end industry resource estimates (e.g.,~\cite{babbush2018low,Chamberland2020BuildingAF,THC_google,2021_Microsoft_catalysis}) of quantum simulation rely on initial state preparation assumptions\footnote{More precisely, they assume the existence of \textit{trial states} with good overlap with the ground state so that running phase estimation provably works~\cite{lin2022heisenberg}. See also~\cite{gharibian2022dequantizing}. }, exposing our ignorance of the complexity of practically relevant states. Often, practitioners turn to heuristic algorithms such as the Variational Quantum Eigensolver (see, e.g.,~\cite{Tilly2021TheVQ}) or the adiabatic algorithm (see, e.g.,~\cite{farhi2000QCompAdiabatic,Albash2016AdiabaticQC}), yet each with concerns for practicality. The former suffers from the so-called Barren Plateau phenomena~\cite{mcClean2018BarrenPlateausInQNN}, and its scalability has been debated; the latter requires a gapped adiabatic path, which appears nontrivial in recent large-scale numerical studies for quantum chemistry applications \cite{lee2022there}. So far, there is a thin consensus on a `go-to' ground state or thermal state quantum algorithm that could work in practice.

This work approaches the state preparation problem via \textit{Quantum Gibbs samplers}. In physics language, this is closely related to open system dynamics where the system of interest is coupled to a thermal bath (see, e.g.,~\cite{Rivas_2012_open_systems}). Here, the conceptual boundary is blurred between the underlying physical process and the algorithm~\cite{terhal2000problem}. If a system thermalizes in nature and our physical model is accurate, we expect the associated quantum Gibbs sampler to converge quickly (i.e., the mixing time or the inverse-spectral gap is small); conversely, proving the latter also gives a rigorous formulation of open-system thermodynamics. This complements the mainstream formulation of closed-system thermodynamics via the Eigenstate Thermalization Hypothesis (see, e.g.,~\cite{ETH_review_2016}), where theoretical progress has been elusive. Practically, our general analysis for open system thermalization could be relevant to analog quantum simulators for Gibbs sampling, although our presentation mainly focuses on fault-tolerant quantum computers. 

In computer science language, quantum Gibbs samplers are the quantum analogs of classical Markov chain Monte Carlo (MCMC) algorithms, most notably Metropolis sampling (see, e.g.,~\cite{Markovchain_mixing}). They proved to be an indispensable pillar in classical computer science, both theoretically and practically, for computational physics and, more recently, optimization problems and machine learning. In a nutshell, the simple yet general idea is a (discrete or continuous) Markov chain whose unique fixed point yields the target distribution $\pi$ (a vector with positive entries); given the energy $E_s$ as a function of the configuration $s$, the Markov chain's transition matrix $\vM$ satisfies
\begin{align}
\vM\pi = \pi \quad \text{where} \quad \pi_s := \frac{\e^{-\beta E_s}}{\sum_{s} \e^{-\beta E_s}}\quad \text{for each}\quad s \quad \text{at temperature}\quad \frac{1}{\beta}. 
\end{align}
The algorithmic cost for preparing a sample from the Gibbs distribution $\pi$ scales directly with the \textit{mixing time}, the number of iterations such that any initial conjugation converges to the stationary distribution $\pi$. The mixing time can be unpredictable and vary wildly depending on the specific problems (e.g.,~\cite{Markovchain_mixing}). Theoretically, rapid mixing can sometimes be proven under suitable assumptions, most notably in lattice Ising models assuming exponential decay of Gibbs state correlation (see, e.g.,~\cite{martinelli1999lectures}). Practically, even when mixing time estimates are elusive, MCMC algorithms often serve as a starting point for more sophisticated algorithms. Given the triumphant impact of classical Gibbs sampling, we argue that Quantum Gibbs samplers have been thus far underexplored in the community and will likely play a central role when more robust quantum computers become available. Indeed, in addition to quantum simulation, quantum Gibbs sampling has been identified as a key subroutine in solving semidefinite programs (SDPs)~\cite{brandao2016QSDPSpeedup,apeldoorn2017QSDPSolvers} and quantum machine learning~\cite{amin2016QBoltzMachine}. To clarify, we will focus on quantum Hamiltonians; quantum algorithms for classical Gibbs states are not in the scope of this work.\footnote{Quantum Gibbs sampler for fast-forwardable Hamiltonian (including commuting Hamiltonians) is already well-defined since one can effectively apply phase estimation to exponential accuracy~\cite{temme2009QuantumMetropolis,wocjan2021szegedy}. The challenges we confront in this work are rooted in the noncommutativity.}

To set the stage for quantum Gibbs sampling, it is instructive to review the classical cousins, which we consider the seminal Metropolis-Hastings algorithm (see, e.g.,~\cite{Markovchain_mixing}) as a representative. This algorithm iterates a \textit{discrete-time Markov chain} as follows: apply a random ``jump'' (or ``update,'' ``move'') $\vA^a$ with probability $p(a)$ (Figure~\ref{fig:metropolis}). If the energy decreases, \textit{accept} the move, otherwise accept only with probability $\e^{-\beta \omega}$ (i.e., rejecting the move with probability $1-\e^{-\beta \omega}$), where $\omega$ being the energy gain. This can be described as a stochastic matrix over pairs of configurations $s's$
\begin{align}\label{eq:discrete_Markov}
    \vM_{s's}:= \underset{\text{``Accept''}}{\underbrace{\sum_{a\in A} p(a)\gamma(E_{s'}-E_{s}) \vA^a_{s's}}} + \underset{\text{``Reject''}}{\underbrace{\vR_{s's}\delta_{s's}}} \quad \text{where} \quad \gamma(\omega):=\min(1, \e^{-\beta \omega}),
\end{align}
and $\vA_{s's}^a$ are stochastic matrices corresponding to each move (e.g., flipping one of the spins). The matrix elements are weighted by the Metropolis factor $\gamma(\omega)$ depending on the energy change. Importantly, the particular function satisfies a particular symmetry (Figure~\ref{fig:metropolis}), known as the \textit{detailed balance} condition 
\begin{align}
    \gamma(\omega)/\gamma(-\omega)=\e^{-\beta\omega}\quad \text{such that}\quad \vM_{s's} \pi_{s} = \pi_{s'} \vM_{s's} \quad\text{for each}\quad s,s' .\label{eq:classical_DB}
\end{align}
Detailed balance ensures that the Gibbs state $\pi$ is a fixed point of the Markov chain $\vM\pi=\pi$. The rejection part $\vR_{s's}$ is a diagonal matrix determined by the probability preserving constraints. Similarly, one may define a \textit{continuous-time Markov chain generator}
\begin{align}
    \vL_{s's} := \sum_{a\in A} p(a)\L(\underset{\text{``transition''}}{\underbrace{\gamma(E_{s'}-E_{s}) \vA^a_{s's}}} - \underset{\text{``decay''}}{\underbrace{\delta_{s's}\sum_{s''} \gamma(E_{s''}-E_{s}) \vA^a_{s''s}}} \R) \quad\text{for each}\quad s,s' .\label{eq:continuous_Markov}
\end{align}
The second term ensures that the generated semi-group $\e^{\vL t}$ preserves probability. The operators $\vA^a$ can be arbitrary nonnegative matrices and need not be stochastic. 

\begin{figure}[t]
	\centering
\includegraphics[width=0.8\textwidth]{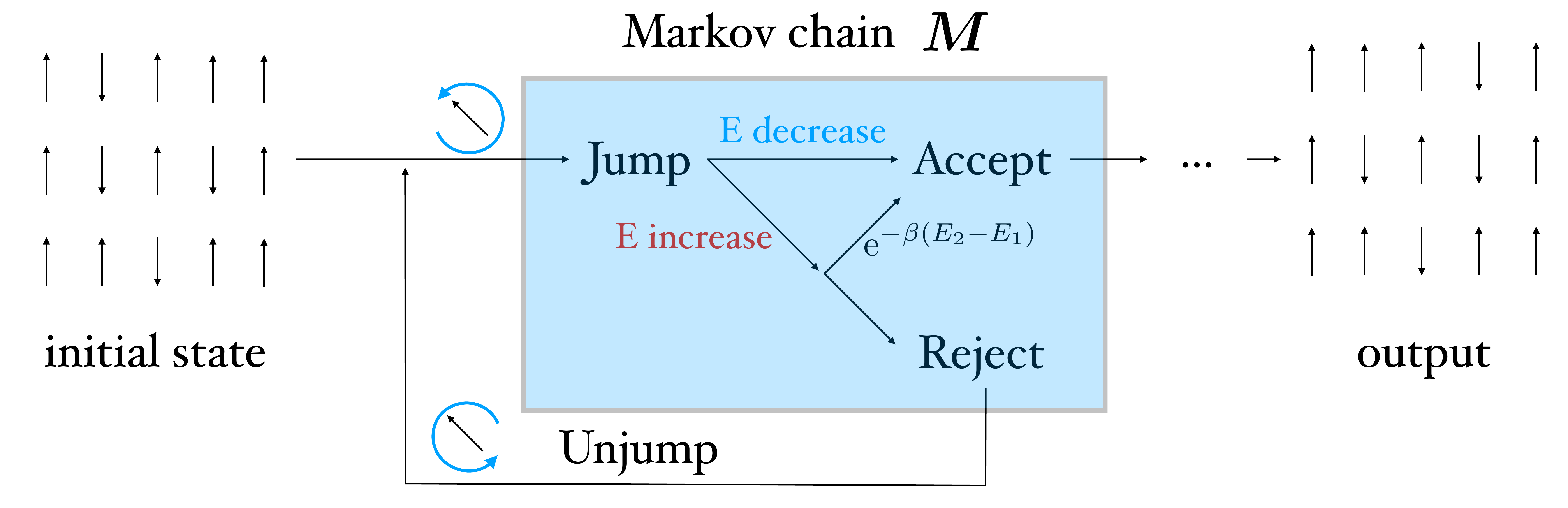}
	\caption{ The Metropolis-Hastings algorithm iterates a Markov chain to sample from the Gibbs distribution. Each step begins with a (random) jump: if the energy decreases, accept; if the energy increases, accept with a carefully chosen probability. Otherwise, reject the move. Remarkably, detailed balance can be enforced in a lazy manner via rejection sampling without storing the whole matrix.\label{fig:metropolis}
	}
\end{figure}
While classical Markov chain Monte Carlo methods have been theoretically and practically mature, the quantum analogs are still in their infancy. The study of Quantum Gibbs sampling currently faces fundamental challenges; surprisingly, even a satisfactory map has not been appropriately defined for general noncommutative Hamiltonians. As the very first step, we need to algorithmically design a quantum analog of Markov chain generator\footnote{We focus on the infinitesimal generators for simplicity. One may alternatively consider discrete quantum channels, also known as completely-positive-trace-preserving (CPTP) maps.} $\CL_{\beta}$, a \emph{\Lword{}},~\cite{wolf2012quantum} whose (unique) fixed point is the \textit{quantum Gibbs state} $\vrho_{\beta}$. More precisely, given a Hamiltonian $\vH$ and an inverse temperature $\beta$ 
\begin{align}
\text{design}\quad \CL_{\beta} \quad\text{such that}\quad  \e^{\CL_{\beta}t}[\vrho_\beta] = \vrho_{\beta}\quad\text{where}\quad \vrho_\beta := \e^{-\beta \vH}/\tr(\e^{-\beta \vH}) \label{eq:main_Gibbs_fixed_point}  
\end{align}
for any $t>0$. Subsequently, we may begin studying the properties of the proposed \Lword{}, especially the mixing time. This work aims to lay the foundation for the first challenge. The second challenge was partially addressed in Ref. \cite{ETH_thermalization_Chen21} using more primitive Gibbs samplers. Unlike the classical case, the construction of quantum Gibbs samplers is nontrivial due to imprecise energy estimates for noncommuting Hamiltonians (i.e., the energy-time uncertainty principle); the fixed point would not be exactly the Gibbs state~\eqref{eq:main_Gibbs_fixed_point}. Previous attempts \cite{terhal2000problem,temme2009QuantumMetropolis,yung2010QuantumQuantumMetropolis,ETH_thermalization_Chen21,wocjan2021szegedy} have their shortcomings, which we discuss in more detail in~\autoref{sec:existing_work} and Table~\ref{table:thermal_algorithms}. Our work, in parallel with the recent paper~\cite{Rall_thermal_22}, provides the first implementable \Lword{} for Gibbs sampling, with provable guarantees and without unrealistic assumptions. To do so, we introduce a robust analytic framework, which additionally
applies to physical \Lword{}s derived in open systems and to \textit{coherent} Gibbs samplers with Szegedy-type speedups.

Our particular construction draws inspiration from thermalization in nature. As the starting point, a system in thermal contact with a bath can be effectively described by the 
so-called \emph{Davies generator} in the Schrödinger Picture in a specific (weak-coupling/infinite-time) limit (\cite{davies74,davies76,Rivas_2012_open_systems}, and see~\cite{Mozgunov2020completelypositive} for a modern discussion)
\begin{equation}\label{eq:davies}
    \mathcal{L}_{\rm Davies}(\vrho) = \sum_{a\in A}\sum_{\nu\in B} \gamma(\nu)\left(\underset{\text{``transition''}}{\underbrace{\vA^a_\nu \vrho (\vA^{a}_\nu})^{\dagger}}-\underset{\text{``decay''}}{\underbrace{\frac{1}{2}((\vA^{a}_\nu)^{\dag}\vA^{a}_\nu\vrho +\vrho (\vA^{a}_\nu)^{\dag} \vA^{a}_\nu)}}\right),
\end{equation}
where $\{\vA^a\}_{a\in A}$ are the set of ``quantum'' jumps and $\nu \in B:=\mathrm{spec}(\vH)-\mathrm{spec}(\vH)$ are the \textit{Bohr frequencies}, the set of energy differences of the Hamiltonian. This resembles its classical Markov chain cousin~\eqref{eq:continuous_Markov}, also featuring two terms: the transition rate and the decay rate. Since we work with density operators $\vrho$ instead of probability vectors, the input $\vrho$ must be formally sandwiched by operators on the left and right. However, if the input states are diagonal in the energy basis and the energy levels are nondegenerate, then the Davies' generator can be faithfully represented as a continuous-time Markov chain~\eqref{eq:continuous_Markov} on the energy eigenstates by \textit{literally} replacing
\begin{align}
s &\rightarrow \ketbra{\psi_i}{\psi_i},\\
\vA^a_{s's} &\rightarrow \labs{\bra{\psi_i} \vA^a \ket{\psi_j}}^2\quad \text{where}\quad \vH =\sum_i E_i\ketbra{\psi_i}{\psi_i} \label{eq:s_to_psi}.
\end{align}

\begin{figure}[t]
	\centering
	\includegraphics[width=0.9\textwidth]{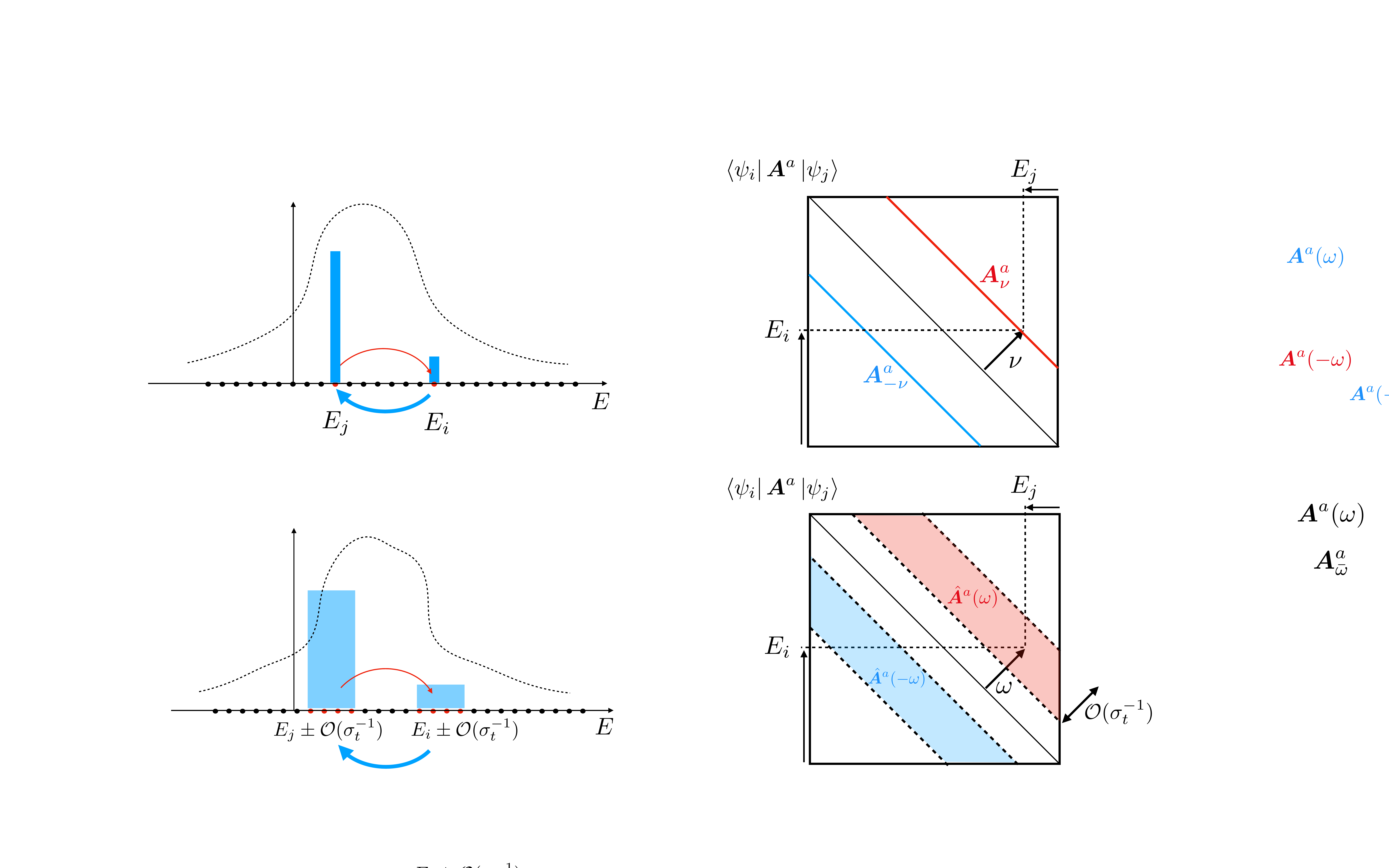}
	\caption{  
		(Up) Davies' generator gives a continuous-time Markov generator on the energy spectrum (assuming the Hamiltonian is nondegenerate and that the input state is diagonal in the energy basis.). The transitions are weighted by $\gamma(\omega)$: the heating transitions (red) are suppressed by a Boltzmann factor relative to the cooling transitions (blue), entailing detailed balance. The operator $\vA_{\nu}^a$ contains the transitions with energy difference $\nu$, which requires an infinite-time Fourier Transform. (Down) Our \Lword{} Gibbs sampler can be considered a ``semi-classical'' random walk where nearby Bohr-frequencies $\omega\pm \CO(\sigma_t^{-1})$ cannot be distinguished. The operator Fourier Transform $\hat{\vA}^a(\omega)$ contains a band of transitions. This breaks the detailed balance condition; the fixed point deviates from the Gibbs state.
	}\label{fig:semi_classical}
	
\end{figure}

A concrete physical example is a geometrically local Hamiltonian on a lattice. The jump operators $\vA^a$ can be one-body Pauli operators on each lattice site, in which case the cardinality of jumps scales with the system size $\labs{A} \propto n$. Of course, the Davies' generator is merely an instance of many possible \Lword{} one can write down (see~\autoref{sec:DB}), which generally may not come from thermodynamics.

In \eqref{eq:davies}, the quantum mechanical transition rate $\vA^a_{\nu}$ is defined as 
\begin{align}
    \vA^a_\nu := \sum_{E_i-E_j = \nu} \vP_{E_i}\vA^a\vP_{E_j} \quad \text{for each Bohr frequency} \quad \nu \in B,
\end{align}
where $\vP_{E}$ denotes energy eigenspace projectors associated with energy $E$ (\autoref{fig:semi_classical}). In general, the dynamics can be inherently quantum-mechanical when the energy subspaces are degenerate; the quantum transition can include \textit{coherent} rotations within the subspaces $\vP_{E}$.

The function $\gamma(\omega)$ depends on the physical model of the bath. Under physical assumptions (thermal bath and Markovianity), the heating transitions are penalized by a Boltzmann factor relative to the cooling transitions\footnote{The sign convention we use (here and also for Fourier Transforms) might differ from that of other works in the open systems literature.} $\gamma(\omega)=\e^{-\omega \beta}\gamma(-\omega)$. Mathematically, this ensures that the Davies' generator $\mathcal{L}_{\rm Davies}$ satisfies the \textit{quantum detailed balance condition} (\autoref{sec:DB}) with respect to the Gibbs state $\vrho_\beta$, implying that $\vrho_\beta$ is a fixed point~\eqref{eq:main_Gibbs_fixed_point}. At first glance, the above properties of Davies generators~\eqref{eq:davies} seem to qualify for a Quantum Gibbs sampler~\eqref{eq:main_Gibbs_fixed_point} - if we were able to simulate it efficiently. 

Unfortunately, the generator \eqref{eq:davies} is generally nonphysical because isolating an exact energy transition $\nu$ requires an \emph{infinite-time} Fourier Transform\footnote{With the exception when the Hamiltonian spectrum takes well-separated discrete values with (roughly) known locations.} over Heisenberg evolution
\begin{equation}\label{eq:Somegadef}
     \vA^a_\nu \propto \int_{-\infty}^\infty \e^{-\ri\omega t} \underset{\vA^a(t):=}{\underbrace{\e^{\ri\vH t} \vA^a \e^{-\ri\vH t}}}\mathrm{d}t  \quad \text{for each}\quad  a \in A\quad \text{and}\quad\nu \in B.
\end{equation} 
This allows the Davies generator to decohere arbitrary close Bohr frequencies $\omega, \omega'$. This contradicts the energy-time uncertainty principle, requiring the runtime to scale inversely with the level spacing, which generally can be \textit{exponentially} small in the system size $n$.
 Unfortunately, at shorter times, the resulting fuzziness of energy resolution breaks detailed balance, which has been central to the analysis of (both classical and quantum) Markov chains. 
Many finite-time versions of the Davies generator have been proposed to capture more realistic physical settings~\cite{Redfield1965, Christian_2013_Coarse_graining, Mozgunov2020completelypositive,ETH_thermalization_Chen21}. Yet, to the best of our knowledge, no \Lword{} arising from a finite-time Fourier Transform has been proven to have a fixed point close to the Gibbs state.\footnote{Ref.~\cite{ETH_thermalization_Chen21} gives a \emph{nonCPTP} generator that does have approximately a Gibbs fixed point.}

Quantum Gibbs sampling algorithms face analogous technical challenges due to a limited algorithmic runtime. Here, the energy-time uncertainty principle incarnates as the statistical uncertainty of the energy measurement via quantum phase estimation. Several works~\cite{rall2021faster,wocjan2021szegedy} evade this issue by imposing a convenient \emph{rounding promise} on the Hamiltonian, requiring its spectrum to be disjoint from certain specific ranges of energy. Such a rounding promise is not physically motivated and does not hold in general but allows for rigorous performance guarantees for Gibbs samplers using boosted phase estimation. Very recently, Ref.~\cite{Rall_thermal_22} circumvents the need for a rounding promise by randomly alternating the phase estimation mesh, but this comes at a high additional algorithmic cost with potentially worsened mixing time due to forbidding certain transitions and seems distant from the physical origins of quantum Gibbs samplers.

In this work, we present quantum Gibbs sampling algorithms inspired by thermalization in Nature. In particular, our construction is a ``smoothed'' version of the Davies' generator~\eqref{eq:davies}. The first algorithm simulates a \Lword{} whose fixed point is approximately a quantum Gibbs state; the second algorithm further ``quantizes'' the \Lword{} to prepare the purified Gibbs state and features a quadratic Szegedy-type speedup. Our algorithms are efficient and comparatively simple to implement while having a provable performance guarantee. The key ingredient in our algorithmic design is to use the weighted \textit{operator Fourier Transform} for the Lindblad operators
\begin{equation}\label{eq:Abomegadef}
\hat{\vA}^a(\omega) = \frac{1}{\sqrt{2\pi}}\int_{-\infty}^{\infty} \e^{-\ri \omega t}f(t)\vA^a(t) \mathrm{d}t \quad \text{for each} \quad  a \in A \quad \text{and}\quad \omega \in \BR.
\end{equation}

In practice, we use a discrete Fourier Transform (which will be denoted by $\hat{\vA}^a(\bomega)$ for discrete frequency label $\bomega$), but for conceptual simplicity, we focus on the continuous case in the introduction. Unlike in ordinary phase estimation where boosting usually adopts median-of-means tricks~\cite{nagaj2009FastAmpQMA}, we weigh the Fourier Transform by a Gaussian distribution $f(t)$ with a tunable width $\sigma_t$. Consequently, the Fourier Transform remains a Gaussian, achieving an analog of a boosted phase estimation with uncertainty $\omega \pm \CO(\sigma_t^{-1})$. 

To give performance guarantees for our construction, the main analytic insight is to define a notion of \textit{approximate detailed balance} (\autoref{sec:error_boltzmann})
\begin{align}
    \vrho_{\beta}^{1/4}\CL^{\dagger}[\vrho_{\beta}^{-1/4}\vO \vrho_{\beta}^{-1/4}]\vrho_{\beta}^{1/4} \approx \vrho_{\beta}^{-1/4}\CL[\vrho_{\beta}^{1/4}\vO \vrho_{\beta}^{1/4}]\vrho_{\beta}^{-1/4} \quad \text{for each operator}\quad \vO
\end{align}
to handle energy uncertainty. In the exact case, this is the quantum generalization of classical detailed balance, where the distribution becomes an operator, and the Markov transition matrix becomes a superoperator. 
 
Our construction and analysis of the Gibbs sampler are physically inspired insofar as it closely resembles the Davies generator of Eqn.\ (\ref{eq:davies}), but we do not know whether it can actually be derived in some physical limit from a weak system-bath coupling. Incidentally, starting from a microscopic system-bath interaction, a recent proposal~\cite{Christian_2013_Coarse_graining, Mozgunov2020completelypositive} specifically derives from first principles a \Lword{} with Lindblad operators
\begin{align} 
\hat{\vA}^a(\omega) :\propto 
\int_{-T/2}^{T/2}\e^{-\ri \omega t}\vA^a(t) \mathrm{d}t  \quad \text{for each} \quad  a \in A \quad\text{and}\quad \omega\in \BR.
\label{eq:CGME}
\end{align}
Here, the Fourier Transform time-scale $T$ sets the energy uncertainty $\omega \pm \CO(T^{-1})$; the fixed point will not be exactly the Gibbs state. Using our new analytic framework, we show that this \Lword{} derived in Refs. \cite{Christian_2013_Coarse_graining, Mozgunov2020completelypositive} have a stationary state close to the Gibbs state. To the best of our knowledge, our work completes the first general proof of many-body Gibbs states in the open system setting (assuming a reasonably short mixing time).

The runtime of both our algorithms has simple dependence on the mixing time or spectral gap of the \Lword{}. In general, the gap will depend sensitively on the details of the physical system, and its calculation for specific Hamiltonians is beyond the scope of this work. We refer to Refs. \cite{kastoryano2013quantum,kastoryano2016commuting, ETH_thermalization_Chen21, capel2021modified} and references therein for a more detailed discussion of mixing times (and spectral gaps) for various Hamiltonian and \Lword{}s. At an intuitive level, we expect the \Lword{} for lattice systems with jump operators on each site to have a constant local \Lword{} gap\footnote{The local gap, in our normalization, is defined as $n \cdot \lambda_{\rm gap}(\CL)$, where $\lambda_{\rm gap}(\CL)$ is the \Lword{} eigenvalue gap and $n$ is the system size. A parallel version of the algorithm could in principle, gain this $n$ factor in the circuit depth. Also, the \Lword{} gap should not be confused with the energy gap of the Hamiltonian, which is not directly relevant to Gibbs sampling at nonvanishing temperatures.} independent of the volume at high enough temperatures or within the same phase. In practice, we believe that quantum Gibbs sampling algorithms will be employed on a case-by-case basis in combination with various heuristics, as is the case with classical Monte Carlo algorithms.


\subsection{Existing work}\label{sec:existing_work}

\begin{table}[h!]\label{table:thermal_algorithms}
\begin{tabular}{|l|l|l|l|}
\hline
Algorithms & Ham. sim. time& assumptions & potential caveats \\ \hline
 \begin{tabular}[c]{@{}l@{}}Quantum Metropolis~\cite{temme2009QuantumMetropolis}:\\ a quantum version of rejection \\ sampling via QPE\end{tabular} & $\poly(\epsilon^{-1}, \beta, n)t_{QPE}\cdot t_{mix} $ & \begin{tabular}[c]{@{}l@{}}shift-invariant, boosted QPE\\  (provably impossible)\end{tabular} & \begin{tabular}[c]{@{}l@{}}without the assumption, \\ the QPE runtime is uncontrolled\end{tabular} \\ \hline
\begin{tabular}[c]{@{}l@{}}\cite{Rall_thermal_22}: simulate a \Lword{} \\ with forbidden energy transitions\end{tabular} & $\tilde{\CO}( \frac{1}{\gamma_{att}}\frac{\beta^3 \tilde{t}_{mix}}{\epsilon^7})$\footnote{Here $\tilde{t}_{mix}$ refers to the mixing time of modified \Lword{}s that forbid certain energies transitions (parameterized by an additional attenuation coefficient $\gamma_{att}$); it is unclear how this restricted connectivity impacts the mixing time. Also, our improved \Lword{} simulation results already improve their complexities from $\tilde{\CO}( \frac{\beta^3 \tilde{t}_{mix}}{\epsilon^7})$ to $\tilde{\CO}( \frac{\beta \tilde{t}_{mix}}{\epsilon^2})$. } & - & \begin{tabular}[c]{@{}l@{}} overhead for randomized rounding;\\ forbidding transitions impacts $\tilde{t}_{mix}$ \end{tabular} \\ \hline
\begin{tabular}[c]{@{}l@{}}\cite{Shtanko2021AlgorithmsforGibbs}: system-bath \\ evolution at weak-coupling\end{tabular} & $\CO(\frac{\beta t_{mix}^3}{\epsilon})$ & 
\begin{tabular}[c]{@{}l@{}}controllable, refreshable bath \\ and ETH\footnote{Ref.~\cite{Shtanko2021AlgorithmsforGibbs} is similar to~\cite{ETH_thermalization_Chen21}, but its proof of convergence assumes the Eigenstate Thermalization Hypothesis (ETH) and a maximally mixed initial state.~\cite{ETH_thermalization_Chen21} has guarantees assuming only the mixing time, and the ETH is one way of bounding the mixing time.}
\end{tabular}
& 
\begin{tabular}[c]{@{}l@{}}large bath; \\ no guarantees without ETH. \end{tabular}
 \\ \hline
 \begin{tabular}[c]{@{}l@{}}\cite{ETH_thermalization_Chen21}: system-bath \\ evolution at weak-coupling\end{tabular} & $\poly(\epsilon^{-1}, \beta, n, t_{mix})$ & controllable, refreshable bath & large bath and large overheads \\ \hline
{\color{blue}\begin{tabular}[c]{@{}l@{}}\autoref{thm:L_correctness}: simulate a \\ \Lword{} via operator FT\end{tabular}} & {\color{blue} $\tCO( \frac{\beta t_{mix}}{\epsilon}\cdot t_{mix})$} & - & -\\ \hline
 \begin{tabular}[c]{@{}l@{}}Quantum${}^2$ Metropolis~\cite{yung2010QuantumQuantumMetropolis}:\\ a coherent version of Quantum \\ Metropolis via QPE\end{tabular} & $\tCO( t_{QPE} \cdot \frac{\beta^2 \braket{\vH^2}_{\beta =0}}{\epsilon\sqrt{\lambda_{gap}} })$\footnote{Ref.~\cite{wocjan2021szegedy} did not include the algorithmic cost of quantum simulated annealing, so we fill in using our modernized version (\autoref{sec:simulated_annealing}). Likewise, the simulated annealing cost of~\cite{yung2010QuantumQuantumMetropolis} could also be improved (still assuming perfect QPE).\label{foot:improvedQSA}
 } & perfect QPE & \begin{tabular}[c]{@{}l@{}}  needs nondegenerate energies,\footnote{The proof of correctness assumes each eigenstate $\vH$ can be perfectly distinguished.}\\ $\e^{\Omega(n)}$ runtime for QPE \end{tabular} \\ \hline
\begin{tabular}[c]{@{}l@{}}\cite{wocjan2021szegedy}:a coherent version of \\ \Lword{} via QPE\end{tabular} & $\tCO( t_{QPE}\cdot \frac{\beta \norm{\vH}}{\sqrt{\lambda_{gap}}})$\textsuperscript{\ref{foot:improvedQSA}} & rounding promise for $\vH$ & \begin{tabular}[c]{@{}l@{}}unknown error for \\ unrounded Hamiltonians\end{tabular} \\ \hline
{\color{blue}\begin{tabular}[c]{@{}l@{}}\autoref{thm:D_correct}: a coherent version \\ of \Lword{} via operator FT\end{tabular}} & {\color{blue}$\tilde{O}(\frac{\beta^2\norm{\vH}}{\lambda_{gap}^{3/2}}+\frac{\beta}{\epsilon\lambda_{gap}^{3/2}})$} & - & -\\ \hline
 & Gate complexity & & \\ \hline
 \begin{tabular}[c]{@{}l@{}}\cite{brandao2019finite_prepare}:\\ patching via recovery maps\end{tabular} & $\e^{ \CO(C^d\ln^d(n/\epsilon))}$\footnote{Assuming the Markov property and clustering quantities both decay exponentially.} & \begin{tabular}[c]{@{}l@{}}Markov and clustering\\ with length scale $C$ \end{tabular} & \begin{tabular}[c]{@{}l@{}}nonconstructive,\\ quasi-local recovery maps \end{tabular}\\ \hline
\begin{tabular}[c]{@{}l@{}}QITE~\cite{Motta_2019_QITE}: imaginary time \\evolution via local tomography\end{tabular} & $\e^{ \tCO( C^d \ln^d(\beta n/\epsilon) ) }$ & correlation length $C$ & \begin{tabular}[c]{@{}l@{}}costly at low-temperature, \\ strong correlation, or high accuracy\end{tabular} \\ \hline
\begin{tabular}[c]{@{}l@{}}\cite{Moussa2019LowDepthQM}: drive transition \\by measurements\end{tabular} & $\tCO( t_{mix} \beta ) $ & certain measurement operator & \begin{tabular}[c]{@{}l@{}} the measurement must have \\ ``good overlap'' with energy basis\footnote{Their exactly detailed-balanced quantum channel seems to qualitatively differ from other quantum MCMC algorithms. It assumes certain efficiently implementable basis measurements that have a \textit{good overlap} with the energy basis. For example if measured in the computational basis, the 1D transverse field Ising model seems to exhibit an \textit{exponential} mixing time at constant temperature, see \cite[Page 5]{Moussa2019LowDepthQM}.}\end{tabular} \\ \hline
\cite{Holmes_2022_thermal}:a perturbative approach & $ \e^{\CO(\beta \norm{\vV})}$ & \begin{tabular}[c]{@{}l@{}}$\vH = \vH_0 + \vV $ \\ for small $\vV$ and simple $\vH_0$\end{tabular} & \begin{tabular}[c]{@{}l@{}}costly at low-temperature \\ or nonperturbative regime\end{tabular} \\ 
\hline
\begin{tabular}[c]{@{}l@{}}QSVT~\cite{gilyen2018QSingValTransf}:\\ directly implementing $\e^{-\beta \vH}$\end{tabular} & $\e^{ \CO( \beta\norm{\vH})}$ & - & not scalable \\ \hline
\end{tabular}
\caption{A comparison of existing thermal state preparation algorithms. We focus on methods with provable guarantees for an $\epsilon$-approximation of the Gibbs state (in trace distance) and list their cost, assumptions, and caveats. We use $\poly(\cdot)$ to denote polynomials and $\tCO(\cdot)$ to absorb logarithmic dependences. 
The first few algorithms are Monte Carlo-style methods, incoherent or coherent; we represent their costs by the total black-box Hamiltonian simulation time. The incoherent ones are based on semi-groups, with complexity being the cost of emulating the semi-group (which is basically dominated by the phase estimation time $t_{QPE}$) multiplied by the mixing time $t_{mix}$. The coherent version instead implements block-encoding for discriminants and prepares the purified Gibbs state via quantum simulated annealing. The number of discriminant calls is $\tilde{\CO}(\beta\norm{\vH} \sqrt{\lambda^{-1}_{gap}})$ (as the counterpart for the mixing time $t_{mix}$) where $\lambda_{\rm gap}$ refers to the minimum gap of discriminants along the adiabatic path. We calculate the costs for~\autoref{thm:L_correctness} and~\autoref{thm:D_correct}, assuming the algorithmic parameters and the mixing time $t_{mix}$ or spectral gap $\lambda_{gap}$ satisfy certain self-consistency constraint. The mixing time and spectral gap can be (loosely) converted to each other, assuming approximate detailed balance. 
Of course, as in all classical MCMC methods, the mixing time $t_{mix}$ or spectral gap $\lambda_{gap}$ can be exponentially small depending on the temperature and the system; indeed, there are systems whose thermal states are expected to have high complexity, such as certain spin glass models, and we do not expect efficient quantum algorithms to exist in general. Optimistically, we often care about the Gibbs state that appears in nature, which suggests the mixing time remains reasonably small.
 }

\end{table}
The first attempt at designing a quantum algorithm for Gibbs sampling with per-update efficiency guarantees is the quantum Metropolis algorithm~\cite{temme2009QuantumMetropolis}, as a quantum analog of~\eqref{eq:discrete_Markov}. The main guiding principle is to ``do Metropolis sampling over the energy spectrum'' in the spirit of~\eqref{eq:s_to_psi}. To do so, two quantitative changes are needed due to quantum mechanics, one algorithmic and one analytic: algorithmically, one needs a subroutine to ``reject'' a quantum state back to the same energy. Classically, one can clone the configuration $s$ before the update. Quantumly, however, we cannot clone the (unknown) quantum state and must be careful not to collapse the quantum state due to energy measurement. Ref. ~\cite{temme2009QuantumMetropolis} handles this issue using the Mariott-Watrous~\cite{marriott2005QAMGames} algorithm (or the ``rewinding'' technique in quantum Cryptography). Crucially, this algorithmic subroutine preserves probability, drawing a distinction from imaginary time evolution or post-selection, but this comes with its limitation and significantly complicates the algorithm. Second, energy measurements based on quantum phase estimation have a finite resolution inversely proportional to the runtime $\delta E \sim 1/T$. Consequently, the detailed balance condition may not generally hold, and one needs to prove that the fixed point remains approximately the Gibbs state. 

The technical results of~\cite{temme2009QuantumMetropolis} contain three approaches based on different phase estimation subroutines: (i) assuming perfect phase estimation (with performance guarantee but with an exponential Hamiltonian simulation time); (ii) un-boosted phase estimation (without performance guarantees); (iii) boosted shift-invariant phase estimation (with performance guarantees).\footnote{The approximate detailed balance argument was later completed in~\cite{ETH_thermalization_Chen21}.} Unfortunately, we recently realized that such a boosted, shift-invariant phase estimation (see~\autoref{sec:impossible}) is provably impossible;\footnote{The authors~\cite{temme2009QuantumMetropolis} communicated with us that there might be ways to fix their algorithm.} we do not know whether Quantum Metropolis sampling~\cite{temme2009QuantumMetropolis}, as stated explicitly, actually works in practice. The quantum metropolis is regarded as an important theoretical milestone, but due to its complicated form (especially due to the rejection subroutine), the particular algorithm largely serves as a high-level inspiration. 

From a physical point of view, one may implement Nature's quantum algorithm~\cite{terhal2000problem} by emulating the global system-bath interaction. However, this black-box approach is a double-edged sword: indeed, this method should work as well as Nature, but as we know from open system thermodynamics, nonasymptotic results are extremely challenging without liberal use of approximations~\cite{terhal2000problem,Rivas_2012_open_systems}, rendering the result qualitative but not quantitative. (For example, it is elusive how big of a bath is needed for the desired accuracy.) Recently, Ref.~\cite{Mozgunov2020completelypositive,ETH_thermalization_Chen21,Shtanko2021AlgorithmsforGibbs} took the physics inspirations seriously and quantitatively studied a system-bath interaction from scratch. 
Ref.~\cite{Mozgunov2020completelypositive} revisits the text-book open system derivation and extracts a \textit{nonasymptotic} version of Davies' generator with explicit error bounds and without unphysical limits. However, it was not known whether the derived \Lword{} has the Gibbs state as the stationary state; ref.~\cite{ETH_thermalization_Chen21} was the first provably polynomial-time algorithm for Gibbs state assuming a reasonably short mixing time, although it assumes good control of the bath and its error bounds are large polynomials and impractical to apply; ref.~\cite{Shtanko2021AlgorithmsforGibbs} is conceptually similar to~\cite{ETH_thermalization_Chen21} but focuses more on near-term feasibility. Technically, its accuracy guarantees require the Eigenstate Thermalization Hypothesis, which significantly simplifies the analysis. Unfortunately, both cases~\cite{ETH_thermalization_Chen21,Shtanko2021AlgorithmsforGibbs} failed to extract a \Lword{} (the generators are not completely positive). In some sense, this motivates us to give a unifying conceptual and analytic perspective on this subject.

Coherent quantum Metropolis sampling~\cite{yung2010QuantumQuantumMetropolis} is a natural generalization of quantum Metropolis sampling~\cite{temme2009QuantumMetropolis} that further gives a quadratic runtime speedup by invoking Szegedy's quantum walk strategy~\cite{szegedy2004QMarkovChainSearch}. Since the dissipative map is quantized as a Hermitian operator, one cannot evolve a semi-group but requires an additional \textit{quantum simulated annealing} step~\cite{somma2007quantum} (a particular adiabatic state preparation along temperatures) to prepare the \textit{purified Gibbs state; see~\autoref{sec:simulated_annealing}}. Unfortunately, Ref.~\cite{yung2010QuantumQuantumMetropolis} assumes perfect phase estimation and a nongenerate Hamiltonian spectrum, and it was unclear how one incorporates imperfect phase estimation in such a coherent algorithm. Ref.~\cite{wocjan2021szegedy} improves and generalizes their result but still makes an unphysical \emph{rounding promise} assumption: the Hamiltonian spectrum has periodic gaps so that phase estimation can be amplified. 

Perhaps inspired by the rounding promise, the recent related work~\cite{Rall_thermal_22} proposes an algorithm that implements a \Lword{} that also provably has approximately Gibbs fixed point using randomized rounding. Their approach is quite different in nature from ours and is not known to enjoy the quadratic speedup. Randomized rounding seems to incur large resource overhead (See Table~\ref{table:thermal_algorithms}) and substantially departs from the physical origin of these ideas.

Our \Lword{} Gibbs samplers build upon the literature for open system simulation~\cite{cleve2016EffLindbladianSim} as well as the coherent Gibbs sampler of~\cite{wocjan2021szegedy}. In both cases, we remove the need for any unphysical assumption (e.g., especially the rounding promise) yet still maintain a simple error bound. This is made possible by identifying the right choice of jump operators in terms of discrete Fourier Transform and refining the analytic technical tool introduced in~\cite{ETH_thermalization_Chen21} (nonasymptotic bounds for secular approximation)\footnote{Some preliminary version of approximate detailed balance was discussed in an earlier version of~\cite{ETH_thermalization_Chen21} regarding quantum Metropolis sampling. That part was completely removed after the authors realized the phase estimation assumption was impossible as stated (\autoref{sec:impossible}).}. In a nutshell, it seems the ``right'' approach to quantum Gibbs sampling is to simulate a (continuous-time) \Lword{}, which Nature implements by default, instead of a (discrete-time) Metropolis-Hastings style quantum channel. The rejection step is handled automatically for any \Lword{}s. We leave for future work to simplify the rejection step in quantum Metropolis sampling~\cite{temme2009QuantumMetropolis} or to design a discrete-time channel with provable guarantees.\footnote{For \Lword{}s, the designer has the freedom to choose \textit{arbitrary} Lindblad operators, and the decay part automatically guarantees trace-preserving. However, it is more challenging to design quantum channels as the trace-preserving condition seems less flexible. Indeed, Quantum Metropolis Sampling~\cite{temme2009QuantumMetropolis} had to invoke Mariot-Wattrous-style rewinding~\cite{marriott2005QAMGames} multiple times to ensure probability is preserved, which unfortunately increases the complexity of the algorithm.} We further note a general distinction between coherent Gibbs samplers and \Lword{} Gibbs samplers regarding obtaining the fixed point. The former relies on the gap of the \Lword{} staying open along the entire adiabatic path from high to low temperature $0\rightarrow \beta$, while the latter (e.g., our construction or~\cite{Rall_thermal_22}) does not. Thus, the two algorithmic costs are not directly comparable. Empirical intuition from classical Gibbs sampling suggests that for ``simple problems'' where the gaps remain largely open throughout the phase of interest, adiabatic and direct sampling methods perform similarly. However, for strongly frustrated systems like spin glasses, adiabatic heuristics are the go-to Monte Carlo method. It is tempting to speculate that the same will be true in the quantum Gibbs sampling case, which lends itself well to our approach.

In addition to Monte Carlo style algorithms, other thermal state preparation algorithms based on quite different principles also exist; we briefly summarize their gate complexities in Table~\ref{table:thermal_algorithms}. We only discuss methods where quantitative arguments are possible and pay less attention to heuristic approaches such as variational circuits~\cite{wu2019variational,zhu2020generation,Martyn2018ProductSA,Sewell2022ThermalME, consiglio2023variational}), energy filtering assuming good initial states, and heuristic quantum assisted Monte Carlo~\cite{Lu2020AlgorithmsFQ,Schuckert2022ProbingFO}.

\subsection{Outline and main results}
Our discussion features the following \Lword{} in the Schrödinger picture
\begin{equation}
    \mathcal{L}_{\beta}[\vrho] := \sum_{a\in A} \int_{-\infty}^{\infty}\gamma(\omega)\left(\hat{\vA}^a(\omega)\vrho \hat{\vA}^a(\omega)^{\dagger}-\frac{1}{2}\L\{\hat{\vA}^a(\omega)^{\dagger}\hat{\vA}^a(\omega),\vrho \R\}\right) \mathrm{d}\omega\label{eq:mainL}
\end{equation}
with the anti-commutator $\{ \vA,\vB\} = \vA\vB+\vB\vA$. We can read out the set of \textit{Lindblad operators} 
\begin{align}
    \{\sqrt{\gamma(\omega)}\hat{\vA}^a(\omega)\}_{a\in A,\omega \in \BR}\quad \text{where}\quad \hat{\vA}^a(\omega) := \frac{1}{\sqrt{2\pi}}\int_{-\infty}^{\infty} \e^{-\ri \omega t} f(t) \vA^a(t) \mathrm{d}t\label{eq:ourLindbladOp}. 
\end{align}
In particular, all \Lword{}s we consider in this work, natural or algorithmic, satisfies the following symmetry and normalization conditions:
\begin{itemize}
    \item The set of \textit{jump operators} $\vA^a$, which ``drives'' the transition, can be arbitrary (and often depends on the Hamiltonian) as long as the set contains their adjoints 
\begin{align}
    \{\vA^a\colon a\in A\}=\{\vA^{a\dagg}\colon a\in A\}\quad \text{and}\quad \nrm{\sum_{a\in A}\vA^{a\dagg} \vA^a}\leq 1.\label{eq:AAdagger}
\end{align}
Indeed, classical Metropolis sampling often starts with a \textit{reversible} Markov chain to algorithmic impose the detailed balance condition; the quantum analog is the adjoint condition. Single-site Pauli operators (which are individually self-adjoint) are handy choices, but few-body operators with arbitrary connectivity are certainly permissible\footnote{In fact, the ability to perform carefully chosen (often not natural) jumps is what empowers classical Gibbs sampling algorithms, e.g., cluster updates.}. The normalization is natural for block-encoding the set of jump operators (\autoref{defn:blockLindladian} and \eqref{eq:blockJumps}). For example, choosing single-site Paulis as jump operators requires dividing them by $\sqrt{|A|}$ (where $\labs{A}$ is the cardinality of the set) to fulfill the normalization requirement.\footnote{ This is slightly different from the physical setting where each jump has operator norm $\norm{\vA^a} =1$. There, the ``strength'' of the \Lword{}~\eqref{eq:L_normalized} scales with the number of jumps $\labs{A}$.}
    \item The Fourier Transform in the time domain is weighted by a \textit{filter} function $f(t)$ that is real and $\ell_2$-normalized 
    \begin{align}
        f^*(t)= f(t) \quad \text{for each}\quad  t \in \BR \quad \text{and}\quad \normp{f}{2}^2 := \int_{-\infty}^{\infty} \labs{f(t)}^2 \rd t =1\label{eq:fnormalized}.
\end{align}
Sometimes, we drop the subscript by $\norm{f}=\norm{f}_2$. The real constraint serves similar purposes as reversibility in classical Gibbs sampling. When considering discrete Fourier Transforms (which is necessary for implementation), we adapt the corresponding (discrete) normalization $\sum_{\bt \in S_{t_0}} \labs{f(\bt)}^2 =1$.
\item The \textit{transition weight} $\gamma(\omega)$ can be any function satisfying the \textit{KMS condition} and the bound
\begin{align}
    \gamma(\omega)/\gamma(-\omega)=\e^{-\beta\omega} \quad \text{and}\quad 0\le \gamma(\omega) \le 1 \quad \text{for fixed}\quad \beta \quad \text{and each}\quad \omega \in \BR\label{eq:gamma_KMS}.
\end{align}
This coincides with the classical recipe for detailed balance~\eqref{eq:classical_DB}. Natural choices include the Metropolis weight $\gamma(\omega)=\min(1, \e^{-\beta \omega})$ or the (smoother) Glauber dynamics weight $\gamma(\omega) = (\e^{\beta \omega}+1)^{-1}$. 
\end{itemize}
To summarize, the above list of symmetry conditions is the key to ensuring (approximate) detailed balance; the above normalization choices are not only natural for implementation but also
conveniently ensures that the ``strength'' of the \Lword{} is normalized 
\begin{align}
\normp{\CL_{\beta}}{1-1} \le 2 \label{eq:L_normalized} 
\end{align}
in the superoperator 1-1 norm.
\subsubsection{Lindbladians from Nature}
Before we discuss Gibbs sampling algorithms, we first address the fundamental question: why do Gibbs states faithfully capture physical systems in thermal equilibrium? In physics, the quantum Gibbs state is often imposed without rigorous justification. As a mathematical physics result, we complete the first proof of open system thermodynamics: the Gibbs state is indeed \textit{approximately} the fixed point of \Lword{}s governing open system dynamics. 

Of course, this further requires a rigorous derivation of \Lword{} from reasonable open system assumptions; this is not the intention of this work, but thankfully, this has been worked out under a Markovian, weak-coupling assumption~\cite{Mozgunov2020completelypositive}. All we need as a black box is that it indeed satisfies the constraints we imposed~\eqref{eq:AAdagger},\eqref{eq:fnormalized}, and \eqref{eq:gamma_KMS}. For simplicity, we have omitted the Hamiltonian part of the \Lword{} and focus only on the dissipative part; see~\autoref{sec:open_system} for the complete results.

\begin{thm}[Gibbs state is thermodynamic] \label{thm:Nature_fixedpoint}
 Any $\CL_{\beta}$ satisfying the symmetry and normalization conditions~\eqref{eq:AAdagger},\eqref{eq:fnormalized}, and \eqref{eq:gamma_KMS} with the particular weight function
\begin{align}
    f(t) = \frac{\indicator(t\le \labs{T}/2)}{\sqrt{T}}
\end{align}
has an approximate Gibbs fixed point
\begin{align}
\normp{\vrho_{fix}( \CL_{\beta} ) - \vrho_{\beta}}{1} &\le \CO\L(\sqrt{\frac{\beta}{T}}t_{mix}(\CL_{\beta})\R)\label{eq:nature_fixedpt_error}.
\end{align}
In particular, (dropping the Hamiltonian term and under suitable normalization) such a \Lword{} can arise from a system (with Hamiltonian $\vH$) interacting weakly with a Markovian bath (with inverse temperature $\beta$) through jump operators $\vA^a$. 
\end{thm}
The above introduces the notion of \textit{mixing time} for Linbladians: the time scale for which any pair of initial states become indistinguishable. Of course, the physical interpretation of this time scale depends on how the Lindbladian is normalized; for our cases, we conveniently have that $\normp{\CL_{\beta}}{1-1} \le 2$~\eqref{eq:L_normalized}.

\begin{defn}[\Lword{} mixing time] For any \Lword{} $\CL$, we define the mixing time $t_{mix}(\CL)$ in the Schrödinger picture to be the shortest time for which
\begin{align}
\lnormp{\e^{\CL t_{mix}}[\vrho-\vrho']}{1} \le \frac{1}{2} \normp{\vrho-\vrho'}{1} \quad \text{for any states}\quad \vrho, \vrho'.
\end{align}
\end{defn}
\autoref{thm:Nature_fixedpoint} states that the approximation of Gibbs state degrades at a low temperature, a poor energy resolution (i.e., a short Fourier Transform time $T$), or a long mixing time. While the parameters $\beta, T$ are tunable parameters; the mixing time is generally Hamiltonian dependent. Still, one may obtain a bound using additional assumptions (such as the decay of correlation of commuting Hamiltonian Gibbs states~\cite{Bardet2021EntropyDF,kastoryano2016commuting} or the Eigenstate Thermalization Hypothesis~\cite{ETH_thermalization_Chen21, Shtanko2021AlgorithmsforGibbs}), empirical intuition, or conversion from a numerically-obtained spectral gap (\autoref{prop:gap_to_mixing}).

Assuming some grasp of the mixing time, how large should the time scale $T$ get to obtain a Gibbs sample? Roughly, according to the error bound~\eqref{eq:nature_fixedpt_error}, the time $T$ should scale with the mixing time by 
\begin{align}
T \sim \beta t_{mix}^2 / \epsilon^2. 
\end{align}

More carefully, the RHS~\eqref{eq:nature_fixedpt_error} does not obviously permit a solution for the desired accuracy $\epsilon$ (similarly for~\autoref{thm:L_correctness} and~\autoref{thm:D_correct}). Indeed, the mixing time $t_{mix}$ can generally depend on other parameters, especially the width $T$; heuristically, one may guess that the mixing time $t_{mix}$ depends mildly on the width $T$, but we leave a careful analysis for future works. An optimistic instance is when the Eigenstate Thermalization Hypothesis holds, and the mixing time at finite energy resolution can be related to the infinite resolution ($T \rightarrow \infty$) case~\cite{ETH_thermalization_Chen21,Shtanko2021AlgorithmsforGibbs}, which can be calculated. 

The main analytic challenge to prove~\autoref{thm:Nature_fixedpoint} is how to control the convergence and fixed point of \Lword{}s without exact detailed balance. Our technical contribution is to formulate an \textit{approximate detailed balance} condition using the appropriate superoperator norm 
\begin{align}
    \lnormp{\vrho_{\beta}^{1/4}\CL^{\dagger}[\vrho_{\beta}^{-1/4}\cdot \vrho_{\beta}^{-1/4}]\vrho_{\beta}^{1/4} -\vrho_{\beta}^{-1/4}\CL[\vrho_{\beta}^{1/4}\cdot \vrho_{\beta}^{1/4}]\vrho_{\beta}^{-1/4}}{2-2}\label{eq:main_ADB_22}.
\end{align}
The above two superoperators are each other's adjoints, and thus, the above quantifies the magnitude of certain \textit{anti-Hermitian} component of the \Lword{} under similarity transformation. Indeed, traditionally, the detailed balance condition is convenient as it effectively reduces the mixing time of superoperators to the spectral theory of Hermitian operators, which is conceptually and technically more transparent. Our observation is that the consequences of detailed balance, including Gibbs fixed point and spectral bounds on mixing time, are remarkably stable against perturbation. Much ink is devoted to bound~\eqref{eq:main_ADB_22}, which is yet another technical challenge. Indeed, we are \textit{inverting} the Gibbs state, which has exponentially small weights. The energy uncertainty in the Fourier Transforms $\hat{\vA}^a(\omega)$ could potentially blow up~\eqref{eq:main_ADB_22}. In response, we further introduce an intermediate \Lword{} by applying a rigorous \textit{secular approximation} (related to the \textit{rotating wave approximation}), such that transitions with large energy deviation are truncated.
\begin{align}
    \hat{\vA}^a(\omega) \approx \hat{\vS}^a(\omega) \quad \text{such that} \quad \bra{\psi_i}\hat{\vS}^a(\omega) \ket{\psi_j} = 0 \quad \textrm{whenever} \quad \labs{(E_i - E_j) - \omega} \gg 0.
\end{align}
The secular approximation interplays nicely with the operator Fourier Transform and should be widely applicable in the rigorous, nonasymptotic analysis of open-system Lindbladians.

Conceptually, there are two \textit{opposite} ways to understand~\autoref{thm:Nature_fixedpoint}: pessimistically, the Gibbs state may not be physical if the mixing time is too long, and we might have to simulate the natural \Lword{} to understand its fixed point; optimistically, if the Gibbs state is indeed physical, we might ignore its physical origin and take a short-cut to design even more efficient Gibbs sampling algorithms. The two perspectives are individually addressed in the following sections.

\subsubsection{Simulating Nature}
Taking a step back from Gibbs sampling, how do we simulate open system dynamics in nature? This boils down to the task of \Lword{} simulation, which has been studied largely restricted to the black-box setting~\cite{cleve2016EffLindbladianSim}. However, we have an \textit{explicit} \Lword{} in mind to simulate. As a result, we had to modify existing black-box input models to capture our Lindbladian~\eqref{eq:ourLindbladOp}; this also inspires us to design \textit{even more} efficient \Lword{} simulation algorithm for our access model. First, we define how we want our \Lword{} to be block-encoded.

\begin{defn}[Block-encoding of a \Lword{}]\label{defn:blockLindladian}
	Given a purely irreversible \Lword{} 
\begin{align}
     \CL[\vrho]:=\sum_{j \in J} \L(\vL_j\vrho \vL_j^{\dagger} - \frac{1}{2} \vL_j^{\dagger}\vL_j\vrho - \frac{1}{2}\vrho\vL_j^{\dagger}\vL_j\R), 
\end{align} 
we say that a matrix\footnote{For implementation purposes $\vU$ will be a unitary quantum circuit, but we also consider nonunitary block-encodings for the sake of analysis.} $\vU$ is a block-encoding of the Lindblad operators $\{\vL_j\}_{j\in J}$ if \footnote{In the first register, we could use any orthonormal basis, sticking to computational basis elements $\ket{j}$ is just for ease of presentation. Intuitively, one can think about $b$ as the number of ancilla qubits used for implementing the operators $\vL_j$, while typically $a-b\approx \log|J|$.} 
 \begin{align}
     (\bra{0^b}\otimes \vI)\cdot \vU\cdot(\ket{0^{c}} \otimes \vI)=\sum_{j\in J} \ket{j} \otimes \vL_j \quad \text{for}\quad b, c \in \mathbb{N}.
 \end{align}
\end{defn}
Indeed, because of the many Lindblad operators $j\in J$, other access models are certainly valid (e.g., given block-encoding for each Lindblad operator $\vL_j$~\cite{li2022SimMarkOpen,cleve2016EffLindbladianSim}). Nevertheless, \autoref{defn:blockLindladian} interplays nicely with our \Lword{} (especially the operator Fourier Transform) and the following efficient \Lword{} simulation algorithm. 

\begin{thm}[Linear-time \Lword{} simulation, simplified]\label{thm:LCUSim_main}
	Suppose $\vU$ is a unitary block-encoding of the \Lword{} $\CL$ as in \autoref{defn:blockLindladian}. Let $t>1$ and $\epsilon\leq 1/2$, then we can simulate the map $\e^{t \CL}$ to error $\epsilon$ in diamond norm using 
	\begin{align}	&\bigO{\L(c+\log(\frac{t}{\epsilon}) \R)\log(\frac{t}{\epsilon})}\quad&\text{ (resettable) ancilla qubits},\\ 
	&\bigOt{t} \quad &\text{(controlled) uses of $\vU$ and $\vU^\dagg$},\\
	\text{and}\quad &\bigOt{t(c+1)}\quad &\text{ other two-qubit gates}.
	\end{align}
\end{thm}
See~\autoref{thm:LCUSim} for the complete result and~\autoref{thm:weakMeasSim} for a simpler algorithm with suboptimal asymptotic scaling. Compared with the best existing results as sum-of-norm $\sum_{j \in J} \norm{\vL_{j}^{\dagger}\vL_{j}}$, we achieve a strictly better scaling with the norm-of-sum $\norm{\sum_{j \in J} \vL_{j}^{\dagger}\vL_{j}}$ when the Lindblad operators are altogether block-encoded as in~\autoref{defn:blockLindladian}; if we are only given block-encodings for each Lindblad operators, we can always convert them to our input model (\autoref{defn:blockLindladian}) and recover the existing scaling $\sum_{j \in J} \norm{\vL_{j}^{\dagger}\vL_{j}}$.

It remains to create a unitary block-encoding (\autoref{defn:blockLindladian}) for our particular Lindblad operators~\eqref{eq:ourLindbladOp}. Since our algorithms run on discrete qubits, our implementation requires discretizing the operator Fourier Transform, with a change of the notation
\begin{equation}\label{eq:DiscAMain}
 \hat{\vA}^a(\omega) \rightarrow \hat{\vA}^a(\bomega):= \sum_{ \bt \in S_{t_0} } \e^{-\ri \bomega \bt}f(\bt)\vA^a(\bt)\quad \text{for each} \quad  a \in A \quad \text{and} \quad \bomega \in S_{\omega_0}.
\end{equation}
The discretized frequency and time labels $\omega \rightarrow \bomega \in S_{\omega_0} =\{0, \pm \omega_0,\cdots \}$ and $t \rightarrow \bt \in S_{t_0} = \{0, \pm t_0,\cdots \}$ corresponds to the discrete Fourier Transform (\autoref{sec:operator_FT}) using a finite grid of size $N = \labs{S_{\omega_0}} = \labs{S_{t_0}}$, which can be stored using $\lceil\log(N)\rceil$ additional ancillas.\footnote{We require $N\omega_0\geq 4 \nrm{\vH}+2/\beta$ to store all possible energy transitions.} While the discretization parameters are needed for explicit algorithmic implementation, conceptually, they merely incur \textit{logarithmic} overhead in the runtime and ancillas (\autoref{apx:cont_limit}). We may now concretely present our algorithmic goal: efficiently construct a block-encoding in the form
\begin{align}\label{eq:opFTGens}
    (\bra{0^b}\otimes I)\vU(|0^{c}\rangle\otimes I)=\sum_{a \in A ,\bomega \in S_{\omega_0}} \sqrt{\gamma(\bomega)} \ket{\bomega,a}\otimes\hat{\vA}^a(\bomega).
\end{align}
\begin{lem}[Efficient block-encoding]\label{lem:nature_block_encoding}
    In the setting of~\autoref{thm:Nature_fixedpoint}, a unitary block-encoding $\vU$ for the (discretized) Lindblad operators~\eqref{eq:opFTGens} can be created using one query of $\vec{Prep}$,  $\vW$, $\vV_{jump}$, $\vec{QFT}$, and 
\begin{align}
&\CO(T) \quad &\text{(controlled) Ham. sim. time for $\vH$}\\
&n+1+ \lceil\log_2(\labs{A})\rceil + \lceil\log_2(N)\rceil  & \text{(resettable) qubits}.
\end{align}
\end{lem}
See~\autoref{lem:block_encoding_L} for the explicit construction. In perceivable usage, we expect the number of Fourier labels registers $N$ to scale polynomially with all other parameters (\autoref{apx:cont_limit})
\begin{align}
    N \sim \poly (n,\beta, \norm{\vH},T,\epsilon^{-1},t, \labs{A})
\end{align}
for a good approximation for the continuous Fourier Transform~\eqref{eq:ourLindbladOp}. Morally, our algorithm extracts the essential functionality of a Markovian bath (which naively may require a substantial number of qubits to implement~\cite{Shtanko2021AlgorithmsforGibbs,ETH_thermalization_Chen21}) by merely polylogarithmic resettable ancillas.

To make the simulation cost transparent, we list the main circuit components required for implementation: the controlled Hamiltonian simulation
\begin{align}
	\sum_{\bt \in S_{t_0}}\ketbra{\bt}{\bt}\otimes e^{\pm \ri \bt \vH},
\end{align}
the unitary for preparing the filter function in superposition 
\begin{align}
    \vec{Prep}: \ket{\bar{0}} \rightarrow \ket{f}:=\sum_{\bar{t}\in S_{t_0}} f(\bt)\ket{\bt},
\end{align}
the controlled rotation for transition weights
\begin{align}
    \vW := \sum_{\bomega\in S_{\omega_0}}  \begin{pmatrix} \sqrt{\gamma(\bomega)} & -\sqrt{1-\gamma(\bomega)}\\ \sqrt{1-\gamma(\bomega)} &  \sqrt{\gamma(\bomega)} \end{pmatrix} \otimes \ketbra{\bomega}{\bomega},
\end{align}
the quantum Fourier Transform $\vec{QFT}$, and the block-encoding $\vV_{jump}$ of the jump operators $\sum_{a\in A} \ket{a} \otimes \vA^a$. In practice, synthesizing the above incurs additional overhead but should be treated as an independent subroutine to study (see~\autoref{sec:block_encodings}); we expect Hamiltonian simulation to be the dominant source of cost, which we present\footnote{To obtain the end-to-end gate complexity, one should specify a Hamiltonian simulation subroutine (e.g.,\cite{lloyd1996UnivQSim,low2016HamSimQSignProc,childs2012HamSimLCU}).} by the accumulated time for the (controlled) unitaries $\e^{\ri \vH \bt}$.

The key idea behind implement the unitary block-encoding \eqref{eq:opFTGens} is the \textit{operator Fourier Transform} (\autoref{sec:block_encodings}) as an alternative to phase estimation (Figure~\ref{fig:OQFT})
\begin{align}
    	\CF[\cdot]: \ket{f}\otimes \vO \rightarrow \sum_{\bomega}\ket{\bomega} \otimes \vO_{\bomega}
\end{align}
which is physically motivated, compatible with our analytic framework, and leads to simple explicit circuits. 

\subsubsection{Improving Nature}
Suppose our goal is to prepare the Gibbs state, then according to~\autoref{thm:Nature_fixedpoint} and~\autoref{lem:nature_block_encoding} we may algorithmically simulate the physical Lindbladian till the mixing time. However, the Fourier Transform occurring in Nature (\autoref{thm:Nature_fixedpoint}), in fact, has a ``heavy tail'' in the frequency domain; this uncertainty in energy may significantly contribute to the Gibbs state error. With full algorithmic freedom, can we do better? In this section, we simply tweak the Lindbladian by considering a nicer Fourier Transform weight
\begin{align}
    f(t) \propto \e^{-\frac{t^2}{4\sigma_t^2 }} \quad \text{with Gaussian width}\quad \sigma_t.
\end{align}
The width sets the Hamiltonian simulation time scale $\sim \sigma_t$ and the energy resolution $\sim \sigma_t^{-1}$.  The Gaussian distribution is particularly nice as it enjoys sharp concentration in \textit{both} time and frequency domains; in principle, other normalized functions are also permissible, such as Kaiser-window functions~\cite{berry2022QuantTopologicalData,McArdle_2022quantumstate}, as long as they can be efficiently generated in superposition $\sum_{\bt \in S_{t_0}} \ket{f(\bt)}$, but we will stick to Gaussians for simplicity.  
\begin{thm}[Approximate Gibbs fixed point]\label{thm:L_correctness}
Any \Lword{} $\CL_{\beta}$~\eqref{eq:mainL} satisfying the symmetry and normalization conditions~\eqref{eq:AAdagger},\eqref{eq:fnormalized}, and \eqref{eq:gamma_KMS} with the normalized Gaussian weight 
\begin{align}
    f(t) \propto \e^{-\frac{t^2}{4\sigma_t^2 }}
\end{align}
has an approximate Gibbs fixed point
\begin{align}
\nrm{ \vrho_{fix}(\CL_\beta) -\vrho_\beta}_1 =\bigOt{\frac{\beta}{\sigma_t} t_{mix}(\CL_{\beta})}
\label{eq:Gaussian_fixedpt_error}.
\end{align}
\end{thm}
See~\autoref{sec:proof_correctness} for the proof. We briefly present the analogous block-encoding costs; the Gaussian width $\sigma_t$ plays a similar role as the time scale $T$, and Gaussians exhibit a better scaling than~\eqref{eq:nature_fixedpt_error} due to its sharp concentration.
\begin{lem}[Efficient block-encoding]
    In the setting of~\autoref{thm:L_correctness}, a unitary block-encoding $\vU$ for the (suitably discretized) Lindblad operators~\eqref{eq:opFTGens} can be constructed using one query each to $\vec{Prep}$,  $\vW$, $\vV_{jump}$, $\vec{QFT}$, and 
\begin{align}
&\tilde{\CO}\L(\sigma_t \R) \quad &\text{(controlled) Ham. sim. time for $\vH$}\\
&n+1+ \lceil\log_2(\labs{A})\rceil + \lceil\log_2(N)\rceil & \text{(resettable) qubits}.
\end{align}
\end{lem}
See~\autoref{lem:block_encoding_L} for the explicit construction (which is essentially the same circuit leading to~\autoref{lem:nature_block_encoding}) and the required $N$ to ensure a good discretization error (\autoref{apx:cont_limit}). Therefore, the Gaussian width merely needs to scale as
\begin{align}
\sigma_t \sim \beta t_{mix} / \epsilon  
\end{align}
to prepare a Gibbs sample; see Table~\ref{table:thermal_algorithms} for the altogether cost.

\subsubsection{Quantum-walk speedup}

With full algorithmic freedom, we may further depart from physics and seek a Szegedy-type speedup by considering a coherent representation of the \Lword{}~\cite{szegedy2004QMarkovChainSearch,wocjan2021szegedy}.\footnote{For convenience, we define the discriminant such that it is shifted by the identity matrix compared to definitions in earlier work.} Inheriting the notation of~\eqref{eq:mainL}, we consider the following \textit{discriminant proxy} 
\begin{align}\label{eq:D_main}
\vec{\CD}_{\beta}&:=  
\sum_{ a\in A}\int \sqrt{\gamma(\omega)\gamma(-\omega)}  \hat{\vA}^a(\omega) \otimes  \hat{\vA}^{a*}(\omega) - \frac{\gamma(\omega)}{2} \L(\vA^{a}(\omega)^\dagg\hat{\vA}^a(\omega) \otimes \vI + \vI \otimes \hat{\vA}^{a}(\omega)^{\dagg*}\hat{\vA}^{a}(\omega)^*\R)\ \mathrm{d} \omega. 
\end{align}
The superscript $\vO^*$ denotes the entry-wise complex conjugation. Indeed, as required by the quantum walk formalism, this operator is Hermitian (\autoref{cor:DiscProxySelfAdjoint})
\begin{align}
    (\vec{\CD}_{\beta})^{\dagger} = (\vec{\CD}_{\beta}) \quad \text{assuming symmetries}\quad~\eqref{eq:AAdagger},\eqref{eq:fnormalized}.
\end{align}
Analogously to the discriminant of classical Markov chains, the discriminant proxy is approximately the Davies-type \Lword{}~\eqref{eq:mainL} superoperator conjugated by the Gibbs state (\autoref{sec:D_circuit}) but \textit{vectorized} into an operator (\autoref{sec:algortihm_coherent}). This construction comes at the cost of duplicating the Hilbert space but, as a bonus, provides access to the following canonical purification of the Gibbs state
\begin{align}
\ket{\sqrt{\vrho_{\beta}}} \propto \sum_i \e^{-\beta E_i/2} \ket{\psi_i} \otimes \ket{\psi_i^*}\quad \text{where}\quad \vH =\sum_i E_i\ketbra{\psi_i}{\psi_i}
\end{align}
as the (approximate) top eigenvector of $\vec{\CD}_{\beta}$. The superscript $\ket{\psi_i^*}$ denotes entrywise complex conjugate in the computational basis\footnote{The above purified state is independent of which basis one applies complex conjugation to.}. If the Hamiltonian is diagonal in the computation basis, this is essentially equivalent to the purified classical distribution. For general Hamiltonians, the state $\ket{\sqrt{\vrho_{\beta}}}$ is also known as the thermofield double state in quantum gravity (see e.g.,~\cite{Maldacena2013CoolHF}). Taking a partial trace recovers the Gibbs state, but the purification can sometimes be more useful, e.g., for efficient verification by a swap test or faster evaluation of observables (see, e.g.,~\cite{knill2007optimal}). 

\begin{thm}[Approximate purified Gibbs state]\label{thm:D_correct} Instantiate the \Lword{} parameters of~\autoref{thm:L_correctness} with Gaussian width $\sigma_t$ for the corresponding discriminant proxy $\vec{\CD}_\beta$~\eqref{eq:D_main}. Then, the leading eigenvector $\ket{\lambda_1(\vec{\CD}_\beta)}$ well-approximates the Gibbs state 
\begin{align}
\nrm{ \ket{\lambda_1(\vec{\CD}_\beta)} -\ket{\vrho_\beta}} = \bigOt{\frac{\beta}{\sigma_t} \frac{1}{\lambda_{gap}(\vec{\CD}_\beta)}} \label{eq:D_error}.  
\end{align}
The quantity $\lambda_{gap}(\vec{\CD}_\beta) := \lambda_1(\vec{\CD}_\beta)-\lambda_2(\vec{\CD}_\beta)$ is the spectral gap.
\end{thm}
See~\autoref{sec:proof_coherent} for the proof of \autoref{thm:D_correct}. The top-eigenvector error resembles the \Lword{} case (\autoref{thm:L_correctness}) with the mixing time $t_{mix}(\CL_{\beta})$ replaced by the inverse spectral gap $\lambda_{gap}^{-1}(\vec{\CD}_{\beta})$. Though, unlike its \Lword{} cousin $\CL_{\beta}$, the discriminant proxy $\vec{\CD}_{\beta}$ does not generate a semi-group and does not by itself prepare its gapped eigenvector; we need to additionally perform \textit{quantum simulated annealing} (\autoref{sec:simulated_annealing}), which is basically adiabatic state preparation supplemented with a natural adiabatic path from high to low temperature $ 0 <\beta' < \beta $. This additional step uses
\begin{align}
    \sim \text{(adiabatic path length)}\quad \beta \norm{\vH} \times \text{(worst root-inverse-gap)}\quad \sqrt{\lambda_{gap}^{-1}}\label{eq:path_times_gap}
\end{align}
queries to block-encodings of $\vI+\vec{\CD}_{\beta'}$ across values of $\beta'$. In other words, the quantum-walk speedup against \Lword{}s boils down to replacing the mixing time by $t_{mix} \rightarrow \beta \norm{\vH}\sqrt{\lambda_{gap}^{-1}}$; this speedup comes with the cost of doubling the number of qubits $n \rightarrow 2n$. Of course, to perform adiabatic state preparation, we must algorithmically construct the block-encoding for $\vI + \vec{\CD}_{\beta}$. 

\begin{lem}[Efficient block-encoding]\label{lem:D_cost}
In the setting of \autoref{thm:D_correct}, a unitary block-encoding of $ \vI + \vec{\CD}_{\beta}$ can be constructed up to $\epsilon$ spectral norm error using $\CO(1)$ query each to $\vec{Prep}$,  $\vW$, $\vV_{jump}$, $\vec{QFT}$, $\vF$, $\vP$, and
\begin{align}
&\tilde{\CO}\L(\sigma_t \R) \quad &\text{(controlled) Ham. sim. time for}\quad \vH\\
& 2n+ \lceil\log_2(\labs{A})\rceil + 
\tCO( 1) & \text{(resettable) qubits}.
\end{align}
\end{lem}
See~\autoref{prop:discriminantBlock} for the circuit for~\autoref{lem:D_cost}. In the above, we have implicitly chosen the appropriate discretization $N$ (\autoref{apx:cont_limit}).  Note that to obtain a quantum-walk speedup, we made use of two additional (low-cost) circuit components (see~\autoref{sec:D_circuit}): the reflection on energy
\begin{align}
		\vF := \sum_{\bomega\in S_{\omega_0}} \ketbra{-\bomega}{\bomega}
\end{align}
and a permutation of the jump operator labels 
\begin{align}
		\vP:= \sum_{a \in A} \ketbra{a'}{a} \quad \text{where} \quad \vA^{a'} = (\vA^a)^{\dagger}\quad \text{for each}\quad a \in A.
\end{align}
To grasp the algorithmic cost, we roughly expect the width along the adiabatic path to scale as
\begin{align}
 \sigma_t(\beta') \sim \frac{\beta'}{\epsilon'\lambda_{gap}(\vec{\CD}_{\beta'})} \quad \text{for}\quad \text{$\epsilon' $-approximated}\quad \ket{\sqrt{\rho_{\beta'}}}; 
\end{align}
see Table~\ref{table:thermal_algorithms} and \autoref{sec:simulated_annealing} for an quantitative accumulated cost for the adiabatic algorithm.

The remainder of the main text is organized by the analytic (\autoref{sec:DB}) and the algorithmic parts (\autoref{sec:main_alg}). We begin the analytic exposition by reviewing basic facts circling detailed balance and mixing time, and then introduce consequences of approximate detailed balance. The algorithmic arguments include black-box \Lword{} simulation and a general recipe to quantize a \Lword{}. These abstract algorithms can be understood assuming merely block-encodings for the Lindblad operator, whose explicit circuit construction is laid out in~\autoref{sec:block_encodings}. We conclude the main text by highlighting plausible future directions in~\autoref{sec:discussion}.

The appendices are organized as follows. We begin with a recapitulation of notations (\autoref{sec:recap_notation}), followed by the supporting details for our key analytic and algorithmic argument: \autoref{sec:operator_FT} discusses properties of the operator Fourier Transform and the secular approximation; \autoref{sec:error_boltzmann} proves approximate detailed balance for the constructed \Lword{}s $\CL_{\beta}$ and discriminant proxies $\vec{\CD}_{\beta}$. 

The rest of the appendices consist of isolated topics. \autoref{apx:cont_limit} discusses the relation between continuous Fourier Transforms, which is conceptually simple, and the discrete Fourier Transform, which we implement. Fortunately, the rule of thumb is that the Fourier Transform register merely needs to be poly-logarithmic for a small discretization error. \autoref{sec:open_system} discusses \Lword{}s arising from a microscopic open system derivation and prove their fixed point accuracy; this requires a moderate generalization of the main analytic framework.~\autoref{sec:spectral_bounds_mixing_times} is devoted to supporting approximate detailed balance (\autoref{sec:DB}), especially on perturbation theory for nonHermitian matrices; these facts tend to be intuitively akin to the Hermitian case but we include the (nonstandard) proofs for completeness. \autoref{sec:simulated_annealing} reviews quantum simulated annealing in a modern quantum algorithm language, which we largely employ as a black box.

\section{Approximate stationarity of the Gibbs state}
\label{sec:DB}
We begin our analysis of the generator in Eqn.~\eqref{eq:mainL} by recalling some general properties of detailed balance \Lword{}s before introducing the key notion of \textit{approximate detailed balance}. At the heart of classical Markov chain Monte Carlo algorithms is a rapid mixing Markov chain whose fixed point yields the desired distribution.
In the quantum setting, central to our discussion is the generator of a quantum dynamical semi-group~\cite{Breuer2006open,wolf2019QCLectureNotes}, the \textit{\Lword{}} in the Schrödinger Picture
\begin{align}
    \CL[\vrho] = -\ri [\vH, \vrho ] + \sum_{j \in J} \L(\vL_j\vrho \vL_j^{\dagger} - \frac{1}{2} \{\vL_j^{\dagger}\vL_j,\vrho\}\R) 
\end{align}
parameterized by a set of \textit{Lindblad operators} $\{\vL_j\}_{j \in J}$ and a Hermitian matrix $\vH$. Mathematically, the above elegant form encompasses all possible \Lword{}s, including, but not restricted to, those arising from a microscopic system-bath derivation. In particular, from an algorithmic perspective, we enjoy the additional freedom of choosing favorable Lindblad operators $\vL_j$ with the hope that
\begin{enumerate}
    \item the \Lword{} can be implemented efficiently, \label{item:L_efficient}
    \item the fixed point is unique and yields the desired state $\vrho_\beta$, and \label{item:L_fixed_point}
    \item the \Lword{} converges rapidly. \label{item:L_converge}
\end{enumerate}

The above summarizes the desirable criteria for a quantum Gibbs sampler.\footnote{One can certainly consider discrete-time quantum channels~\cite{temme2009QuantumMetropolis} as Gibbs sampler candidates. However, the continuous-time  \Lword{}, inspired by physics, appears technically nicer for our purposes.} The convergence depends on the particular Hamiltonian of interest and is generally nontrivial to analyze. Fortunately, the \textit{detailed balance condition} enables systematic analysis of quantum dynamical semi-groups, similarly to how detailed balance is central in analyzing classical Markov chains (see, e.g.,~\cite{Markovchain_mixing}).
\begin{defn}[Detailed balance condition]\label{defn:DB}
For a normalized, full-rank state $ \vrho\succ 0$, we say that an endomorphism $\CL$ satisfies $\vrho$-detailed balance whenever the associated {\rm discriminant} is self-adjoint with respect to $\vrho$, i.e., 
\begin{align}
\CD(\vrho, \CL) &:= \vrho^{-1/4}\CL[ \vrho^{1/4}\cdot  \vrho^{1/4}]  \vrho^{-1/4}.\\
  &=\vrho^{1/4}\CL^{\dagger}[ \vrho^{-1/4}\cdot  \vrho^{-1/4}]  \vrho^{1/4}= \CD(\vrho, \CL)^{\dagger}.  
\end{align}
\end{defn}
In the above definition (and in the rest of the paper), we define the adjoint of a superoperator with respect to the Hilbert-Schmidt inner product $\ipc{\vX}{\vY}_{HS}=\tr(\vX^\dagg \vY)$. Explicitly,
\begin{align}
&\mathcal{C}[\cdot]=\sum_j  \alpha_j \vA_j[\cdot]\vB_j \quad \text{implies}\quad \mathcal{C}^{\dagger}[\cdot]=\sum_j  \alpha^*_j \vA^{\dagger}_j[\cdot]\vB^{\dagger}_j\quad \text{for any}\quad \vA_j, \vB_j\\
&\text{since}\quad\tr(\vX^\dagg \mathcal{C}[\vY])=\tr(\vX^\dagg \sum_j  \alpha_j \vA_j\vY\vB_j)
    =\tr((\sum_j \alpha_j^*\vA_j^\dagg\vX\vB_j^\dagg)^\dagg\vY) = \tr[ \L(\mathcal{C}^{\dagger}(\vX) \R)^{\dagger} \vY]\label{eq:superoperator_adjoint}.
\end{align}
This will be revisited when defining the vectorization (\autoref{sec:algortihm_coherent}). In particular, the superoperator adjoint for Lindbladians coincides with converting between the Heisenberg and Schrödinger pictures
\begin{align}
     \tr[\vrho\CL^{\dagger}[\vO]]= \tr[\CL[\vrho]\vO ]\quad \text{for each}\quad \vrho \succeq 0\quad \text{and}\quad\vO
\end{align}
using that $\vrho^{\dagger} = \vrho$ and that Lindbladians $\CL[\vrho] = (\CL[\vrho])^{\dagger}$ preserves Hermiticity.

Analogously to the classical case, the detailed balance condition considers a similarity transformation according to the target distribution weights $\vrho$.\footnote{Technically, an alternative definition of detailed balance may distribute the power somewhat arbitrarily $\vrho^{s}[\cdot]\vrho^{1/2-s}$, but we stick to the symmetric case $\vrho^{1/4}[\cdot]\vrho^{1/4}$ for simplicity.} The detailed balance condition brings about two desirable properties. First, it ensures that the state $\vrho$ is a fixed point (Point~\ref{item:L_fixed_point}).
\begin{prop}[Gibbs fixed point \cite{temme2010chi}]\label{prop:Gibbs_fixed_point}
If a superoperator $\CL$ generates a trace-preserving map and satisfies $\vrho$-detailed balance, then it annihilates the state $\CL[ \vrho]= 0 $, or equivalently, $\CD(\vrho, \CL)[\sqrt{\vrho}] = 0$. 
\end{prop}
\begin{proof}
	We know that the infiniestimal exponential map $\e^{\varepsilon\CL}[\cdot ]= [\cdot] + \varepsilon\CL[\cdot] +\bigO{\varepsilon^2}[\cdot]$ is trace-preserving, thus the leading order term must satisfy $\CL^{\dagger}[\vI] =0$. Therefore, 
\begin{align}
    0=\vrho^{1/4}\CL^{\dagger}[ \vI]  \vrho^{1/4} =\CD(\vrho, \CL)^{\dagger}[\sqrt{\vrho}]=\CD(\vrho, \CL)[\sqrt{\vrho}]
\end{align}  
 using $\vrho$-detailed balance in the last equality.
\end{proof}
Second, it relates the \Lword{} mixing time to the spectral gap (Point~\ref{item:L_converge}). We only state the following result here but later prove a qualitatively more robust statement in \autoref{prop:gap_to_mixing} applicable to the approximate case.

\begin{restatable}[Mixing time from spectral gap \cite{kastoryano2013quantum}]{prop}{mixingtimegapDB}\label{prop:mixDetail}
	If a \Lword{} $\CL$ satisfies $ \vrho$-detailed balance, then 
    \begin{align}
        t_{mix}(\CL) \le \frac{ \ln(2\norm{\vrho^{-1/2}})}{\lambda_{gap}(\CL)},
    \end{align}
    where $\lambda_{gap}(\CL)$ is the eigenvalue gap of the \Lword{}, and the mixing time $t_{mix}$ is the smallest time for which
\begin{align}
\lnormp{\e^{\CL t_{mix}}[\vrho_1-\vrho_2]}{1} \le \frac{1}{2} \normp{\vrho_1-\vrho_2}{1} \quad \text{for any states}\quad \vrho_1, \vrho_2.
\end{align}
\end{restatable}

The analysis of a superoperator gap $\lambda_{gap}(\CL)$ is perhaps more tractable than the mixing time $t_{mix}$ but still nontrivial and instance specific.\footnote{For the experts, the gap may not give the tightest possible mixing time bounds; techniques beyond gap-based bounds typically require proving a Log-Sobolev inequality, which can be very challenging in the noncommuting cases.} Otherwise, we see that the detailed balance condition readily addresses two criteria (Point~\ref{item:L_fixed_point} and Point~\ref{item:L_converge}) for Gibbs samplers. 


\subsection{Approximate detailed balance}
\label{sec:L_ADB}

Unfortunately, we do not know of general \textit{efficient} constructions of quantum Gibbs samplers satisfying the detailed balance condition \textit{exactly} (Point~\ref{item:L_efficient})\footnote{This problem is resolved in a follow-up work~\cite{exactDB}.}; this is rooted in the energy-time uncertainty principle where quantum algorithms only access the energies of a quantum system approximately. As our main technical contribution, we formulate the \textit{$\epsilon$-approximate detailed balance condition} that addresses all three requirements for a quantum Gibbs sampler.

\begin{defn}[Approximate detailed balance condition]\label{defn:ADB}
For any \Lword{} $\CL$ and full-rank state $\vrho$, take a similarity transformation and decompose into the Hermitian and the anti-Hermitian parts
\begin{align}\label{eq:almostSA}
    \CD(\vrho, \CL) = \vrho^{-1/4}\CL[\vrho^{1/4}\cdot\vrho^{1/4}]\vrho^{-1/4} &= \CH(\vrho,\CL) + \CA(\vrho,\CL).\\
    \CD(\vrho, \CL)^{\dagger} = \vrho^{1/4}\CL^{\dagger}[\vrho^{-1/4}\cdot \vrho^{-1/4}]\vrho^{1/4} & = \CH(\vrho,\CL) - \CA(\vrho,\CL)
\end{align}    
We say the \Lword{} $\CL$ satisfies the $\epsilon$-approximate $\vrho$-detailed balance condition if the anti-Hermitian part $\CA$ is small
\begin{align}
    \frac{1}{2}\lnormp{\CD(\vrho, \CL)-\CD(\vrho, \CL)^{\dagger}}{2-2} =  \norm{\CA(\vrho,\CL)}_{2-2} \le \epsilon.
\end{align}
\end{defn}
If the anti-Hermitian part vanishes, we recover the exact detailed balance condition $ \CD(\vrho, \CL)^{\dagger} = \CD(\vrho, \CL)$. If not, we show that the fixed point still approximates the state $\vrho$ (Point~\ref{item:L_fixed_point}).
\begin{restatable}[Fixed point accuracy]{cor}{Lfixedpointerrormix}\label{cor:fixed_point_error_mix}
	If a \Lword{} $\CL$ satisfies the $\epsilon$-approximate $\vrho$-detailed balance condition, then its fixed point $\vrho_{fix}(\CL)$ deviates from $\vrho$ by at most
	\begin{align}
		\nrm{\vrho_{fix}(\CL)-\vrho}_1\leq20 t_{mix}(\CL) \epsilon.
	\end{align}	
\end{restatable}
See~\autoref{sec:proof_for_approx_DB} for the proof. We see that the error bound deteriorates if the map has a large anti-Hermitian component $\CA$ or if the \Lword{} mixes slowly.\footnote{We actually prove a stronger statement in \autoref{sec:proof_for_approx_DB} that gives a bound in terms of the gap $\lambda_{gap}(\CH)$ of the Hermitian part.} The anti-Hermitian component involves the inverse Gibbs state $\vrho^{-1}$, and might be difficult to bound directly. As a remedy, it is helpful to introduce an intermediate \Lword{} $\CL'$ for which approximate detailed balance is easier to show. In that case, we can write
\begin{align}
    \normp{\vrho_{fix}(\CL) - \vrho}{1} &\le \normp{\vrho_{fix}(\CL) - \vrho_{fix}(\CL')}{1}+ \normp{\vrho_{fix}(\CL') - \vrho}{1}.
\end{align}
The first term on the RHS does not directly involve the inverse $\vrho^{-1}$ and can be controlled by a \Lword{} perturbation bound as follows.
\begin{restatable}[Fixed point difference]{lem}{mixingtimetofixedpoint}\label{fact:mixing_time_to_fixedpoint}
For any two \Lword{}s $\CL_1$ and $\CL_2$, the difference of their fixed points (in the Schrödinger picture) is bounded by
\begin{align}
    \nrm{ \vrho_{fix}(\CL_1) - \vrho_{fix}(\CL_2)}_1 \le 4 \normp{ \CL_1 - \CL_2 }{1-1}\cdot t_{mix}(\CL_1).
\end{align}
\end{restatable}
See \autoref{sec:gap_and_convergence} for the proof. Conveniently, even without detailed balance, the mixing time $t_{mix}$ remains controlled by spectral properties of the Hermitian part (addressing Point~\ref{item:L_converge}):
\begin{restatable}[Mixing time from Hermitian gap]{prop}{gaptomixing}\label{prop:gap_to_mixing}
	For any \Lword{} $\CL$ and a full-rank state $\vrho$, suppose the self-adjoint component $\CH = \CH (\vrho,\CL)$ satisfies
\begin{align}
	\frac{\lambda_1(\CH)}{\lambda_{gap}(\CH)} \le \frac{1}{100}, \quad \text{then}\quad t_{mix}(\CL) \le 3\frac{\ln(3\norm{\vrho^{-1/2}})}{\lambda_{gap}(\CH)}.
\end{align}
\end{restatable}

See \autoref{sec:gap_and_convergence} for the proof. In particular, the top eigenvalue can be bounded by the anti-Hermitian part $\lambda_1(\CH)\le \norm{\CA}_{2-2}$ for any \Lword{}~\eqref{eq:topEigenvalueBound}; therefore, it remains to provide an efficient construction of the \Lword{} (\autoref{sec:block_encodings}) and prove approximate detailed balance (\autoref{sec:error_boltzmann}).

\subsection{Proof of fixed point correctness ({\autoref{thm:L_correctness}}) }
\label{sec:proof_correctness}

We are now in a position to prove our first main theorem; the proximity of the stationary state $\vrho_{fix}$ to the Gibbs state $\vrho_\beta$. Most of the technical definitions and lemmata are relegated to Appendix \ref{sec:secular}. Here, we address the essential features of the proof together with some essential tools. The main technical argument introduces an intermediate \Lword{} $\CL_{sec}$
\begin{align}
    \normp{\vrho_{fix}(\CL_\beta) - \vrho_\beta}{1} &\leq\normp{\vrho_{fix}(\CL_\beta) - \vrho_{fix}(\CL_{sec})}{1} + \normp{\vrho_{fix}(\CL_{sec}) - \vrho_\beta}{1}
\end{align}
and uses the fixed point error bounds (\autoref{fact:mixing_time_to_fixedpoint}) for the first term and (\autoref{prop:fixed_point_error}) for the second term.

The first error arises from the \textit{secular approximation} (\autoref{sec:secular}), defined by truncating the transitions in the frequency domain
\begin{align}
    \hat{\vA}^a(\bomega) \rightarrow \hat{\vS}^a(\bomega) \quad \text{such that} \quad \bra{\psi_i}\hat{\vS}^a(\bomega) \ket{\psi_j} = 0 \quad \textrm{whenever} \quad \labs{(E_i - E_j) - \bomega} > \bmu
\end{align}
for a tunable truncation parameter $\bmu$. See~\eqref{eqn:Aoperator}, \eqref{eq:genSecDef} for the precise definition of the secular approximated jump operators $\hat{\vS}^a(\bomega)$. The purpose of this truncation is to ensure approximate detailed balance: conjugating $\CL_{sec}$ with the Gibbs state $\vrho$, as required in comparing with the similarity transformation, remains well-behaved (\autoref{sec:secular}). The truncation parameter $\bar{\mu}$ is not physical but rather a proof artifact. Intuitively, our choice of Gaussian weight ensures its Fourier Transform to remain (approximately) another Gaussian (see~\autoref{sec:tail_bounds}), which has a rapidly decaying tail. Thus, we expect the error from truncating the Gaussian tail to be small whenever $\bmu \gtrsim \sigma_t^{-1}$. Thus, with the Gaussian weight, the secular approximation incurs a mild error; this error becomes more severe with the step-function weights given by nature, whose Fourier Transform has a heavy tail (\autoref{prop:uniformTail}).
 
The second error is the most technical part, showing that the secular-approximated operator $\CL_{sec}$ satisfies approximate detailed balance (See~\autoref{sec:error_boltzmann}). We highlight the full technical statement as follows.

\begin{restatable}[Approximate detailed balance]{lem}{ADBopensys}\label{lem:L_Approx_DB}
	Consider a \Lword{} in the following form
	\begin{align}
		\CL &= 
		\sum_{a\in A, \bomega\in S_{\omega_0}} \gamma(\bomega)  \hat{\vS}^a(\bomega)[\cdot] \hat{\vS}^a(\bomega)^{\dagger} -\frac{\gamma(\bomega)}{2} \{\hat{\vS}^a(\bomega)^{\dagger} \hat{\vS}^a(\bomega) ,\cdot \},
	\end{align}
where $\gamma(\bomega)/\gamma(-\bomega)=\e^{-\beta\bomega}$ for each $\bomega\in S_{\omega_0}$.
Suppose there exists $\bmu\leq \beta^{-1}$ such that the operators satisfy
\begin{align}
		\bra{\psi_i}\hat{\vS}^a(\bomega) \ket{\psi_j} = 0 \quad \textrm{whenever} \quad \labs{(E_i - E_j) - \bomega} > \bmu \label{eq:PSPass_L}
\end{align}
	for the eigenvalue decomposition of $\vH=\sum_j E_j\ketbra{\psi_j}{\psi_j}$,
	and there is a permutation $\vP\colon a\rightarrow a'$ such that $\hat{\vS}^a(\bomega)^\dagg = \hat{\vS}^{a'}(-\bomega)$ for each $a,\bomega$. Then, for the Gibbs state $\vrho=\e^{-\beta \vH}/\tr[\e^{-\beta \vH}]$ we have
\begin{align}
\frac{1}{2}\lnormp{\CD(\vrho, \CL) - \CD(\vrho, \CL)^{\dagger}}{2-2} \le \CO\L(\beta\bmu 
\nrm{\sum_{ a\in A}\sum_{\bomega\in S_{\omega_0}}\!\!\!\gamma(\bomega) \hat{\vS}^a(\bomega)^{\dagger} \hat{\vS}^a(\bomega)}\R).
\end{align}
\end{restatable}
The above is a simplified version of \autoref{lem:error_from_Boltzmann}, which we prove in~\autoref{sec:error_boltzmann}. Our normalization further simplifies the RHS to $\CO(\beta \bmu)$. We now combine the above estimates to prove Theorem~\ref{thm:L_correctness}.  

\begin{proof}[Proof of Theorem~\ref{thm:L_correctness}]
While Theorem~\ref{thm:L_correctness} is stated in the continuum limit $N\rightarrow \infty$~\eqref{eq:mainL}, we give general error bounds at finite $N$~\eqref{eq:LbetaAgain} and then take the $N\rightarrow \infty$ limit~\eqref{eq:mainL}. Introduce the secular-approximated Lindblad operator $\CL_{sec}$ to bound the fixed point error
\begin{align}
    \normp{\vrho_{fix}(\CL_\beta) - \vrho_\beta}{1} &\le \normp{\vrho_{fix}(\CL_\beta) - \vrho_{fix}(\CL_{sec})}{1} + \normp{\vrho_{fix}(\CL_{sec}) - \vrho_\beta}{1}\\
    &\le  2\normp{\CL_\beta - \CL_{sec}}{1-1} t_{mix}(\CL_\beta) + 10\lnorm{\CD(\vrho, \CL_{sec}) - \CD(\vrho, \CL_{sec})^{\dagger}}t_{mix}(\CL_{sec})\\
    &\le \CO\L( \L(\normp{\CL_\beta - \CL_{sec}}{1-1}+\lnorm{\CD(\vrho, \CL_{sec}) - \CD(\vrho, \CL_{sec})^{\dagger}}\R)t_{mix}(\CL_\beta) \R)\label{eq:L_fL_sec}.
\end{align}
The second inequality uses \autoref{fact:mixing_time_to_fixedpoint} for the first term and \autoref{cor:fixed_point_error_mix} for the last term. The third inequality uses that 
$t_{mix}(\CL_{sec}) \le t_{mix}(\CL_\beta)\left\lceil\frac{\ln(1/2)}{\ln(1/2+t_{mix}(\CL_\beta) \normp{\CL_\beta-\CL_{sec}}{1-1})}\right\rceil$ (\autoref{prop:mixingtime_diff}), which further simplifies to $t_{mix}(\CL_{sec}) = \CO(t_{mix}(\CL_\beta))$ since we must have $\normp{\CL_\beta - \CL_{sec}}{1-1} t_{mix}(\CL_\beta) = \CO(1)$ otherwise the trace distance bound becomes vacuous. 

Now, we evaluate approximate detailed balance (\autoref{lem:L_Approx_DB}) and the secular approximation error using~\autoref{lem:secular}, \autoref{prop:Gaussian_tail}, and that the Gaussian tail in the time domain is bounded directly by $\sqrt{\sum_{ \labs{\bt}\ge  T} \labs{f(\bt)}^2} = \CO(\sqrt{T/\sigma_t}^{-1}\e^{-T^2/4\sigma_t^2} )$
\begin{align}
    \eqref{eq:L_fL_sec}&\le \CO\L( \L(T\bomega_0 + \e^{-T^2/4\sigma_t^2} + \e^{ - N^2t_0^2/16\sigma_t^2} +\e^{ - N^2\omega_0^2\sigma_t^2/2}+\e^{ - \bmu^2\sigma_t^2}+\beta \bmu \R)\cdot t_{mix}(\CL_\beta)\R)\\
    &\le \CO\L( (\sigma_t\omega_0 \sqrt{\log(1/(\sigma_t\omega_0))}+ \frac{\beta}{\sigma_t}\sqrt{\log({\sigma_t}/{\beta})} + \e^{ - N^2\omega_0^2\sigma_t^2/2} )\cdot t_{mix} (\CL_{\beta}) \R)
\end{align}
The second inequality chooses the free parameter $T = 2\sigma_t \sqrt{\ln(1/(\sigma_t\omega_0))}$ and $\bmu = \frac{\beta}{\sigma_t} \sqrt{\ln({\sigma_t}/{\beta})}$ and uses that $\e^{ - N^2t_0^2/16\sigma_t^2}=\e^{ - \pi^2/4\omega_0^2\sigma_t^2} = \CO( \sigma_t\omega_0 \sqrt{\log(1/(\sigma_t\omega_0))}) $ to simplify the expression. 
For the continuum case~\eqref{eq:LbetaAgain}, we have the simpler bound
\begin{align}
        \normp{\vrho_{fix}(\CL_\beta) - \vrho_\beta}{1} = \CO\L( \frac{\beta}{\sigma_t}\sqrt{\log({\sigma_t}/{\beta})} \cdot t_{mix} (\CL_{\beta}) \R)\quad \text{if}\quad N \omega_0 \rightarrow \infty,\ \omega_0 \rightarrow 0
\end{align}
where discretization parameter $\omega_0$ and $N$ disappears in the continuum limit. 
\end{proof}

\section{Quantum algorithms for Gibbs sampling}
\label{sec:main_alg}
In this section, we present two algorithms for approximately preparing the Gibbs state $\vrho_\beta$, both of which are inspired by the dynamical semi-group generated by the Lindbladian $\CL_\beta$. This first algorithm, which we call the \textit{incoherent Gibbs sampling algorithm}, directly simulates the time evolution $\e^{\CL_\beta t}$ by introducing ancillas. The second, which we call the \textit{coherent Gibbs sampling algorithm}, is a Szegedy-type quantum walk algorithm. It enables implementing an orthogonal projector onto the coherent Gibbs state $\ketbra{\sqrt{\vrho_\beta}}{\sqrt{\vrho_\beta}}$ with a quadratic speedup with respect to the real spectral gap of the generator $\CL_\beta$. This projector can then be used in conjunction with simulated annealing (\autoref{sec:simulated_annealing}) to prepare the purified Gibbs state.  

In the circuit constructions, we will extensively use the following rotation gates
\begin{align}
\vY_\theta:=\e^{-\ri\arcsin\!\sqrt{\theta}\vY}
=\begin{pmatrix} \sqrt{1-\theta} & -\sqrt{\theta}\\ \sqrt{\theta} &  \sqrt{1-\theta} \end{pmatrix}
\quad \text{with the Pauli-Y matrix}\quad \vY=\begin{pmatrix} 0 & -\ri \\ \ri &  0 \end{pmatrix}. \label{eq:Y_theta}  
\end{align}

\subsection{Our quantum Gibbs sampling algorithms}
We describe two \Lword{} simulation algorithms: the first exhibits Trotter-like scaling and repeatedly uses a simple (randomized) and weak-measurement gadget (\autoref{thm:weakMeasSim}, \autoref{cor:rndWeakMeasSim}); the second is inspired by~\cite{cleve2016EffLindbladianSim} and has asymptotically almost optimal scaling with time and error \autoref{thm:LCUSim} but requiring a more involved circuit and slightly more ancilla qubits. Both arguments are general as they assume merely a block-encoding of the \Lword{} (\autoref{defn:blockLindladian}); the particular block-encoding for our proposed Gibbs sampler are constructed explicitly in another section (\autoref{sec:block_encodings}).

 Further, we ``quantize'' the \Lword{}s and present \textit{coherent} Gibbs sampling algorithms that prepare the (canonical) purification of an approximate Gibbs state via simulated annealing (\autoref{sec:simulated_annealing}). The procedure assumes that we have a block-encoding of the discriminant matrix of our \Lword{}, which then enables a Szegedy-type quadratic speedup in the simulation time. However, the total speedup is only sub-quadratic on the gap dependence because of the cost to block-encode the discriminant matrix.

\subsubsection{Incoherent \Lword{} simulation algorithms}
Following~\cite{cleve2016EffLindbladianSim} we propose two different implementation methods for incoherent (trajectory-based) simulation of the \Lword{}s that describe our Gibbs sampler. 
The first method is based on a product formula and repeatedly uses a weak measurement scheme\footnote{Our weak measurement scheme is very similar to the short-time evolution by the auxiliary Hamiltonian $J$ utilized in~\cite{cleve2016EffLindbladianSim}, however our approach is a bit more direct and made it clear that a block-encoding of the jump operators suffices as input.} for implementing a small time step. The resulting scaling is analogous to the performance of ``vanilla'' Trotter-based Hamiltonian simulation: the complexity for an $\epsilon$-accurate-time-$t$ \Lword{} evolution scales as $t^2/\epsilon$. Our weak measurement scheme gives rise to simple and low-depth circuits for simulating \Lword{}s given block-encoding access.

The usefulness of weak measurements should come as no surprise, as they are also extremely helpful in other noncommutative state preparation tasks as well ( see, e.g.,~\cite{gilyen2016PrepGapHamEffQLLL}), and the very recent independent work of~\cite{cubitt2023DissipativeStatePrep}. The common theme in these applications is the exploitation of some quantum Zeno-like effect,\footnote{In our case, the quantum Zeno-like effect is manifest in the quadratically reduced amplitude of $\ket{0^c \perp}$ in \eqref{eq:weakMeasPsi}.} but on a higher level, these applications also show some conceptual differences. We leave it for future work to explore whether there is a more fundamental connection between our weak measurement scheme and that of \cite{gilyen2016PrepGapHamEffQLLL,cubitt2023DissipativeStatePrep}.

The second method is based on the algorithm of \cite{cleve2016EffLindbladianSim}, which achieves a close-to-optimal scaling with respect to time and accuracy. Although the asymptotical complexity is much improved, the corresponding circuits are more complicated as they use a linear combination of unitaries (LCU), oblivious amplitude amplification, and advanced ``compression'' techniques. We leave it to future work to determine how the two schemes perform in practice. 

For both algorithms, it suffices to assume that a  \emph{purely irreversible} \Lword{} without the Hamiltonian term $\CL[\cdot]=\sum_{j\in J} \vL_j[\cdot]\vL_j^\dagg - \frac{1}{2} \{\vL_j^\dagg \vL_j,\cdot\}$
is provided in the form of a ``block-encoding'' (i.e., dilation) as~\autoref{defn:blockLindladian}.\footnote{Recent work~\cite{li2021EffSimNonMarkov,li2022SimMarkOpen} assumes the \Lword{} jumps are individually block-encoded while we assume the \textit{entire} set of jumps is encoded in a \textit{single} unitary. We give strictly better complexity for simulating \Lword{}s under this input model, which holds for our Gibbs sampling algorithm and that of~\cite{Rall_thermal_22} (leading to direct improvement for the latter). Remarkably, even if the jumps are individually block-encoded~\cite{li2021EffSimNonMarkov,li2022SimMarkOpen}, these can be converted to our input model. Still, even accounting for the conversion overhead, we recover (up to polylogarithmic factors) their complexity for \Lword{} simulation. The main innovation here seems to be the generalization of the input model, as the earlier Lindbladian simulation algorithms also seems to work~\cite{wang2023email} under this more general input assumption.} In particular, recall our proposed \Lword{} Gibbs sampler (as discretization of~\eqref{eq:mainL}) 
\begin{align}\label{eq:LbetaAgain}
	\CL_{\beta}&:=\sum_{a\in A, \bomega\in S_{\omega_0}} \gamma(\bomega) \L( \hat{\vA}^{a}(\bomega)[\cdot] \hat{\vA}^{a}(\bomega)^{\dagger} -\frac{1}{2} \{\hat{\vA}^{a}(\bomega)^{\dagger} \hat{\vA}^{a}(\bomega) ,\cdot \}\R)\\
 &\text{with Lindblad operators}\quad \{\sqrt{\gamma(\bomega)} \hat{\vA}^a(\bomega)\}_{a\in A,\bomega\in \BR},
\end{align}
and its block-encoding can be found in Eqn.~\ref{eqn:Aoperator} in~\autoref{sec:block_encodings}. However, working with abstract block encodings makes our simulation results general and also simplifies our presentation and proofs, as the operator Fourier Transform naturally fits this definition (\autoref{fig:L_circuit}).
Our weak-measurement scheme is not only simple but also improves, e.g., the sparse \Lword{} simulation algorithm of~\cite[Theorem 9]{childs2016SparseLindbladianSim}.\footnote{Indeed, the complexity is improved by about a factor of $k^4$, where $k$ is the sparsity.} Also, the lower bound on the ``total evolution time'' for simple iterative circuits in \cite{cleve2016EffLindbladianSim} suggests that the performance of similar schemes may be optimal.

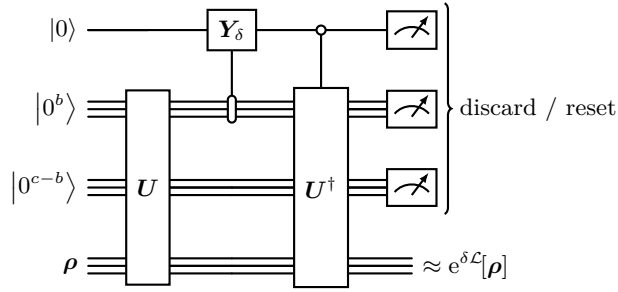
\begin{figure}[!ht]
\begin{quantikz}[wire types={q,b,b,b},classical	gap=1mm]
 	\lstick{\ket{0}}		&\qw			&\gate{\vY_\delta}	&\octrl{1}		&\meter{}\rstick[wires=3]{\kern6.5mm discard / reset}\\
 	\lstick{\ket{0^b}}		&\gate[3]{\vU}	&\octrl{-1} 	&\gate[3]{\vU^\dagg}&|[meter]| \qw \\
 	\lstick{\ket{0^{c-b}}}	& 				&\qw			&\qw 				&|[meter]| \qw \\
 	\lstick{$\vrho$}		&				&\qw			&\qw				& \qw \rstick{$\approx \e^{\delta \CL}\![\vrho]$}
\end{quantikz}
\caption{Quantum circuit implementation of an approximate $\delta$-time step via a weak measurement scheme.\footnote{ 
		The scheme can be extended to general \Lword{}s that include the coherence term $-\ri[\vH,\vrho]$ by applying $\mathcal{O}(\delta^2)$-precise Hamiltonian time-evolution for time $\delta$ on the system register before the above circuit is applied. For example, one could use Trotterized time-evolution. (In case $\nrm{\vH}>1$, the entire \Lword{} should be first scaled down by a factor of $\nrm{\vH}$.)}}\label{fig:weakMeasCircuit}
\end{figure}

\begin{restatable}[Weak-measurement for incoherent \Lword{} simulation]{thm}{weakMeasSim}\label{thm:weakMeasSim}
	Suppose $\vU$ is a block-encoding of the purely irreversible \Lword{} $\CL$ as in \autoref{defn:blockLindladian}. We can simulate the action of the superoperator $\e^{t \CL}$ to precision $\epsilon$ in diamond norm using 
 \begin{align}
 c+1&\quad \text{(resettable) ancilla qubits}, \\
 \bigO{t^2/\epsilon}& \quad\text{(controlled) uses of}\quad \vU, \vU^\dagg,\\ 
 \text{and}\quad \bigO{(b+1)t^2/\epsilon} &\quad \text{other two-qubit gates}.
 \end{align}
\end{restatable}
\begin{proof}
	We can simulate an approximate $\delta$-time step by $\CL$ using the following weak-measurement scheme displayed in \autoref{fig:weakMeasCircuit}. 
 \begin{enumerate}
     \item Apply $\vU$.
     \item Append an ancilla qubit in state $\ket{0}$ and rotate it with angle $\arcsin{\sqrt{\delta}}$ controlled on the $\ket{0^b}$ state (indicating the successful application of a jump).
     \item Apply $\vU^\dagger$ controlled on the ancilla qubit being $0$.
     \item Measure and discard all but the system register.
 \end{enumerate}
Assuming the system register is in the pure state $\ket{\psi}$, this circuit $\vC$ acts as follows:
	\begin{align}
		\ket{0}\cdot \ket{0^c}\ket{\psi}
		&\stackrel{(1)}{\rightarrow} \ket{0}\cdot \vU\ket{0^c}\ket{\psi}\nonumber\\&
		\stackrel{(2)}{\rightarrow}  \L(\sqrt{1-\delta}\ket{0}+\sqrt{\delta}\ket{1}\R)\cdot\L(\ketbra{0^b}{0^b}\otimes \vI\R)\vU\ket{0^c}\ket{\psi}
  \ \ + \ \ \ket{0}\cdot (\vI-\ketbra{0^b}{0^b}\otimes \vI) \vU\ket{0^c}\ket{\psi}\nonumber\\&
		=\ket{0}\cdot \vU\ket{0^c}\ket{\psi}\ \ +\ \ \sqrt{\delta} \ket{1}\cdot \ket{0^b}\underset{\ket{\psi'_0} :=}{\underbrace{(\bra{0^b}\otimes \vI)\vU\ket{0^c}\ket{\psi}}} \ \ - \ \ (1-\sqrt{1-\delta})\ket{0}\cdot\L(\ketbra{0^b}{0^b}\otimes \vI\R)\vU\ket{0^c}\ket{\psi}\nonumber\\&
		\stackrel{(3)}{\rightarrow}  \ket{0}\cdot \ket{0^c}\ket{\psi}\ \ + \ \ \sqrt{\delta} \ket{1}\cdot \ket{0^b}\ket{\psi'_0} 
		\ \ - \ \ (1-\sqrt{1-\delta})\ket{0}\cdot \vU^\dagg\L(\ketbra{0^b}{0^b}\otimes \vI\R)\vU\ket{0^c}\ket{\psi}\nonumber\\&
		=\ket{0}\cdot \ket{0^c}\ket{\psi}\ \ +\ \  \sqrt{\delta} \ket{1}\cdot \ket{0^b}\ket{\psi'_0}
		\ \ - \ \  (1-\sqrt{1-\delta})\ket{0}\cdot \ket{0^c}(\bra{0^c}\otimes \vI)\vU^\dagg(\ket{0^b}\otimes \vI)\cdot(\bra{0^b}\otimes \vI)\vU\ket{0^c}\ket{\psi}\nonumber\\&
		\phantom{=\ket{0}\ket{0^c}\ket{\psi} \ \ +\ \  \sqrt{\delta} \ket{1}\cdot \ket{0^b}\ket{\psi'_0}\ \ }\kern1mm
	\ \	-\ \ (1-\sqrt{1-\delta})\ket{0}\cdot (\vI-\ketbra{0^c}{0^c}\otimes \vI) \vU^\dagg\L(\ketbra{0^b}{0^b}\otimes \vI\R)\vU\ket{0^c}\ket{\psi}\nonumber\\&   
		=\ket{0}\cdot \ket{0^c}\L(\vI-\underset{\frac\delta2+\bigO{\delta^2}}{\underbrace{(1-\sqrt{1-\delta})}} \sum_{j\in J}\vL_j^\dagg \vL_j\R) \ket{\psi} + \sqrt{\delta}\ket{1}\cdot\ket{0^b}\sum_{j\in J} \ket{j}\vL_j\ket{\psi} - \underset{\frac\delta2+\bigO{\delta^2}}{\underbrace{(1-\sqrt{1-\delta})}}\ket{0}\cdot \ket{0^c \perp},\label{eq:weakMeasPsi}
	\end{align}
	where $\ket{0^c \perp}$ is some quantum state such that $\nrm{\ket{0^c \perp}}\leq 1$ and $(\bra{0^c}\otimes \vI)\cdot \ket{0^c \perp}=0$. Tracing out the first $a+1$ qubits, we get that the resulting state is $\bigO{\delta^2}$-close to the desired state. Indeed, let $\ket{\psi'}$ denote the final state above in \eqref{eq:weakMeasPsi}; we now show that
	\begin{align}
		\lnormp{(\mathcal{I}+\delta \CL)[\ketbra{\psi}{\psi}]-\tr_{c+1}\L[\ketbra{\psi'}{\psi'}\R]}{1}=\bigO{\delta^2}\label{eq:almostLin}
	\end{align}
	by observing that 
	\begin{align*}
		\tr_{c+1}[\ketbra{\psi'}{\psi'}]&
		= \tr_{c}\bigg[(\bra{0}\otimes \vI)\cdot \ketbra{\psi'}{\psi'}\cdot (\ket{0}\otimes \vI)\bigg] + \tr_{c}\bigg[(\bra{1}\otimes \vI)\cdot \ketbra{\psi'}{\psi'}\cdot (\ket{1}\otimes \vI)\bigg]\\&
		= (\bra{0^{c+1}}\otimes \vI)\cdot\ketbra{\psi'}{\psi'}\cdot(\ket{0^{c+1}}\otimes \vI) 
		+\tr_{c}\bigg[\L(\bra{0}\otimes \vI-\ketbra{0^c}{0^c}\R)\cdot\ketbra{\psi'}{\psi'}\cdot\L(\ket{0}\otimes \vI-\ketbra{0^c}{0^c}\R)\bigg]\\
		&+\delta\sum_{j\in J} \vL_j \ketbra{\psi}{\psi} \vL_j^\dagg\\&	
		=
		\L(\vI-\frac\delta2\sum_{j\in J}\vL_j^\dagg \vL_j+\bigO{\delta^2}\R)\ketbra{\psi}{\psi}\L(\vI-\frac\delta2\sum_{j\in J}\vL_j^\dagg \vL_j+\bigO{\delta^2}\R) \\
  &\quad+ \bigO{\delta^2} \tr_{c}\ketbra{0^c \perp}{0^c \perp}+\delta\sum_{j\in J} \vL_j \ketbra{\psi}{\psi} \vL_j^\dagg\\&	
		=\ketbra{\psi}{\psi} +\delta\sum_{j\in J} \vL_j \ketbra{\psi}{\psi}\vL_j^\dagg 
		- \frac\delta2 \Big\{\sum_{j\in J}\vL_j^\dagg \vL_j,\ketbra{\psi}{\psi}\Big\} + \bigO{\delta^2}\\&	
		=(\mathcal{I}+\delta \CL)[\ketbra{\psi}{\psi}] +\bigO{\delta^2}.
	\end{align*}
	Convexity implies \eqref{eq:almostLin} also holds for mixed input states. To extend to the diamond norm, observe that $\CL[\cdot]\otimes \vI[\cdot]\vI$ has Lindblad operators $\vL_j\otimes \vI$ and therefore $\vU\otimes \vI$ is a block-encoding for $\CL[\cdot]\otimes \vI[\cdot]\vI$.
	This implies that the trace-norm bound of \eqref{eq:almostLin} holds with respect to $\CL[\cdot]\otimes \vI[\cdot]\vI$ as well, and so we can conclude that
	\begin{align}
		\lnormp{(\mathcal{I}+\delta \CL)[\cdot]-\tr_{c+1}\vC\left[\ketbra{0^{c+1}}{0^{c+1}}\otimes \cdot\right]\vC^\dagg}{\Diamond}=\bigO{\delta^2} \label{eq:almostLinDiam}
	\end{align}
	The triangle inequality then implies that the implemented map is $\bigO{\delta^2}$-close in diamond distance to $\e^{\delta \CL}$,  since $\nrm{(\mathcal{I}+\delta \CL)-\e^{\delta \CL}}_\Diamond=\bigO{\delta^2}$ as shown by, e.g., \cite[Appendix B]{cleve2016EffLindbladianSim}.\footnote{Here we implicitly used the fact that a block-encoded \Lword{} has norm at most one. This follows from the observation that $\nrm{\sum_{j\in J}\vL_j^\dagg\vL_j}\leq 1$, which is a direct consequence of \autoref{prop:rejectBlockEncoded}.}
	
	Choosing $\delta=\Theta(\frac{\epsilon}{t})$ ensures that the error in a single time-step is bounded by $\bigO{\frac{\epsilon^2}{t^2}}$, and repeating the process $\Theta(\frac{t^2}{\epsilon})$-times induces an error that is bounded by $\epsilon$ for the entire time-$t$ evolution. The complexity is then $\Theta(\frac{t^2}{\epsilon})$-times the cost of implementing the circuit in \autoref{fig:weakMeasCircuit}.
\end{proof}
In addition to purely irreversible \Lword{}, as noted below (\autoref{fig:weakMeasCircuit}), the above weak measurement scheme can be amended with the Hamiltonian evolution term.

In our Gibbs sampling algorithm, the original random process typically proceeds by a random ``jump'' operator $\vA^a$ for a  uniformly random $a\in A$. We then obtain the final generators by applying the operator Fourier Transform to these ``jump'' operators. Na\"{i}vely applying our weak measurement scheme to such a \Lword{} would require us to use all the ``jump'' operators in each iteration. However, we show in the next corollary that it suffices to randomly pick a single ``jump'' operator in each iteration. In some situations, we could hope for further improvement by parallelization of these jumps if their operator Fourier Transform remains localized.

\begin{cor}[Improved randomized simulation for convex combinations of \Lword{}s]\label{cor:rndWeakMeasSim}
	Suppose that a purely irreversible \Lword{} $\CL[\cdot]=\sum_{i}p_i\CL_i[\cdot]$ is a convex combination of the purely irreversible \Lword{}s $\CL_i[\cdot]$ which are given by their respective block-encodings. In \autoref{thm:weakMeasSim} we can replace each weak-measurement gadget for $\CL[\cdot]$ by an independently sampled weak-measurement gadget for $\CL_i[\cdot]$ according to the distribution $p_i$ while keeping the same asymptotic iteration count  $\bigO{t^2/\epsilon}$.
\end{cor}
\begin{proof}
	It suffices to show that 
		\begin{align}
		\nrm{\e^{\delta \CL}[\cdot]-\sum_{i}p_i\tr_{c+1}\vC_i\L(\ketbra{0^{c+1}}{0^{c+1}}\otimes 		\left[\cdot\right]\R)\vC_i^\dagg}_\Diamond\!\!=\bigO{\delta^2}.\label{eq:almostLinRandom}
	\end{align}		
	From \eqref{eq:almostLinDiam} we know that the weak measurement gadget $\vC_i$ in \autoref{fig:weakMeasCircuit} for $\CL^{\dagger}_i[\cdot]$ satisfies 
	\begin{align}
		\nrm{(\mathcal{I}+\delta \CL_i)[\cdot]-\tr_{c+1}\vC_i\L(\ketbra{0^{c+1}}{0^{c+1}}\otimes \left[\cdot\right]\R)\vC_i^\dagg}_\Diamond=\bigO{\delta^2}.
	\end{align}
	By linearity and the triangle inequality, it follows that 
	\begin{align}
		\nrm{(\mathcal{I}+\delta \sum_{i}p_i\CL_i)[\cdot]-\sum_{i}p_i\tr_{c+1}\vC_i\L(\ketbra{0^{c+1}}{0^{c+1}}\otimes \left[\cdot\right]\R)\vC_i^\dagg}_\Diamond=\bigO{\delta^2}.
	\end{align}
	Since $\nrm{(\mathcal{I}+\delta \CL)-\e^{\delta \CL}}_\Diamond\!\!=\bigO{\delta^2}$, 
	by the triangle inequality, we get the sought inequality in \eqref{eq:almostLinRandom}. 
\end{proof}

Now, we turn to our second incoherent simulation result that is roughly based on the algorithm of~\cite{cleve2016EffLindbladianSim} but contains further improvements and fixes. We obtain improved complexity because we assume that the \Lword{} is provided via a block-encoding, while effectively~\cite{cleve2016EffLindbladianSim} construct a (potentially suboptimal) block-encoding within their algorithm. Their complexity depends on $\sum_{j\in J}\nrm{\vL_j^\dagg \vL_j}$, while our algorithm can in principle achieve a dependence like $\nrm{\sum_{j\in J}\vL_j^\dagg \vL_j}$ when an efficient block-encoding is provided -- which is the case for our explicit block-encodings outlined in the next section (\autoref{sec:block_encodings}). To our knowledge, this is the first \Lword{} simulation algorithm that achieves both near-linear time dependence and a complexity that scales with $\nrm{\sum_{j\in J}\vL_j^\dagg \vL_j}$. Note that this improvement looks similar to how \cite{apeldoorn2022QTomographyWStatePrepUnis} improved over \cite{huggins2021QAlgMultipleExpectationValues} on the complexity of estimating multiple expectation values, but the techniques are very different. Here, the improvement stems from the following efficient block-encoding construction. 

\begin{prop}\label{prop:rejectBlockEncoded}
	Given a block-encoding of a \Lword{} (\autoref{defn:blockLindladian}), we get a block-encoding of 
 \begin{align}
  \sum_{j\in J}\vL_j^\dagg  \vL_j \quad \text{via}\quad \vV:= (\vY_{\frac{1}{2}}\otimes \vU^\dagg)\cdot \L(2\ketbra{0^{b+1}}{0^{b+1}}\otimes\vI-\vI\R)\cdot(\vY_{\frac{1}{2}}\otimes \vU),
 \end{align}
where $\ket{\pm} := (\ket{0}\pm\ket{1})/\sqrt{2}$. 
\end{prop}
\begin{proof}
We calculate
\begin{align}
	(\bra{0^{c+1}}\otimes \vI)\cdot \vV\cdot (\ket{0^{c+1}}\otimes \vI)
	&=\bigg(\bra{-}\otimes (\bra{0^{c}}\otimes \vI)\vU^\dagg \bigg)\cdot\L(2\ketbra{0^{b+1}}{0^{b+1}}\otimes\vI-\vI\R) \cdot\bigg(\ket{+}\otimes \vU(\ket{0^{c}}\otimes \vI)\bigg)\\&
	=\bigg(\bra{-}\otimes (\bra{0^{c}}\otimes \vI)\vU^\dagg \bigg)\cdot\L(2\ketbra{0^{b+1}}{0^{b+1}}\otimes\vI\R) \cdot \bigg(\ket{+}\otimes \vU(\ket{0^{c}}\otimes \vI)\bigg)\\&
	=(\bra{0^{c}}\otimes \vI)\cdot \vU^\dagg\cdot (\ketbra{0^{b}}{0^{b}}\otimes\vI)\cdot \vU\cdot (\ket{0^{c}}\otimes \vI)\\&
	=\L(\sum_{j\in J}\bra{j}\otimes \vL_j^\dagg \R) \L(\sum_{j'\in J}\ket{j'}\otimes \vL_{j'}\R)
	=\sum_{j\in J}\vL_j^\dagg  \vL_j.\tag*{\qedhere}
\end{align}
\end{proof}
This block-encoding construction and the following generic \Lword{} simulation algorithm answers an open question\footnote{See~\cite[Section 7]{Rall_thermal_22} ``That one special Kraus operator involves all the $L_j$’s. Does there exist any special treatment of this special Kraus operator so that we can leverage the special structure of the oracle $\sum_j \ket{j}\otimes L_j$ to get rid of the $\bigO{m}$ dependence?''} recently posed by Rall, Wang, and Wocjan~\cite{Rall_thermal_22}, and can significantly improve their complexity.
Although we do not use the above block-encoding explicitly, this observation is implicitly used in our weak measurement schemes \autoref{fig:weakMeasCircuit}-\autoref{fig:postWeakMeasCircuit} that enable us to prove the following result, whose proof is presented in~\autoref{apx:improvLind}).
\begin{restatable}[Compressed incoherent \Lword{} simulation algorithm]{thm}{LCUSim}\label{thm:LCUSim}
	Suppose $\vU$ is a block-encoding of the Lindblad operators of a purely irreversible \Lword{arg1} $\CL$ as in \autoref{defn:blockLindladian}. Let $\epsilon\leq 1/2$, then we can simulate the action of the superoperator $\e^{t \CL}$ to precision $\epsilon$ in diamond norm using 
	\begin{align}
	&\bigO{\left(c+\log((t+1)/\epsilon)\right)\log((t+1)/\epsilon)}\quad&\text{ (resettable) ancilla qubits},\\ 
	&\bigO{(t+1)\frac{\log((t+1)/\epsilon)}{\log\log((t+1)/\epsilon)}} \quad &\text{(controlled) uses of $\vU$ and $\vU^\dagg$},\\
	\text{and}\quad &\bigO{(t+1)(c+1)\mathrm{polylog}((t+1)/\epsilon)}\quad &\text{ other two-qubit gates}.
	\end{align}
	If the \Lword{} has a coherent part $-\ri[\vH,\vrho]$, and we have access to a block-encoding of $\vH=(\bra{0^c}\otimes \vI)\vV(\ket{0^c}\otimes \vI)$, then we can simulate $\e^{t \CL}$ with $\bigO{(t+1)\frac{\log((t+1)/\epsilon)}{\log\log((t+1)/\epsilon)}}$ additional (controlled) uses of $\vV$ and $\vV^\dagg$.
\end{restatable}
Crucially, the complexity scales almost linearly with time $t$ and poly-logarithmic with the precision $\epsilon$ while using very few ancillas, representing a large asymptotic speedup compared to the $t^2/\epsilon$ complexity of the weak-measurement scheme (\autoref{thm:weakMeasSim}). 

\subsubsection{Coherent \Lword{} simulation algorithms}
\label{sec:algortihm_coherent}

With a quantum computer, we further ask for a \textit{coherent} Gibbs sampler that outputs the \textit{purified} distribution ($\ket{\sqrt{\vrho_{\beta}}} \propto \sum_i \e^{-\beta E_i/2} \ket{\psi_i} \otimes \ket{\psi_i^*}$) on two copies of the Hilbert space. A desirable coherent Gibbs sampler should satisfy the following conditions:
\begin{enumerate}
    \item A Hermitian operator $\vec{\CC}$ can be efficiently block-encoded on the duplicated Hilbert space,\label{item:coherent_block}
    \item its top-eigenvector is unique and yields the purified state $\ket{\sqrt{\vrho}}$, and\label{item:coherent_fixed}
    \item there exists an adiabatic path of operators $\vec{\CC}(s)$ whose top-eigenvalue-gap remains open. \label{item:coherent_gap}
\end{enumerate}

A general coherent Gibbs sampler may not refer to an existing \Lword{}. However, a natural candidate of the operator $\vec{\CC}$ is to take the \textit{vectorized} discriminant $\vec{\CD}(\vrho,\CL)^{\dagger}$ associated with a detailed balance \Lword{} (\autoref{defn:DB}), as how one quantizes classical Markov chains~\cite{temme2010chi}. Formally, we define \emph{vectorization} of a superoperator by\footnote{One might be tempted to use $\vB^\dagg$ instead in the vectorization, but that definition leads to inconsistencies. Indeed, if we would use $\vB^\dagg$ for vectorization then the two different representations of the scalar $1\otimes 1 = 1 = \ri \otimes -\ri$ would lead to different vectorizations $\pm 1$.}
\begin{align*}
\mathcal{C}[\cdot]=\sum_j  \alpha_j \vA_j[\cdot]\vB_j \rightarrow \vec{\mathcal{C}}=\sum_j \alpha_j \vA_j\otimes\vB^T_j \quad \text{(vectorization)},
\end{align*}
where $\vB^T_j$ denotes the transpose of the matrix $\vB_j$ in the computational basis $\ket{i}$. We use curly fonts $\CC$ for superoperators and bold fonts $\vec{\CC}$ for the vectorized superoperators (which is, a matrix).\footnote{Note that $\vec{\mathcal{C}}^\dagg$ is well defined. The (matrix) adjoint of the vectorized operator is $\sum_j \alpha_j^* \vA_j^\dagg\otimes\vB_j^{*}$. On the other hand, the superoperator adjoint $\mathcal{C}^\dagg[\cdot]$ is $\sum_j \alpha_j^* \vA_j^\dagg[\cdot]\vB_j^\dagg$~\eqref{eq:superoperator_adjoint}, whose vectorization is then the same $\sum_j \alpha_j^* \vA_j^\dagg[\cdot]\vB_j^\dagg$.\label{foot:wellDefinedAdjoint}} For a matrix $\vA$, let us denote its vectorized (or purified) version by 
\begin{align}
\ket{\vA}:= (I \otimes T^{-1})\vA \quad \text{(purification)}
\end{align}
using the ``transpose'' map $T \ket{i}=\bra{i}$. This automatically ensures the correctness of the fixed point (Point~\ref{item:coherent_fixed}). 
\begin{prop}
    For any full-rank state $\vrho$ and any \Lword{} $\CL$, we have that $\vec{\CD}(\vrho,\CL)^{\dagger}\ket{\sqrt{\vrho}}=0$. Further, if $\CL$ satisfies $\vrho$-detailed balance, we also have that $\vec{\CD}(\vrho,\CL) \ket{\sqrt{\vrho}}=0$.
\end{prop}
The above follows from a direct calculation using that any \Lword{} is trace-preserving $\CL^{\dagger}[\vI] = 0$. However, to turn the above into the advertised efficient algorithm (\autoref{thm:D_correct}), we need two key components reminiscent of the incoherent case. First, we need a proxy for the discriminant with efficient block-encoding. As we mentioned~\eqref{eq:D_main}, we consider 
\begin{align}\label{eq:vec_main_result}
	\vec{\CD}_{\beta}&=  \sum_{ a\in A, \bomega\in S_{\omega_0}} \sqrt{\gamma(\bomega)\gamma(-\bomega)}  \hat{\vA}^a(\bomega) \otimes  \hat{\vA}^{a}(\bomega)^{*} - \frac{\gamma(\bomega)}{2} \L(\hat{\vA}^{a}(\bomega)^\dagg\hat{\vA}^a(\bomega)\otimes \vI + \vI \otimes \hat{\vA}^{a}(\bomega)^{*\dagger}\hat{\vA}^{a}(\bomega)^{*}\R)\\ 
 & \text{as a proxy for}\quad \vec{\CD}(\vrho,\CL_{\beta})^{\dagger},
\end{align}
where the $\hat{\vA}^a(\bomega)$ are the same operator Fourier Transforms as in the incoherent case~\eqref{eq:Abomegadef}; the block-encoding for $\vec{\CD}$ can be obtained given the block-encoding for $\hat{\vA}^a(\bomega)$ (\autoref{sec:D_circuit}). The map is self-adjoint $\vec{\CD}_{\beta} = \vec{\CD}_{\beta}^{\dag}$ due to Hermiticity $\vA^a = \vA^{a\dagger}$ and properties of weighted Fourier Transform (\autoref{sec:operator_FT}). More carefully, we do not implement exactly the discriminant $\vec{\CD}_{\beta}$, but merely an approximation $\vec{\CD}_{\rm impl}\approx \vec{\CD}_{\beta}$ due to additional implementation errors for the Gaussian weight and truncation errors for the Gaussian tail. 

Second, we need to formulate a notion of approximate detailed balance for the above discriminant proxy.
\begin{defn}[$\epsilon$-Discriminant proxy]\label{defn:discriminant_proxy}
    We say a Hermitian matrix $\vec{\CD}$ is an $\epsilon$-discriminant proxy for \Lword{} $\CL$ and a full-rank state $\vrho$ if
    \begin{align}
        \norm{\vec{\CD} - \vec{\CD}(\vrho,\CL)^{\dagger}}\le \epsilon.
    \end{align}
\end{defn}
Indeed, this implies approximate detailed balance $\CD(\vrho,\CL)^{\dagger} \approx \CD(\vrho,\CL)$ for the \Lword{} $\CL$ by taking the adjoints. We can think of~\autoref{defn:discriminant_proxy} as a different form of the approximate detailed balance condition (Definition~\ref{defn:ADB}) that controls the top eigenvector error up to the spectral gap (Point~\ref{item:coherent_fixed}).
\begin{restatable}[Fixed point error]{prop}{Dfixedpointerror}\label{prop:Dfixed_point_error}
Suppose a gapped Hermitian operator $\vec{\CD}$ is an $\epsilon$-discriminant proxy for a \Lword{} $\CL$ and a full-rank state $\vrho$. Then, its top eigenvector is approximately the purified state $\ket{\sqrt{\vrho}}$
\begin{align}
    \lnorm{\ket{\lambda_1(\vec{\CD})} - \ket{\sqrt{\vrho}}} \le 4\sqrt{2}\frac{\epsilon}{\lambda_{gap}(\vec{\CD})}.
\end{align}
\end{restatable}

Here, the gap dependence naturally arises from eigenvalue (\autoref{cor:eigenvalue_perturb}) and eigenvector perturbation (\autoref{prop:eigenvector_perturb_gen}) arguments. Unlike \Lword{}s, the cost for preparing the coherent Gibbs state scales directly with the gap via quantum simulated annealing; the mixing time of the original \Lword{} is not linked directly to the algorithmic cost. 

Thirdly, the block-encoding by itself does not prepare the desired top eigenvector, unlike a \Lword{}; this additionally requires a standard subroutine called \textit{quantum simulated annealing}~\cite{Wocjan_2008_quantum_sampling,yung2010QuantumQuantumMetropolis}: adiabatically change the inverse temperature from $\beta' = 0 \rightarrow \beta' = \beta$. The algorithmic cost is associated with the gaps along the adiabatic path~\cite{boixo2010QAlgTraversingEigStatePaths}; see~\autoref{sec:simulated_annealing}.

\subsection{Explicit block-encodings}\label{sec:block_encodings}
In this section, we lay out the circuit ingredients to construct the advertised \Lword{}s $\CL_\beta$ and discriminants~$\vec{\CD}_\beta$. 
First, we show how to construct a block-encoding of the discretized \Lword{} with Lindblad operators \eqref{eq:DiscAMain} from a block-encoding of the jump operators $\vA^a$.
Then, we further construct a block-encoding of the corresponding discriminant proxy - with the additional assumption that the set of jump operators is self-adjoint ($\{\vA^a\colon a\in A\}=\{\vA^{a\dagg}\colon a\in A\}$) and the Fourier weight function $f$ is real. It is not surprising that implementing the discriminant proxy requires more symmetry constraints, as its definition already implicitly draws from these symmetries.

\subsubsection{Block-encoding $\CL_\beta$}
\label{sec:L_circuit}
For both incoherent algorithms (\autoref{thm:weakMeasSim},\autoref{thm:LCUSim}), we have assumed that a purely-irreversible \Lword{} $\CL[\cdot]=\sum_{j\in J} \vL_j[\cdot]\vL_j^\dagg + \{\vL_j^\dagg \vL_j,\cdot\}$ is given by a unitary block-encoding $\vU$ (\autoref{defn:blockLindladian}). Here, we construct a block-encoding unitary $\vU$ for the advertised Lindblad operators $\sqrt{\gamma(\bomega)}\hat{\vA}^a(\bomega)$ labeled by $a, \bomega$:
\begin{align}
    \sum_{a \in A,\bomega \in S_{\omega_0}}\sqrt{\gamma(\bomega)} \ket{\bomega} \otimes \ket{a}\otimes \hat{\vA}^a(\bomega)
\end{align}
for the advertised \Lword{} (discretization of~\eqref{eq:mainL}, recap of~\eqref{eq:LbetaAgain}) .
 \begin{align} 
 	\CL_{\beta}:=\sum_{a\in A, \bomega\in S_{\omega_0}} \gamma(\bomega) \L( \hat{\vA}^{a}(\bomega)[\cdot] \hat{\vA}^{a}(\bomega)^{\dagger} -\frac{1}{2} \{\hat{\vA}^{a}(\bomega)^{\dagger} \hat{\vA}^{a}(\bomega) ,\cdot \}\R) .
 \end{align}

We begin by laying out the registers explicitly, including the additional ancillae for block-encoding.
\begin{align}
\text{registers:}\quad \underset{\text{Boltz. weight}}{\underbrace{\ket{0}}}\otimes \underset{\text{Bohr freq.}}{\underbrace{\ket{\bomega}}} \otimes \underset{\text{block. enc. anc.}}{\underbrace{\ket{0^b}}} \otimes\underset{\text{jump}}{\underbrace{\ket{a}}}\otimes \underset{\text{system}}{\underbrace{\hat{\vA}^a(\bomega)}}
\end{align}

From right to left, the registers individually correspond to: the physical system of interest; the jump labels $\ket{a}$ and additional ancillae to accommodate block-encoding access (indicating successful application by the all-zero state $\ket{0^b}$); the frequency register is dedicated to the operator Fourier Transform, storing the weight $\ket{f}$ or the Bohr frequencies $\ket{\bomega}$; finally, an ancilla qubit for storing the Bohr-frequency dependent Boltzmann weights in the amplitudes $\sqrt{\gamma(\bomega)}\ket{0}+ \sqrt{1-\gamma(\bomega)}\ket{1}$.

We specify the discrete Fourier Transform parameters that determine the dimension of the Bohr frequency register as follows. The Fourier frequencies $\bomega$ and times $\bt$ are integer multiples of $\omega_0$ and $t_0$ respectively such that 
\begin{align} 
    \omega_0 t_0  = \frac{2\pi}{N}, \quad \text{and} \quad 
    S^{\lceil N \rfloor} &:= \bigg\{ -\left\lceil(N-1)/2\right\rceil, \ldots, -1,0,1,\ldots, \left\lfloor(N-1)/2\right\rfloor \bigg\},\\ 
    \text{and}\quad S^{\lceil N \rfloor}_{\omega_0}&:= \omega_0 \cdot S^{\lceil N \rfloor}, \quad S^{\lceil N \rfloor}_{t_0}:= t_0 \cdot S^{\lceil N \rfloor}.
\end{align} 
We use a ``bar'' to denote discretized variables; the (Bohr) frequency register takes values $\bomega \in S_{\omega_0}$. To implement the Fourier Transform when $N=2^n$, we specify the signed binary representation for the integers $S^{\lceil N \rfloor}$ as follows:
\begin{align}
    10^{n-1},\ldots, 1^{n} \quad &\text{for each}\quad -N/2, \ldots, -1,\\
    0^n,\ldots, 01^{n-1} \quad &\text{for each}\quad  0, 1, \ldots, N/2-1.
\end{align}
Where it does not cause confusion we will drop $\lceil N \rfloor$ from the superscript and will simply write $\bomega \in S_{\omega_0}$ and $\bt \in S_{t_0}$.
 We will set the value of $N$ and $\omega_0$ such that the Bohr frequencies $B$ = spec$(\vH)$ $-$ spec$(\vH)$ are contained within the range of energies
\begin{align}
\norm{\vH} \leq \frac{N}{2}\omega_0.
\end{align}

 The only ``physical'' energy scale in the above is the Hamiltonian strength $\norm{\vH}$. Indeed, setting the scale to $N$ requires only $\log(N)$ qubits in the readout register.
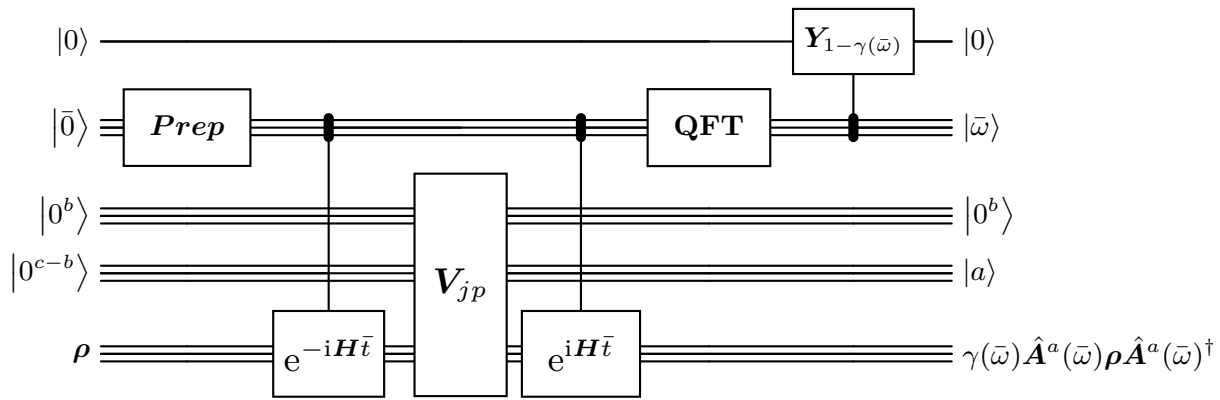
\begin{figure}[t]

\begin{center}
	\newcommand{\scalea}{1.2}
	\newcommand{\scaleb}{1.5}
	\begin{quantikz}[wire types={q,b,b,b,b},classical	gap=1mm]
	 	\lstick{\scalebox{\scalea}{$\ket{0}$}}		&\qw &	\qw		&\qw	&\qw				&\qw &\gate[style={inner xsep=0mm, inner ysep=1mm}]{\scalebox{\scalea}{$\vY_{1-\gamma(\bomega)}$}}&\qw \rstick{\scalebox{\scalea}{$\ket{0}$}}\\		
	 	\lstick{\scalebox{\scalea}{$\ket{\bar{0}}$}}&\gate[style={inner xsep=2mm, inner ysep=2mm}]{\scalebox{\scalea}{$\vec{Prep}$}} \qw&\ctrl{3}	&\qw	&\ctrl{3}			&\gate[style={inner xsep=2mm, inner ysep=2mm}]{\scalebox{\scalea}{\textbf{QFT}}}\qw & \ctrl{-1}&\qw \rstick{\scalebox{\scalea}{$\ket{\bomega}$}}\\
	 	\lstick{\scalebox{\scalea}{$\ket{0^b}$}}		&\qw&\qw	&\gate[3,style={inner xsep=1mm, inner ysep=2mm}]{\scalebox{\scaleb}{$\vV_{jp}$}}	&\qw	&\qw 	&\qw	&\rstick{\scalebox{\scalea}{$\ket{0^b}$}} \qw \\[-1mm]
	 	\lstick{\kern-2mm\scalebox{\scalea}{$\ket{0^{c-b}}$}}	&\qw&\qw	& 					&\qw		&\qw 			& \qw & \qw\rstick{\scalebox{\scalea}{$\ket{a}$}} \\[2mm]
	 	\lstick{\scalebox{\scalea}{$\vrho$}}		&\qw&\gate[style={inner xsep=0mm, inner ysep=2mm}]{\scalebox{\scaleb}{$\e^{-\ri\vH \bar{t}}$}}\qw	&	\qw	&\gate[style={inner xsep=2mm, inner ysep=2mm}]{\scalebox{\scaleb}{$\e^{\ri\vH \bar{t}}$}}\qw		&\qw		&\qw		& \qw \rstick{\scalebox{\scalea}{$\gamma(\bomega)\hat{\vA}^a(\bomega)\vrho\hat{\vA}^a(\bomega)^\dagg$}}
	 \end{quantikz}
\end{center} 
  \caption{Circuit $\vU$ for block-encoding the \Lword{}. Practically, if we use the simpler weak-measurement-based simulation (\autoref{thm:weakMeasSim}), then by~\autoref{cor:rndWeakMeasSim}, we can use a single randomly chosen Lindblad operator $\vA^a$ at a time. Moreover, if $\vA^a$ is unitary, we can simply replace $\vV_{jump}$ with $\vA^a$, implying $b=c=0$, i.e., the third and the forth registers can be omitted, thus $n+\lceil\log(N)\rceil+2$ qubits suffice to simulate the \Lword{} $\e^{\CL t }$}.\label{fig:L_circuit}
\end{figure}
The circuit consists of the following ingredients, each acting on some appropriate subset of the registers:
\begin{itemize}
	\item Block-encoding $\vV_{jump}$ of the jump operators $\vA^a$ in the form of~\autoref{defn:blockLindladian}:
	\begin{align}\label{eq:blockJumps}
		(\bra{0^b}\otimes \vI_{a}\otimes \vI_{sys})]\cdot\vV_{jump}\cdot (\ket{0^c}\otimes \vI_{sys})=\sum_{a\in A} \ket{a} \otimes \vA^a.
	\end{align}
The operators $\vA^a$ need not be self-adjoint nor proportional to a unitary. Still, one may conveniently choose $\sqrt{\labs{A}} \vA^a$ to be unitary for all $a \in A$, (e.g., few-body unitary operators). Then, we can set $b=0$ and choose 
	\begin{align}\label{eq:exampleJumps}
	\vV_{jump} = \L(\sum_{a \in A} \ketbra{a}{a} \otimes \vA^a \R)\cdot ( \vB \otimes \vI_{sys})\quad \text{where}
		\quad \vB\ket{0^c} = \sum_{a \in A}\frac{\ket{a}}{\sqrt{\labs{A}}}. 
	\end{align}
	Note that implementing the \Lword{} does not require the set of jump operators to contain the adjoints $\{\vA^a\colon a\in A\}=\{\vA^{a\dagg}\colon a\in A\}$;
	this assumption is only used for approximate detailed balance and the fixed point correctness (\autoref{thm:L_correctness}).
    \item Controlled Hamiltonian simulation
    \begin{align}
	\sum_{\bt \in S_{t_0}}\ketbra{\bt}{\bt}\otimes e^{\pm \ri \bt \vH}.
\end{align}
    \item Quantum Fourier Transform
    \begin{align}
        \vec{QFT}_N: \ket{\bt} \rightarrow \frac{1}{\sqrt{N}} \sum_{\bomega \in S_{\omega_0}} \e^{-\ri \bomega \bt}\ket{\bomega}.
    \end{align}
        \item State preparation unitary for the Fourier Transform weights, acting on the frequency register
        \begin{align}
            \vec{Prep}_{f}\quad \text{such that}\quad \vec{Prep}_{f}\ket{\bar{0}} = \ket{f}. 
        \end{align}
        Naturally, the weight $f(\bt)$ as amplitudes of 
 a state is normalized
	\begin{align}
	\sum_{\bt \in S_{t_0}} \labs{f(\bt)}^2 = 1.
	\end{align}
        It could be, e.g., an easily preparable step function or a Gaussian whose tail decays rapidly. Gaussian states are attractive because they are relatively easy to prepare~\cite{McArdle_2022quantumstate}, but as a matter of fact, any other so-called \emph{window function} could be used, such as the Kaiser-window~\cite{berry2022QuantTopologicalData,McArdle_2022quantumstate} potentially providing further overhead improvements. 
	\item Controlled filter for the Boltzmann factors acting on the frequency register and the Boltzmann weight register
	\begin{align}
	\vW := \sum_{\bomega\in S_{\omega_0}}  \vY_{1-\gamma(\bomega)}\otimes \ketbra{\bomega}{\bomega} 
	\quad &\text{where} \quad 0 \le \gamma(\bomega) \le 1 \quad\text{and}\quad \gamma(\bomega) = \gamma(-\bomega)\e^{-\beta\bomega}. 
	\end{align}
	The constraint $0\le\gamma(\bomega)\le 1$ ensures the matrix $\vY_{1-\gamma(\bomega)}$ is unitary; the symmetry (i.e., the KMS condition) $\gamma(\bomega) = \gamma(-\bomega)\e^{-\beta\bomega}$ gives lower weights for ``heating'' transitions and is closely related to the detailed balance condition. Important examples of weight functions are 
 \begin{align}
     \text{(Metropolis)}\quad\gamma(\bomega)=\min(1, \e^{-\beta \bomega})\quad \text{and}\quad\text{(Glauber)}\quad\gamma(\bomega) = \frac{1}{\e^{\beta \bomega}+1},
 \end{align}
 which both reduce to the step function in the $\beta \rightarrow \infty$ limit. Note that the range of energy labels $\bomega$ is finite; we choose a large energy readout range $N \omega_0 \ge 4\norm{\vH}+ \frac{2}{\beta}$ to ensure all possible transitions are covered by the discretization range (after secular approximation $\mu \le \frac{1}{\beta}$). We may generally synthesize the controlled filter from elementary gates at cost\footnote{By first applying controlled Hamiltonian simulation for $\vY$ rotation $\sum_{\bomega} \e^{\ri (\bomega/2\norm{\vH})\vY} \otimes \ketbra{\bomega}{\bomega}$ and then apply QSVT to map $\bomega/2\norm{\vH} \rightarrow \arcsin(\sqrt{1-\gamma(\bomega)})$. This is reminiscent of~\cite{haah2018ProdDecPerFuncQSignPRoc}.}
\begin{align}
    \text{ (polynomial degree of $\gamma$)} \times \poly (\log (N), \log(1/\epsilon)), 
\end{align}
which is $\tCO(1+\beta \norm{\vH})$ for the Glauber weight. For the Metropolis weight, one can achieve the same scaling by manually switching between $1$ and $\e^{-\beta \bomega}$ at $\bomega =0$.  In principle, since we are merely controlling a qubit, we may directly implement any efficient computable function (perhaps with terrible overhead).
\end{itemize}

Further, combining the controlled Hamiltonian simulation and Quantum Fourier Transform yields the advertised operator Fourier Transform (\autoref{fig:OQFT}) acting on the frequency and system register \footnote{If the operator $\vO$ maps between larger Hilbert spaces than $\vH$, we formally extend its action trivially so that $\vO(t)= (\vI\otimes\e^{\ri \vH t}) \vO (\vI'\otimes\e^{-\ri \vH t})$. Of course, this formal extension does not incur any additional cost.\label{foot:OFTExtension}}
	\begin{align}
	\CF[\cdot]: \ket{f(\bt)}\otimes \vO \rightarrow \sum_{\bomega \in S_{\omega_0}}\ket{\bomega} \otimes \hat{\vO}_{f}(\bomega) \quad \text{where} \quad \hat{\vO}_{f}(\bomega): = \frac{1}{\sqrt{N}}\sum_{\bt \in S_{t_0}} \e^{-\ri \bomega \bt}f(\bt)\vO(\bt)  \quad \text{and} \quad \vO(t):= \e^{\ri \vH t} \vO \e^{-\ri \vH t}.
	\end{align}
	See \autoref{sec:operator_FT} for basic properties of the operator Fourier Transform. Our implementation is inspired by, but differs from \cite{wocjan2021szegedy}; they sandwich the jump operators with phase estimation and its inverse. The operator picture, inspired by physics, is more natural and tangible. Our construction allows for flexibility in the choice of the weight function $f(\bt)$. 

As shown in \autoref{fig:L_circuit}, we assemble the above ingredients to obtain the unitary $\vU$ such that 
\begin{align}
	\vU(\vI \otimes \ket{\bar{0}}\otimes \vI_{\vV})& = (\vW\otimes \vI_{\vV}) \cdot \L(\vI \otimes \CF[\vec{Prep}_{f}\ket{\bar{0}}\otimes \vV_{jump}]  \R),
 \label{eq:unitary_for_L_block_encoding}
\end{align} 
where $\vI$ is the single qubit identity and $\vI_{\vV}$ is the identity on the registers on which $\vV_{jump}$ acts. As described in Footnote~\ref{foot:OFTExtension}, intuitively speaking, we \textit{only} apply the operator Fourier Transform on the ``system'' register; indeed, in the corresponding circuit (\autoref{fig:OQFT}), the jump label register $\ket{a}$ and the block-encoding ancillae~$\ket{0^b}$ is only affected by $\vV_{jump}$.
\begin{figure}[t]
	\begin{quantikz}[wire types={b,b},classical	gap=1mm]
	 	\lstick{\scalebox{1.2}{\ket{\bar{0}}}}	&\gate[style={inner xsep=2mm, inner ysep=2mm}]{\scalebox{1.2}{$\vec{Prep}$}}&\ctrl{1}	&\qw			&\ctrl{1}				&\gate[style={inner xsep=2mm, inner ysep=2mm}]{\scalebox{1.2}{$\vec{QFT}$}} & \qw \rstick{\scalebox{1.2}{$\ket{\bomega}$}}\\
	 	\lstick{\scalebox{1.2}{$\vrho$}}		&\qw&\gate[style={inner xsep=0mm, inner ysep=2mm}]{\scalebox{1.5}{$\e^{-\ri\vH \bar{t}}$}} &\gate[style={inner xsep=2mm, inner ysep=2.25mm}]{\scalebox{1.5}{$\vO$}}&\gate[style={inner xsep=1mm, inner ysep=2mm}]{\scalebox{1.5}{$\e^{\ri\vH \bar{t}}$}}	&\qw		& \qw \rstick{\scalebox{1.2}{$\vO_{\bomega}\vrho(\vO_{\bomega})^\dagg$}}
	 \end{quantikz}
  \caption{Circuit for operator Fourier Transform $\CF$ for an operator $\vO$ acting on the system $\vrho$. Of course, in our use, the operator may also act nontrivially on other ancillas.
  }\label{fig:OQFT}
\end{figure}
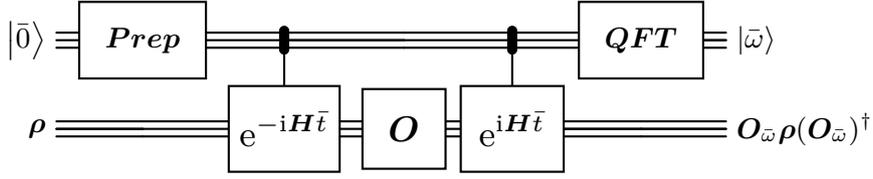

\begin{lem}[Explicit Block-encoding]\label{lem:block_encoding_L}
The untiary $\vU$ in~\eqref{eq:unitary_for_L_block_encoding} gives a block-encoding for our \Lword{}~\eqref{eq:mainL}
    \begin{align}
        \L(\bra{0}\otimes\vI_{\bomega}\otimes \bra{0^b}\otimes \vI_{a} \otimes \vI_{sys}\R)\cdot \vU\cdot \L(\ket{0}\otimes\ket{\bar{0}}\otimes\ket{0^c}\otimes \vI_{sys}\R) =& \sum_{a \in A,\bomega \in S_{\omega_0}}\sqrt{\gamma(\bomega)} \ket{\bomega} \otimes \ket{a}\otimes \hat{\vA}_{f}^a(\bomega).\label{eq:blockEncodedLindblad}
    \end{align}
\end{lem}
\begin{proof}
\begin{align}
	&(\text{LHS of }\eqref{eq:blockEncodedLindblad})\\ 
 &= \L(\bra{0}\otimes\vI_{\bomega}\otimes \bra{0^b}\otimes \vI_{a} \otimes \vI_{sys}\R)\cdot (\vW\otimes \vI_{\vV}) \cdot \L(\ket{0}\otimes \CF[\ket{f}\otimes\vV_{jump}] \L(\ket{0^c}\otimes \vI_{sys}\R)  \R) \tag*{(by \eqref{eq:unitary_for_L_block_encoding})}\\&
	=((\bra{0}\otimes\vI_{\bomega})\vW(\ket{0}\otimes\vI_{\bomega})\otimes \vI_{a} \otimes \vI_{sys}) \cdot \L(\L(\vI_{\bomega}\otimes\bra{0^b}\otimes \vI_{a} \otimes \vI_{sys}\R) \CF[\ket{f}\otimes\vV_{jump}] \L(\ket{0^c}\otimes \vI_{sys}\R)  \R)\\&	
	=\sum_{a\in A, \bomega \in S_{\omega_0}}(\bra{0} \vY_{1-\gamma(\bomega)} \ket{0}\otimes\ketbra{\bomega}{\bomega}\otimes \vI_{a} \otimes \vI_{sys}) \cdot \L(\CF[\ket{f}\otimes\ket{a} \otimes \vA^a]\R)\tag*{(by \eqref{eq:blockJumps} and \autoref{fig:OQFT})}\\&
	= \sum_{a \in A,\bomega \in S_{\omega_0}}\sqrt{\gamma(\bomega)} \ket{\bomega} \otimes \ket{a}\otimes \hat{\vA}_{f}^a(\bomega).\tag*{(by OFT and controlled filter)\qedhere}	
\end{align}
\end{proof}

\subsubsection{Block-encoding $\vec{\CD}_{\beta}$}
\label{sec:D_circuit}
We now describe the explicit and efficient circuit that implements the advertised vectorized discriminant proxy (discretization of~\eqref{eq:D_main}, recap of~\eqref{eq:vec_main_result})
\begin{align}
	\vec{\CD}_{\beta}&= \sum_{a \in A, \bomega\in S_{\omega_0}} \sqrt{\gamma(\bomega)\gamma(-\bomega)}  \hat{\vA}^a(\bomega) \otimes  \hat{\vA}^{a}(\bomega)^{*} - \frac{\gamma(\bomega)}{2} \L(\hat{\vA}^{a}(\bomega)^\dagg\hat{\vA}^a(\bomega)\otimes \vI + \vI \otimes \hat{\vA}^{a}(\bomega)^{*\dagger}\hat{\vA}^{a}(\bomega)^{*}\R), 
\end{align}
assuming that the set of jump operators is self-adjoint in the sense that $\{\vA^a\colon a\in A\}=\{\vA^{a\dagg}\colon a\in A\}$ and the Fourier weight function $f$ is real.
Combining the circuit with simulated annealing then leads to the advertised quadratic speedup.

Our discriminant proxy and its block-encoding is an instantiation of the following general construct that is self-adjoint as a superoperator (and hence Hermitian after vectorization). 

\begin{prop}[Self-adjoint discriminant proxies\footnote{If we additionally have $\sum_{j\in J}\frac12\vL_j\otimes \vL_{j'}^{*\dagg} + \frac12\vL_{j'}^{\dagg}\otimes \vL_{j}^{*}\approx\sum_{j\in J}\L(\vrho^{-\frac14}\otimes (\vrho^{*})^{-\frac14}\R)\vL_j\otimes \vL_{j}^{*}\L(\vrho^{\frac14}\otimes (\vrho^{*})^{\frac14}\R)$ and $\sum_{j\in J}\vL_j^\dagg \vL_j\approx\sum_{j\in J} \vrho^{\frac14}\vL_j^\dagg \vL_j\vrho^{-\frac14}$, then the resulting discriminant proxy $\vec{\CD}$ is close to the discriminant $\vec{\CD}(\vrho,\CL)^{\dagger}$. This is exactly what we show in \ref{sec:secular} for the discriminant proxy \eqref{eq:vec_main_result}.}]\label{prop:genDiscProxy}
	Given a purely irreversible \Lword{} with Lindblad operators $\vL_j$ for $j\in J$, and a permutation $\vP\colon j\rightarrow j'$ on the set $J$, the following superoperator  (and hence its vectorization)
	is self-adjoint:
	 \begin{align}\label{eq:generalDiscProxy}
		\CD(\vP,\{\vL_j\})&:=\frac{1}{2}\sum_{j\in J} \vL_j[\cdot]\vL_{j'}^{\dagg} +\vL_{j}^{\dagg}[\cdot]\vL_{j'}-\{\vL_j^\dagg \vL_j, \cdot\},\\ 
		\vec{\CD} (\vP,\{\vL_j\})&=\frac{1}{2}\sum_{j\in J} \vL_j\otimes \vL_{j'}^{*\dagg} +\vL_{j}^{\dagg}\otimes \vL_{j'}^{*}- \vL_j^\dagg \vL_j \otimes \vI-\vI\otimes\vL_j^{*\dagg} \vL_j^*.\label{eq:generalDiscProxyVec}
	\end{align}
\end{prop}
\begin{proof}
 \begin{align}
 	\phantom{f}\\[-13mm]
      \CD (\vP,\{\vL_j\})^{\dagger}&=\frac12\sum_{j\in J} \vL_{j}^{\dagg}[\cdot]\vL_{j'} +\vL_{j}[\cdot]\vL_{j'}^{\dagg}  - \{\vL_j^\dagg \vL_j, \cdot\} = \CD (\vP,\{\vL_j\}).\tag*{\qedhere}
 \end{align}
 The vectorization is hence also self-adjoint.
	\end{proof}
\begin{cor}\label{cor:DiscProxySelfAdjoint}
    If $f$ is real~\eqref{eq:fnormalized} and the set of jump operators is self-adjoint~\eqref{eq:AAdagger}, then the discriminant proxy \eqref{eq:vec_main_result} is Hermitian. 
\end{cor}
\begin{proof}
Due to \autoref{foot:wellDefinedAdjoint} it suffices to verify that the superoperator is self-adjoint:
This follows from~\autoref{prop:genDiscProxy} by setting the permutation $\vP\colon (\bomega,a)\rightarrow  (-\bomega,a')$ such that $\vA^{a\dagg}=\vA^{a'}$ and using the operator Fourier Transform property
	$\hat{\vA}^{a'}(-\bomega)^{\dagger} = \hat{\vA^{a'\dagger}}(\bomega)$ for real weight $f$ \eqref{eq:OpFTDaggProperty}, implying that $\hat{\vA}^{a}(\bomega) \otimes \hat{\vA}^{a'}(-\bomega)^{\dagger*}=\hat{\vA}^{a}(\bomega) \otimes \hat{\vA}^{a}(\bomega)^{*}=\hat{\vA}^{a'}(-\bomega)^{\dagger}\otimes \hat{\vA}^{a}(\bomega)^{*}$.
\end{proof}
Now that we have verified the symmetries of the desired discriminant proxy, we move on to our explicit construction. By the standard quantum walk recipe~\cite{szegedy2004QMarkovChainSearch}, we design an isometry and a reflection such that
\begin{align}\label{eq:VecDiscrImp}
	\vI+ \vec{\CD} (\vP,\{\vL_j\})= \vT^{'\dagger} \vR \vT',
\end{align}
which is block-encoded as in Figure~\ref{fig:D_circuit}.
\begin{prop}[A block-encoding for discriminant proxies]\label{prop:discriminantBlock}
	Using the notation of \autoref{prop:genDiscProxy}, let
 	\begin{align}
		\vR := \vI - (\vI \otimes \vPi) + \underset{=:\vR_0}{\underbrace{\vZ \otimes |0^{b'}\rangle\!\langle 0^{b'}| \otimes \vP \otimes \vI_{sys} \otimes \vI_{sys'}}}  \quad \text{where}\quad \vPi:=  |0^{b'}\rangle\!\langle 0^{b'}| \otimes \vI_J \otimes\vI_{sys} \otimes \vI_{sys'},
	\end{align}
	and $\vZ$ is the Pauli-Z operator such that $\vZ \ket{\pm} = \ket{\mp}$ for $\ket{\pm}:=(\ket{0}\pm\ket{1})/\sqrt{2}$. If $\vU$ is a unitary block-encoding of the \Lword{} such that
	\begin{align}
		\L(\langle 0^{b'}|\otimes  \vI_J \otimes \vI_{sys} \R)\cdot \vU\cdot \L(|0^{c'}\rangle\otimes \vI_{sys}\R) =& \sum_{j\in J}\ket{j}\otimes \vL_j,\\[-10mm]
	\end{align}
then we obtain a block encoding for the (shifted) discriminant proxy
	\begin{align}
		\L(\langle 0^{c'+1}|\otimes \vI_{sys} \otimes \vI_{sys'}\R)\cdot \vU_{\vec{\CD} (\vP,\{\vL_j\})}\cdot \L(|0^{c'+1}\rangle\otimes \vI_{sys} \otimes \vI_{sys'}\R) = \vI + \vec{\CD} (\vP,\{\vL_j\}).
	\end{align}
using 
 \begin{align}
		\vU_{\vec{\CD} (\vP,\{\vL_j\})}:=\vU^{'\dagger} \cdot \vR \cdot \vU'\quad \text{where}\quad \vU'= \bigg(\ketbra{+}{+}\otimes\vU\otimes\vI_{sys'} + \ketbra{-}{-}\otimes\vI_{sys}\otimes\vU^*\bigg)
	\end{align}
and the unitary $\vU^*$ is the conjugate of $\vU$ but acting on a copy of the system register $(sys')$.
\end{prop}
\begin{proof}
	Consider the isometries $\vT:=\vU(|0^{c'}\rangle\otimes \vI_{sys})$ and $\vT^*:=\vU^*(|0^{c'}\rangle\otimes \vI_{sys'})$, and 
	\begin{align}
		\vT' &:=\vU'\L(|0^{c'+1}\rangle\otimes \vI_{sys}\otimes \vI_{sys'}\R)= \frac{1}{\sqrt{2}} \big(\ket{+} \otimes \vT \otimes \vI_{sys'} +\ket{-} \otimes\vI_{sys}\otimes \vT^* \big). 
	\end{align} 
	To understand the product $\vT^{'\dagger} \vR \vT'$, we first calculate $(\vI \otimes \vPi)\vT' $ and $\vR_0\vT'$:
	\begin{align}
		(\vI \otimes \vPi)\vT' &=\frac{1}{\sqrt{2}}\sum_{j\in J}\ket{+} \otimes |0^{b'}\rangle \otimes \ket{j} \otimes \vL_j\otimes \vI_{sys'}\\&
		+ \frac{1}{\sqrt{2}}\sum_{j\in J}\ket{-} \otimes |0^{b'}\rangle \otimes \ket{j} \otimes \vI_{sys}\otimes \vL^*_j,\label{eq:PiTPrime}\\	
		\vR_0\vT' &=\frac{1}{\sqrt{2}}\sum_{j\in J} \ket{-} \otimes |0^{b'}\rangle \otimes \ket{j'} \otimes \vL_j\otimes \vI_{sys'}\\&
		+ \frac{1}{\sqrt{2}}\sum_{j\in J}\ket{+} \otimes |0^{b'}\rangle \otimes \ket{j'} \otimes \vI_{sys}\otimes \vL^{*}_j,\label{eq:XTPrime}
	\end{align}
	using the bit-flip $\vZ$ ($+\leftrightarrow -$), and the permutation $\vP$ ($j \leftrightarrow j'$). Finally, we get that
	\begin{align}
		\vT^{'\dagger}\vR\vT' &= 	\vT^{'\dagger}\vI\vT' -  \vT^{'\dagger}(\vI \otimes \vPi)\vT' +  \vT^{'\dagger}\vR_0\vT'\\&
		=\vI -  \vT^{'\dagger}(\vI \otimes \vPi)\cdot (\vI \otimes \vPi)\vT' + \vT^{'\dagger}(\vI \otimes \vPi)\cdot \vR_0\vT'\tag*{(since $(\vI \otimes \vPi) \vR_0 = \vR_0$)}\\&
		=\vI -\frac{1}{2}\sum_{j\in J} \vL_j^\dagg\vL_j\otimes \vI_{sys'} + \vI_{sys}\otimes  \vL^{*\dagg}_j\vL^*_j +\frac12\sum_{j\in J}\vL_j \otimes \vL^{*\dagg}_{j'} + \vL_{j'}^\dagg \otimes \vL^*_j\tag*{(by \eqref{eq:PiTPrime}-\eqref{eq:XTPrime})}\\&
		=\vI + \vec{\CD} (\vP,\{\vL_j\}).\tag*{\qedhere}
	\end{align}
\end{proof}
Specializing the above recipe for \eqref{eq:vec_main_result} yields a block-encoding of $\vec{\CD}_{\beta}$ using the following ingredients:
\begin{itemize}
	\item A unitary block-encoding for the Lindbladian 
	\begin{align}
		\L(\langle 0^{b'}|\otimes  \vI_{\bomega} \otimes \vI_a \otimes \vI_{sys} \R)\cdot \vU\cdot \L(|0^{c'}\rangle\otimes \vI_{sys}\R) =& \sum_{a \in A,\bomega \in S_{\omega_0}}\sqrt{\gamma(\bomega)} \ket{\bomega} \otimes \ket{a}\otimes  \hat{\vA}^a(\bomega). 
	\end{align}
	An example would be the block-encoding \eqref{eq:blockEncodedLindblad} instantiating the parameters $b'=b+1$ and $|0^{c'}\rangle=\ket{0^{c+1}}\ket{\bar{0}}$ after appropriately rearranging the registers.
	\item Negation on the Bohr frequency register 
	\begin{align}
		\vF := \sum_{\bomega\in S_{\omega_0}} \ketbra{-\bomega}{\bomega} \quad \text{such that}\quad \vF^{2} = \vI_{\bomega}.
	\end{align}
	\item Permutation (involution) of the jump operator labels 
	\begin{align}\label{eq:AdjointLabelPermutation}
		\vP:= \sum_{a \in A} \ketbra{a'}{a} \quad \text{where} \quad \vA^{a'} = (\vA^a)^{\dagger}\quad \text{for each}\quad a \in A.
	\end{align}
	Note that if the jump operators are Hermitian, e.g., Pauli matrices, then we can simply take the permutation to be the identity $\vP=\vI_a$.
\end{itemize}

\begin{figure}[t]

\begin{center}
	\newcommand{\scalea}{1.2}
	\newcommand{\scaleb}{1.5}
	\begin{quantikz}[wire types={q,b,b,b,b,b},classical	gap=1mm]
	\lstick{\scalebox{\scalea}{$\ket{0}$}}		&\qw &\gate[7,style={inner xsep=1mm, inner ysep=2mm}]{\scalebox{\scaleb}{$\vU^{'}$}}		& \gate[style={inner xsep=1mm, inner ysep=2mm}]{\scalebox{\scaleb}{$\vZ$}}	&	\gate[7,style={inner xsep=1mm, inner ysep=2mm}]{\scalebox{\scaleb}{$\vU^{'\dagger}$}}			&\qw &\qw&\qw \rstick{\scalebox{\scalea}{$\bra{0}$}}\\ 			
	 	\lstick{\scalebox{\scalea}{$\ket{0^{b+1}}$}}		&\qw&\qw	&\octrl{1}\wire[u]{q} &\qw	&\qw 	&\qw	&\rstick{\scalebox{\scalea}{$\bra{0^{b+1}}$}} \qw \\[-1mm]
	 	\lstick{\kern-2mm\scalebox{\scalea}{$\ket{0^{c-b}}$}}	&\qw&\qw	& 	\gate[style={inner xsep=2mm, inner ysep=2mm}]{\scalebox{\scalea}{$\vP$}}				&\qw		&\qw 			& \qw & \qw\rstick{\scalebox{\scalea}{$\bra{0^{c-b}}$}} \\
   \lstick{\scalebox{\scalea}{$\ket{\bar{0}}$}} & \qw&	&\qw	&	&\qw & \qw &\qw \rstick{\scalebox{\scalea}{$\bra{\bar{0}}$}}\\
\lstick{\scalebox{\scalea}{sys}}	&\qw& \qw	&	\qw	&\qw		&\qw		&\qw		& \qw \\
   \lstick{\scalebox{\scalea}{sys'}} &\qw& \qw	&	\qw	& \qw		&\qw		&\qw		& \qw 
	 \end{quantikz}
\end{center} 
  \caption{Circuit $\vU'^{\dagger} \vR \vU' $ for block-encoding the discriminant in the fashion of Szegedy quantum walk.
  }\label{fig:D_circuit}
\end{figure}
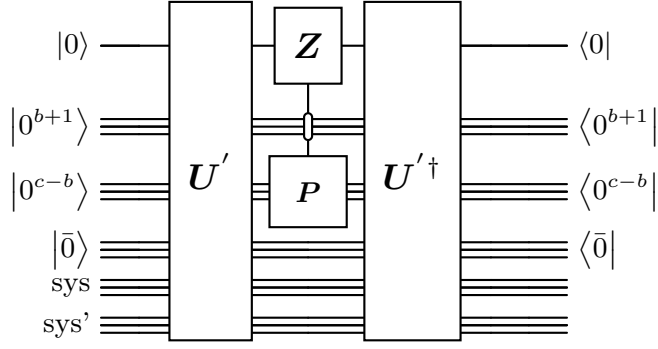
To reiterate, compared to the incoherent case, the discriminant proxy can be implemented with two extra unitaries $\vF$ and $\vP$, an additional copy of the system register (which we denoted by $system'$), and an additional ancilla qubit.
\begin{align}\label{eq:Dfregsiters}
\text{registers:}\quad\underset{\vT/\vT^*\text{ selector}}{\underbrace{\ket{\pm}}}\otimes\underset{\text{block. enc. anc.}}{\underbrace{|0^{b'}\rangle}}\otimes \underset{\text{Bohr freq.}}{\underbrace{\ket{\bomega}}}&\otimes\underset{\text{jump}}{\underbrace{\ket{a}}}\otimes \underset{\text{system}}{\underbrace{\hat{\vA}^a(\bomega)}} \otimes \underset{\text{system'}}{\underbrace{\vA^{a}(\bomega)^*}}
\end{align}

If we combine the constructions of \autoref{sec:L_circuit}-\autoref{sec:D_circuit}, we can see that the number of qubits is
\begin{align}
	\#qubits = 2n + \lceil\log_2(N)\rceil + c + 2,
\end{align}
coming from the two copies of the system register, the frequency register, the ancillae for the block-encoding of the jumps, and one additional ancillae introduced in each of \autoref{sec:L_circuit}-\autoref{sec:D_circuit}. When the normalized jump operators $\sqrt{|A|}\vA^a$ are unitaries, we can have $c$ as small as $\lceil\log_2(|A|)\rceil$, see for example \eqref{eq:exampleJumps}.

\subsubsection{Proof of coherent Gibbs sampler ({\autoref{thm:D_correct}})}\label{sec:proof_coherent}
In this section, we prove guarantees for our coherent Gibbs sampler (\autoref{thm:D_correct}) in a similar vein as the \Lword{} case (\autoref{thm:L_correctness}).
The required lemmas and propositions are analogous but refer to the spectral gap instead of the mixing time.

\begin{proof}[Proof of \autoref{thm:D_correct}]
We present general bounds on finite $N$~\eqref{eq:vec_main_result} and then take the large $N$ limit for the continuum~\eqref{eq:D_main}. We bound the eigenvector distance by the operator norm bounds (\autoref{prop:Bauer-Fike}, \autoref{prop:eigenvector_perturb_gen} and \autoref{prop:Dfixed_point_error}): secular approximation error (\autoref{lem:secular} and \autoref{prop:Gaussian_tail}), and discriminant proxy (\autoref{lem:error_from_Boltzmann} and $\CD_{sec}=\CD_{sec}^{\dagger}$)\footnote{Here, we implicitly assume that $\vec{\CD}_{impl}$ is Hermitian, which holds for example if $\vR$ in the block-encoding is implemented exactly.}
\begin{align}
     \norm{ \ket{\lambda_1(\vec{\CD}_{\beta})} - \ket{\sqrt{\vrho}} }&\le \frac{6\norm{\vec{\CD}_{\beta} - \vec{\CD}(\vrho,\CL_{sec})^{\dagger}}}{\lambda_{gap}(\vec{\CD}_{\beta})}\\
     &\le \frac{6}{\lambda_{gap}(\vec{\CD}_{\beta})} \L(\normp{\vec{\CD}_{\beta} - \vec{\CD}_{sec}}{}+ \normp{\vec{\CD}_{sec} - \vec{\CD}(\vrho,\CL_{sec})^{\dagger}}{}\R)\\
    &\le \CO\L(\frac{(\e^{ - N^2\omega_0^2\sigma_t^2/2} +\e^{ - N^2t_0^2/16\sigma_t^2}+\e^{ - \bmu^2\sigma_t^2}+\e^{-T^2/4\sigma_t^2}+\omega_0T)+\beta \bmu }{\lambda_{gap}(\vec{\CD}_{\beta})}\R)\\
    &\le \CO\L(\frac{\sigma_t \omega_0\sqrt{\log(1/(\sigma_t\omega_0))}+ (\beta/\sigma_t) \sqrt{\log(\sigma_t/\beta)} +\e^{ - N^2\omega_0^2\sigma_t^2/2}}{\lambda_{gap}(\vec{\CD}_\beta)}\R).
\end{align}
The fourth inequality chooses the free parameter $T = 2\sigma_t \sqrt{\ln(1/(\sigma_t\omega_0))}$ and $\bmu = \frac{\beta}{\sigma_t} \sqrt{\ln(\frac{\sigma_t}{\beta})}$ and uses that $\e^{ - N^2t_0^2/16\sigma_t^2}=\e^{ - \pi^2/4\omega_0^2\sigma_t^2} = \CO( \sigma_t\omega_0 \sqrt{\log(1/(\sigma_t\omega_0))}) $ to reduce the expression. 

In the continuum limit~\eqref{eq:D_main}, the discretization parameters $\omega_0$ and $N$ disappear, and the RHS becomes 
\begin{align}
 \norm{ \ket{\lambda_1(\vec{\CD}_{\beta})} - \ket{\sqrt{\vrho}} } = \CO\L(\frac{\beta}{\sigma_t}\frac{ \sqrt{\log(\sigma_t/\beta)}}{\lambda_{gap}(\vec{\CD}_\beta)}\R)   
\end{align}
by taking  the limit $N \omega_0 \rightarrow \infty,\ \omega_0 \rightarrow 0$.
\end{proof}
Note the user only chooses the time limit $T$, Gaussian width $\sigma_t$, and the Discrete Fourier Transform resolution $\omega_0$ and the number of points $N$; the truncation parameter $\bmu$ only appears implicitly in the analysis of secular approximation. Compared with the fixed point error for \Lword{}s (\autoref{sec:proof_correctness}), the Hermiticity and gap substantially simplifies the analysis.

\subsection{Metropolis sampling with arbitrary spectral target weights}\label{sec:general_Metropolis}

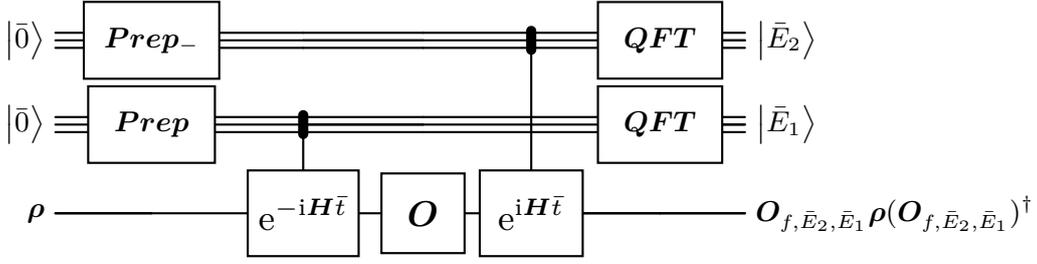
\begin{figure}[t]
\label{fig:two_sided_OQFT}
	\begin{quantikz}[wire types={b,b},classical	gap=1mm]
	 	\lstick{\scalebox{1.2}{\ket{\bar{0}}}}	&\gate[style={inner xsep=1.25mm, inner ysep=2mm}]{\scalebox{1.2}{$\vec{Prep}_-$}}&	&\qw			&\ctrl{2}				&\gate[style={inner xsep=2mm, inner ysep=2mm}]{\scalebox{1.2}{$\vec{QFT}$}} & \qw \rstick{\scalebox{1.2}{$\ket{\bE_2}$}}\\
   \lstick{\scalebox{1.2}{\ket{\bar{0}}}}	&\gate[style={inner xsep=2mm, inner ysep=2mm}]{\scalebox{1.2}{$\vec{Prep}$}}&\ctrl{1}	&\qw			&				&\gate[style={inner xsep=2mm, inner ysep=2mm}]{\scalebox{1.2}{$\vec{QFT}$}} & \qw \rstick{\scalebox{1.2}{$\ket{\bE_1}$}}\\
	 	\lstick{\scalebox{1.2}{$\vrho$}}		&\qw&\gate[style={inner xsep=0mm, inner ysep=2mm}]{\scalebox{1.5}{$\e^{-\ri\vH \bar{t}}$}} &\gate[style={inner xsep=2mm, inner ysep=2.25mm}]{\scalebox{1.5}{$\vO$}}&\gate[style={inner xsep=1mm, inner ysep=2mm}]{\scalebox{1.5}{$\e^{\ri\vH \bar{t}}$}}	&\qw		& \qw \rstick{\scalebox{1.2}{$\vO_{f,\bE_2,\bE_1}\vrho(\vO_{f,\bE_2,\bE_1})^\dagg$}}
	 \end{quantikz}
  \caption{The circuit for two-sided operator Fourier Transform. The gate $\vec{prep}_-$ prepares the flipped function $f_- (t) = f(-t)$. Unlike the one-sided version (Figure~\ref{fig:OQFT}), now we cannot understand the expression by Heisenberg evolution; this is, in spirit, more similar to doing two consecutive phase estimations. 
  }
\end{figure}

Looking beyond sampling Gibbs states $\vrho \propto \sum_{i} \e^{-\beta E_i} \ketbra{\psi_i}{\psi_i}$, we may modify the circuit to sample from arbitrary weight function $ \vrho \propto \sum_{i} p(E_i)\ketbra{\psi_i}{\psi_i}$ that could be useful in other contexts. We reserve this section for pointing out the required ingredients and adaptations; we will stick to Gibbs sampling for the rest of the paper, but both our coherent and incoherent algorithms should apply in the general setting as well. 

Instead of applying the operator Fourier Transform for the Bohr frequencies (the energy differences), sampling from arbitrary weight requires accessing \textit{both} energies before and after the jump, resembling~\cite{temme2009QuantumMetropolis,wocjan2021szegedy,Rall_thermal_22}. The appropriate \Lword{} takes the following form
\begin{align}\label{eq:TwoSidedLi}
    \mathcal{L}_{general}[\vrho] := \sum_{a\in A, \bE_2,\bE_1\in S_{\omega_0}} \gamma(\bE_2,\bE_1)\left(\hat{\vA}^a_{f}(\bE_2,\bE_1) \vrho \hat{\vA}^a_{f}(\bE_2,\bE_1)^\dagger-\frac{1}{2}\{\hat{\vA}^a_{f}(\bE_2,\bE_1)^\dagger\hat{\vA}^a_{f}(\bE_2,\bE_1),\vrho\}\right).
\end{align}
To implement a block-encoding of the above \Lword{}, we need different Fourier Transform components and a controlled filter for the two-argument Metropolis weight.
\begin{itemize}
   \item Phase estimation isometry (in the Schrödinger picture)
    \begin{align}
        \Phi_f := \frac{1}{\sqrt{N}}\sum_{\bE\in S_{\omega_0}}\sum_{\bt \in S_{\bt_0}}f(\bt)\ket{\bE}\otimes \e^{-\ri (\bE-\vH) \bt} .
    \end{align}
    This circumvents the impossibility results as it does not take the shift-invariant form (\autoref{sec:impossible}).
    \item Two-sided operator Fourier Transform (\autoref{fig:two_sided_OQFT})
	\begin{align}
	\CF[\ket{f_2}\otimes \ket{f_1}\otimes \vO] &:= (\Phi_{f_2}\otimes \vI_{\bE_1})\circ (\vI_{\bE_1}\otimes \vO)\circ \Phi_{f_1} = \sum_{\bE_2,\bE_1\in S_{\omega_0}}\ket{\bE_2}\otimes \ket{\bE_1}\otimes\hat{\vO}_{f_1,f_2}(\bE_2,\bE_1), \\ 
	\text{where} \quad \hat{\vO}_{f_1,f_2}(\bE_2,\bE_1) &: = \frac{1}{N}\sum_{\bt_2,\bt_1 \in S_{\bt_0}}f_2(\bt_2) f_1(\bt_1) \e^{- \ri (\bE_2-\vH) \bt_2}  \vO  \e^{-\ri (\bE_1-\vH) \bt_1} .
 \end{align}
 In our particular case, we set $f_1(\bt) = f(\bt)$ and $f_2(\bt) = f(-\bt)=:f_-(\bt)$ for a normalized function $f$ to get
 \begin{align}
 \vO_{f,f_-}(\bE_2,\bE_1) = \frac{1}{N}\sum_{\bt_2,\bt_1 \in S_{\bt_0}} f(-\bt_2) f(\bt_1) \e^{- \ri (\bE_2-\vH) \bt_2}  \vO  \e^{-\ri (\bE_1-\vH) \bt_1} \label{eq:O_ff-}    
 \end{align}
 which will be short-handed by $\vO_{f,\bE_2,\bE_1}$. This construction is reminiscent of~\cite{wocjan2021szegedy}, but it does not require a rounding promise. 
    \item Controlled filter for the Metropolis weight
	\begin{align}
	\vW := \sum_{\bE_2,\bE_1 \in S_{\omega_0}}  \vY_{1-\gamma(\bE_2,\bE_1)} \otimes \ketbra{\bE_2}{\bE_2} \otimes \ketbra{\bE_1}{\bE_1} \\
 \text{such that} \quad 0\le\gamma(\bE_2,\bE_1) \le 1\quad \text{and} \quad\frac{\gamma(\bE_2,\bE_1)}{\gamma(\bE_1,\bE_2)}=\frac{p(\bE_2)}{p(\bE_1)}.
	\end{align}
The ratio constraint ensures approximate detailed balance. 
\end{itemize}
If the set of jump operators is self-adjoint and $f$ is real, we may construct a valid Hermitian discriminant proxy as outlined in \autoref{prop:genDiscProxy} -- the corresponding involution is $\textbf{SWAP}\otimes \vP$, where \textbf{SWAP} acts on the two phase estimation registers, and $\vP$ is a permutation as in \eqref{eq:AdjointLabelPermutation}.
This construction yields a discriminant proxy analogous to \eqref{eq:vec_main_result} due to the following ``skew symmetry'' of the two-sided operator Fourier Transform. We also include the two-index version of operator Parseval's identity. 

\begin{restatable}[Parseval's identity]{prop}{propertiesTwoSidedDFT}\label{prop:properties_twosided_DFT}
	For a set of matrices $\vA^a\in\mathbb{C}^{n\times n}$, consider their two-sided discrete Fourier Transform weighted by a complex-valued function $f\colon S_{t_0}\rightarrow\BC$ as in~\eqref{eq:O_ff-}. Then, $\hat{\vA}^a_{f}(\bE_2,\bE_1)^{\dagger} = (\hat{\vA^{a\dagger}})_{f^*}(\bE_2,\bE_1)$, and  
	\begin{align}
		 \sum_{a\in A} \sum_{\bE_1,\bE_2\in S_{\omega_0}} \hat{\vA}^a_{f}(\bE_2,\bE_1)^{\dagger} \hat{\vA}^a_{f}(\bE_2,\bE_1)
		&=\nrm{f}_2^2 \sum_{\bt_1\in S_{t_0}} \labs{f(\bt_1)}^2	 \e^{-\ri \vH \bt_1}\vA^{a\dagger} \vA^a \e^{\ri \vH \bt_1}
		\preceq  \nrm{f}_2^4\nrm{\sum_{a\in A} \vA^{a\dagger}\vA^a} \cdot \vI\quad \text{and}\\
		\sum_{a\in A}\sum_{\bE_1,\bE_2\in S_{\omega_0}}  \hat{\vA}^a_{f}(\bE_2,\bE_1)\hat{\vA}^a_{f}(\bE_2,\bE_1)^{\dagger}
		&=\nrm{f}_2^2\sum_{a\in A} \sum_{\bt_2 \in S_{t_0}}  \labs{f(\bt_2)}^2	\e^{\ri \vH \bt_2} \vA^a \vA^{a\dagger} \e^{-\ri \vH \bt_2}
		\preceq \nrm{f}_2^4\nrm{\sum_{a\in A} \vA^a\vA^{a\dagger}}\cdot \vI.
	\end{align}
\end{restatable}
\begin{proof}For each $a\in A$, we have
\begin{align}
	& \\[-13mm]
    \hat{\vA}^a_{f}(\bE_2,\bE_1)^{\dagger} & = \L(\frac{1}{N}\sum_{\bt_2,\bt_1 \in S_{\bt_0}} f(-\bt_2) f(\bt_1)\e^{-\ri (\bE_2-\vH) \bt_2}  \vA^a  \e^{-\ri (\bE_1-\vH) \bt_1} \R)^{\!\!\dagger}\\
	& = \L(\frac{1}{N}\sum_{\bt_2,\bt_1 \in S_{\bt_0}} f(\bt_2) f(\bt_1)\e^{\ri (\bE_2-\vH) \bt_2}  \vA^a  \e^{-\ri (\bE_1-\vH) \bt_1} \R)^{\!\!\dagger}\\    
    & = \frac{1}{N}\sum_{\bt_2,\bt_1 \in S_{\bt_0}} f(\bt_2)^* f(\bt_1)^*\e^{\ri (\bE_1-\vH) \bt_1}\vA^{a\dagger} \e^{-\ri (\bE_2-\vH) \bt_2}  = (\hat{\vA^{a\dagger}})_{f^*}(\bE_2,\bE_1).\tag*{\qedhere}
\end{align}
\end{proof}
Our analytic arguments (\autoref{sec:operator_FT}, \autoref{sec:error_boltzmann}) can be adapted to the above \Lword{} \eqref{eq:TwoSidedLi} and discriminant variant, but we will stick to the operator Fourier Transform for simplicity throughout the paper.

\section{Discussion}
\label{sec:discussion}
Our work aimed to lay the algorithmic and analytic foundation for Monte Carlo-style Quantum Gibbs samplers. We have presented families of efficiently implementable algorithms for Gibbs sampling in terms of \Lword{}s with guarantees for fixed-point accuracy. We have confronted technical issues from energy uncertainty (which have haunted quantum Gibbs sampling algorithms for a decade) by highlighting a key algorithmic component, the operator Fourier Transform, and introducing a general analytic framework, the secular approximation and approximate detailed balance. These arguments are compatible with the Szegedy-type speedup and preparation of the purified Gibbs state.

Our construction is conceptually simple as it draws inspiration from the physical mechanism of open-system thermalization, especially the Davies' generator. Conversely, our argument completes the first proof of Gibbs state stationarity for physically derived \Lword{}s, especially the coarse-grained master equations~\cite{Christian_2013_Coarse_graining,Mozgunov2020completelypositive}. Potentially, this could lead to proposals of quantum Gibbs samplers on analog quantum simulators. Still, our analysis is restricted to the open system setting where the bath is Markovian by assumption; we hope the precise statement in the open system settings inspires further insight toward closed-system thermodynamics. 

Would quantum Gibbs samplers be the ultimate solution to the ground state preparation problem by setting $\beta \gg 1$? Our work only answers the first half of this problem by writing down some candidate \Lword{}s with a provably accurate Gibbs fixed point and efficient simulation algorithm. Still, the missing piece of the puzzle is the mixing time (or the spectral gap) of the proposed \Lword{}s. A scientifically informative first step is to directly benchmark the performance of quantum Gibbs samplers numerically for viable system sizes. This would give concrete estimates of the realistic costs of quantum simulation. From a mathematical physics perspective, there have been efforts to prove rapid mixing~\cite{Bardet2021EntropyDF,capel2021modified} (convergence at a logarithmic depth $\log(n)$) of Quantum Gibbs samplers for lattice Hamiltonian in the spirit of classical Ising models~\cite{martinelli1999lectures}. There, most results have been restricted to commuting Hamiltonians due to the lack of a satisfactory formulation of noncommuting Gibbs samplers, which this work provides\footnote{The followup work~\cite{exactDB} gives an even nicer noncommutative Gibbs sampler with \textit{exact detailed balance}.}. As a direct implication of this work, we provide a candidate algorithm for preparing a gapped ground state at a potentially very low depth: simply setting the Gaussian width to be $\sigma_t \sim \tilde{\CO}(\frac{1}{\Delta_{gap}})$ and $\beta = \tilde{\CO}(\frac{1}{\Delta_{gap}})$ ensures the ground state to be approximately the common kernel of $\tilde{\CO}(\frac{1}{\Delta_{gap}})$-local \Lword{}s\footnote{This observation was later exploited in~\cite{ding2023single} for ground state preparation.}. 

From a complexity perspective, quantum Gibbs samplers provide a new \textit{dynamic} angle to study the complexity of thermal states and even ground states. Indeed, existing complexity results for gapped ground states, especially the area law, have beaten the static properties to death (local gap, decay of correlation, etc.). Taking a step back, these approaches neglect the instinctive experimental origin of ground states: cool the system in a fridge. It would be curious to bridge this thermodynamics process to the area-law literature (e.g.,~\cite{Hastings2007AnAL,landau2015polynomial,anshu2022area}). Practically, we hope noncommuting Gibbs samplers will inspire new tensor network algorithms or even new ansatz, which could lead to a better grasp of 2D gapped physics.

\section*{Acknowledgments}
We thank Simon Apers, Mario Berta, Garnet Chan, Alex Dalzell, Zhiyan Ding, Hsin-Yuan (Robert) Huang, Lin Lin, Yunchao Liu, Sam McArdle,  Jonathan Moussa,  Evgeny Mozgunov, 
  Tobias Osborne, Wocjan Pawel, Patrick Rall, Mehdi Soleimanifar, Kristan Temme, Umesh Vazirani, and Tong Yu for helpful discussions. We also thank anonymous referees for their helpful feedback. AG thanks Chunhao Wang, Dávid Matolcsi, Cambyse Rouzé, and Daniel Stilck França for useful discussions. CFC is supported by the Eddlemen Fellowship and the AWS Center for Quantum Computing internship. AG acknowledges funding from the AWS Center for Quantum Computing.

\bibliographystyle{alphaUrlePrint.bst}
\bibliography{ref,qc_gily}

\appendix
\section*{Nomenclature}\label{sec:recap_notation}
This appendix recapitulates notations.
We write scalars, functions, and vectors in normal font, matrices in bold font $\vO$, and superoperators in curly font~$\CL$. Natural constants $\e, \ri, \pi$ are denoted in Roman font.\pagebreak[0]
\begin{align}
\vH &= \sum_i E_i \ketbra{\psi_i}{\psi_i}&\text{the Hamiltonian of interest with eigen decomposition}\\
\text{Spec}(\vH) &:= \{ E_i \} & \text{the spectrum of the Hamiltonian}\\
\nu \in B = B(\vH)&:= \text{Spec}(\vH) - \text{Spec}(\vH) &\text{the set of Bohr frequencies}\\
\vP_{E}&:= \sum_{i:E_i = E} \ketbra{\psi_i}{\psi_i}&\text{eigenspace projector for energy $E$}\\
\CL:& & \text{a Lindbladian in the Schrodinger Picture}\\
\CL^{\dagger}:& & \text{a Lindbladian in the Heisenberg Picture}\\
n: & &\text{ system size (number of qubits) of the Hamiltonian $\vH$}\\
\beta: & &\text{ inverse temperature}\\
\vrho:& &\text{a density matrix}\\
\vrho_{\beta}&:= \frac{\e^{-\beta \vH }}{\tr[ \e^{-\beta \vH }]} \quad &\text{the Gibbs state with inverse temperature $\beta$}\\
\ket{\sqrt{\vrho_{\beta}}} &:= \frac{1}{\sqrt{\tr[ \e^{-\beta \vH }]}} \sum_i \e^{-\beta E_i/2} \ket{\psi_i} \otimes \ket{\psi_i^*}\kern-5mm &\text{the purified Gibbs state}\\
\{\vA^a\}_{a \in A}: & &\text{set of jump operators}\\
\labs{A}: & & \text{cardinality of the set of jumps}\\
\vI:& &\text{the identity operator}\\
\bigOt{\cdot},\tOmega (\cdot) :& &\text{complexity expression ignoring (poly)logarithmic factors}
\end{align}
Fourier Transform notations:
\begin{align}
\bomega &\in S_{\omega_0} \subset \BZ \omega_0 &\text{discrete frequency labels for Fourier Transform}\\
\bt &\in S_{t_0} \subset \BZ t_0 &\text{discrete time labels for Fourier Transform}\\
\bar{\vH}: & &\text{the discretized Hamiltonian with eigenvalues in $\BZ \bomega_0$}\\
N:& &\text{number of Fourier Transform labels such that $\omega_0t_0 = \frac{2\pi}{N}$}\\
\quad \vA(t)&:= \e^{\ri \vH t} \vA \e^{-\ri \vH t} & \text{Heisenberg evolution for operator $\vA$}\\
\hat{\vA}_{(f)}(\bomega)&:= \frac{1}{\sqrt{N}} \sum_{\bt \in S_{t_0}} \e^{-\ri \bomega \bt} f(\bt) \vA(\bt)  & \text{discrete operator Fourier Transform for $\vA$ weighted by $f$}\\
\hat{\vA}_{(f)}(\omega) &:= \frac{1}{\sqrt{2\pi}}\int_{-\infty}^{\infty} \e^{-\ri \omega t}f(t) \vA(t)\mathrm{d}t& \text{continuous operator Fourier Transform for $\vA$ weighted by $f$}\\
\hat{f}(\omega)&=\CF(f)=\lim_{K\rightarrow  \infty}\frac{1}{\sqrt{2\pi}}\int_{-K}^{K}\e^{-\ri\omega t} f(t)\mathrm{d}t & \text{the Fourier Transform of a scalar function $f$ over inputs $t$}\\		
\vA_\nu&:=\sum_{E_2 - E_1 = \nu } \vP_{E_2} \vA \vP_{E_1} &\text{operator $\vA$ at exact Bohr frequency $\nu$}\end{align}
Norms: 
\begin{align}
	\norm{f(x)}_p&: = \L(\int_{x} \labs{f(x)}^p \mathrm{d} x \R)^{\! 1/p} \quad &  \text{the $p$-norm of a scalar function $f$ over inputs $x$ for $p\in[1,\infty]$}\\
   \ell_p(\BR)&:= \{f:\BR \rightarrow \BC,\quad \norm{f}_p < \infty\}\quad &\text{the set of integrable functions}\\ 
 \norm{f(\bar{x})}_p&: = \L(\sum_{\bar{x}} \labs{f(\bar{x})}^p\R)^{\! 1/p} \quad &  \kern-5mm \text{the $p$-norm of a scalar function $f$ over discrete inputs $x$ for $p\in[1,\infty]$}\\
	\norm{f(x)}_I&: = 	\norm{f(x)\cdot \indicator(x\in I )}_\infty \quad &  \text{the sup-norm of a scalar function $f$ over the interval $I$}\\
	\norm{f(x)}&: = \norm{f(x)}_2 = \sqrt{\int_{x} \labs{f(x)}^2 \mathrm{d} x} \quad & \text{the 2-norm of a scalar function $f$ over inputs $x$}\\	
    \norm{f(\bar{x})}&: = \norm{f(\bar{x})}_2 = \sqrt{\sum_{\bar{x}} \labs{f(\bar{x})}^2 } \quad & \text{the 2-norm of a scalar function $f$ over discrete inputs $x$}\\
	\norm{\ket{\psi}}&: \quad &\text{the Euclidean norm of a vector $\ket{\psi}$}\\
	\norm{\vO}&:= \sup_{\ket{\psi},\ket{\phi}} \frac{\bra{\phi} \vO \ket{\psi}}{\norm{\ket{\psi}}\cdot \norm{\ket{\phi}}} \quad &\text{the operator norm of a matrix $\vO$}\\
  	\norm{\vO}_p&:= (\tr \labs{\vO}^p)^{1/p}\quad&\text{the Schatten p-norm of a matrix $\vO$}\\
  \norm{\CL}_{p-p} &:= \sup_{\vO} \frac{\normp{\CL[\vO]}{p}}{\normp{\vO}{p}}\quad&\text{the induced $p-p$ norm of a superoperator $\CL$}
\end{align}
Linear algebra:
\begin{align}
    \lambda_i(\vO): & \quad &\text{ the $i$-th largest eigenvalue of a matrix $\vO$ sorted by their real parts}\\
 	\lambda_{gap}(\vO)&:=\Re\lambda_1(\vO)-\Re\lambda_2(\vO)\ge 0 \quad &\text{the real spectral gap of a matrix $\vO$}\\
 \varsigma_i(\vO): & \quad &\text{\kern-2cm the $i$-th largest singular of a matrix $\vO$}\\
	\vO^*: & \quad & \text{the entry-wise complex conjugate of a matrix $\vO$}\\
 	\vO^\dagger: & \quad & \text{the Hermitian conjugate of a matrix $\vO$}\\
  \ket{\psi^*}&: \quad&\text{ entry-wise complex conjugate of a vector $\ket{\psi}$}
\end{align}

\section{Operator Fourier Transform: properties and error bounds}\label{sec:operator_FT}
In this section, we study properties of the \textit{operator Fourier Transform}. Given a Hamiltonian $\vH$, an operator $\vA$, and a complex-valued function $f: \BR \rightarrow \BC$, let 
\begin{align}\label{eq:OFT_appendix}
\hat{\vA}_{f}(\bomega): = \frac{1}{\sqrt{N}} \sum_{\bt \in S_{t_0}} \e^{-\ri \bomega \bt} f(\bt) \vA(\bt)  \quad \text{where} \quad \vA(t):= \e^{\ri \vH t} \vA \e^{-\ri \vH t}.
\end{align}
Note the normalization $\frac{1}{\sqrt{N}}$. The transformed operators $\hat{\vA}_{f}(\bomega)$ satisfy the desirable \emph{exact} symmetry of a standard Fourier Transform as well as an operator version of Parseval's identity.

\begin{restatable}[Symmetry and operator Parseval's identity]{prop}{propertiesDFT}\label{prop:properties_DFT}
	For a set of matrices $\{\vA^a\}_{a\in A}$ and a Hamiltonian $\vH$, consider their discrete operator Fourier Transform weighted by a complex-valued function $f\colon S_{t_0}\rightarrow\BC$ as in \eqref{eq:OFT_appendix}. Then, the symmetry holds $\hat{\vA}^a_{f}(\bomega)^{\dagger} = (\hat{\vA^{a\dagger}})_{f^*}(-\bomega)$, moreover
	\begin{align}
	\sum_{a\in A}\sum_{\bomega\in S_{\omega_0}} \hat{\vA}^a_{f}(\bomega)^{\dagger} \hat{\vA}^a_{f}(\bomega) 
	&=\sum_{a\in A}\sum_{\bt \in S_{t_0}} \labs{f(\bt)}^2	 \e^{\ri \vH \bt}\vA^{a\dagger} \vA^a \e^{-\ri \vH \bt}
	\preceq \nrm{\sum_{a\in A} \vA^{a\dagger}\vA^a} \nrm{f}_2^2 \cdot \vI\quad \text{and}\label{eq:discFourierIdentity}\\
	\sum_{a\in A}\sum_{\bomega\in S_{\omega_0}}  \hat{\vA}^a_{f}(\bomega) \hat{\vA}^a_{f}(\bomega)^{\dagger}
	&=\sum_{a\in A}\sum_{\bt \in S_{t_0}} \labs{f(\bt)}^2	\e^{\ri \vH \bt} \vA^a \vA^{a\dagger} \e^{-\ri \vH \bt}
	\preceq \nrm{\sum_{a\in A} \vA^a\vA^{a\dagger}} \nrm{f}_2^2 \cdot \vI.\label{eq:discFourierIdentity2}
	\end{align}
\end{restatable}
For our \Lword{} Gibbs samplers, the weights will be normalized $\sum_{\bt \in S_{t_0}} \labs{f(\bt)}^2=1$, which means that they can be implemented by amplitudes of a state. 
In the special case $\sum_{a\in A}\vA^{a\dagger}\vA^a =\vI$ (i.e., these operators can be interpreted as a quantum channel), then the inequality \eqref{eq:discFourierIdentity} hold with equality, and as a consequence, the operators resolve the identity $\sum_{a\in A} \sum_{\bomega\in S_{\omega_0}} \hat{\vA}^a_{f}(\bomega)^{\dagger} \hat{\vA}^a_{f}(\bomega) = \vI.$ 

\begin{proof}First observe that by definition
	\begin{align}\label{eq:OpFTDaggProperty}
	\hat{\vA}^a_{f}(\bomega)^{\dagger} &= \frac{1}{\sqrt{N}}\sum_{\bt \in S_{t_0}}  \e^{\ri \bomega \bt} f^*(\bt) (\vA^a(\bt))^{\dagger}
	= \frac{1}{\sqrt{N}}\sum_{\bt \in S_{t_0}}\e^{\ri \bomega \bt} f^*(\bt) \vA^{a\dagger}(\bt) = (\hat{\vA^{a\dagger}})_{f^*,-\bomega}.
	\end{align}
	Next, we prove \eqref{eq:discFourierIdentity} by direct computation as follows
	\begin{align}
	\sum_{a\in A}\sum_{\bomega \in S_{\omega_0}}\hat{\vA}^a_{f}(\bomega)^{\dagger} \hat{\vA}^a_{f}(\bomega) 
	&=\sum_{a\in A}\frac{1}{N} \sum_{\bomega \in S_{\omega_0}} \sum_{\bt' \in S_{t_0}} \e^{\ri \bomega \bt'} f^*(\bt) \vA^a(\bt')^{\dagger} \sum_{\bt \in S_{t_0}} \e^{-\ri \bomega \bt} f(\bt)\vA^a(\bt) \\&
	= \sum_{a\in A}\sum_{\bt \in S_{t_0}} \labs{f(\bt)}^2 \vA^a(\bt)^{\dagger}\vA^a(\bt)  \\&
	= \sum_{a\in A} \sum_{\bt \in S_{t_0}} \labs{f(\bt)}^2 \e^{\ri \vH \bt}  \vA^{a\dagger}\vA^a\e^{-\ri \vH \bt}\\&
	\preceq \nrm{\sum_{a\in A} \vA^{a\dagger}\vA^a}  \sum_{\bt \in S_{t_0}} \labs{f(\bt)}^2 \e^{\ri \vH \bt}\vI\e^{-\ri \vH \bt} \\&	
	= \vI \cdot \nrm{\sum_{a\in A} \vA^{a\dagger}\vA^a} \sum_{\bt \in S_{t_0}} \labs{f(\bt)}^2 .
	\end{align}
	The second equality uses the Fourier representation of the discrete delta function
	\begin{align}\label{eq:discreteDiracDelata}
	\sum_{\bomega \in S_{\omega_0}} \e^{-\ri \bomega (\bt-\bt')} = N \delta_{\bt,\bt'}.
	\end{align}
	The proof of \eqref{eq:discFourierIdentity2} is completely analogous.
\end{proof}

We also include the analogous analysis in the continuum limit where the discretization parameter $N$ disappears. 
We will assume throughout that the weight function $f$ is square integrable, i.e., $f\in\ell_2(\mathbb{R})$. In the continuous case, the operator Fourier Transform is a matrix-valued function, and to emphasize this, we change the notation to $\hat{\vA}_{f}(\bomega)\rightarrow \hat{\vA}_{f}(\omega)$. We could directly copy the above proof; however, arguing about the Dirac delta function in the continuous case is tricky. We resolve this by relying on Parseval-Plancherel's identity. 

For studying the operator Fourier Transform, it is useful to decompose the operator according to the Bohr frequencies
\begin{align}
\label{eqn:Aoperator}
\e^{\ri \vH t} \vA \e^{-\ri \vH t} &= \sum_{\nu\in B(\vH)} \e^{\ri \nu t} \vA_\nu \\
\text{where}\quad \vA_\nu&:=\sum_{E_2 - E_1 = \nu } \vP_{E_2} \vA \vP_{E_1} \quad \text{satisfies that} \quad (\vA_{\nu})^{\dagger} = (\vA^{\dagger})_{-\nu},
\end{align}
where $\vP_{E}$ denotes the orthogonal projector onto the subspace spanned by energy $E$ eigenstates of $\vH$. If $f\in\ell_1(\mathbb{R})$, using this decomposition, we can conveniently express the operator Fourier Transform as follows 
\begin{align}
	\hat{\vA}_{f}(\omega) &=  \frac{1}{\sqrt{2\pi}}\int_{-\infty}^{\infty}\e^{-\ri \omega t} f(t) \e^{\ri \vH t} \vA \e^{-\ri \vH t}\mathrm{d}t \\
	&=\frac{1}{\sqrt{2\pi}}\int_{-\infty}^{\infty} \e^{-\ri \omega t} f(t)  \e^{\ri \vH t} \sum_{\nu \in B(\vH)} \vA_{\nu} \e^{-\ri \vH t}\mathrm{d}t \\
	&=\sum_{\nu \in B(\vH)}  \frac{1}{\sqrt{2\pi}}\int_{-\infty}^{\infty} \e^{-\ri (\omega -\nu) t} f(t) \vA_{\nu} \mathrm{d}t =\sum_{\nu \in B(\vH)}  \hat{f}(\omega-\nu) \vA_{\nu},  \label{eq:BohrDecomposedOFT}
\end{align}
where $\hat{f}(\omega)=\frac{1}{\sqrt{2\pi}}\int_{-\infty}^{\infty} f(t) \e^{-\ri \omega t}\mathrm{d}t$ is the Fourier Transform of the weight function $f(t)$. More generally, if $f\in\ell_2(\mathbb{R})$, then we use \eqref{eq:BohrDecomposedOFT} as the definition of the operator Fourier Transform because the Fourier Transform uniquely extends to a unitary map $\CF\colon \ell_2(\mathbb{R})\rightarrow \ell_2(\mathbb{R})$.

\begin{prop}[Symmetry and operator Parseval's identity]\label{prop:cont_operator_FT}
	For a set of matrices $\{\vA^a\}_{a\in A}$ and a Hamiltonian $\vH$, consider their continuous operator Fourier Transform weighted by a complex-valued function $f\in \ell_2(\mathbb{R})$
	\begin{align}
		\hat{\vA}^a_{f}(\omega): =  \sum_{\nu \in B(\vH)}  \hat{f}(\omega-\nu) \vA^a_{\nu},
	\end{align}
	then 
	$(\hat{\vA}^a_{f}(\omega))^{\dagger} = (\hat{\vA^{a\dagger}})_{f^*}(-\omega),$
	moreover
	\begin{align}
		\sum_{a\in A}\int_{-\infty}^{\infty} \hat{\vA}^a_{f}(\omega)^{\dagger}\hat{\vA}^a_{f}(\omega) \mathrm{d}\omega
		&= \sum_{a\in A}\int_{-\infty}^{\infty} \labs{f(t)}^2\e^{\ri \vH t}\vA^{a\dagger}\vA^a\e^{-\ri \vH t}\mathrm{d}t
		\preceq  \nrm{f}_2^2 \nrm{\sum_{a\in A} \vA^{a\dagger}\vA^a}  \cdot \vI, \quad \text{and}\label{eq:contFourierIdentity}\\
		\sum_{a\in A}\int_{-\infty}^{\infty} \hat{\vA}^a_{f}(\omega)\hat{\vA^a}_{f}(\omega)^{\dagger} \mathrm{d}\omega
		&= \sum_{a\in A}\int_{-\infty}^{\infty}\labs{f(t)}^2 \e^{\ri \vH t}\vA^a\vA^{a\dagger}\e^{-\ri \vH t}\mathrm{d}t 
		\preceq  \nrm{f}_2^2 \nrm{\sum_{a\in A} \vA^a\vA^{a\dagger}}  \cdot \vI.	\label{eq:contFourierIdentity2}
	\end{align}
\end{prop}
Similarly as before, if $\sum_{a\in A} \vA^{a\dagger}\vA^a = \vI $, then the inequality \eqref{eq:contFourierIdentity} hold with equality.
\begin{proof}First observe that by definition
	\begin{align}
		\hat{\vA}^a_{f}(\omega)^{\dagger} 
		= \sum_{\nu \in B(\vH)} (\hat{f}(\omega-\nu))^*(\vA^a_\nu)^\dagg 
		= \sum_{\nu \in B(\vH)} \widehat{f^*}(\nu-\omega)(\vA^{a\dagg})_{-E} 
		= \sum_{-\nu \in B(\vH)} \widehat{f^*}(-\omega-\nu)(\vA^{a\dagg})_{\nu} 
		= \hat{\vA^{a\dagger}}_{f^*}(-\omega).	
	\end{align}
	Next, we prove \eqref{eq:contFourierIdentity} by direct computation as follows
	\begin{align}
		\sum_{a\in A}\int_{-\infty}^{\infty} \hat{\vA}^a_{f}(\omega)^{\dagger}\hat{\vA}^a_{f}(\omega)  \mathrm{d}\omega
		&= \sum_{a\in A}\int_{-\infty}^{\infty}  \sum_{\nu \in B(\vH)} (\vA^a_{\nu})^\dagg (\hat{f}(\omega-\nu))^* \sum_{\nu' \in B(\vH)} \hat{f}(\omega - \nu')\vA^a_{\nu'}\mathrm{d}\omega \tag*{(by defintion)}\\&
		=  \sum_{a\in A}\sum_{\nu,\nu' \in B(\vH)} (\vA^a_{\nu})^\dagg \vA^a_{\nu'}\int_{-\infty}^{\infty}  (\hat{f}(\omega))^*\hat{f}(\omega-(\nu'-\nu))\mathrm{d}\omega \tag*{(shift $\omega\rightarrow\omega-\nu$ and use $|B|\leq\infty$)}\\&
		= \sum_{a\in A}\sum_{\nu'' \in B(\vH)} (\vA^{a\dagg} \vA^a)_{\nu''} \int_{-\infty}^{\infty}  (\hat{f}(\omega))^*\hat{f}(\omega - \nu'')\mathrm{d}\omega \tag*{(by the definition of $\vA^a_\nu$)}\\&		
		= \sum_{a\in A} \sum_{\nu'' \in B(\vH)} (\vA^{a\dagg} \vA^a)_{\nu''}\int_{-\infty}^{\infty} (f(t))^* f(t)\e^{\ri \nu'' t}\mathrm{d}t\tag*{(since $\CF$ is unitary)}\\&
		= \sum_{a\in A}\sum_{\nu'' \in B(\vH)}\int_{-\infty}^{\infty} |f(t)|^2 \e^{\ri \vH t}  (\vA^{a\dagg} \vA^a)_{\nu''} \e^{-\ri \vH t} \mathrm{d}t  \tag*{(as in \eqref{eq:BohrDecomposedOFT})}\\&
		= \sum_{a\in A}\int_{-\infty}^{\infty} \labs{f(t)}^2 \e^{\ri \vH t}\vA^{a\dagger}\vA^a\e^{-\ri \vH t}\mathrm{d}t\tag*{(since $|B(\vH)|\leq\infty$)}\\&
		\preceq  \nrm{\sum_{a\in A} \vA^{a\dagger}\vA^a} \int_{-\infty}^{\infty} \labs{f(t)}^2\e^{\ri \vH t}\vI\e^{-\ri \vH t}\mathrm{d}t\tag*{(since $\sum_{a\in A} \vA^{a\dagger}\vA^a\preceq\nrm{\sum_{a\in A} \vA^{a\dagger}\vA^a} \vI$)}\\&
		= \vI \cdot \nrm{\sum_{a\in A} \vA^{a\dagger}\vA^a} \int_{-\infty}^{\infty} \labs{f(t)}^2\mathrm{d}t.
	\end{align}
The proof of \eqref{eq:contFourierIdentity2} is completely analogous.\footnote{Intuitively speaking the fourth line can be viewed as consequence of the Fourier representation of the Dirac delta distribution $\int_{-\infty}^{\infty} \e^{\ri \omega t} \mathrm{d}\omega = 2 \pi\delta(t)$, analogous to \eqref{eq:discreteDiracDelata}. Not introducing delta functions makes the proof completely general.} 
\end{proof}
\subsection{Secular approximation}\label{sec:secular}
In this section, we define the secular approximation of the Fourier Transformed operators $\hat{\vA}_{f}(\bomega)$ and analyze the resulting error. The secular approximation applies truncation to the Fourier-transformed operators in the frequency domain by suppressing Bohr frequencies $\nu \in B(\vH)$ that deviate substantially from the frequency label $\omega$ via some filter function $s\in \ell_\infty(\mathbb{R})$. For example truncation at energy difference $\mu$ can be achieved by setting $s(\omega):=\indicator(\labs{\omega} < \mu)$
and defining the following secular-approximated operators as follows
\begin{align}\label{eq:genSecDef}
	\hat{\vS}_{f,s}(\omega) := \sum_{\nu \in B(\vH)}  \hat{f}(\omega-\nu)s(\omega-\nu) \vA_{\nu}.
\end{align}
In some cases, we will consider alternative filter functions $s\colon \BR \rightarrow \BC$, so we will treat general $s$ throughout our discussion.

\begin{figure}[t]
	\centering
	\includegraphics[width=0.9\textwidth]{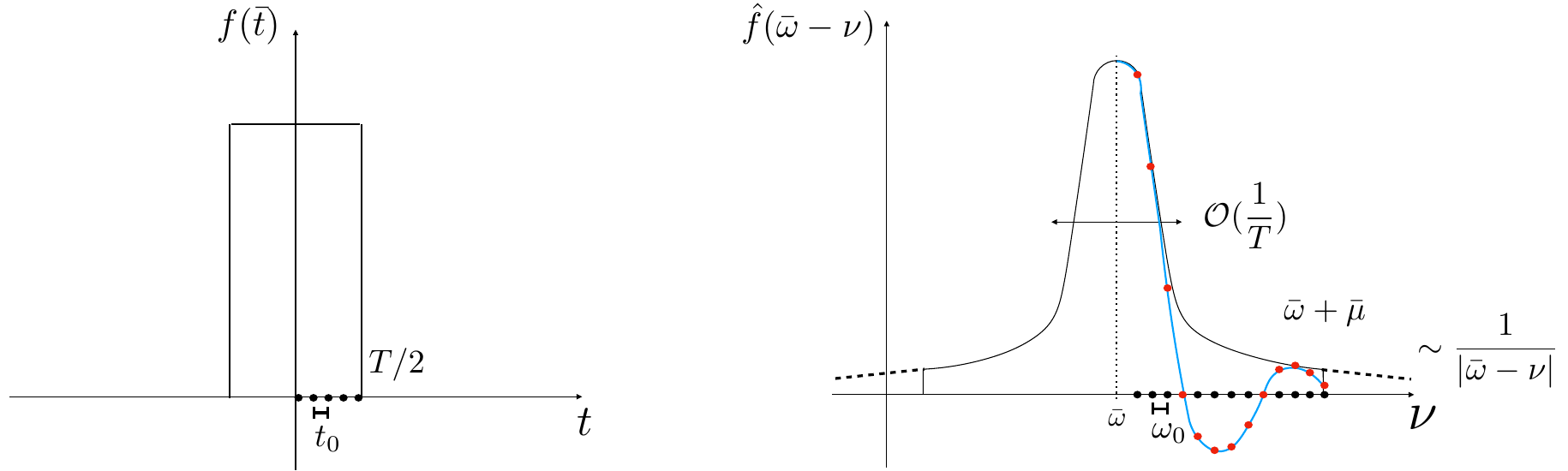}
	\caption{ Left: the weight as the step function. Right: An illustration for the Fourier Transformed amplitudes 
		$\vA_{\bomega} = \sum_{\nu \in S_{\omega_0}} \vA_{\nu} f(\bomega-\nu) =\sum_{\nu \in S_{\omega_0}} \vA_{\nu} \frac{\e^{\ri (\nu-\bomega) T} - \e^{ -\ri (\nu-\bomega) T}}{\e^{\ri  (\nu-\bomega) t_0} - 1}$
		for Bohr frequency $\nu$ given the energy label $\bomega$. The expression coincides with the phase estimation profile. It peaks near energy $\nu = \bomega$ with a width $\sim T^{-1}$ and decays polynomially. The profile in absolute value oscillates (blue), but we also display the norm bound to guide the eye (black). The secular approximation truncates the profile at an energy $\bmu$ far in the tail $\bmu \gg T^{-1}$. 
	}
	\label{fig:flat}
\end{figure}

The key observation in our analysis is that due to the definition of the operator Fourier Transform \eqref{eq:BohrDecomposedOFT} we have
\begin{align}\label{eq:contSecFourierTrunc}
	\hat{\vS}_{f,s}(\omega) = \hat{\vA}_{f_{s}}(\omega),
\end{align}
where $f_{s}=\CF^{-1}(\hat{f}\cdot s)$ is the inverse Fourier Transform of the function $\hat{f}(\omega)s(\omega)$.
Furthermore, this also implies that 
\begin{align}
	\hat{\vA}_{f}(\omega)-\hat{\vS}_{f,s}(\omega) = \hat{\vA}_{f-f_{s}}(\omega),
\end{align}
where $f-f_{s}=f_{1-s}$ is the inverse Fourier Transform of the function $\hat{f}(\omega)(1-s(\omega))$.
The significance of this equation is that it enables us to bound the error induced by the secular approximation via bounding the tail $\norm{\hat{f}(\omega)(1-s(\omega))}_2$.

\begin{prop}\label{prop:realSecReal}
	If $f(t)$ is real and $s(t)$ is real and even, then $f_s(t)$ is real.
\end{prop}
\begin{proof}
	If $f$ is real, then $\hat{f}$ has an even real and odd imaginary part, which remains true for $\hat{f}\cdot s$, and therefore $\CF^{-1}(\hat{f}\cdot s)$ is real as well.\footnote{If we work with the discrete Fourier Transform, and $N$ is even, then $-\omega_0 \cdot N/2$ (or $-t_0 \cdot N/2 $ if we work in the time domain) should be treated as its own inverse due to reasons of parity and modular arithmetic. In particular an even function can take arbitrary value on $-\omega_0 \cdot N/2$, but an odd function must be $0$, similarly to how such functions must behave on $0$.}
\end{proof}

Analogously, we define the secular approximation for the discrete Fourier Transform as 
\begin{align}\label{eq:discSecDef}
	\hat{\vS}_{f,s}(\bomega) :=\sum_{\nu \in B(\vH)}  \bar{\CF}\left(f(\bt) \cdot \e^{\ri \nu \bt}\right)(\bomega) \cdot s(\bomega - \nu) \vA_{\nu},
\end{align}
where $\bar{\CF}$ denotes the discrete Fourier Transform. In case $\nu$ is an integer multiple of the base frequency $\omega_0$ of $\bar{\CF}$, then $\bar{\CF}\left(f(\bt) \cdot \e^{(\ri \nu \bt)}\right)(\bomega)$ above simplifies to $\hat{f}(\bomega-\nu)=\bar{\CF}\left(f\right)(\bomega-\nu)$. Indeed, the discrete Fourier Transform $\hat{f}$ is only defined at points $\bomega \in S_{\omega_0}$, and thus translation by some value $\nu$ which is not an integer multiple of $\omega_0$ can cause troubles.

The simple analysis working nicely in the continuous case can be directly translated to the discrete case if all the Bohr frequencies are multiples of the base frequency $\omega_0$ of $\bar{\CF}$, i.e., $B\subset \omega_0 \BZ$.
This is the reason why we introduce a slightly rounded Hamiltonian in the analysis of the secular approximation.

\begin{prop}[Truncation via modifying weight function]\label{prop:TruncationByWeighing}
	Let $f\colon S_{t_0}\rightarrow \BC$ and $s\colon S_{\omega_0}\rightarrow \BC$ and suppose that $B(\vH)\subset \omega_0 \BZ$, then the secular-approximated operator can be represented as plain operator Fourier Transform corresponding to a perturbed weight function
	\begin{align}
		\hat{\vS}_{f,s}(\bomega) =  \hat{\vA}_{f_{s}}(\bomega)
	\end{align}
	where $f_{s}(\bomega)$ is the inverse discrete Fourier Transform of the function $\hat{f}(\bomega)s(\bomega)$ and
	$\hat{f}(\bomega) := \frac{1}{\sqrt{N}} \sum_{\bt\in S_{t_0}}  \e^{-\ri \bomega \bt} f(\bt)$.
\end{prop}
\begin{proof}
	\begin{align}
		\hat{\vS}_{f,s}(\bomega) &=\sum_{\bnu \in B(\vH)} \hat{f}(\bomega-\bnu) \cdot s(\bomega - \bnu) \vA_{\bnu} \\&
		=\sum_{\bnu \in B(\vH)}  \hat{f}_{s}(\bomega-\bnu) \vA_{\bnu}\\&
		=\sum_{\bnu \in B(\vH)} \frac{1}{\sqrt{N}}\sum_{\bt\in S_{t_0}} \e^{-\ri (\bomega - \nu) \bt} f_{s}(\bt) \vA_{\bnu}\\&
		=\frac{1}{\sqrt{N}}\sum_{\bt\in S_{t_0}}   \e^{-\ri \bomega \bt} f_{s}(\bt) \vA(\bt).
	\end{align}
	The last equality recombines the Bohr frequencies $\sum_{\bnu \in B(\vH)} \e^{\ri \bnu\bt} \vA_{\bnu}  = \vA(\bt)$ analogously to \eqref{eq:BohrDecomposedOFT}.
\end{proof}
To reiterate, the condition $B\subseteq \omega_0 \BZ$ need not hold for the original Hamiltonian $\vH$. Proceeding with the discretized Hamiltonian $\bvH$ introduces a small additive error.\footnote{This differs from the \emph{unphysical} rounding assumption~\cite{wocjan2021szegedy} where the Hamiltonian needs to have ``large'' gaps in the spectrum. Here, $\omega_0$ is not related to the resolution of energy estimates but rather the discretization of the register.} We present error bounds for both the \Lword{} and our discriminant proxy.

\begin{lem}[Perturbation bounds]\label{lem:OpFtDiffBound}
	Let $f,f'\in \mathbb{C}^{S_{t_0}}$, and $\gamma\in \mathbb{C}^{S_{\omega_0}}$ such that $\nrm{f}_2,\nrm{f'}_2,\nrm{\gamma}_\infty\leq 1$.
	If $\nrm{\sum_{a\in A}\vA^{a\dagg} \vA^a}\leq 1$, then for every $T>0$ and $f_T(t):=f(t)\indicator(|t|\leq T)$ we have that 
	\begin{align*}
		&\nrm{ \sum_{\bomega \in S_{\omega_0},a\in A}\kern-2.5mm \sqrt{\gamma(\bomega)} \ketbra{\bomega,a}{\bar{0}}\otimes\hat{\vA}^a_{(f,\vH)}(\bomega)-
			\sqrt{\gamma(\bomega)} \ketbra{\bomega,a}{\bar{0}}\otimes\hat{\vA}^a_{(f',\vH')}(\bomega)}
		\leq \nrm{f-f_T}_2+\nrm{f_T-f'}_2+2T\nrm{\vH-\vH'}.
	\end{align*}
\end{lem}
\begin{proof}
	This directly follows from the (not necessarily unitary) block-encoding construction of \autoref{fig:OQFT}-\autoref{fig:L_circuit} and triangle inequalities. Indeed, for any function $g\in \mathbb{C}^{S_{t_0}}$ let $\vec{Prep}_g:=\ketbra{g}{\bar{0}}$ and let $\vB_{(g,\vH)}$ denote the block-encoding given by \autoref{fig:L_circuit} when setting $\vec{Prep}\leftarrow \vec{Prep}_g$,  $\vV_{jp}\leftarrow \sum_a \ket{a}\otimes\vA^a$, and using the Hamiltonian $\vH$. Then, we have that 
	\begin{align*}
		\nrm{\vB_{(g,\vH)}-\vB_{(f_T,\vH')}}&
		\leq\nrm{\vB_{(g,\vH)}-\vB_{(f_T,\vH)}}+\nrm{\vB_{(f_T,\vH)}-\vB_{(f_T,\vH')}}\\&
		\leq\nrm{g-f_T}_2 \nrm{\vV_{jp}} +2T\nrm{\vH-\vH'}\nrm{\vV_{jp}}\\&
		\leq\nrm{g-f_T}_2 +2T\nrm{\vH-\vH'}.
	\end{align*}
	Using the above inequality twice we obtain the desired result
	\begin{align*}
		\nrm{ \sum_{\bomega \in S_{\omega_0}, a\in A} \sqrt{\gamma(\bomega)} \ketbra{\bomega,a}{\bar{0}}\otimes\left(\hat{\vA}^a_{(f,\vH)}(\bomega)-
			\hat{\vA}^a_{(f',\vH')}(\bomega)\right)}
		&=\nrm{\vB_{(f,\vH)}-\vB_{(f',\vH')}}\\&
		\leq\nrm{\vB_{(f,\vH)}-\vB_{(f_T,\vH')}}+\nrm{\vB_{(f_T,\vH')}-\vB_{(f',\vH')}}\\&
		\leq\nrm{f-f_T}_2+\nrm{f_T-f'}_2+2T\nrm{\vH-\vH'}.\tag*{\qedhere}
	\end{align*}
\end{proof}
\begin{cor}[Perturbation bounds on Lindbladians and discriminant proxies]\label{cor:fHPerturBound}
	Consider
	\begin{align*}
		\bar{\CL}_{(f,\vH)}&:= \kern-3mm\sum_{a\in A, \bomega\in S_{\omega_0}} \kern-3mm \gamma(\bomega) \L( \hat{\vA}^{a}_{(f,\vH)}(\bomega)[\cdot] \hat{\vA}^{a}_{(f,\vH)}(\bomega)^{\dagger} -\frac{1}{2} \{\hat{\vA}_{(f,\vH)}^{a}(\bomega)^{\dagger} \hat{\vA}_{(f,\vH)}^{a}(\bomega) ,\cdot \}\R),\quad\text{and}\\
		\bar{\vec{\CD}}_{(f,\vH)}&:=  \kern-3mm \sum_{ a\in A, \bomega\in S_{\omega_0}} \kern-3mm \sqrt{\gamma(\bomega)\gamma(-\bomega)}  \hat{\vA}^{a}_{(f,\vH)}(\bomega) \otimes  \hat{\vA}^{a}_{(f,\vH)}(\bomega)^{*} - \frac{\gamma(\bomega)}{2} \L(\hat{\vA}^{a}_{(f,\vH)}(\bomega)^\dagg\hat{\vA}^{a}_{(f,\vH)}(\bomega)\otimes \vI + \vI \otimes \hat{\vA}^{a}_{(f,\vH)}(\bomega)^{*\dagger}\hat{\vA}^{a}_{(f,\vH)}(\bomega)^{*}\R).
	\end{align*}
	Assuming the conditions and the notation of \autoref{lem:OpFtDiffBound} hold, we have 
	\begin{align*}
		\nrm{\bar{\CL}_{(f,\vH)} - \bar{\CL}_{(f',\vH')}}_{1-1}\leq 4 \L(\nrm{f-f_T}_2+\nrm{f_T-f'}_2+2T\nrm{\vH-\vH'}\R).
	\end{align*}
	Further assuming the symmetry and normalization conditions~\eqref{eq:AAdagger},\eqref{eq:fnormalized} are satisfied by $f,f'$ and $\{\vA^a\colon a\in A\}$, we have 
	\begin{align*}
		\norm{\bar{\vec{\CD}}_{(f,\vH)} - \bar{\vec{\CD}}_{(f',\vH')}} \le 4 \L(\nrm{f-f_T}_2+\nrm{f_T-f'}_2+2T\nrm{\vH-\vH'}\R).
	\end{align*}
\end{cor}
\begin{proof}
	The superoperator $\CL_{(f,\vH)}$ acts as follows
	\begin{align}
		\bar{\CL}_{(f,\vH)}[\vrho] &= \tr_{a,\bomega}\L[ \L(\sum_{a\in A,\bomega\in S_{\omega_0}} \sqrt{\gamma(\bomega)}\ketbra{\bomega,a}{\bar{0}} \hat{\vA}^a(\bomega)  \R)\ketbra{\bar{0}}{\bar{0}}\otimes\vrho \L(\sum_{a\in A,\bomega\in S_{\omega_0}}\sqrt{\gamma(\bomega)}\ketbra{\bar{0}}{\bomega,a}\hat{\vA}^a(\bomega)^{\dagger}\R) \R] \\
		&- \frac{1}{2} \tr_{\bar{0}}\L\{ \L(\sum_{a\in A,\bomega\in S_{\omega_0}}\sqrt{\gamma(\bomega)}\ketbra{\bar{0}}{\bomega,a}\hat{\vA}^a(\bomega)^{\dagger}\R)\L(\sum_{a\in A,\bomega\in S_{\omega_0}} \sqrt{\gamma(\bomega)}\ketbra{\bomega,a}{\bar{0}} \hat{\vA}^a(\bomega)  \R), \ketbra{\bar{0}}{\bar{0}}\otimes\vrho \R\}. 
	\end{align}
	The conclusion about $\nrm{\bar{\CL}_{(f,\vH)} - \bar{\CL}_{(f',\vH')}}_{1-1}$ follows from \autoref{lem:OpFtDiffBound} using the triangle and Hölder inequalities and that taking partial trace contracts trace-distance.
	
	The proof of \autoref{prop:discriminantBlock} shows that if $\vB$ is a (nonunitary) block-encoding of $\bar{\CL}_{(f,\vH)}$, then $\vB^{'\dagger} \cdot (\vR-\vI) \cdot \vB'$ is a block-encoding of $\bar{\vec{\CD}}_{(f,\vH)}$. Since $\nrm{f}_2,\nrm{f'}_2,\nrm{\gamma}_\infty\leq 1$ without loss of generality we can assume $\nrm{\vB'}\leq 1$, which together with $\nrm{\vR-\vI}\leq 2$ implies the bound on $\norm{\bar{\vec{\CD}}_{(f,\vH)} - \bar{\vec{\CD}}_{(f',\vH')}}$ via a triangle inequality.
\end{proof}

Note that under the conditions of \autoref{thm:discToCont}, the same bounds also hold in the continuous case, as can be shown by a limit argument using the results of \autoref{thm:discToCont}.

\begin{cor}[Perturbation bounds on continuous Lindbladians and discriminant proxies]\label{cor:fHContPerturBound}
	Consider
	\begin{align*}
		\CL_{(f,\vH)}&:= \kern-1mm\sum_{a\in A}\int_{-\infty}^\infty \kern-1mm \gamma(\omega) \L( \hat{\vA}^{a}_{(f,\vH)}(\omega)[\cdot] \hat{\vA}^{a}_{(f,\vH)}(\omega)^{\dagger} -\frac{1}{2} \{\hat{\vA}_{(f,\vH)}^{a}(\omega)^{\dagger} \hat{\vA}_{(f,\vH)}^{a}(\omega) ,\cdot \}\R)\mathrm{d}\omega,\quad\text{and}\\
		\vec{\CD}_{(f,\vH)}&:=  \kern-1mm \sum_{a\in A}\int_{-\infty}^\infty \kern-1mm \sqrt{\gamma(\omega)\gamma(-\omega)}  \hat{\vA}^{a}_{(f,\vH)}(\omega) \otimes  \hat{\vA}^{a}_{(f,\vH)}(\omega)^{*} - \frac{\gamma(\omega)}{2} \L(\hat{\vA}^{a}_{(f,\vH)}(\omega)^\dagg\hat{\vA}^{a}_{(f,\vH)}(\omega)\otimes \vI + \vI \otimes \hat{\vA}^{a}_{(f,\vH)}(\omega)^{*\dagger}\hat{\vA}^{a}_{(f,\vH)}(\omega)^{*}\R)\mathrm{d}\omega.
	\end{align*}
	If $\gamma\in \ell_\infty(\mathbb{R})$, $f\in \ell_2(\mathbb{R})$, and $\gamma$, $f$ are continuous almost everywhere (i.e., the set of points of discontinuity has measure zero) while $f$ is bounded on every finite interval, then
	assuming the conditions and the notation of \autoref{lem:OpFtDiffBound} hold, we have 
	\begin{align*}
		\lnormp{\CL_{(f,\vH)} - \CL_{(f',\vH')}}{1-1}\leq 4 \L(\nrm{f-f_T}_2+\nrm{f_T-f'}_2+2T\nrm{\vH-\vH'}\R).
	\end{align*}
	Further assuming the symmetry and normalization conditions~\eqref{eq:AAdagger},\eqref{eq:fnormalized} are satisfied by $f,f'$ and $\{\vA^a\colon a\in A\}$, we have 
	\begin{align*}
		\norm{\vec{\CD}_{(f,\vH)} - \vec{\CD}_{(f',\vH')}} \le 4 \L(\nrm{f-f_T}_2+\nrm{f_T-f'}_2+2T\nrm{\vH-\vH'}\R).
	\end{align*}
\end{cor}
\begin{proof}
	The objects $\CL_{(f,\vH)},\CL_{(f',\vH')},\vec{\CD}_{(f,\vH)},\vec{\CD}_{(f,\vH)}$ can be obtained as limits of their respective discretizations $\bar{\CL}_{(f,\vH)},\bar{\CL}_{(f',\vH')},\bar{\vec{\CD}}_{(f,\vH)},\bar{\vec{\CD}}_{(f,\vH)}$ as per \autoref{thm:discToCont}, for which the discretized versions of these bounds hold due to \autoref{cor:fHContPerturBound}. As shown in the proof of \autoref{thm:discToCont}, if a function $g\in \ell_2(\mathbb{R})$ is continuous almost everywhere while also bounded on every finite interval, then $\nrm{g}_2^2=\lim_{K\rightarrow \infty}\lim_{N\rightarrow \infty}\sum_{\bt\in S^{\lceil N  \rfloor}_{t_0}} |\bar{g}_K(\bt)|^2$, where $t_0=\sqrt{2\pi/N}$ and $\bar{g}_K(t)=\sqrt{t_0}g(t)\indicator(|t|\leq K)$, therefore the RHS of the discretized bounds also converge to their continuous counterpart implying the validity of the continuous versions of these bounds.
\end{proof}

\begin{restatable}[Secular approximation]{lem}{SecularApprox}\label{lem:secular}
	Let $\norm{\gamma}_{\infty} \le 1$, 
	and consider the \Lword{} $\CL_{\beta}$~\eqref{eq:LbetaAgain} and discriminant $\vec{\CD}_{\beta}$~\eqref{eq:vec_main_result} with $\hat{\vA}^a(\bomega)$ being the operator Fourier Transforms of $\vA^a$ with $\vH$ and their secular approximations
	\begin{align}
		\CL_{sec}[\cdot] &= 
		\sum_{a\in A, \bomega\in S_{\omega_0}} \gamma(\bomega) \L( \hat{\vS}^a(\bomega)[\cdot] \hat{\vS}^a(\bomega)^{\dagger} -\frac{1}{2} \{\hat{\vS}^a(\bomega)^{\dagger} \hat{\vS}^a(\bomega) ,\cdot \}\R), \\		
		\vec{\CD}_{sec}&= \sum_{a\in A, \bomega \in S_{\omega_0} } \sqrt{\gamma(\bomega)\gamma(-\bomega)} \hat{\vS}^a(\bomega)\otimes \hat{\vS}^{a}(\bomega)^*- \frac{\gamma(\bomega)}{2} \L(\hat{\vS}^a(\bomega)^{\dagger}\hat{\vS}^a(\bomega)\otimes \vI + \vI \otimes \hat{\vS}^{a}(\bomega)^{\dagger*}\hat{\vS}^{a}(\bomega)^*\R),
	\end{align}	
	with the operators
 \begin{align}
    \hat{\vS}^a(\bomega):=\sum_{\bnu \in B(\bvH)}  \hat{f}_s(\bomega-\bnu) \vA^a_{\bnu}  \quad \text{where}\quad \hat{f}_s(\bomega):=\hat{f}(\bomega) \cdot \indicator(\labs{\bomega} < m \omega_0)
 \end{align}
 defined by the discretized Hamiltonian $\bvH$ and cut-off frequency $m \omega_0$. If $\nrm{f}_2\leq 1$ and $\nrm{\sum_{a\in A}\vA^{a\dagg} \vA^a}\leq 1$, then for every $T>0$ and $f_T(t):=f(t)\indicator(|t|\leq T)$, we have 
	\begin{align}   
		\lnormp{\CL_{\beta}  - \CL_{sec}}{1-1} \le 4\lVert\hat{f}-\hat{f}_s\rVert_2 + 8\nrm{f-f_T}_2 + 4T\omega_0.
	\end{align}
	Moreover, assuming the symmetry and normalization conditions~\eqref{eq:AAdagger},\eqref{eq:fnormalized}, we have $\vec{\CD}_{sec} = \vec{\CD}_{sec}^{\dagger}$ and
	\begin{align}
		\norm{\vec{\CD}_{\beta} - \vec{\CD}_{sec}} \le 4\lVert\hat{f}-\hat{f}_s\rVert_2 + 8\nrm{f-f_T}_2 + 4T\omega_0.
	\end{align}
\end{restatable}

The truncation introduces an error scaling with the tail in the frequency domain, while the last two error terms arise from discretizing the Hamiltonian spectrum $\vH \rightarrow \bvH$ for discrete Fourier Transforms; this is more of a technical artifact and merely introduces a minor error shrinking with finer Fourier frequency resolution $\omega_0$.

\begin{proof}
	 Since the secular approximation amounts to changing the real function and discretizing the Hamiltonian, i.e., $\CL_{sec}=\CL_{(f_{s},\bvH)}$ we can apply \autoref{cor:fHPerturBound}. The final bound follows using the observation that $\nrm{\vH - \bvH}\leq \omega_0/2$, and $\nrm{f-f_s}=\lVert \hat{f}-\hat{f}_s\rVert$ since the discrete Fourier Transformation is unitary.
	 
	 Finally, since $\vec{\CD}_{sec} = \vec{\CD}_{(f_{s},\bvH)}$ we have that $\vec{\CD}_{sec}$ is self-adjoint due to \autoref{prop:realSecReal} and \autoref{cor:DiscProxySelfAdjoint}. 
\end{proof}

Let us quickly note that the above argument controls the implementation error for truncating the Hamiltonian simulation. Indeed, in practice, we will only implement Hamiltonian simulation up to time $T$, and this is perfectly accounted for by \autoref{lem:OpFtDiffBound},\autoref{cor:fHPerturBound} by setting $f' = f_T$ and $\vH' = \vH$.

\subsection{Uniform weights}
Consider the simplest Fourier Transform with uniform weights
\begin{align}
    \hat{\vA}_{f}(\bomega): = \frac{1}{\sqrt{2NT/t_0}} \sum_{ -T\le \bt < T }  \e^{-\ri \bomega \bt} \vA(\bt)
    .\label{eq:unif_weight} 
\end{align}

\begin{prop}[Preparing uniform weights]
Suppose that $T/t_0 = 2^k$. Then, the state
\begin{align}
    \sum_{\bt }f(\bt)\ket{\bt} \quad \text{for}\quad     f(\bt) = \frac{1}{\sqrt{2T/t_0}} \cdot \begin{cases}
     1 \quad&\text{if}\quad -T \le \bt < T \\
     0\quad&\text{else}.
    \end{cases}
\end{align}
can be prepared using $k+1$ Hadamard gates and $n-k-1$ CNOT gates.
\end{prop}
\begin{proof}
Prepare with the GHZ state on the first $n-k$ qubits using 1 Hadamard gate and $n-k-1$ CNOT gates
\begin{align}
    \frac{1}{\sqrt{2}}\L(\ket{1^{n-k}}+ \ket{0^{n-k}} \R) \ket{0^k}
\end{align}
and then apply Hadamard gates on the last $k$ qubits.
\end{proof}

Since the weights are real $f(t) = f^*(t)$ and normalized $\sum_{\bt\in S_{t_0}}\labs{f(\bt)}^2 = 1$, the transformed operator satisfies the properties listed in \autoref{prop:properties_DFT}. 
\subsection{Gaussian ansatz}\label{sec:Gaussian_ansatz}
Instead of the plain Fourier Transform, consider the Gaussian-weighted Fourier Transform
\begin{align}
    \hat{\vA}_{f}(\bomega): = \frac{1}{\sqrt{N}} \sum_{ \bt \in S_{t_0}}\e^{-\ri \bomega \bt} f(\bt)\vA(\bt) \quad \text{for} \quad f(\bt):=\frac{1}{\sqrt{\sum_{\bt \in S_{t_0}} \e^{-\frac{\bt^2}{2\sigma_t^2}}}} \sum_{\bt \in S_{t_0}} \e^{-\frac{\bt^2}{4\sigma_t^2}} \ket{\bt}. \label{eq:Gaussian_weight} 
\end{align}

Again, since the weight is real $f(t) = f^*(t)$ and normalized $\int_{-\infty}^{\infty}\labs{f(t)}^2\mathrm{d}t =1$, the transformed operator satisfies the symmetry properties listed in \autoref{prop:properties_DFT}. To implement the above operator, we just need to prepare the initial state approximately.

\begin{prop}[Preparing a truncated Gaussian state~\cite{McArdle_2022quantumstate}]
Suppose $(Nt_0)^2/16 \sigma_t^2\ge \log(1/\epsilon)$. Then, the state
\begin{align}
\frac{1}{\sqrt{\sum_{\bt \in S_{t_0}} \e^{-\frac{\bt^2}{2\sigma_t^2}}}} \sum_{\bt \in S_{t_0}} \e^{-\frac{\bt^2}{4\sigma_t^2}} \ket{\bt}
\end{align}
can be prepared using $\CO(n\log(1/\epsilon)^{5/4})$ gates up to error $\epsilon$. 
\end{prop}

The main advantage of using a Gaussian weight is that its Fourier Transform remains a Gaussian, which has a rapidly decaying tail. Indeed, for the \emph{continuous} Gaussian, we can evaluate the Gaussian integral by completing the square
\begin{align}
    \frac{1}{\sqrt{\sigma_t\sqrt{2\pi}}} \int_{-\infty}^{\infty} \e^{-\ri \omega t} \e^{-\frac{t^2}{4\sigma_t^2}} \mathrm{d}t
    &= \sqrt{\sigma_t\sqrt{2\pi}} \e^{-\omega^2\sigma^2_t}.
\end{align}
The uncertainty in energy is inversely proportional to the uncertainty in time $\sigma_t^{-1}$, as a manifestation of the energy-time uncertainty principle.

\subsection{Discretizing continuous functions via periodic summation}
It is not obvious how to carefully derive bounds on the discretization errors that appear in the Riemann sums of the discrete Fourier Transform as the Fourier phases are highly oscillatory. Nevertheless, the discrete Fourier Transform remains Gaussian, up to a well-controlled error.
\begin{prop}[DFT of Gaussian]\label{prop:DFT of Gaussian}
There is a choice of parameter $(Nt_0)^2/\sigma_t^2= \Omega( \log(1/\epsilon))$ and $(N\omega_0)^2\sigma_t^2= \Omega( \log(1/\epsilon))$ such that the discrete Fourier Transform for 
\begin{align}
f(\bt) = \frac{1}{\sqrt{\sum_{\bt \in S_{t_0}} \e^{-\frac{\bt^2}{2\sigma_t^2}}}} \e^{-\frac{\bt^2}{4\sigma_t^2}} \quad \text{is approximately} \quad \frac{1}{\sqrt{\sum_{\bomega \in S_{\omega_0}} \e^{-\bomega^2\sigma_t^2}}} \e^{-\bomega^2\sigma^2_t}
\end{align}
up to error $\CO( \frac{1}{Nt_0\sigma_t}\e^{ - N^2v_0^2\sigma_t^2/2} + \frac{1}{ Nt_0/\sigma_t}\e^{ - N^2t_0^2/16\sigma_t^2} )$ in 2-norm.
\end{prop}

In order to relate the continuous Fourier Transform to the discrete one, we apply the discretization to a continuous-variable function after periodic summation. 
This is related to the \emph{Poisson Summation Formula}~\cite[Chapter II \S 13]{zygmund2003TrigSeries}, but pushes the idea one step further to the realm of discrete Fourier Transform. Similar ideas are used, e.g., in lattice cryptography (c.f.,~\cite{Regev2009OnLL}), but we include a self-contained treatment for completeness. 

To state the following general result we introduce the notation $\ell_1(\mathbb{R})$ for (Lebesque) integrable $\mathbb{R}\rightarrow \mathbb{C}$ functions.

\begin{fact}\label{fact:periodic_sum}
	Consider the Fourier Transform $\hat{f}(\omega) := \frac{1}{\sqrt{2\pi}}\int_{-\infty}^{\infty} \e^{-\ri \omega t} f(t)\mathrm{d}t$ of a function $f\in \ell_1(\mathbb{R})$.
	Suppose that a ``wrapped around'' version of $f$ can be defined such that $p(t) = \sum_{n = -\infty}^{\infty} f(t + n N t_0)$ for almost every $t\in \mathbb{R}$ (i.e., the set of points where the equality does not hold has Lebesgue measure $0$), $p(t)$ is continuous at every $\bt \in S_{t_0}$ and Riemann integrable on the interval $[-\frac{N}{2}t_0,\frac{N}{2}t_0]$. If the sequence $k\cdot \hat{f}(k\omega_0)\colon k\in \mathbb{Z}$ is bounded in absolute value, then the limit $\hat{p}(\bomega):=\lim_{B \rightarrow \infty}\sum_{\ell=-B}^{B} \hat{f}(\bomega+\ell N \omega_0)$ exists for every $\bomega\in  S_{\omega_0}$ and $\hat{p}(\bomega)$ is the discrete Fourier Transform of $p(\bt)$, i.e., 
	\begin{align}
	\hat{p}(\bomega)=\frac{t_0}{\sqrt{2\pi}}\sum_{\bt \in S_{t_0}} \e^{-\ri \bomega \bt} p(\bt) . 
	\end{align}
\end{fact}
\begin{proof}
	For every $\omega \in  \omega_0\cdot\mathbb{Z}$ we have that
	\begin{align}
	\hat{f}(\omega) &= \frac{1}{\sqrt{2\pi}}\int_{-\infty}^{\infty}\e^{-\ri \omega t} f(t)\mathrm{d}t  \tag*{(since $f\in \ell_1(\mathbb{R})$)}\\&
	= \frac{1}{\sqrt{2\pi}}\sum_{n = -\infty}^{\infty}\int_{-\frac{N}{2}t_0}^{\frac{N}{2}t_0} \e^{-\ri \omega (t + nNt_0) } f(t + nNt_0)\mathrm{d}t \tag*{(by Fubini's theorem since $f\in \ell_1(\mathbb{R})$)}\\&
	= \frac{1}{\sqrt{2\pi}}\sum_{n = -\infty}^{\infty}\int_{-\frac{N}{2}t_0}^{\frac{N}{2}t_0}\e^{-\ri \omega t} f(t + nNt_0)\mathrm{d}t  \tag*{(since $\omega \in \omega_0\cdot\mathbb{Z}$ and $\omega_0 t_0=\frac{2\pi}{N}$)}\\&	
	= \frac{1}{\sqrt{2\pi}}\int_{-\frac{N}{2}t_0}^{\frac{N}{2}t_0}\sum_{n = -\infty}^{\infty}  \e^{-\ri \omega t} f(t + nNt_0)\mathrm{d}t  \tag*{(by Fubini's theorem since $f\in \ell_1(\mathbb{R})$)}\\&
	= \frac{1}{\sqrt{2\pi}}\int_{-\frac{N}{2}t_0}^{\frac{N}{2}t_0} \e^{-\ri \omega t} p(t)\mathrm{d}t \tag*{(since $p(t) = \sum_{n = -\infty}^{\infty} f(t + nNt_0)$ a.s.)}.
	\end{align}
	The above equality means that $c_k:=\frac{\sqrt{2\pi}}{N t_0}\cdot \hat{f}(k\cdot\omega_0)$ is the Fourier series of $p(t)$. Since 
	$p$ is Riemann integrable on the interval $[-\frac{N}{2}t_0,\frac{N}{2}t_0]$, continuous at every $\bt \in S_{t_0}$ point and $k \cdot c_k$ is bounded by assumption, we have for all $\bt\in S_{t_0}$ that
	\begin{align}
	p(\bt)&= \lim_{B \rightarrow \infty} \sum_{k=-B}^B c_k \e^{\ri \bt k \omega_0} \tag*{(due to \cite[Theorem 15.3]{korner1988FourierAnalysis})}\\&	
	= \lim_{B \rightarrow \infty} \sum_{k=-BN}^{BN} c_k \e^{\ri \bt k  \omega_0}\\&		
	= \lim_{B \rightarrow \infty} \sum_{k=-BN-\left\lceil(N-1)/2\right\rceil}^{BN+\left\lfloor(N-1)/2\right\rfloor} c_k \e^{\ri \bt k  \omega_0} \tag*{(since $|c_k|=\bigO{\frac{1}{k}}$)}\\&	
	= \lim_{B \rightarrow \infty} \sum_{\ell=-B}^{B} \sum_{j=-\left\lceil(N-1)/2\right\rceil}^{\left\lfloor(N-1)/2\right\rfloor} c_{\ell N + j} \e^{\ri \bt (\ell N + j)  \omega_0} \tag*{(set $k=\ell N + j$)}\\&		
	= \lim_{B \rightarrow \infty} \sum_{\ell=-B}^{B} \sum_{\bomega\in S_{\omega_0}} \frac{\sqrt{2\pi}}{N t_0}\hat{f}\left(\ell N\omega_0+\bomega\right) \e^{\ri \bt  \bomega}.\tag*{(set $\bomega=j\omega_0$; use $\bt \in  t_0\cdot\mathbb{Z}$ and $\omega_0 t_0=\frac{2\pi}{N}$)}
	\end{align}
	Finally, for all $\bomega'\in S_{\omega_0}$ we have
	\begin{align}
		\frac{t_0}{\sqrt{2\pi}}\sum_{\bt\in S_{t_0}}\e^{-\ri\bt \bomega'}p(\bt)&
		=\sum_{\bt\in S_{t_0}}\e^{-\ri\bt \bomega'}\lim_{B \rightarrow \infty} \sum_{\ell=-B}^{B} \sum_{\bomega\in S_{\omega_0}} \frac{1}{N}\hat{f}\left(\bomega + \ell N\omega_0\right) \e^{\ri \bt  \bomega}\\&
		=\lim_{B \rightarrow \infty} \sum_{\ell=-B}^{B} \sum_{\bomega\in S_{\omega_0}} \frac{1}{N}\hat{f}\left(\bomega + \ell N\omega_0\right) \sum_{\bt\in S_{t_0}}\e^{-\ri\bt \bomega'}\e^{\ri \bt  \bomega}\\&		
		=\lim_{B \rightarrow \infty} \sum_{\ell=-B}^{B} \hat{f}\left(\bomega' + \ell N\omega_0\right) \tag*{$\left(\text{since }\sum_{\bt\in S_{t_0}}\e^{\ri \bt  (\bomega-\bomega')}=N\cdot\delta_{\bomega, \bomega'}\right)$}\\&
		= \hat{p}(\bomega').\tag*{\qedhere}
	\end{align}	
\end{proof}

Using the above, we prove \autoref{prop:DFT of Gaussian}.
\begin{proof}[Proof of \autoref{prop:DFT of Gaussian}]
Apply Fact~\ref{fact:periodic_sum} for the Gaussian $f(t) =
\sqrt{g_{\sigma}(t)} = \frac{1}{\sqrt{\sigma_t\sqrt{2\pi}}}\e^{-\frac{t^2}{4\sigma_t^2}}$. The problem reduces to implementing the periodic sum $p(\bt)$ (\autoref{fig:Gaussian}) approximately. Up to a small error from the Gaussian tail, we may only keep the centered Gaussian $(n=0)$ as long as the Gaussian is largely confined in the window for both the time domain $(Nt_0)^2 \gg \sigma_t^2$ and frequency domain $(N\omega_0)^2\gg\sigma_t^2$.
\end{proof}

\begin{figure}[t]
    \centering
    \includegraphics[width=0.9\textwidth]{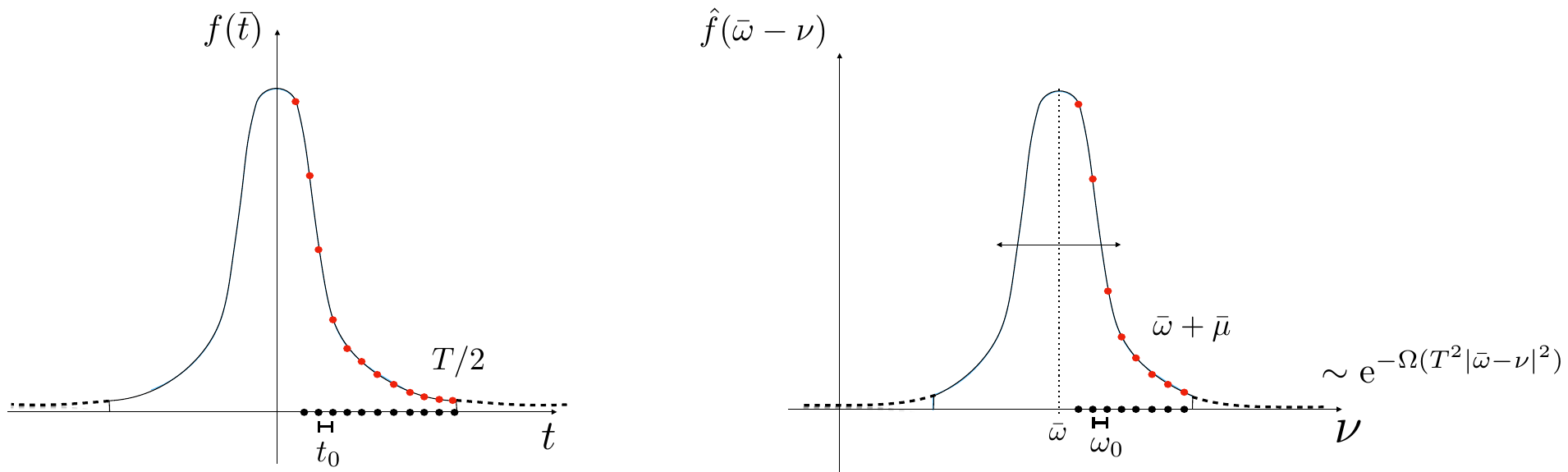}
    \caption{ Left: the weight function as (approximately) the Gaussian distribution but truncated at $T$. Right: An illustration for the Fourier Transformed amplitudes, which is also (approximately) Gaussian. It peaks near energy $\nu = \bomega$ with a width $\sim T^{-1}$ and decays exponentially. The secular approximation truncates the profile at an energy $\bmu$ in the tail $\bmu \sim T^{-1}$. 
    }
    \label{fig:Gaussian}
\end{figure}
\subsection{Tail bounds}\label{sec:tail_bounds}
We evaluate the tail bounds that appear in the analysis of the secular approximation. 
First, we consider the case of the uniform weights whose Fourier Transform has a heavy tail impacting the accuracy.

\begin{prop}[Tail bound for uniform weights]\label{prop:uniformTail}
	Let	$f(\bt):=\frac{\indicator(-T \le \bt < T)}{\sqrt{2T/t_0}}\e^{\ri\nu\bt}$. Then, its discrete Fourier Transform is 
	\begin{align}
		\hat{f}(\bomega)=\frac{1}{\sqrt{2TN/t_0}}\frac{\e^{\ri  (\nu-\bomega) T} - \e^{ - \ri  (\nu-\bomega) T}}{\e^{\ri  (\nu-\bomega) t_0} - 1}
	\end{align}
	with a tail bound
	\begin{align}
		\sum_{\labs{\bomega} > m \omega_0} \labs{\hat{f}(\bomega)}^2\le\frac{ \pi}{2m\omega_0 T}.
	\end{align}
\end{prop}
\begin{proof}
The Fourier Transform gives a geometric series with ratio $\e^{\ri (\nu-\bomega) \bt}$
\begin{align}
    \frac{1}{\sqrt{N}}\sum_{\bt \in S_{t_0}}\e^{-\ri \bomega \bt}f(\bt) = \frac{1}{\sqrt{2TN/t_0}} \sum_{-T \le \bt < T} \e^{\ri (\nu-\bomega) \bt} = \frac{1}{\sqrt{2TN/t_0}}\frac{\e^{\ri  (\nu-\bomega) T} - \e^{ - \ri  (\nu-\bomega) T}}{\e^{\ri  (\nu-\bomega) t_0} - 1}.
\end{align}
This function scales inversely with $\bomega$
\begin{align}
		\sum_{\labs{\bomega} > m \omega_0} \labs{\hat{f}(\bomega)}^2  &\le \frac{1}{2NT/t_0}\sum_{\labs{\bomega} > m \omega_0} \frac{4}{\labs{\e^{\ri \bomega t_0 } - 1}^2}\\&
		\le \frac{t_0}{2NT} \sum_{\labs{\bomega} > m \omega_0} \frac{\pi^2}{(\bomega t_0)^2}\tag*{(since $\labs{\e^{\ri x} - 1} \ge \frac{2}{\pi} \labs{x}$ for $x\in[-\pi, \pi]$)}\\& 
		\le \frac{\pi^2}{NTt_0\omega_0^2} \sum_{n = m + 1 }^\infty \frac{1}{n^2}\\& 
		\le \frac{\pi}{2T\omega_0} \sum_{n = m + 1 }^\infty \frac{1}{n(n-1)}\tag*{(since $\omega_0 t_0=\frac{2\pi}{N}$)}\\&
		= \frac{ \pi}{2m\omega_0 T}. \tag*{\qedhere}
\end{align}
\end{proof}
In retrospect, it is very important that we consider the 2-norm of the tail here; the 1-norm would be divergent.

Now, we consider Gaussians, which have a rapidly decaying tail, greatly improving the accuracy.
\begin{restatable}[Tail bound for Gaussian weights]{prop}{GaussianTail}\label{prop:Gaussian_tail}
	For the function $f(\bt) = ({\sum_{\bt \in S_{t_0}} \e^{-\frac{\bt^2}{2\sigma_t^2}}})^{-1/2} \e^{-\frac{\bt^2}{4\sigma_t^2}}$, the Fourier-transformed tail satisfies
	\begin{align}
		\sqrt{\sum_{\underset{\labs{\bomega}\ge \bmu}{\bomega\in S^{\lceil N \rfloor}_{\omega_0}} } \labs{\hat{f}(\bomega)}^2} \le \CO\L(\frac{1}{\sqrt{N\omega_0\sigma_t}}\e^{ - N^2\omega_0^2\sigma_t^2/2} + \frac{1}{ \sqrt{Nt_0/\sigma_t}}\e^{ - N^2t_0^2/16\sigma_t^2} +\frac{1}{\sqrt{\bmu\sigma_t}}\e^{ - \bmu^2\sigma_t^2}\R).
	\end{align}
\end{restatable}
\begin{proof}
	By \autoref{prop:DFT of Gaussian}, the discrete Fourier Transform $\hat{f}(\bomega)$ is approximately
	$\frac{1}{\sum_{\bomega \in S_{\omega_0}} \e^{-\bomega^2\sigma_t^2}} \e^{-2\bomega^2\sigma^2_t}$ up to error 
	$\CO( \frac{1}{\sqrt{N\omega_0\sigma_t}}\e^{ - N^2v_0^2\sigma_t^2/2} + \frac{1}{ \sqrt{Nt_0/\sigma_t}}\e^{ - N^2t_0^2/16\sigma_t^2} )$. We then control the tail bound
	\begin{align}
		\frac{1}{\sum_{\bomega \in S_{\omega_0}} \e^{-\bomega^2\sigma_t^2}} \sum_{\labs{\bomega}\ge \bmu }  \e^{-2\bomega^2\sigma^2_t} = \CO(\frac{1}{\bmu\sigma_t}\e^{ - 2\bmu^2\sigma_t^2}).\tag*{\qedhere}
	\end{align}
\end{proof}

\section{Proving approximate detailed balance}\label{sec:error_boltzmann}
In this section, we prove approximate detailed balance (or discriminant proxy) for the constructed discriminant. It amounts to controlling the error arising from Boltzmann factors due to the finite resolution of the operator Fourier Transform. 

\subsection{A simpler but weaker bound}\label{sec:weaker}
We begin with a simpler but weaker bound. This will be enough for the Gaussian-damped discriminant due to its rapidly decaying tail. We can bootstrap this weaker bound using a more careful truncation scheme for the special case of uniform weight (which has a heavy tail) as shown in \autoref{sec:improved}. To prove our error bound, we introduce two useful technical lemmas.
\begin{lem}[Norm bound on block-band matrices]\label{lem:bandNorm}
	Let $V_i\subseteq \CH$ and $W_i\subseteq \CH'$ be systems of mutually orthogonal subspaces of $\CH$ and $\CH'$ respectively. If $\vM=\bigoplus_i \vB_i$ where $\vB_i\colon V_i\rightarrow W_i$, then $\nrm{\vM}=\max_i \nrm{\vB_i}$.
\end{lem}
\begin{proof}
	We can get a singular value decomposition of $\vM$ by taking singular value decompositions of each $\vB_i$ and then merging them. Since $\nrm{\vM}$ is the largest singular value, we get the claimed equality.
\end{proof}

\begin{lem}[Norm bounds on sums of tensor products of matrices]\label{lem:tensorSumBound}
	Let $|I|\leq \infty$ and $\vA_i\in\mathbb{C}^{n \times m}$, $\vB_i\in\mathbb{C}^{n' \times m'}$ for each $i \in I$, then
	\begin{equation}\label{eq:CPBound}
		\nrm{\sum_{i\in I}\vA_i[\cdot]\vB_i^\dagg}_{2-2}=
		\nrm{\sum_{i\in I} \vA_i \otimes \vB^*}
		\leq \sqrt{\nrm{\sum_{i\in I} \vA_i \vA_i^\dagg} \nrm{\sum_{i\in I} \vB_i^\dagg \vB_i }}.
	\end{equation}
\end{lem}

\begin{proof}
Define the maps
\begin{align}
    \vV: = \sum_{i\in I} \vA_i \otimes \vI \otimes \bra{i} \quad \text{and}\quad \vU := \sum_{i\in I}\vI\otimes \vB^* \otimes \ket{i}, \qquad \text{then}
\end{align}
\begin{align}
    \nrm{\sum_{i\in I} \vA_i \otimes \vB^*} = \norm{ \vV \vU} \leq \norm{\vV}\norm{\vU}
   \leq \nrm{\sum_{i\in I}\vA_i\vA_i^\dagg}^{1/2} \nrm{\sum_{i\in I}\vB^{*\dagg}\vB^*}^{1/2}.
\end{align}
Take complex conjugate to conclude the proof.
\end{proof}
\begin{restatable}[Secular approximation gives discriminant proxy]{lem}{boltzmannError}\label{lem:error_from_Boltzmann}
	Consider the following discriminant proxy and a closely related \Lword{}
	\begin{align}
     \vec{\CD}_{sec}&:= \sum_{ a\in A}\sum_{\bomega\in S_{\omega_0}} \sqrt{\gamma(\bomega)\gamma(-\bomega)} \hat{\vS}^a(\bomega)\otimes \hat{\vS}^{a}(\bomega)^*- \frac{\gamma(\bomega)}{2} \L(\hat{\vS}^a(\bomega)^\dagg\hat{\vS}^a(\bomega)\otimes \vI + \vI \otimes \hat{\vS}^{a}(\bomega)^{\dagg*}\hat{\vS}^{a}(\bomega)^*\R),\\
		\CL_{sec} &= 
	\sum_{ a\in A}\sum_{\bomega\in S_{\omega_0}} \gamma(\bomega)  \hat{\vS}^a(\bomega)[\cdot] \hat{\vS}^a(\bomega)^\dagg -\frac{\gamma(\bomega)}{2} \{\hat{\vS}^a(\bomega)^\dagg \hat{\vS}^a(\bomega) ,\cdot \},
	\end{align}
 such that the nonnegative weights satisfy\footnote{Potentially allowing zero values of $\gamma(\bomega)=\gamma(-\bomega)=0$ might be needed for dealing with the case when $N$ is even and therefore the smallest label in $S_{\omega_0}$ would be its own inverse (due to parity and the modulo arithmetic of $S_{\omega_0}$). Also note that if $\hat{\vS}^a(\bomega)=\vS^{a}_{-\bomega}=0$ for all $a\in A$, then we can assume without loss of generality that the corresponding weight are $\gamma(\bomega)=\gamma(-\bomega)=0$. \label{foot:gammaZeroPossibility}}
\begin{align}\label{eq:Symass} 
		\gamma(\bomega)/\gamma(-\bomega)=\e^{-\beta\bomega}\quad \text{or}\quad \gamma(\bomega)=\gamma(-\bomega)=0\quad \text{for each} \quad \bomega\in S_{\omega_0}. 
	\end{align}
Suppose the operators satisfy that 
\begin{align}
		\bra{\psi_i}\hat{\vS}^a(\bomega) \ket{\psi_j} = 0 \quad \textrm{whenever} \quad \labs{(E_i - E_j) - \bomega} > \bmu \label{eq:PSPass_D}
\end{align}
	for the eigenvalue decomposition of $\vH=\sum_j E_j\ketbra{\psi_j}{\psi_j}$.	
	Then, for any $\beta, \bmu >0$ such that $\beta\bmu\leq 1$ and the Gibbs state $\vrho=\e^{-\beta \vH}/\tr[\e^{-\beta \vH}]$,
	\begin{align}\label{eq:DiscBound1}
		&\nrm{\vec{\CD}_{sec}- \vec{\CD}(\vrho,\CL_{sec})}\\
		&\le \beta \bmu \sqrt{\nrm{\sum_{a\in A, \bomega\in S_{\omega_0}}\gamma(\bomega) \hat{\vS}^a(\bomega)^\dagg \hat{\vS}^a(\bomega)}}\left(7\sqrt{\nrm{\sum_{a\in A, \bomega\in S_{\omega_0}}\gamma(\bomega) \hat{\vS}^a(\bomega)^\dagg\hat{\vS}^a(\bomega)}}+125\sqrt{\nrm{\sum_{a\in A, \bomega\in S_{\omega_0}}\gamma(-\bomega)\hat{\vS}^a(\bomega) \hat{\vS}^a(\bomega)^\dagg}}\right)
	\end{align}
where $\vec{\CD}(\vrho,\CL_{sec})$ is the vectorization of $\CD(\vrho,\CL_{sec})$.
 Finally, if there is a permutation $\vP\colon a\rightarrow a'$ such that\textsuperscript{\ref{foot:gammaZeroPossibility}} $\sqrt{\gamma(-\bomega)}\hat{\vS}^a(\bomega)^\dagg = \sqrt{\gamma(-\bomega)}\hat{\vS}^{a'}(-\bomega)$ for each $a$ and $\bomega$, then we have that $\CD_{sec} = \CD_{sec}^{\dagg}$ and that
\begin{align}
    \nrm{\vec{\CD}_{sec}- \vec{\CD}(\vrho,\CL_{sec})^{\dagger}} = \nrm{\vec{\CD}_{sec}- \vec{\CD}(\vrho,\CL_{sec})} \le 132\beta \bmu \nrm{\sum_{ a\in A}\sum_{\bomega\in S_{\omega_0}}\!\!\!\gamma(\bomega) \hat{\vS}^a(\bomega)^\dagg \hat{\vS}^a(\bomega)}.
\end{align}
\end{restatable}
This also quickly leads to approximate detailed balance for \Lword{}s (\autoref{lem:L_Approx_DB}) by a triangle inequality
 \begin{align}
\lnorm{\CD(\vrho, \CL_{sec}) - \CD(\vrho, \CL_{sec})^{\dagger}}_{2-2} &\le \norm{\CD(\vrho, \CL_{sec}) - \CD_{sec}}_{2-2} +\norm{ \CD_{sec} - \CD(\vrho, \CL_{sec})^{\dagger}}_{2-2}\\
& = \norm{\vec{\CD}(\vrho, \CL_{sec}) - \vec{\CD}_{sec}} +\norm{ \vec{\CD}_{sec} - \vec{\CD}(\vrho, \CL_{sec})^{\dagger}}.
 \end{align}

 Note that the term $\norm{\sum_{a,\bomega} \hat{\vS}^a(\bomega)^{\dagger} \hat{\vS}^a(\bomega)}$ can be thought of as the ``strength'' of the interaction, and in our case can be 
simply bounded by $1$ due to \autoref{prop:properties_DFT} as follows
\begin{align}
	\norm{\sum_{a,\bomega} \hat{\vS}^a(\bomega)^\dagg \hat{\vS}^a(\bomega)} 
	= \norm{\sum_{a,\bomega} \L(\vA^a_{f_{sec}}(\bomega)\R)^\dagg \vA^a_{f_{sec}}(\bomega)}
	= \norm{\sum_{a\in A}\vA^{a\dagg}\vA^a}\cdot \nrm{\ket{f_{sec}}}^2 \le 1.
\end{align}

\begin{proof}[Proof of \autoref{lem:error_from_Boltzmann}] Our proof adapts from the strategy of~\cite{ETH_thermalization_Chen21} for dealing with approximate detailed balance. First, let us define projectors that partition the spectrum per truncation frequency $\bmu$ as follows
	\begin{align}
		\vP_i :=\sum_{ j\colon \frac{E_{j}}{\bmu}\in[(i-\frac12),(i+\frac12)) }\ketbra{\psi_j}{\psi_j} \quad\text{for each}\quad i \in \mathbb{Z}\quad \text{such that}\quad \sum_i \vP_i = \vI \quad \text{and}\quad \vP_i\vH=\vH\vP_i.
	\end{align}
	In other words, these projectors provide a resolution of the identity, and moreover, they commute with $\vH$. We proceed by decomposing the difference of the operators $\vec{\CD}_{sec}- \vec{\CD}'$ as follows
	\begin{align}
		\vec{\CD}_{sec}- \vec{\CD}' &= \overset{\delta\vec{\CA}:=}{\overbrace{\sum_{a, \bomega} \sqrt{\gamma(\bomega)\gamma(-\bomega)} \hat{\vS}^a(\bomega)\otimes \hat{\vS}^{a}(\bomega)^* - \gamma(\bomega) \vrho^{-1/4}\hat{\vS}^a(\bomega)\vrho^{1/4}\otimes \vrho^{*-1/4}\hat{\vS}^{a}(\bomega)^* \vrho^{*1/4}}} \label{eq:dA}\\
		\kern-2.5mm&+ \underset{\frac{1}{2}\delta\vec{\CR}\otimes \vI:=}{\underbrace{\sum_{a, \bomega}\frac{\gamma(\bomega)}{2} \left(\hat{\vS}^a(\bomega)^\dagg \hat{\vS}^a(\bomega)\otimes \vI - \vrho^{-1/4}\hat{\vS}^a(\bomega)^\dagg \hat{\vS}^a(\bomega)\vrho^{1/4}\otimes \vI\right)}}
		\!\\
  &+\! \underset{\frac{1}{2} \vI\otimes\delta\vec{\CR}^*=}{\underbrace{\sum_{a, \bomega}\frac{\gamma(\bomega)}{2} \left(\vI\otimes\hat{\vS}^{a}(\bomega)^{\dagg*} \hat{\vS}^{a}(\bomega)^* - \vI\otimes \vrho^{*-1/4}\hat{\vS}^{a}(\bomega)^{*\dagg} \hat{\vS}^{a}(\bomega)^*\vrho^{*1/4} \right)}}\!.\kern5.5mm\label{eq:dR}
	\end{align}
	
	By the triangle inequality we have that $\nrm{\vec{\CD}_{sec}- \vec{\CD}'}\leq\nrm{\delta\vec{\CA}}+\frac{1}{2}\nrm{\delta\vec{\CR}\otimes \vI}+\frac{1}{2}\nrm{\vI\otimes \delta\vec{\CR}^*}=\nrm{\delta\vec{\CA}}+\nrm{\delta\vec{\CR}}$.
	Now, we bound the above two terms individually, starting from the term $\nrm{\delta\vec{\CR}}$. Our proof crucially relies on the fact that the frequency label of $\hat{\vS}^a(\bomega)$ closely approximates the true Bohr frequency, up to the truncation frequency $\bmu$ as expressed by \eqref{eq:PSPass_D}. This implies that $\hat{\vS}^a(\bomega)^\dagg \hat{\vS}^a(\bomega)$ roughly preserves energy that
\begin{align}
	\bra{\psi_i}\hat{\vS}^a(\bomega)^\dagg \hat{\vS}^a(\bomega)\ket{\psi_j}=0 \quad \text{whenever}\quad \labs{E_i - E_j} > 2\bmu.
\end{align}
Indeed, 
\begin{align}
\bra{\psi_i}\hat{\vS}^a(\bomega)^\dagg\hat{\vS}^a(\bomega)\ket{\psi_j}=\sum_k\bra{\psi_i}\hat{\vS}^a(\bomega)^\dagg\ketbra{\psi_k}{\psi_k} \hat{\vS}^a(\bomega)\ket{\psi_j}	
		=\sum_{k\colon \substack{|(E_k-E_i)-\bomega|\leq\bmu \,\&\\ |(E_k-E_j)-\bomega|\leq\bmu\phantom{\,\&}}}\bra{\psi_i}\hat{\vS}^a(\bomega)^\dagg\ketbra{\psi_k}{\psi_k} \hat{\vS}^a(\bomega)\ket{\psi_j},
		\tag*{(due to \eqref{eq:PSPass_D})}
	\end{align}
	since $\labs{E_i - E_j}=\labs{(E_k-E_j)-\bomega-\left((E_k-E_i)-\bomega\right)}\leq \labs{(E_k-E_j)-\bomega}+\labs{(E_k-E_i)-\bomega}$, meaning that the above summands can only be nonzero when $\labs{E_i - E_j} \leq 2\bmu$.  This observation enables us to introduce the following decomposition
	\begin{align}\label{eq:doubleBisection}
		\hat{\vS}^a(\bomega)^\dagg \hat{\vS}^a(\bomega)
		=\sum_{i,j}\vP_i\hat{\vS}^a(\bomega)^\dagg\hat{\vS}^a(\bomega) \vP_{j}
		=\sum_{\ell=-2}^2\sum_{i}\vP_i \hat{\vS}^a(\bomega)^\dagg\hat{\vS}^a(\bomega)\vP_{i+ \ell}. 
	\end{align}
	
	Let us define $\delta\vH_j:=\frac{\beta}{4}(\vP_j\vH\vP_j-j\bmu\vP_j)$ and $\delta\vH:=\sum_j\delta\vH_j$. Since $\vH$ and $\vP_j$ commute and $\vP_j=\vP_j^2$ we have
	\begin{align}\label{eq:sigmaPj}
		\vrho^{\mp \frac14}\vP_j=\vP_j\vrho^{\mp \frac14} =(\tr(\e^{-\beta \vH}))^{\pm \frac14}\cdot \e^{\pm\frac{\beta\bmu}{4} j}\vP_j\e^{\pm \delta\vH}.		
	\end{align}
	
	We use \eqref{eq:doubleBisection}-\eqref{eq:sigmaPj} to exploit the ``approximate energy preservation'' of the operator $\hat{\vS}^a(\bomega)^\dagg \hat{\vS}^a(\bomega)$ as follows
	\begin{align}
		\vrho^{-1/4}\hat{\vS}^a(\bomega)^\dagg \hat{\vS}^a(\bomega) \vrho^{1/4}  &=\sum_{\ell=-2}^2\sum_{i}\vrho^{-1/4}\vP_{i+ \ell}\hat{\vS}^a(\bomega)^\dagg\hat{\vS}^a(\bomega) \vP_i\vrho^{1/4}\\ 	
		&=\sum_{\ell=-2}^2\sum_{i} \e^{\frac{\beta\bmu}{4}\ell}\vP_{i+ \ell} \e^{\delta\vH}\hat{\vS}^a(\bomega)^\dagg\hat{\vS}^a(\bomega) \e^{-\delta\vH}\vP_i. \label{eq:energyLoc}
	\end{align}
	Let us define $\vS':=\sum_{a,\bomega} \gamma(\bomega) \hat{\vS}^a(\bomega)^\dagg \hat{\vS}^a(\bomega)$, then we get the following bound on $\nrm{\delta\vec{\CR}}$:
	\begin{align}
		\nrm{\delta\vec{\CR}}&=\lnorm{\sum_{a,\bomega} \gamma(\bomega) \L(\vrho^{-1/4}\hat{\vS}^a(\bomega)^\dagg \hat{\vS}^a(\bomega) \vrho^{1/4}-\hat{\vS}^a(\bomega)^\dagg \hat{\vS}^a(\bomega) \R)} \tag*{(by \eqref{eq:dR})}\\&
		=\lnorm{\sum_{a,\bomega} \gamma(\bomega) \L(
			\sum_{\ell=-2}^2\sum_{i} \e^{\frac{\beta\bmu}{4}\ell}\vP_{i+ \ell} \e^{\delta\vH}\hat{\vS}^a(\bomega)^\dagg\hat{\vS}^a(\bomega) \e^{-\delta\vH}\vP_i
			-\vP_{i+ \ell}\hat{\vS}^a(\bomega)^\dagg \hat{\vS}^a(\bomega)\vP_i \R)} \tag*{(by \eqref{eq:doubleBisection}-\eqref{eq:energyLoc})}\\&
		\leq\sum_{\ell=-2}^2\lnorm{\sum_{i}\vP_{i+ \ell}\sum_{a,\bomega} \gamma(\bomega) \L(
			\e^{\frac{\beta\bmu}{4}\ell} \e^{\delta\vH}\hat{\vS}^a(\bomega)^\dagg\hat{\vS}^a(\bomega) \e^{-\delta\vH}
			-\hat{\vS}^a(\bomega)^\dagg\hat{\vS}^a(\bomega) \R)\vP_i} \tag*{(by triangle inequality)}\\&
		=\sum_{\ell=-2}^2\max_i\lnorm{\vP_{i+ \ell}\sum_{a,\bomega} \gamma(\bomega) \L(
			\e^{\frac{\beta\bmu}{4}\ell} \e^{\delta\vH}\hat{\vS}^a(\bomega)^\dagg\hat{\vS}^a(\bomega) \e^{-\delta\vH}
			-\hat{\vS}^a(\bomega)^\dagg \hat{\vS}^a(\bomega) \R)\vP_i} \tag*{(by \autoref{lem:bandNorm})}.
\end{align}
We may now drop the project $\vP_i$ and $\vP_{i+\ell}$ and simplify via elementary bounds.
\begin{align}
		(cont'\!d)&\leq\sum_{\ell=-2}^2\lnorm{\e^{(\delta\vH-\frac{\beta\bmu}{4}\ell \vI)}\vS' \e^{-\delta\vH}-\vS'}\\&		
		\leq\sum_{\ell=-2}^2\lnorm{\e^{(\delta\vH-\frac{\beta\bmu}{4}\ell \vI)}\vS' \e^{-\delta\vH}-\vS' \e^{-\delta\vH}}+\lnorm{\vS' \e^{-\delta\vH}-\vS'}\\&
		\leq\sum_{\ell=-2}^2\nrm{\e^{(\delta\vH-\frac{\beta\bmu}{4}\ell \vI)}-\vI}\nrm{\vS'} \nrm{\e^{-\delta\vH}}+\nrm{\vS'}\lnorm{\e^{-\delta\vH}-\vI}\\&
		\leq\nrm{\vS'}\sum_{\ell=-2}^2\nrm{2\delta\vH-\frac{\beta\bmu}2\ell \vI} \left(1+\nrm{2\delta\vH}\right)+\nrm{2\delta\vH} \tag*{(since $|\e^x-1|\leq 2|x|$ for $|x|\leq\frac{5}{4}$)}\\&
		\leq\nrm{\vS'}\sum_{\ell=-2}^2\beta\bmu\left(\frac{1}{4}+\frac{|\ell|}{2}\right) \left(1+\frac{\beta\bmu}{4}\right)+\frac{\beta\bmu}{4} \tag*{(since $\nrm{\delta\vH}\leq \frac{\beta\bmu}{8}$ and $\beta\bmu\leq 1$)}\\&	
		\leq 7\beta\bmu\nrm{\vS'}.\label{eq:dRbound}
	\end{align}
	Next, we bound $\delta\vec{\CA}$ in a similar fashion. The expression will be more cumbersome because of the double Hilbert spaces. Decompose $\delta\vec{\CA}=\sum_{a\in A} \delta\vec{\CA}^a$, where 
	\begin{align}
		\delta\vec{\CA}^a
		:=&\sum_{\bomega\in S_{\omega_0}} \underset{\mathring{\gamma}(\bomega):=}{\underbrace{\sqrt{\gamma(\bomega)\gamma(-\bomega)}}} \hat{\vS}^a(\bomega)\otimes \hat{\vS}^{a}(\bomega)^*
		-\sum_{\bomega\in S_{\omega_0}} \gamma(\bomega) \vrho^{-1/4}\hat{\vS}^a(\bomega)\vrho^{1/4}\otimes \vrho^{*-1/4}\hat{\vS}^{a}(\bomega)^* \vrho^{*1/4}\\
		\overset{(\text{by } \eqref{eq:Symass})}{=}&\sum_{\bomega\in S_{\omega_0}} \mathring{\gamma}(\bomega)\hat{\vS}^a(\bomega)\otimes \hat{\vS}^{a}(\bomega)^*- \mathring{\gamma}(\bomega)\e^{-\frac\beta2\bomega} \vrho^{-1/4}\hat{\vS}^a(\bomega)\vrho^{1/4}\otimes \vrho^{*-1/4}\hat{\vS}^{a}(\bomega)^* \vrho^{*1/4}  
		 \label{eq:FlippedA}
	\end{align}
	
	Let $\lfloor \bomega \rceil$ denote the rounding of $\frac{\bomega}{\bmu}$ to the closest integer, and suppose that $\vP_i\hat{\vS}^a(\bomega)\vP_j\neq 0$. Then, there must exist some $\ket{\psi_i},\ket{\psi_j}$ eigenvectors in the images of $\vP_i,\vP_j$ respectively such that $\bra{\psi_i}\hat{\vS}^a(\bomega)\ket{\psi_j}\neq 0$. Due to \eqref{eq:PSPass_D} we have that $|E_i-E_j-\bomega|\leq \bmu$. Then $|i-j-\lfloor \bomega \rceil|\leq|i-\frac{E_i}{\bmu}|+|\frac{E_j}{\bmu}-j|+|\frac{\bomega}{\bmu}-\lfloor \bomega \rceil|+|\frac{E_i}{\bmu}-\frac{E_j}{\bmu}-\frac{\bomega}{\bmu}|\leq \frac{1}{2}+\frac{1}{2}+\frac{1}{2}+1<3$, thus we can define a bisection analogously to \eqref{eq:doubleBisection} as follows
	\begin{align}\label{eq:singleBisection}
		\hat{\vS}^a(\bomega)
		=\sum_{i,j}\vP_i\hat{\vS}^a(\bomega) \vP_{j}
		=\sum_{i,j\colon |i-j-\lfloor \bomega \rceil|< 3}\vP_i \hat{\vS}^a(\bomega) \vP_{j}
		=\sum_{\ell=-2}^2\sum_{i}\vP_i \hat{\vS}^a(\bomega) \vP_{i-\lfloor \bomega \rceil+\ell},	
	\end{align}
	which leads to the following ``tensor-slicing'' assuming that $\vP_n$ and $\vM_m$ commute: 
	\begin{align}
		\vM_1 \hat{\vS}^a(\bomega) \vM_2 \otimes \vM^*_3 \hat{\vS}^{a}(\bomega)^* \vM^*_4
		&=\sum_{\ell,\ell'=-2}^2\sum_{i,j}  (\vP_j\otimes\vP^*_{j+i})\left(\vM_1\hat{\vS}^a(\bomega)\vM_2  \otimes  \vM^*_3\hat{\vS}^{a}(\bomega)^*\vM^*_4\right) (\vP_{j-\lfloor \bomega \rceil+\ell}\otimes\vP^*_{j+i-\lfloor \bomega \rceil+\ell'}).\label{eq:tensorBisection}
	\end{align}
	Using this we get the following decomposition analogously to \eqref{eq:energyLoc} by expressing $\vrho^{-1/4} \hat{\vS}^a(\bomega) \vrho^{1/4}$ via \eqref{eq:FlippedA} and \eqref{eq:sigmaPj}
	\begin{align}
		&\delta\vec{\CA}^a\! =\sum_{\ell,\ell'=-2}^2\sum_{i}\\
  &\underset{\delta\vec{\CA}^a_{i,\ell,\ell'}:=}{\underbrace{\sum_{j,\bomega}\!\mathring{\gamma}(\bomega)(\vP_j\otimes\vP^*_{j+i})\!
		\left(\hat{\vS}^a(\bomega)  \otimes  \hat{\vS}^{a}(\bomega)^* 
		-\e^{\frac{\beta\bmu}2(\lfloor\bomega\rceil-\frac{\bomega}{\bmu}-\frac{\ell+\ell'}{2})} \e^{\delta\vH}\hat{\vS}^a(\bomega) \e^{-\delta\vH}\! \otimes \! \e^{\delta\vH^*} \hat{\vS}^{a}(\bomega)^* \e^{-\delta\vH^*}\right)\!(\vP_{j-\lfloor \bomega \rceil+\ell}\otimes\vP^*_{j+i-\lfloor \bomega \rceil+\ell'})}}.
	\end{align}
	Since $|\frac{\bomega}{\bmu}-\lfloor \bomega \rceil|\leq \frac12$ and $|\ell|,|\ell'|\leq 2$, we have that $|\lfloor\bomega\rceil-\frac{\bomega}{\bmu}-\frac{\ell+\ell'}{2}|\leq \frac52$ and therefore the above factor is close to $1$:
	\begin{align}\label{eq:BoltzmannBound}
		|\e^{\frac{\beta\bmu}2(\lfloor\bomega\rceil-\frac{\bomega}{\bmu}-\frac{\ell+\ell'}{2})}-1|\leq \beta\bmu  \labs{\lfloor\bomega\rceil-\frac{\bomega}{\bmu}-\frac{\ell+\ell'}{2}} \leq \frac52\beta \bmu.
	\end{align}
	
	At this point, it seems intuitively clear that the error coming from the Boltzmann factor is small. However, we need to argue how the sum over $j$ does not blow up the error. We proceed by using the triangle inequality over $\ell,\ell'$ and then \autoref{lem:bandNorm} over $i$ to get that
	\begin{align}
	    \nrm{\delta\vec{\CA}}=\nrm{\sum_{a\in A} \delta\vec{\CA}^a} \le 5^2\max_{i,\ell,\ell'} \Bigg\lVert\underset{\delta\vec{\CA}_{i,\ell,\ell'}:=}{\underbrace{\sum_{a}\delta\vec{\CA}^a_{i,\ell,\ell'}}}\Bigg\rVert.
	\end{align}

	The key for bounding the norm of $\delta\vec{\CA}_{i,\ell,\ell'}$ is applying \autoref{lem:tensorSumBound} to estimate the following for some weights $\alpha(\bomega)$:
	\begin{align}
		&\nrm{\sum_{j,a,\bomega}\alpha(\bomega)\mathring{\gamma}(\bomega)(\vP_j\otimes\vP^*_{j+i})
			\left(\hat{\vS}^a(\bomega)  \otimes  \hat{\vS}^{a}(\bomega)^* \right)(\vP_{j-\lfloor\bomega\rceil+\ell}\otimes\vP^*_{j+i-\lfloor\bomega\rceil+\ell'})}^2\\&
		=	\nrm{\sum_{j,a,\bomega}\alpha(\bomega)
			\vP_j\sqrt{\gamma(-\bomega)}\hat{\vS}^a(\bomega)\vP_{j-\lfloor\bomega\rceil+\ell}  \otimes  \vP^*_{j+i}\sqrt{\gamma(\bomega)}\hat{\vS}^{a}(\bomega)^*\vP^*_{j+i-\lfloor\bomega\rceil+\ell'} }^2\\&
		\leq \nrm{\sum_{j,a,\bomega}|\alpha(\bomega)|^2\gamma(-\bomega)
			\vP_j\underset{\preceq \sum_k \hat{\vS}^a(\bomega)\vP_{k}\hat{\vS}^a(\bomega)^\dagg=\hat{\vS}^a(\bomega)\hat{\vS}^a(\bomega)^\dagg}{\underbrace{\hat{\vS}^a(\bomega)\vP_{j-\lfloor\bomega\rceil+\ell}\hat{\vS}^a(\bomega)^\dagg}}\vP_{j}}\cdot 
		\nrm{\sum_{j,a,\bomega}\gamma(\bomega)\vP_{j+i-\lfloor\bomega\rceil+\ell'}\underset{\preceq \hat{\vS}^a(\bomega)^\dagg\hat{\vS}^a(\bomega)}{\underbrace{\hat{\vS}^a(\bomega)^\dagg\vP_{j+i}\hat{\vS}^a(\bomega)}}\vP_{j+i-\lfloor\bomega\rceil+\ell'}} \tag*{(by \autoref{lem:tensorSumBound})}\\& \leq \nrm{\alpha}^2_\infty 
		\Bigg\lVert\sum_{j}
		\vP_j\underset{\vS'':=}{\underbrace{\sum_{a,\bomega}\gamma(-\bomega)\hat{\vS}^a(\bomega)\hat{\vS}^a(\bomega)^\dagg}}\vP_{j}\Bigg\rVert
		\Bigg\lVert\sum_{j}
		\vP_j\underset{=\vS'}{\underbrace{\sum_{a,\bomega}\gamma(\bomega)\hat{\vS}^a(\bomega)^\dagg\hat{\vS}^a(\bomega)}}\vP_{j}\Bigg\rVert\tag*{since $0\preceq \vA \preceq \vB \Rightarrow \nrm{\vA}\leq\nrm{\vB}$}\\&
		= \nrm{\alpha}^2_\infty \underset{\text{due to \autoref{lem:bandNorm}}}{\underbrace{  \max_{j''}\Big\lVert\vP_{j''}\vS''\vP_{j''}\Big\rVert \cdot \max_{j'}\Big\lVert\vP_{j'}\vS'\vP_{j'}\Big\rVert}}
		\leq \nrm{\alpha}^2_\infty \nrm{\vS'}\nrm{\vS''}.\label{eq:ProjSumBound}
	\end{align}
	
	Let us introduce a telescoping sum $\delta\vec{\CA}_{i,\ell,\ell'}=\sum_{s=1}^3\delta\vec{\CA}^{(s)}_{i,\ell,\ell'}$, where $\mathring{\gamma}_{\ell,\ell'}(\bomega):=\e^{\frac{\beta\bmu}2(\lfloor\bomega\rceil-\frac{\bomega}{\bmu}-\frac{\ell+\ell'}{2})}\mathring{\gamma}(\bomega)$ and
	\begin{align}
		\delta\vec{\CA}^{(1)}_{i,\ell,\ell'}
		:=\sum_{j,a,\bomega}\mathring{\gamma}(\bomega)(\vP_j\otimes\vP^*_{j+i})
				\left(\hat{\vS}^a(\bomega)  \otimes  \hat{\vS}^{a}(\bomega)^* 
				-\e^{\frac{\beta\bmu}2(\lfloor\bomega\rceil-\frac{\bomega}{\bmu}-\frac{\ell+\ell'}{2})} \hat{\vS}^a(\bomega)  \otimes   \hat{\vS}^{a}(\bomega)^* \right)(\vP^*_{j-\lfloor \bomega \rceil+\ell}\otimes\vP^*_{j+i-\lfloor \bomega \rceil+\ell'}),
	\end{align}
	\begin{align}
		\delta\vec{\CA}^{(2)}_{i,\ell,\ell'}
		:=\sum_{j,a,\bomega}\mathring{\gamma}_{\ell,\ell'}(\bomega)(\vP_j\otimes\vP^*_{j+i})
		\left(\hat{\vS}^a(\bomega)  \otimes  \hat{\vS}^{a}(\bomega)^* 
		- \e^{-\delta\vH}\hat{\vS}^a(\bomega) \otimes  \e^{-\delta\vH^*} \hat{\vS}^{a}(\bomega)^* \right)(\vP^*_{j-\lfloor \bomega \rceil+\ell}\otimes\vP^*_{j+i-\lfloor \bomega \rceil+\ell'}),
	\end{align}				
	\begin{align}
		&\delta\vec{\CA}^{(3)}_{i,\ell,\ell'}
		:=\sum_{j,a,\bomega}\\
  &\mathring{\gamma}_{\ell,\ell'}(\bomega)(\vP_j\otimes\vP^*_{j+i})
		\left(\e^{-\delta\vH}\hat{\vS}^a(\bomega)  \otimes  \e^{-\delta\vH^*}\hat{\vS}^{a}(\bomega)^* 
		- \e^{-\delta\vH}\hat{\vS}^a(\bomega) \e^{\delta\vH} \otimes  \e^{-\delta\vH^*} \hat{\vS}^{a}(\bomega)^* \e^{\delta\vH^*}\right)(\vP^*_{j-\lfloor \bomega \rceil+\ell}\otimes\vP^*_{j+i-\lfloor \bomega \rceil+\ell'}).
	\end{align}

Due to \eqref{eq:BoltzmannBound} we can bound 
 $\nrm{\delta\vec{\CA}^{(1)}_{i,\ell,\ell'}}\leq \frac52\beta \bmu \sqrt{\nrm{\vS'}\nrm{\vS''}}$ via \eqref{eq:ProjSumBound}. For bounding $\nrm{\delta\vec{\CA}^{(2)}_{i,\ell,\ell'}}$ observe that
	\begin{align}
		\nrm{\delta\vec{\CA}^{(2)}_{i,\ell,\ell'}}&=
		\nrm{\sum_{j,a,\bomega}\mathring{\gamma}_{\ell,\ell'}(\bomega)(\vP_j\otimes\vP^*_{j+i})
			\left(\hat{\vS}^a(\bomega)  \otimes  \hat{\vS}^{a}(\bomega)^* 
			- \e^{-\delta\vH}\hat{\vS}^a(\bomega) \otimes  \e^{-\delta\vH^*} \hat{\vS}^{a}(\bomega)^* \right)(\vP^*_{j-\lfloor \bomega \rceil+\ell}\otimes\vP^*_{j+i-\lfloor \bomega \rceil+\ell'})}\\&
		=\nrm{(\vI\otimes \vI - \e^{-\delta\vH} \otimes \e^{-\delta\vH^*})\sum_{j,\bomega}\mathring{\gamma}_{\ell,\ell'}(\bomega)(\vP_j\otimes\vP^*_{j+i})
			\left(\hat{\vS}^a(\bomega)  \otimes  \hat{\vS}^{a}(\bomega)^*\right)(\vP^*_{j-\lfloor \bomega \rceil+\ell}\otimes\vP^*_{j+i-\lfloor \bomega \rceil+\ell'})}\\&
		\leq\underset{\leq\frac27\beta\bmu \text{ since }\nrm{\delta\vH}\leq \frac{\beta\bmu}{8}}{\underbrace{\nrm{\vI\otimes \vI - \e^{-\delta\vH} \otimes \e^{-\delta\vH^*}}}}
		\underset{\leq (1+\frac52\beta \bmu) \sqrt{\nrm{\vS'}\nrm{\vS''}}\text{ due to }\eqref{eq:BoltzmannBound}\text{ and }\eqref{eq:ProjSumBound}}{\underbrace{\nrm{\sum_{j,\bomega}\mathring{\gamma}_{\ell,\ell'}(\bomega)(\vP_j\otimes\vP^*_{j+i})
					\left(\hat{\vS}^a(\bomega)  \otimes  \hat{\vS}^{a}(\bomega)^*\right)(\vP^*_{j-\lfloor \bomega \rceil+\ell}\otimes\vP^*_{j+i-\lfloor \bomega \rceil+\ell'})}}}\\&
		\leq \beta\bmu\sqrt{\nrm{\vS'}\nrm{\vS''}}.
	\end{align}
	Analogously we can bound $\nrm{\delta\vec{\CA}^{(3)}_{i,\ell,\ell'}}$ as follows
	\begin{align}
		\nrm{\delta\vec{\CA}^{(3)}_{i,\ell,\ell'}}&=\nrm{(\e^{-\delta\vH} \otimes \e^{-\delta\vH^*})\sum_{j,\bomega}\mathring{\gamma}_{\ell,\ell'}(\bomega)(\vP_j\otimes\vP^*_{j+i})
			\left(\hat{\vS}^a(\bomega)  \otimes  \hat{\vS}^{a}(\bomega)^*\right)(\vP^*_{j-\lfloor \bomega \rceil+\ell}\otimes\vP^*_{j+i-\lfloor \bomega \rceil+\ell'})(\vI\otimes \vI - \e^{\delta\vH} \otimes \e^{\delta\vH^*})}\\&
		\leq\underset{\leq1+\frac27\beta\bmu }{\underbrace{\nrm{\e^{-\delta\vH} \otimes \e^{-\delta\vH^*}}}}
		\underset{\leq (1+\frac52\beta \bmu) \sqrt{\nrm{\vS'}\nrm{\vS''}}\text{ due to }\eqref{eq:BoltzmannBound}\text{ and }\eqref{eq:ProjSumBound}}{\underbrace{\nrm{\sum_{j,\bomega}\mathring{\gamma}_{\ell,\ell'}(\bomega)(\vP_j\otimes\vP^*_{j+i})
					\left(\hat{\vS}^a(\bomega)  \otimes  \hat{\vS}^{a}(\bomega)^*\right)(\vP^*_{j-\lfloor \bomega \rceil+\ell}\otimes\vP^*_{j+i-\lfloor \bomega \rceil+\ell'})}}}
				\underset{\leq\frac27\beta\bmu \text{ since }\nrm{\delta\vH}\leq \frac{\beta\bmu}{8}}{\underbrace{\nrm{\vI\otimes \vI - \e^{-\delta\vH} \otimes \e^{-\delta\vH^*}}}}\\&
		\leq \frac{3}{2}\beta\bmu\sqrt{\nrm{\vS'}\nrm{\vS''}}.
	\end{align}
	Putting everything together, we get that 
	\begin{align}
		\nrm{\delta\vec{\CA}}&\leq 25\max_{i,\ell,\ell'}\nrm{\delta\vec{\CA}_{i,\ell,\ell'}}
		\leq 25\max_{i,\ell,\ell'}\left(\nrm{\delta\vec{\CA}^{(1)}_{i,\ell,\ell'}}+\nrm{\delta\vec{\CA}^{(2)}_{i,\ell,\ell'}}+\nrm{\delta\vec{\CA}^{(3)}_{i,\ell,\ell'}}\right)
		\leq125\beta\bmu \sqrt{\nrm{\vS'}\nrm{\vS''}},	
	\end{align}
 which concludes the proof for the first bound.

Finally, if $\sqrt{\gamma(-\bomega)}\hat{\vS}^a(\bomega)^\dagg = \sqrt{\gamma(-\bomega)}\hat{\vS}^{a'}(-\bomega)$, then $\vS'=\vS''$ and $\CD_{sec} =\CD^{\dagger}_{sec}$, so we easily get the other bound 
\begin{align}
\nrm{\vec{\CD}_{sec}- \vec{\CD}(\vrho,\CL_{sec})^{\dagger}} &= \nrm{\CD_{sec}- \CD(\vrho,\CL^{\dagger}_{sec})}_{2-2} = \nrm{\CD_{sec}^{\dagger}- \CD(\vrho,\CL_{sec})}_{2-2}\\ 
&=\nrm{\CD_{sec}- \CD(\vrho,\CL_{sec})}_{2-2} = \nrm{\vec{\CD}_{sec}- \vec{\CD}(\vrho,\CL_{sec})}\leq 132\beta\bmu \nrm{\vS'}. \tag*{\qedhere}
\end{align}
\end{proof}

\subsection{Bootstrapping the secular approximation}\label{sec:improved}

\begin{lem}[Bootstrapping the secular approximation]\label{lem:slicing}
	Consider the decomposition $f=\sum_{j\in J}f_j$ of the weight function, where $f_j\colon S_{t_0}\rightarrow \mathbb{C}$, and let $\mu_j:=\min\left\{\mu\geq 0\colon \nrm{\hat{f}_j(\omega)\cdot\indicator(|\omega|> \mu)}=0\right\}$. If the Hamiltonian $\bvH$ has discretized spectrum so that $B\subset \omega_0\BZ$, $\beta\mu_j\leq 1$, $\max_{\bomega\in B, j\in J}\bomega + \mu_j=:\nu\in S_{\omega_0}$, $\gamma\colon S_{\omega_0}\rightarrow \BR_+$ is such that $\gamma(\bomega)/\gamma(-\bomega)=\e^{-\beta\bomega}$ for all $\bomega\in[-\nu,\nu]\cap S_{\omega_0}$, and the set of jumps is self-adjoint $\{\vA^a\colon a\in A\}=\{\vA^{a\dagger}\colon a\in A\}$, then\footnote{Note that here we use notation $\CD_f$, $\CL_{f}$ instead of $\CD_\beta$, $\CL_{\beta}$ to spell out the dependence on $f$ instead of $\beta$.}
\begin{align}
\nrm{\vec{\CD}_f-\vec{\CD}(\vrho,\CL_f)^{\dagger}} = \nrm{\vec{\CD}_f-\vec{\CD}(\vrho,\CL_f)}
\le\sum_{i\in J}\nrm{f_i}\mu_i\sum_{j\in J}\nrm{f_j} \nrm{\gamma}_\infty 1056\beta\nrm{\sum_{ a \in A} \vA^{a \dagg} \vA^{a}}
	\end{align}
	where $\vrho=\e^{-\beta \vH}/\tr[\e^{-\beta \vH}]$,
	\begin{align}
		\CD_f&:= \vI[\cdot] \vI+ \sum_{a, \bomega \in S_{\omega_0} } \sqrt{\gamma(\bomega)\gamma(-\bomega)} \hat{\vA}^{a}_{f}(\bomega)[\cdot] \hat{\vA}^{a}_{f}(\bomega)^{\dagger}- \frac{\gamma(\bomega)}{2} \{\hat{\vA}^{a}_{f}(\bomega)^{\dagger} \hat{\vA}^{a}_{f}(\bomega) ,\cdot \},\text{ and}\\
		\CL_{f} &:= 
		\sum_{a\in A, \bomega\in S_{\omega_0}} \gamma(\bomega)  \hat{\vA}^{a}_{f}(\bomega)[\cdot] \hat{\vA}^{a}_{f}(\bomega)^{\dagger} -\frac{\gamma(\bomega)}{2} \{\hat{\vA}^{a}_{f}(\bomega)^{\dagger} \hat{\vA}^{a}_{f}(\bomega) ,\cdot \}.
	\end{align}
\end{lem}
\begin{proof}
	The proof builds on the following ``polarization'' identity: for all $\vM_i$ matrices and $c_1,c_2\in\BR$
	\begin{align}
		\vM_1\star\vM_4+\vM_2\star\vM_3=\frac{c_1c_2}{2}\sum_{s=\pm 1}s\left(\frac{\vM_1}{c_1}+s\frac{\vM_2}{c_2}\right)\star\left(\frac{\vM_3}{c_1}+s\frac{\vM_4}{c_2}\right),		
	\end{align}
	where $\star$ stands for any operation that is distributive with $+$, e.g., matrix product $\star=\cdot$ or tensor product $\star=\otimes$.
	
	Due to the linearity of the operator Fourier Transform, we have that $\hat{\vA}^{a}_{f}(\bomega)=\sum_{j\in J}\vA^{a}_{f_j,\bomega}$ and consequently 
	\begin{align}
		\vec{\CD}_f&=\sum_{i,j\in J}\left(1-\frac{\delta_{ij}}{2}\right)\frac{\nrm{f_i}\nrm{f_j}}{2}\sum_{s=\pm 1}s\vec{\CD}_{\!\!\frac{f_i}{\nrm{f_i}}+s\frac{f_j}{\nrm{f_j}}},\\
		\vec{\CD}(\vrho,\CL_f)&=\sum_{i,j\in J}\left(1-\frac{\delta_{ij}}{2}\right)\frac{\nrm{f_i}\nrm{f_j}}{2}\sum_{s=\pm 1}s\vec{\CD}(\vrho,\CL_{\!\frac{f_i}{\nrm{f_i}}+s\frac{f_j}{\nrm{f_j}}}).
	\end{align}
	Due to the properties of the operator Fourier Transform (\autoref{prop:properties_DFT}), we have that 
	\begin{align}
		\nrm{\sum_{a\in A}\sum_{\bomega\in S_{\omega_0}} \gamma(\bomega)\vA^{a}_{g}(\bomega)^{\dagg} \hat{\vA}^a_{g}(\bomega) }
		\leq \nrm{\gamma}_\infty\nrm{\sum_{a\in A}\sum_{\bomega\in S_{\omega_0}} \vA^{a}_{g}(\bomega)^{\dagg} \hat{\vA}^a_{g}(\bomega) }
		=\nrm{\sum_{a\in A}\sum_{\bt \in S_{t_0}} \labs{g(\bt)}^2\e^{\ri \bvH \bt}\vA^{a \dagger} \vA^a \e^{-\ri \bvH \bt}}
		\leq \nrm{\sum_{a\in A}\vA^{a \dagger} \vA^a }\nrm{g}^2\!,
	\end{align}
	since $\nrm{\frac{f_i}{\nrm{f_i}}+s\frac{f_j}{\nrm{f_j}}}\leq 2$, $\beta\mu_j\leq 1$ for all $j\in J$, and\textsuperscript{\ref{foot:gammaZeroPossibility}} $\nu\in S_{\omega_0}$ (i.e., no wrapping around), by \autoref{lem:error_from_Boltzmann} we get
	\begin{align}
		\nrm{\vec{\CD}_{\!\!\frac{f_i}{\nrm{f_i}}+s\frac{f_j}{\nrm{f_j}}} - \vec{\CD}(\vrho,\CL_{\! \frac{f_i}{\nrm{f_i}}+s\frac{f_j}{\nrm{f_j}}})}
  &\le \max\{\mu_i,\mu_j\}{{\nrm{\gamma}_\infty4\beta\sqrt{\nrm{\sum_{ a \in A} \vA^{a \dagg} \vA^{a}}}\left(7\sqrt{\nrm{\sum_{ a \in A} \vA^{a \dagg} \vA^{a}}}+125\sqrt{\nrm{\sum_{ a \in A} \vA^{a} \vA^{a \dagg}}}\right)}}\\
	&	\le \max\{\mu_i,\mu_j\} \underset{K:=}{\underbrace{\nrm{\gamma}_\infty528\beta{\nrm{\sum_{ a \in A} \vA^{a \dagg} \vA^{a}}}}}.
	\end{align}
 The second line uses that $\{\vA^a\colon a\in A\}=\{\vA^{a\dagger}\colon a\in A\}$ thus $\nrm{\sum_{ a \in A} \vA^{a \dagg} \vA^{a}}=\nrm{\sum_{ a \in A} \vA^{a} \vA^{a \dagg}}$. 
 
 Finally, by the triangle inequality, we get that 
	\begin{align}
		\nrm{\vec{\CD}_f-\vec{\CD}(\vrho,\CL_f)}&\leq \sum_{i,j\in J}\left(1-\frac{\delta_{ij}}{2}\right)\frac{\nrm{f_i}\nrm{f_j}}{2}\sum_{s=\pm 1}\nrm{\vec{\CD}_{\!\!\frac{f_i}{\nrm{f_i}}+s\frac{f_j}{\nrm{f_j}}} - \vec{\CD}(\vrho,\CL_{\!\frac{f_i}{\nrm{f_i}}+s\frac{f_j}{\nrm{f_j}}})}\\&
		\leq \sum_{i,j\in J}\left(1-\frac{\delta_{ij}}{2}\right)\nrm{f_i}\nrm{f_j}\max\{\mu_i,\mu_j\}K\\&
		\leq \sum_{i,j\in J}\nrm{f_i}\nrm{f_j}(\mu_i+\mu_j)K\\&	
		= 2\sum_{i\in J}\nrm{f_i}\mu_i\sum_{j\in J}\nrm{f_j}K.\tag*{\qedhere}
	\end{align}
\end{proof}

\begin{cor}[Improved bounds for uniform weights]\label{cor:improved_boltzmann}
 In the setting of~\autoref{lem:slicing}, consider the uniform weight function $f(\bt) = \sqrt{\frac{t_0}{2T}}\indicator(-T \le \bt < T)$ and the secular approximation with $s(\bomega)=\indicator(|\bomega|\leq \mu)$ for some $\mu\geq \pi/T$, 
then
	\begin{align}
\nrm{\vec{\CD}_{sec}- \vec{\CD}(\vrho,\CL_{sec})^{\dagger}}=
	\nrm{\vec{\CD}_{sec}- \vec{\CD}(\vrho,\CL_{sec})} = \bigO{\beta\nrm{\gamma}_\infty\sqrt{\frac{\mu}{T}} \nrm{\sum_{ a \in A} \vA^{a \dagg} \vA^{a}}}.
	\end{align}
\end{cor}
\begin{proof}
	Let $f_s\!:=\CF^{*-1}(\hat{f}\cdot s)$; by \autoref{prop:TruncationByWeighing}, we know that $\vec{\CD}_{sec}$ can be obtained by utilizing the weight \nolinebreak function \nolinebreak $f_s$.
	
	We decompose $s$ into exponentially increasing intervals. We set $s_0:=\indicator(|\bomega|\leq \pi/T)$ and
	\begin{align}
		s_j:=\indicator(4^{j-1}\pi/T < |\bomega|\leq \min\{4^{j}\pi/T,\mu\}) \quad \forall j\in \BZ_+.
	\end{align}
	Let $f_j:=\CF^{*-1}(\hat{f}\cdot s_j)$; since $s=\sum_{j=0}^{\ceil{\log_4(\mu T/\pi)}}s_j$ we have $f_s=\sum_{j=0}^{\ceil{\log_4(\mu T/\pi)}}f_j$. Observe that due to \autoref{prop:uniformTail}
	\begin{align}
		\nrm{f_j}=\nrm{\hat{f}\cdot s_j}\leq\nrm{\hat{f}\cdot \indicator(4^{j-1}\pi/T < |\bomega|)}\leq \sqrt{4^{1-j}}=2^{1-j}, \quad \text{and thus} \quad \sum_{j}\nrm{f_j}\leq4.
	\end{align}
	The result follows from \autoref{lem:slicing} since $\hat{f}_j=\hat{f}\cdot s_j$ and by the definition of $s_j$ we have $\mu_j\leq 4^{j}\pi/T$ and thus
	\begin{align}
		\sum_{j=0}^{\ceil{\log_4(\mu T/\pi)}}\nrm{f_j}\mu_j
		\leq\sum_{j=0}^{\ceil{\log_4(\mu T/\pi)}}\frac{2\pi}{T}2^j
		\leq\frac{4\pi}{T}2^{\ceil{\log_4(\mu T/\pi)}}	
		\leq\frac{8\pi}{T}2^{\log_4(\mu T/\pi)}
		=\frac{8\pi}{T}\sqrt{\mu T/\pi}
		=8\sqrt{\frac{\pi\mu}{T}}.\tag*{\qedhere}
	\end{align}
\end{proof}

\subsection{Fourier Transform with uniform weights}

For simpler implementation, we can also work with the Fourier Transform with uniform weight (which is not smooth), leading to slightly worse bounds than the Gaussian damped case of \autoref{thm:L_correctness}.
\begin{thm}[Uniform weight for Fourier Transform]\label{thm:non_damped_Gibbs_error}
Consider the discriminant proxy $\vec{\CD_{\beta}}$~\eqref{eq:vec_main_result} with the plain Fourier Transform 
   $ \hat{\vA}^{a}(\bomega):\propto \sum_{ -T \le \bt < T} {\vA^{a}(\bt)} \e^{-\ri \bomega \bt}$. Let $\nu:=\frac{1}{\beta}+\max_{\omega\in B}\omega$ such that $\nu\le\max_{\bomega\in S_{\omega_0}}$, $\gamma\colon S_{\omega_0}\rightarrow \BR_+$ is such that $\gamma(\bomega)/\gamma(-\bomega)=\e^{-\beta\bomega}$ for all $\bomega\in[-\nu,\nu]\cap S_{\omega_0}$ and the set of jumps are self-adjoint and normalized~\eqref{eq:AAdagger},
then the (normalized) top eigenvector approximates the purified Gibbs state $\ket{\sqrt{\vrho_{\beta}}}$ such that
\begin{align}
    \lnormp{ \ket{\lambda_1(\vec{\CD}_{\beta}}) - \ket{\sqrt{\vrho_{\beta}}} }{} \le \CO\L(\frac{1}{\lambda_{gap}(\vec{\CD}_{\beta})}(\omega_0T + \sqrt{\frac{\beta}{T}})\R).
\end{align}
The block-encoding for the discriminant proxy can be implemented exactly using Hamiltonian simulation time $\CO(T)$ using the construction outlined in \autoref{sec:L_circuit}-\autoref{sec:D_circuit}.
\end{thm}
Even though with a worse asymptotic bound, the plain Fourier Transform is simpler to implement and closer to thermalization in nature (\autoref{sec:open_system}). The proof is even simpler than the Gaussian case, partly because $\vec{\CD} = \vec{\CD}_{impl}$ as the uniform weights can be prepared exactly.
\begin{proof}[Proof of \autoref{thm:non_damped_Gibbs_error}]
We can assume without loss of generality that $T/\beta\geq \pi$, since otherwise, the bound is vacuous. We bound the eigenvector distance by the operator norm using \autoref{prop:Dfixed_point_error} and recall the secular approximation (\autoref{lem:secular}, \autoref{prop:uniformTail}) and the improved bounds on approximate detailed balance (\autoref{cor:improved_boltzmann}):
\begin{align}
     \lnormp{ \ket{\lambda_1(\vec{\CD}_{\beta})} - \ket{\sqrt{\vrho}} }{} \le \frac{6\lnorm{\vec{\CD}_{\beta} - \vec{\CD}(\vrho,\CL_{sec})^{\dagger} }}{\lambda_{gap}(\CD_{\beta})} &\le \frac{6}{\lambda_{gap}(\vec{\CD}_{\beta})} \L(\norm{\vec{\CD}_{\beta} - \vec{\CD}_{sec}}+ \norm{\vec{\CD}_{sec} - \vec{\CD}(\vrho,\CL_{sec})^{\dagger}}\R)\\
    &\le \CO\L(\frac{1}{\lambda_{gap}(\vec{\CD}_{\beta})}(\omega_0T + \frac{1}{\sqrt{\mu T}}+\beta \sqrt{\frac{\mu}{T}})\R)\tag*{(now set $\mu := 1/\beta$)}\\
    &\le \CO\L(\frac{1}{\lambda_{gap}(\vec{\CD}_{\beta})}(\omega_0T +\sqrt{\frac{\beta}{T}})\R).\tag*{\qedhere}
\end{align}
\end{proof}

Unfortunately, the above result suggests that the Hamiltonian simulation time needs to scale with the inverse gap squared $\lambda_{gap}^{-2}$; we do not know if better bounds are possible.

\section{Discretization error for \Lword{}s and discriminant proxies}\label{apx:cont_limit}
In this appendix, we bound the discretization error for continuous \Lword{}s. We use the notation established in \autoref{sec:operator_FT}. In addition, for a function $f\colon\mathbb{R}\rightarrow\mathbb{C}$, by $\bar{\CF}\left(f(\bt)\right)$, we mean the discrete Fourier Transform of the vector obtained by evaluating $f$ at the points $\bt\in S^{\lceil N \rfloor}_{t_0}$. Also, we define the ``discretized'' version $\bar{f}(t):=\sqrt{t_0}f(t)$ with a natural rescaling.

We begin with a seemingly loose bound that will, however, be sufficient.
\begin{lem}\label{lem:discToContDecomp}
	If $f,g,h\in \ell_2(\mathbb{R})$ and $\gamma\in \ell_\infty(\mathbb{R})$, then for any norm $\vertiii{\cdot}$, we have that
	\begin{align}
		&\vertiii{\int_{-\infty}^{\infty} \gamma(\omega) \hat{\vA}_{f}(\omega)^\dagg\star\hat{\vA}_{f}(\omega)  \mathrm{d}\omega
			-\sum_{\bomega\in S^{\lceil N \rfloor}_{\omega_0}}	g(\bomega) \hat{\vA}_{h}(\bomega)^\dagg \star \hat{\vA}_{h}(\bomega)}\\
		&\quad\leq\sum_{\nu,\nu'\in B}\vertiii{(\vA_{\nu})^\dagg\star\vA_{\nu'}}\left|\int_{\!-\infty}^{\infty} \! 
		 \gamma(\omega)  \hat{f}^*(\omega\!-\!\nu)\hat{f}(\omega\!-\!\nu')\mathrm{d}\omega
		-\!\!\!\!\sum_{\bomega\in S^{\lceil N \rfloor}_{\omega_0}}\! g(\bomega) \bar{\CF}\left(h(\bt) \cdot \e^{(\ri \nu \bt)}\right)^{\!\!*}\!\!(\bomega)\bar{\CF}\left(h(\bt) \cdot \e^{(\ri \nu' \bt)}\right)\!(\bomega)\right|,
	\end{align}
	where $\star$ stands for any operation that is distributive with $+$, e.g., matrix product $\star=\cdot$ or tensor product $\star=\otimes$.
\end{lem}
Directly applying the above for the original Hamiltonian $\vH$ suffers from the number of the Bohr frequencies $\labs{B(\vH)}$, which can generally scale with the Hilbert space dimension if the eigenvalue differences are nondegenerate. However, we will see that the above becomes sufficiently stringent if we consider a \textit{rounded} Hamiltonian $\bvH$, substantially reducing the number of distinct Bohr frequencies $\labs{B(\bvH)}$ while staying close to the original Hamiltonian $\vH$.
\begin{proof}
	We use the defining decomposition of the continuous operator Fourier Transform from \autoref{prop:cont_operator_FT}
	\begin{align}
		\hat{\vA}_{f}(\omega) = \sum_{\nu \in B}\hat{f}(\omega-\nu) \vA_{\nu} ,
	\end{align}
	and its discrete counterpart 
	\begin{align}
		\hat{\vA}_{h}(\bomega) =\sum_{\nu \in B}  \bar{\CF}\left(h(\bt) \cdot \e^{\ri \nu \bt}\right)(\bomega) \vA_{\nu},
	\end{align}	
	where $\bar{\CF}$ denotes the discrete Fourier Transform defined via $\bar{\CF}\left(h(\bt)\right)(\bomega)= \frac{1}{\sqrt{N}}\sum_{\bt \in S^{\lceil N \rfloor}_{t_0}} h(\bt)\e^{-\ri \bomega \bt}$, where $t_0=\frac{2\pi}{N\omega_0}$.
	
	Due to the distributivity of $+$ and $\star$ we have 
	\begin{align*}
		&\int_{-\infty}^{\infty}  \gamma(\omega) \hat{\vA}_{f}(\omega)^\dagg\star\hat{\vA}_{f}(\omega) \mathrm{d}\omega
		-\sum_{\bomega\in S^{\lceil N \rfloor}_{\omega_0}}	g(\bomega) \hat{\vA}_{h}(\bomega)^\dagg \star \hat{\vA}_{h}(\bomega)\\
		&\quad=\sum_{\nu,\nu'\in B}(\vA_{\nu})^\dagg\star\vA_{\nu'}\left(\int_{\!-\infty}^{\infty} \! 
		\gamma(\omega)  \hat{f}^*(\omega\!-\!\nu)\hat{f}(\omega\!-\!\nu')\mathrm{d}\omega
		-\sum_{\bomega\in S^{\lceil N \rfloor}_{\omega_0}}\! g(\bomega) \bar{\CF}\left(h(\bt) \cdot \e^{(\ri \nu \bt)}\right)^{\!\!*}\!\!(\bomega)\bar{\CF}\left(h(\bt) \cdot \e^{(\ri \nu' \bt)}\right)\!(\bomega)\right).
	\end{align*}
	We conclude the proof by using the triangle inequality.
\end{proof}
We proceed by controlling the discretization error for the scalar integral. To do so, we also need to regularize the filter function and the transition weight $\gamma(\omega)$ by truncations
\begin{align}
f_T(t)&:=f(t)\cdot \indicator(t\in [-T/2,T/2))\label{eq:f_T}\\
	\gamma_W(\omega)&:= \gamma(\omega)\cdot \indicator(\omega\in [-W/2,W/2))~\label{eq:gammaW}.
\end{align}

\begin{lem}[Discretization error bounds for the integral]\label{lem:ContToDiscDecomp}
	Let $\bar{\CF}$ denote the discrete Fourier Transform with parameters $N,\omega_0,t_0$ and consider $f\in\ell_2(\mathbb{R})$ and $\gamma\in\ell_\infty(\mathbb{R})$ with truncation parameters $T,W$ satisfying $N \ge T/t_0\in\mathbb{Z}$ and $N\ge W/\omega_0\in\mathbb{Z}$. Then, for each $\nu,\nu'\in[-K,K]$, 
	\begin{align}
		\left|\int_{-\infty}^\infty \underset{g(\omega):=}{\underbrace{\gamma_W(\omega)\hat{f}^*(\omega-\nu)\hat{f}(\omega - \nu')}}\mathrm{d}\omega
		-\sum_{\bomega\in S^{\lceil N  \rfloor}_{\omega_0}}\gamma_W(\bomega)\bar{\CF}\left(\sqrt{t_0}f_T(\bt) \cdot \e^{(\ri \nu \bt)}\right)^{\!\!*}\!\!(\bomega)\bar{\CF}\left(\sqrt{t_0}f_T(\bt) \cdot \e^{(\ri \nu' \bt)}\right)\!(\bomega)\right|\leq \epsilon
	\end{align}
	holds provided the following conditions: \footnote{Note that the definition of $g(\omega)$ depends on the values $\nu,\nu'$, however we do not explicitly indicate this dependence for ease of notation.}
	\begin{align}\label{eq:DiscDeltaReq1}
		\sum_{k\in \mathbb{Z}}\omega_0 \nrm{g(\omega)-g(k\omega_0)}_{[k\omega_0,(k+1)\omega_0)}\leq \frac{\epsilon}{2},
	\end{align}	
	and
	\begin{align}\label{eq:DiscDeltaReq2}
		\sum_{\bt \in S^{\lceil N  \rfloor}_{t_0}}t_0 \nrm{f_T(t)-f_T(\bt)}_{[\bt,\bt+t_0)}\leq \delta, \qquad	t_0\nrm{f_T(t)}_\infty T(W+K)\leq \delta, \qquad \nrm{f-f_T}_1\leq (\sqrt{2\pi}-2)  \delta,
	\end{align}
	where $\delta=\min\left( \frac{\epsilon}{4W\nrm{\gamma_W}_\infty(\nrm{f}_1+1)},2\right)$.\footnote{Note that in case $N$ is odd in order to match the spacing of the grid $S^{\lceil N  \rfloor}$, the intervals $[k\omega_0,(k+1)\omega_0)$ should be changed to $[(k-\frac{1}{2})\omega_0,(k+\frac{1}{2})\omega_0)$ and analogously the endpoints of $[\bt,\bt+t_0)$ should be shifted to $\bt \pm t_0/2$.}
\end{lem}
\begin{proof}
	By \eqref{eq:DiscDeltaReq1} we have that
	\begin{align}\label{eq:RiemannDiff1}
		\left|\int_{\mathbb{R}} g(\omega) \mathrm{d}\omega - \sum_{k\in \mathbb{Z}} g(k\omega_0) \omega_0\right|
		\leq \sum_{k\in \mathbb{Z}}\left|\int_{k\omega_0}^{(k+1)\omega_0} g(\omega) - g(k\omega_0)\mathrm{d}\omega\right|
		\leq \sum_{k\in \mathbb{Z}}\omega_0 \nrm{g(\omega)-g(k\omega_0)}_{[k\omega_0,(k+1)\omega_0)}
		\leq \frac{\epsilon}{2}.\quad
	\end{align}
	Also observe that due to \eqref{eq:DiscDeltaReq2} we have for all $\omega\in \mathbb{R}$ that
	\begin{align}
		\sum_{\bt \in S^{\lceil N  \rfloor}_{t_0}}t_0 &\nrm{f_T(t)\e^{-\ri\omega t}-f_T(\bt)\e^{-\ri\omega\bt}}_{[\bt,\bt+t_0)}\\
		&\leq \sum_{\bt \in S^{\lceil N  \rfloor}_{t_0}}t_0 \nrm{f_T(t)\e^{-\ri\omega t}-f_T(\bt)\e^{-\ri\omega t}}_{[\bt,\bt+t_0)}+t_0 \nrm{f_T(\bt)\e^{-\ri\omega t}-f_T(\bt)\e^{-\ri\omega \bt}}_{[\bt,\bt+t_0)}\\
		&\leq \sum_{\bt \in S^{\lceil N  \rfloor}_{t_0}}t_0\nrm{f_T(t)-f_T(\bt)}_{[\bt,\bt+t_0)}+t_0|f_T(\bt)|\nrm{\e^{-\ri\omega t}-\e^{-\ri\omega \bt}}_{[\bt,\bt+t_0)}\\
		&\leq \sum_{\bt \in S^{\lceil N  \rfloor}_{t_0}}t_0\nrm{f_T(t)-f_T(\bt)}_{[\bt,\bt+t_0)}+t_0|f_T(\bt)|t_0|\omega|\\&
		\leq \delta+t_0\nrm{f_T(t)}_\infty T|\omega|.\label{eq:RiemannDiff2}
	\end{align}
	Next we define $\tilde{\hat{f}}_T(\omega):=\frac{t_0}{\sqrt{2\pi}}\sum_{\bt \in S_{t_0}}f_T(\bt)\e^{-\ri\omega \bt}$. The above two inequalities imply similarly to $\eqref{eq:RiemannDiff1}$ that for all $\omega$ satisfying $|\omega|\leq W+K$ we have
	\begin{align}
		|\tilde{\hat{f}}_T(\omega)-\hat{f}(\omega)|
		&\leq\left|\tilde{\hat{f}}_T(\omega)-\frac{1}{\sqrt{2\pi}}\int f_T(t)\e^{-\ri \omega t}\mathrm{d}t\right|+\left|\frac{1}{\sqrt{2\pi}}\int (f_T(t)-f(t))\e^{-\ri \omega t}\mathrm{d}t\right|\\&
		\leq\sum_{ \bt \in S_{t_0}}\left|\frac{t_0}{\sqrt{2\pi}} f_T(\bt)\e^{-\ri\omega \bt} -\frac{1}{\sqrt{2\pi}}\int_{\bt}^{\bt+t_0} f_T(t)\e^{-\ri \omega t}\mathrm{d}t\right|+(\sqrt{2\pi}-2) \frac{\delta}{\sqrt{2\pi}}\\&
		=\frac{1}{\sqrt{2\pi}}\sum_{\bt \in S_{t_0}}\left|\int_{\bt}^{\bt+t_0}f_T(\bt)\e^{-\ri\omega \bt} - f_T(t)\e^{-\ri \omega t}\mathrm{d}t\right|+(\sqrt{2\pi}-2)   \frac{\delta}{\sqrt{2\pi}}\\&
		\leq \frac{2\delta}{\sqrt{2\pi}} + (\sqrt{2\pi}-2)\frac{\delta}{\sqrt{2\pi}}\\& 
		= \delta =\min\left( \frac{\epsilon}{4W\nrm{\gamma_W}_\infty(\nrm{f}_1+1)},2\right).\label{eq:DiscDeltaDer1}
	\end{align}
	Let $\tilde{g}(\omega):=\gamma_W(\omega)\tilde{\hat{f}}_T^*(\omega-\nu)\tilde{\hat{f}}_T(\omega - \nu')$.
	Considering that $|\nu|,|\nu'|\leq K$ we get that for all $\omega\in\mathbb{R}$
	\begin{align}
		&|g(\omega)-\tilde{g}(\omega)|\\&
		\leq|\gamma_W(\omega)|\left(|\hat{f}^*(\omega-\nu)\hat{f}(\omega - \nu')-\hat{f}^*(\omega-\nu)\tilde{\hat{f}}_T(\omega - \nu')|+|\hat{f}^*(\omega-\nu)\tilde{\hat{f}}_T(\omega - \nu')-\tilde{\hat{f}}_T^*(\omega-\nu)\tilde{\hat{f}}_T(\omega - \nu')|\right)\\&
		=|\gamma_W(\omega)|\left(|\hat{f}^*(\omega-\nu)||\hat{f}(\omega - \nu')-\tilde{\hat{f}}_T(\omega - \nu')|+|\hat{f}^*(\omega-\nu)-\tilde{\hat{f}}_T^*(\omega-\nu)||\tilde{\hat{f}}_T(\omega - \nu')|\right)\\&
		\leq \delta|\gamma_W(\omega)|\left(|\hat{f}^*(\omega-\nu)|+|\tilde{\hat{f}}_T(\omega - \nu')|\right)\tag*{(by \eqref{eq:DiscDeltaDer1})}\\&
		\leq\delta|\gamma_W(\omega)|\left(|\hat{f}^*(\omega-\nu)|+|\hat{f}(\omega - \nu')|+|\tilde{\hat{f}}_T(\omega - \nu')-\hat{f}(\omega - \nu')|\right)\\&
		\leq \delta|\gamma_W(\omega)|(2\nrm{f}_1+2)\tag*{(by \eqref{eq:DiscDeltaDer1})}\\&
		\leq\frac{\epsilon}{2W}. \tag*{(by \eqref{eq:DiscDeltaDer1})}
	\end{align}
	This implies 
	\begin{align}\label{eq:RiemannDiff3}
		\left|\sum_{k\in \mathbb{Z}}\omega_0 (g(k\omega_0)-\tilde{g}(k\omega_0))\right|
		\leq \omega_0\sum_{k\in \mathbb{Z}} \left|g(k\omega_0)-\tilde{g}(k\omega_0)\right|
		\leq \sum_{\underset{k\omega_0\in [-W/2,W/2)}{k\in \mathbb{Z}}} \frac{\epsilon \omega_0}{2W}
		\leq \frac{\epsilon}{2},
	\end{align}
	showing that 
	\begin{align}
		\left|\int_{\mathbb{R}} g(\omega)\mathrm{d}\omega-\sum_{k\in \mathbb{Z}}\omega_0\tilde{g}(k\omega_0)\right|
		&=\left|\int_{\mathbb{R}} \gamma_W(\omega)\hat{f}^*(\omega-\nu)\hat{f}(\omega - \nu')\mathrm{d}\omega-\sum_{k\in \mathbb{Z}}\omega_0\gamma_W(k\omega_0)\tilde{\hat{f}}_T^*(k\omega_0-\nu)\tilde{\hat{f}}_T(k\omega_0-\nu')\right|\\&
		\leq \left|\int_{\mathbb{R}}\gamma_W(\omega)\hat{f}^*(\omega-\nu)\hat{f}(\omega - \nu')\mathrm{d}\omega-\sum_{k\in \mathbb{Z}}\omega_0\gamma_W(k\omega_0)\hat{f}^*(k\omega_0-\nu)\hat{f}(k\omega_0-\nu')\right|\\&
		+\left|\omega_0\sum_{k\in \mathbb{Z}}\gamma_W(k\omega_0)\hat{f}^*(k\omega_0-\nu)\hat{f}(k\omega_0-\nu')-\gamma_W(k\omega_0)\tilde{\hat{f}}_T^*(k\omega_0-\nu)\tilde{\hat{f}}_T(k\omega_0-\nu')\right|\\&
		\leq \epsilon. \tag*{(by \eqref{eq:RiemannDiff1} and \eqref{eq:RiemannDiff3})}
	\end{align}
	We conclude the proof by observing that due to $W\leq N\omega_0$, $T\leq N t_0$, and $\frac{\omega_0 t_0}{2\pi}=\frac{1}{N}$ we have
	\begin{align}
		\sum_{k\in \mathbb{Z}}\omega_0\tilde{g}(k\omega_0)
		&=\omega_0\sum_{\bomega\in \mathbb{Z}\omega_0}\gamma_W(\bomega)\tilde{\hat{f}}_T^*(\bomega-\nu)\tilde{\hat{f}}_T(\bomega-\nu')\\&
		=\frac{\omega_0 t_0^2}{2\pi}\sum_{\bomega\in \mathbb{Z}\omega_0}\gamma_W(\bomega)\left(\sum_{\bt\in\mathbb{Z}t_0}f_T(\bt)\e^{-\ri(\bomega-\nu)\bt}\right)^{\!\!*}\left(\sum_{\bt\in\mathbb{Z}t_0}f_T(\bt)\e^{-\ri(\bomega-\nu')\bt}\right)\\&
		=\frac{1}{N}\sum_{\bomega\in S^{\lceil N  \rfloor}_{\omega_0}}\gamma_W(\bomega)\left(\sum_{\bt\in S^{\lceil N  \rfloor}_{t_0}}\sqrt{t_0}f_T(\bt)\e^{\ri\nu\bt}\e^{-\ri\bomega\bt}\right)^{\!\!*}\left(\sum_{\bt\in S^{\lceil N  \rfloor}_{t_0}}\sqrt{t_0}f_T(\bt)\e^{\ri\nu'\bt}\e^{-\ri\bomega\bt}\right)\\&	
		=\sum_{\bomega\in S^{\lceil N  \rfloor}_{\omega_0}}\gamma_W(\bomega)\bar{\CF}\left(\sqrt{t_0}f_T(\bt) \cdot \e^{(\ri \nu \bt)}\right)^{\!\!*}\!\!(\bomega)\bar{\CF}\left(\sqrt{t_0}f_T(\bt) \cdot \e^{(\ri \nu' \bt)}\right)\!(\bomega).\tag*{\qedhere}
	\end{align}
\end{proof}
As a sanity check, the above implies that the discretized Lindbladian convergences to the continuum \textit{in the limit}.
\begin{restatable}[Discretizations converge to the continuum]{thm}{discLindToCont}\label{thm:discToCont}
	If $\gamma\in \ell_\infty(\mathbb{R})$, $f\in \ell_2(\mathbb{R})$, and $\gamma$, $f$ are continuous almost everywhere (i.e., the set of points of discontinuity has measure zero) while $f$ is bounded on every finite interval, then
	\begin{align}
		\lim_{W,T\rightarrow \infty}\lim_{N\rightarrow \infty} &\sum_{\bomega\in S^{\lceil N  \rfloor}_{\omega_0}}
		\gamma_W(\bomega) \left( \hat{\vA}_{\bar{f}_T}(\bomega)^\dagg[\cdot]\hat{\vA}_{\bar{f}_T}(\bomega) -\frac{1}{2}\{\hat{\vA}_{\bar{f}_T}(\bomega)^\dagg\hat{\vA}_{\bar{f}_T}(\bomega),\cdot \} \right)\\
		= &\int_{-\infty}^{\infty}  \gamma(\omega) \left( \hat{\vA}_{f}(\omega)^\dagg[\cdot]\hat{\vA}_{f}(\omega) -\frac{1}{2}\{\hat{\vA}_{f}(\omega)^\dagg\hat{\vA}_{f}(\omega),\cdot \} \right) \mathrm{d}\omega, \label{eq:cont_Lind}
	\end{align}
	and if $\gamma\geq 0$, then also
	\begin{align}
		\!\lim_{W,T\rightarrow \infty}\lim_{N\rightarrow \infty} &\sum_{\bomega\in S^{\lceil N  \rfloor}_{\omega_0}}\!\! \sqrt{\gamma_W(\bomega)\gamma_W(-\bomega)} \hat{\vA}_{\bar{f}_T}(\bomega)\otimes \hat{\vA}_{\bar{f}_T}(\bomega)^*- \frac{\gamma(\bomega)}{2} \L(\hat{\vA}_{\bar{f}_T}(\bomega)^\dagg\hat{\vA}_{\bar{f}_T}(\bomega)\otimes \vI + \vI \otimes \hat{\vA}_{\bar{f}_T}(\bomega)^{\dagg*}\hat{\vA}_{\bar{f}_T}(\bomega)^*\R)\!\\
		= &\int_{-\infty}^{\infty} \sqrt{\gamma(\omega)\gamma(-\omega)}  \hat{\vA}_{f}(\omega) \otimes  \hat{\vA}_{f}(\omega)^* - \frac{\gamma(\omega)}{2} \L(\hat{\vA}_{f}(\omega)^\dagg\hat{\vA}_{f}(\omega) \otimes \vI + \vI \otimes \hat{\vA}_{f}(\omega)^{\dagg*}\hat{\vA}_{f}(\omega)^*\R)\ \mathrm{d} \omega, \label{eq:cont_Disc}
	\end{align}	
	where $\omega_0=t_0=\sqrt{2\pi/N}$ and $\bar{f}_T(t):=\sqrt{t_0}f_T(t)$.
\end{restatable}
\begin{proof}
	Due to \autoref{lem:discToContDecomp}, it suffices to prove for all $\nu,\nu'\in B$ that 
	\begin{align}
		\lim_{W,T\rightarrow \infty}\lim_{N\rightarrow \infty} &\kern-3mm\sum_{\kern3mm\bomega\in S^{\lceil N  \rfloor}_{\omega_0}}\kern-1.5mm
		\gamma_W(\bomega) \bar{\CF}\left(\bar{f}_T(\bt) \cdot \e^{(\ri \nu \bt)}\right)^{\!\!*}\!\!(\bomega)\bar{\CF}\left(\bar{f}_T(\bt) \cdot \e^{(\ri \nu' \bt)}\right)\!(\bomega)
		= \int_{-\infty}^{\infty}  \gamma(\omega) \hat{f}^*(\omega\!-\!\nu)\hat{f}(\omega\!-\!\nu')\mathrm{d}\omega,\label{eq:intLim1}\\
		\lim_{W,T\rightarrow \infty}\lim_{N\rightarrow \infty} &\kern-3mm\sum_{\kern3mm\bomega\in S^{\lceil N  \rfloor}_{\omega_0}}
		\kern-4mm\sqrt{\gamma_W(\bomega)\gamma_W(-\bomega)} \bar{\CF}\left(\bar{f}_T(\bt) \cdot \e^{(\ri \nu \bt)}\right)\!(\bomega)\bar{\CF}\left(\bar{f}_T(\bt) \cdot \e^{(\ri \nu' \bt)}\right)^{\!\!*}\!\!(\bomega)
		=\! \int_{-\infty}^{\infty}  \kern-3mm\sqrt{\gamma(\omega)\gamma(-\omega)} \hat{f}(\omega\!-\!\nu)\hat{f}^*(\omega\!-\!\nu')\mathrm{d}\omega.\!	\label{eq:intLim2}	
	\end{align}
	Since $f\in \ell_2(\mathbb{R})$, we have $\hat{f}\in \ell_2(\mathbb{R})$ and therefore by Hölder's inequality we get that $\gamma(\omega)\hat{f}^*(\omega-\nu)\hat{f}(\omega - \nu')\in \ell_1(\mathbb{R})$, which then implies 
	\begin{align}\label{eq:LimIntW}
		\int_{-\infty}^{\infty}   \gamma(\omega)  \hat{f}^*(\omega-\nu)\hat{f}(\omega - \nu')\mathrm{d}\omega=
		\lim_{W\rightarrow \infty} \int_{-W}^{W}  \gamma(\omega)\hat{f}^*(\omega-\nu)\hat{f}(\omega - \nu')\mathrm{d}\omega.
	\end{align}
	Since $f=\lim_{T\rightarrow \infty}f_T$, Parseval's Theorem implies $\hat{f}=\lim_{T\rightarrow \infty}\hat{f}_T$ and so by Hölder's inequality we get 
	\begin{align}\label{eq:LimIntT}
		\int_{-W}^{W}  \gamma(\omega)\hat{f}^*(\omega-\nu)\hat{f}(\omega - \nu')\mathrm{d}\omega 
		=\lim_{T\rightarrow \infty} \int_{-W}^{W} \gamma(\omega)\hat{f}_T^*(\omega-\nu)\hat{f}_T(\omega - \nu')\mathrm{d}\omega .
	\end{align}
	As $f_T$ is bounded, continuous almost everywhere, and has compact support, the Lebesgue-Vitali Theorem~\cite[Theorem 11.33]{rudin1976PrinciplesMatAnal,muger2006LebesueRiemannInt} implies that it is Riemann integrable. Therefore,  $\hat{f}_T$ is bounded and continuous which similarly implies that $g(\omega):=\gamma_W(\omega)\hat{f}_T^*(\omega-\nu)\hat{f}_T(\omega - \nu')$ is Riemann integrable on $[-W,W]$.
	
	Since $f\leftarrow f_T$, $\lim_{N\rightarrow 0}\omega_0=\lim_{N\rightarrow 0}t_0=0$, and $f_T(t),g(\omega)$ are Riemann integrable (c.f.\ \cite[Theorem 8.26.]{clark2014HonorsCalculus}), for every $\epsilon>0$ \eqref{eq:DiscDeltaReq1}-\eqref{eq:DiscDeltaReq2} are satisfied for large enough $N$. Thus \autoref{lem:ContToDiscDecomp} implies that 
	\begin{align}\label{eq:RiemannConvFinal}
		\int_{-W}^{W}  \gamma(\omega)\hat{f}_T^*(\omega-\nu)\hat{f}_T(\omega - \nu')\mathrm{d}\omega
		&=\lim_{N\rightarrow \infty}\sum_{\bomega\in S^{\lceil N  \rfloor}_{\omega_0}}\gamma_W(\bomega)\bar{\CF}\left(\bar{f}_T(\bt) \cdot \e^{(\ri \nu \bt)}\right)^{\!\!*}\!\!(\bomega)\bar{\CF}\left(\bar{f}_T(\bt) \cdot \e^{(\ri \nu' \bt)}\right)\!(\bomega).
	\end{align}
	
	Equation~\eqref{eq:intLim2} can be analogously proven after replacing $\gamma(\omega)$ with $\sqrt{\gamma(\omega)\gamma(-\omega)}$ throughout the argument.
\end{proof}

We believe that a similar result can be shown for any $\gamma\in \ell_\infty(\mathbb{R})$, $f\in \ell_2(\mathbb{R})$ without the other assumptions in \autoref{thm:discToCont} by applying a further approximation with the help of mollifiers. While such an asymptotic result is conceptually elegant, for quantum algorithm implementation, we need quantitative, nonasymptotic error bounds for the particular functions we encounter.

\begin{restatable}[Discretization error of \Lword{}s and discriminant proxies]{prop}{TruncatedContToDisc}\label{prop:TruncatedContToDisc}
In the setting of \autoref{lem:ContToDiscDecomp}, assume continuity and boundedness assumption for $f,\gamma$ as in \autoref{thm:discToCont} with normalization $\nrm{\gamma}_\infty,\nrm{f}_2\leq1$. Consider a single jump operator $\vA$ with $\nrm{\vA}\leq 1$, assume that $\gamma(\omega)$ and $\sqrt{\gamma(\omega)\gamma(-\omega)}$ are $C$-Lipschitz continuous, $\hat{f}$ is $(D\cdot\lVert\hat{f}\rVert_\infty)$-Lipschitz continuous, $f_T$ is $(L\cdot\lVert f_T\rVert_\infty)$-Lipschitz continuous on $[-T/2,T/2)$, and assume the following conditions:
	\begin{align}
 \nrm{f_T-f}_2&\leq\frac{\epsilon}{64}, \quad \sum_{\bomega\in S^{\lceil N \rfloor}_{t_0}}t_0\left|f(\bt)-f_T(\bt)\right|^2\leq\left(\frac{\epsilon}{64}\right)^2,\quad \sum_{\bomega\in S^{\lceil N \rfloor}_{t_0}}t_0\left|f(\bt)\right|^2\leq1\\
		\int_{-\infty}^\infty |\hat{f}(\omega) \indicator(|\omega|\geq W - 2 \nrm{\vH})|^2\mathrm{d}\omega   &\leq \frac{\epsilon}{8\left(\frac{256 \nrm{\vH} T }{\epsilon}+1\right)^{\!\! 2}},\label{eq:ContTailWeight}
	\end{align}
	and $\nrm{f-f_T}_1\leq (\sqrt{2\pi}-2)  \delta$ for $\delta=\min\bigg( \frac{\epsilon}{32 \left(\frac{256 \nrm{\vH} T }{\epsilon}+1\right)^{\! 2} W(\nrm{f}_1+1)},2\bigg)$, and
	\begin{align}
		t_0&\leq \frac{\delta}{T \nrm{f_T}_\infty} \min\bigg(\frac{1}{W+2\nrm{\vH}},\frac{1}{L} \bigg),\\
  \omega_0&\leq\frac{\epsilon}{16}\cdot\min\left(\frac{1}{8 T}, \frac{1}{ W (C+2D)\lVert\hat{f}\rVert^2_\infty\left( \frac{256 \nrm{\vH} T }{\epsilon}+1\right)^{\!\! 2}}\right)\label{eq:BaseFreqBound}.
	\end{align}
	Then, in the notation of \autoref{cor:fHPerturBound}-\ref{cor:fHContPerturBound},
	\begin{align}
		\nrm{\CL_{(f,\vH)}-\bar{\CL}^{(W)}_{(f,\vH)}}_{1-1}\leq \epsilon,\qquad
		\nrm{\vec{\CD}_{(f,\vH)}-\bar{\vec{\CD}}^{(W)}_{(f,\vH)}}\leq \epsilon,
	\end{align}
	where the discretized \Lword{} and discriminant proxy uses $\gamma_W(\omega)$ instead of $\gamma(\omega)$.
	Moreover, if 
	\begin{align}
		\sum_{\bomega\in S^{\lceil N \rfloor}_{\omega_0}}\left|\bar{\CF}\left(\bar{f}(\bt)  \e^{(\ri \nu' \bt)}\right)\!(\bomega)\indicator(|\bomega|\geq W)\right|^2&\leq \frac{\epsilon}{8\left(\frac{256 \nrm{\vH} T }{\epsilon}+1\right)^{\!\! 2}}\quad \text{for each}\quad |\nu|\leq 2\nrm{\vH},\label{eq:DiscTailWeight}
	\end{align}	
	then 
	\begin{align}
		\nrm{\CL_{(f,\vH)}-\bar{\CL}_{(f,\vH)}}_{1-1}\leq \epsilon,\qquad
		\nrm{\vec{\CD}_{(f,\vH)}-\bar{\vec{\CD}}_{(f,\vH)}}\leq \epsilon.
	\end{align}
\end{restatable}
\begin{proof}
	Let $\bar{\vH}$ be the Hamiltonian obtained by rounding the spectrum of $\vH$ (down in absolute value) to $\mathbb{Z}\eta$; this is crucial before we invoke~\autoref{lem:discToContDecomp}.
	Our proof proceeds using the following triangle inequalities:
	\begin{align}
		\nrm{\CL_{(f,\vH)}-\tilde{\CL}_{(f,\vH)}}_{1-1}
		&\leq \nrm{\CL_{(f,\vH)}-\CL_{(f,\bar{\vH})}}_{1-1}
		+\nrm{\CL_{(f,\bar{\vH})}-\tilde{\CL}_{(f,\bar{\vH})}}_{1-1}
		+\nrm{\tilde{\CL}_{(f,\bar{\vH})}-\tilde{\CL}_{(f,\vH)}}_{1-1}\\
		\nrm{\vec{\CD}_{(f,\vH)}-\tilde{\vec{\CD}}_{(f,\vH)}}
		&\leq \nrm{\vec{\CD}_{(f,\vH)}-\vec{\CD}_{(f,\bar{\vH})}}
		+\nrm{\vec{\CD}_{(f,\bar{\vH})}-\tilde{\vec{\CD}}_{(f,\bar{\vH})}}
		+\nrm{\tilde{\vec{\CD}}_{(f,\bar{\vH})}-\tilde{\vec{\CD}}_{(f,\vH)}}
	\end{align}
	where $\tilde{\CL}_{(f,\vH)}$ stands for either $\bar{\CL}^{(W)}_{(f,\vH)}$ or $\bar{\CL}_{(f,\vH)}$ and
	$\tilde{\vec{\CD}}_{(f,\vH)}$ for either $\bar{\vec{\CD}}^{(W)}_{(f,\vH)}$ or $\bar{\vec{\CD}}_{(f,\vH)}$.
	
	Choosing $\eta:=\frac{\epsilon}{64 T }$, we get by \autoref{cor:fHContPerturBound} that
	\begin{align*}
		\nrm{\CL_{(f,\vH)}-\CL_{(f,\bar{\vH})}}_{1-1}
		&\leq \nrm{\CL_{(f,\vH)}-\CL_{(f_T,\bar{\vH})}}_{1-1}
		+ \nrm{\CL_{(f_T,\bar{\vH})}-\CL_{(f,\bar{\vH})}}_{1-1}
		&\leq 8 (T\eta + \nrm{f_T-f}_2)\leq\frac{\epsilon}{8},\\
		\nrm{\vec{\CD}_{(f,\vH)}-\vec{\CD}_{(f,\bar{\vH})}}
		&\leq \nrm{\vec{\CD}_{(f,\vH)}-\vec{\CD}_{(f_T,\bar{\vH})}}
		+ \nrm{\vec{\CD}_{(f_T,\bar{\vH})}-\vec{\CD}_{(f,\bar{\vH})}}
		&\leq 8 (T\eta + \nrm{f_T-f}_2)\leq\frac{\epsilon}{8},
	\end{align*}
	and similarly by \autoref{cor:fHPerturBound} that 
	\begin{align*}
		\nrm{\tilde{\CL}_{(f,\vH)}-\tilde{\CL}_{(f,\bar{\vH})}}_{1-1}
		&\leq \nrm{\tilde{\CL}_{(f,\vH)}-\tilde{\CL}_{(f_T,\bar{\vH})}}_{1-1}
		+ \nrm{\tilde{\CL}_{(f_T,\bar{\vH})}-\tilde{\CL}_{(f,\bar{\vH})}}_{1-1}
		\leq \frac{\epsilon}{8},\\
		\nrm{\tilde{\vec{\CD}}_{(f,\vH)}-\tilde{\vec{\CD}}_{(f,\bar{\vH})}}
		&\leq \nrm{\tilde{\vec{\CD}}_{(f,\vH)}-\tilde{\vec{\CD}}_{(f_T,\bar{\vH})}}
		+ \nrm{\tilde{\vec{\CD}}_{(f_T,\bar{\vH})}-\tilde{\vec{\CD}}_{(f,\bar{\vH})}}
		\leq \frac{\epsilon}{8}.
	\end{align*}
	
	Therefore, it suffices to show that $\nrm{\CL_{(f,\bar{\vH})}-\tilde{\CL}_{(f,\bar{\vH})}}_{1-1}$, $\nrm{\vec{\CD}_{(f,\bar{\vH})}-\tilde{\vec{\CD}}_{(f,\bar{\vH})}}\leq\frac{3\epsilon}{4}$. Let $\tilde{\gamma}$ be either $\gamma$ or $\gamma_W$ matching the definition of $\tilde{\CL}$, and $\tilde{\vec{\CD}}$. We apply \autoref{lem:discToContDecomp} with $\star=\cdot\vrho\cdot$ and $\star=\cdot$, showing that the difference between the discrete and continuous generators $\nrm{\CL_{(f,\bar{\vH})}-\tilde{\CL}_{(f,\bar{\vH})}}_{1-1}$, $\nrm{\vec{\CD}_{(f,\bar{\vH})}-\tilde{\vec{\CD}}_{(f,\bar{\vH})}}$ can be bounded by 
	\begin{align}\label{eqn:UnifTarget}
		\kern-6mm 2\left(\frac{256 \nrm{\vH} T }{\epsilon}+1\right)^{\!\!2}\left|\int_{\!-\infty}^{\infty}
		\gamma(\omega)  \hat{f}^*(\omega\!-\!\nu)\hat{f}(\omega\!-\!\nu')\mathrm{d}\omega 
		-\sum_{\bomega\in S^{\lceil N \rfloor}_{\omega_0}}\! \tilde{\gamma}(\bomega) \bar{\CF}\left(\bar{f}_T(\bt) \cdot \e^{(\ri \nu \bt)}\right)^{\!\!*}\!\!(\bomega)\bar{\CF}\left(\bar{f}_T(\bt) \cdot \e^{(\ri \nu' \bt)}\right)\!(\bomega)\right|,
	\end{align}
	since the number of Bohr frequencies for the discretized Hamiltonian $\bar{\vH}$ satisfies $\labs{B(\bvH)}\leq\frac{4\nrm{\bar{\vH}}}{\eta}+1\leq\frac{4\nrm{\vH}}{\eta}+1=\frac{256 \nrm{\vH} T }{\epsilon}+1$. 
	
	We further bound $\nrm{\CL_{(f,\bar{\vH})}-\tilde{\CL}_{(f,\bar{\vH})}}_{1-1}$ in three steps according to the following triangle inequality:
	\begin{align}
		&\left|\int_{\!-\infty}^{\infty}   
		\gamma(\omega)  \hat{f}^*(\omega\!-\!\nu)\hat{f}(\omega\!-\!\nu')\mathrm{d}\omega
		-\sum_{\bomega\in S^{\lceil N \rfloor}_{\omega_0}}\! \tilde{\gamma}(\bomega) \bar{\CF}\left(\bar{f}(\bt)  \e^{(\ri \nu \bt)}\right)^{\!\!*}\!\!(\bomega)\bar{\CF}\left(\bar{f}(\bt)  \e^{(\ri \nu' \bt)}\right)\!(\bomega)\right|\\&\		
		\leq \left|\int_{\!-\infty}^{\infty} \! \gamma(\omega)  \hat{f}^*(\omega\!-\!\nu)\hat{f}(\omega\!-\!\nu')\mathrm{d}\omega 
		-\int_{\!-\infty}^{\infty} 
		\gamma_W(\omega)\hat{f}^*(\omega\!-\!\nu)\hat{f}(\omega\!-\!\nu')\mathrm{d}\omega\right|  \label{eqn:UnifTarget1}\\&
		+ \left|\int_{\!-\infty}^{\infty} 
		\gamma_W(\omega)\hat{f}^*(\omega\!-\!\nu)\hat{f}(\omega\!-\!\nu')\mathrm{d}\omega
		-\!\!\!\!\sum_{\bomega\in S^{\lceil N \rfloor}_{\omega_0}}\! \gamma_W(\bomega) \bar{\CF}\left(\bar{f}(\bt)  \e^{(\ri \nu \bt)}\right)^{\!\!*}\!\!(\bomega)\bar{\CF}\left(\bar{f}(\bt)  \e^{(\ri \nu' \bt)}\right)\!(\bomega)\right|.\label{eqn:UnifTarget2}\\&
		+ \left|\sum_{\bomega\in S^{\lceil N \rfloor}_{\omega_0}}\! (\gamma_W(\bomega)-\tilde{\gamma}(\bomega)) \bar{\CF}\left(\bar{f}(\bt)  \e^{(\ri \nu \bt)}\right)^{\!\!*}\!\!(\bomega)\bar{\CF}\left(\bar{f}(\bt)  \e^{(\ri \nu' \bt)}\right)\!(\bomega)\right|.\label{eqn:UnifTarget3}
	\end{align}	
	Considering that $|\nu|,|\nu'|\leq 2\nrm{\vH}$ and using Hölder's inequality, we can see that truncation at threshold $W$ introduces error $\leq \frac{\epsilon}{8 \big(\frac{256 \nrm{\vH} T }{\epsilon}+1\big)^{\! 2}}\leq \frac{\epsilon}{8 |B(\bvH)|^2}$ in \eqref{eqn:UnifTarget1}, inducing no more than $\frac{\epsilon}{4}$ error in \eqref{eqn:UnifTarget}. 
	If $\tilde{\gamma}:=\gamma$, and \eqref{eq:DiscTailWeight} holds, then the analogous argument shows the same bound for \eqref{eqn:UnifTarget3}.
	
	We complete our proof by showing that \eqref{eqn:UnifTarget2} is bounded by $\frac{\epsilon}{8 \big(\frac{256 \nrm{\vH} T }{\epsilon}+1\big)^{\! 2}}$, by applying \autoref{lem:ContToDiscDecomp}, i.e., showing that \eqref{eq:DiscDeltaReq1}-\eqref{eq:DiscDeltaReq2} are fulfilled.
	In \autoref{lem:ContToDiscDecomp} we set $K=2\nrm{\vH}$ and our accuracy goal $\epsilon\leftarrow\frac{\epsilon}{8 \big(\frac{256 \nrm{\vH} T }{\epsilon}+1\big)^{\! 2}}$.
	To bound \eqref{eq:DiscDeltaReq1} we observe that 
	\begin{align}
		g(\omega):=\gamma_W(\omega)\hat{f}^*(\omega-\nu)\hat{f}(\omega - \nu')
	\end{align}
	is $(C+2D)\lVert\hat{f}\rVert^2_\infty$-Lipschitz continuous:
	\begin{align*}
		\left|g(\omega+\delta)-g(\omega)\right|\leq& 
		\left|\gamma_W(\omega+\delta)\hat{f}^*(\omega+\delta-\nu)\hat{f}_T(\omega+\delta - \nu')-\gamma_W(\omega)\hat{f}^*(\omega+\delta-\nu)\hat{f}_T(\omega+\delta - \nu')\right|\\&
		+\left|\gamma_W(\omega)\hat{f}^*(\omega+\delta-\nu)\hat{f}_T(\omega+\delta - \nu')-\gamma_W(\omega)\hat{f}^*(\omega-\nu)\hat{f}_T(\omega+\delta - \nu')\right|\\&
		+\left|\gamma_W(\omega)\hat{f}^*(\omega-\nu)\hat{f}_T(\omega+\delta - \nu')-\gamma_W(\omega)\hat{f}^*(\omega-\nu)\hat{f}_T(\omega - \nu')\right|\\
		\leq& (C+2D)\lVert\hat{f}\rVert^2_\infty\cdot\delta,
	\end{align*}
	therefore by \eqref{eq:BaseFreqBound} we can upper bound the left-hand side of \eqref{eq:DiscDeltaReq1} as 
	\begin{align}
		W\omega_0 (C+2D)\lVert\hat{f}\rVert^2_\infty\leq \frac{\epsilon}{16 \big(\frac{256 \nrm{\vH} T }{\epsilon}+1\big)^{\! 2}}.
	\end{align}
	Now observe that due to the $(L\cdot\lVert f_T\rVert_\infty)$-Lipschitz continuity of $f_T$ we have
	\begin{align}
		\sum_{\bt \in S^{\lceil N  \rfloor}_{t_0}}t_0 \nrm{f_T(t)-f_T(\bt)}_{[\bt,\bt+t_0)}\leq \sum_{\bt \in S^{\lceil N  \rfloor}_{t_0}\cap  [-T/2,T/2)}t_0^2(L\cdot\lVert f_T\rVert_\infty)
		=Tt_0(L\cdot\lVert f_T\rVert_\infty)\leq \delta.
	\end{align}

	Repeating the above argument replacing $\gamma(\omega)$ by $\sqrt{\gamma(\omega)\gamma(-\omega)}$ proves the same bound for $\nrm{\vec{\CD}_{(f,\bar{\vH})}-\tilde{\vec{\CD}}_{(f,\bar{\vH})}}$.
\end{proof}

\begin{lem}\label{lem:gammaLipschitz}
	Both functions $\gamma_G(\omega):=(\e^{4\beta \omega}+1)^{-1}$ and $\gamma_M(\omega):=\min(1, \e^{-\beta \omega})$ are $\beta$-Lipschitz continuous for all $\beta\geq 0$. Moreover, $\sqrt{\gamma_G(\omega)\gamma_G(-\omega)}$ and $\sqrt{\gamma_M(\omega)\gamma_M(-\omega)}$ are $\frac{\beta}{2}$-Lipschitz continuous for all $\beta\geq 0$.
\end{lem}
\begin{proof}
	A simple calculation shows that the absolute value of the derivative of $(\e^{4\beta \omega}+1)^{-1}$ is largest at $0$, where it is $\beta$, therefore it is $\beta$-Lipschitz continuous. Similarly, since $\min(1, \e^{-\beta \omega})$ is continuous, and the absolute value of the (right) derivative of $\min(1, \e^{-\beta \omega})$ is bounded by $\beta$ for every $\omega\in\mathbb{R}$, it is also $\beta$-Lipschitz continuous. Similar elementary calculation shows the $\frac{\beta}{2}$-Lipschitz continuity of $\sqrt{\gamma_G(\omega)\gamma_G(-\omega)}$ and $\sqrt{\gamma_M(\omega)\gamma_M(-\omega)}$.
\end{proof}

Note that since we apply our generic bound~\autoref{prop:TruncatedContToDisc} to the following nice functions, we get rather loose estimates of $N$, which are certainly off by polynomial factors from the tight values. However, since algorithmically, we only pay (poly-)logarithmic cost in $N$, the looseness of our bounds probably only results in constant overheads.

\begin{restatable}[Discretization error of ``finitie-time'' Davies generators]{cor}{UnifContToDisc}\label{cor:UnifContToDisc}
	Let $f(t)=\sqrt{\frac{1}{T}}\indicator(t\in [-T/2,T/2))$ and $\epsilon>0$. Assuming the normalization condition \eqref{eq:AAdagger}, if $\nrm{\gamma}_\infty\leq1$ and $\sqrt{\gamma(\omega)\gamma(-\omega)}$ are $\beta$-Lipschitz continuous (e.g., Metropolis or Glauber), $\nrm{\gamma}_\infty\leq1$,  then setting $N=2^n$ for $n=\bigTh{\log\left(\frac{2|A|(\nrm{\vH}+1)(T+1)(\beta+1)}{\epsilon}\right)}$ with appropriate constants ensures that
	\begin{align}
		\nrm{\CL_{(f,\vH)}-\bar{\CL}_{(f,\vH)}}_{1-1}\leq \epsilon,\qquad
		\nrm{\vec{\CD}_{(f,\vH)}-\bar{\vec{\CD}}_{(f,\vH)}}\leq \epsilon.
	\end{align}
	where $\omega_0=\frac{2}{T}\sqrt{\frac{2\pi}{(C+1)N}}$, $t_0=\frac{T}{2}\sqrt{\frac{2\pi(C+1)}{N}}$. 
\end{restatable}
\begin{proof}
	We use \autoref{prop:TruncatedContToDisc} for each $\vA^{a}$ setting $\epsilon'\leftarrow\frac{\epsilon}{|A|}$ to prove the claim. First observe that $f_T=f$, and $\nrm{f}_2=1=\sum_{\bomega\in S^{\lceil N \rfloor}_{t_0}}t_0\left|f(\bt)\right|^2=1$.
	As $\hat{f}(w)= \sinc(T\omega/2)$, and $|\sinc(T\omega/2)|\leq 1$, and $|\frac{\omega}{\mathrm{d}\omega}\sinc(T\omega/2)|\leq T/4$ we get that $D \leq\frac{T}{4}$ and trivially $L=0$.
	
	Let us recall the tail bound on the discrete Fourier Transform from \autoref{prop:uniformTail}
	\begin{align}\label{eq:dicsPhError}
		\sum_{\labs{\bomega} > m \omega_0} \labs{\bar{\CF}\left(\bar{f}(\bt) \cdot \e^{(\ri \nu \bt)}\right)\!(\bomega-\nu)}^2\le\frac{ \pi}{Tm\omega_0}.
	\end{align}
	Similar bound holds for the continuous Fourier Transform $\hat{f}(\omega)=\sqrt{\frac{T}{2\pi}}\sinc(T\omega/2)$ of the uniform weight function:
	\begin{align}\label{eq:UnifTailDisc}
		\int_{-\infty}^\infty |\hat{f}(\omega) \indicator(|\omega|\geq W)|^2\mathrm{d}\omega  \leq \frac{\pi}{TW}.
	\end{align}
	so choosing $W:=\bigTh{\left(\frac{(\nrm{\vH}+1)(T+1)}{\epsilon'}\right)^{\!3}}\geq 2 \nrm{\vH}+8\pi\frac{T}{\epsilon'}\left(\frac{256 \nrm{\vH} T }{\epsilon'}+1\right)^{\!\! 2}$ ensures the necessary tail bounds. Therefore it is easy to see that setting $N=\bigTh{\left(|A|\frac{(\nrm{\vH}+1)(T+1)(\beta+1)}{\epsilon}\right)^{\!\!\bigO1}}$ satisfies all requirements of \autoref{prop:TruncatedContToDisc}.
\end{proof}

\begin{restatable}[Discretization error of Gaussian \Lword{}s and discriminant proxies]{cor}{GaussContToDisc}\label{cor:GaussContToDisc}
	Let $f(t)=\frac{1}{ \sqrt{\sigma\sqrt{2 \pi }}}e^{-\frac{t^2}{4 \sigma^2}}$ and $\epsilon>0$. Assuming the normalization condition \eqref{eq:AAdagger}, if $\nrm{\gamma}_\infty\leq1$ and $\sqrt{\gamma(\omega)\gamma(-\omega)}$ are $\beta$-Lipschitz continuous (e.g., Metropolis or Glauber), , then setting $N=2^n$ for $n=\bigTh{\log\left(\frac{2|A|(\nrm{\vH}+1)(\sigma+1/\sigma)(\beta+1)}{\epsilon}\right)}$ with appropriate constants ensures 
	\begin{align}
		\nrm{\CL_{(f,\vH)}-\bar{\CL}_{(f,\vH)}}_{1-1}\leq \epsilon,\qquad
		\nrm{\vec{\CD}_{(f,\vH)}-\bar{\vec{\CD}}_{(f,\vH)}}\leq \epsilon,
	\end{align}
	where $\omega_0=\frac{1}{\sigma}\sqrt{\frac{2\pi}{N}}$, $t_0=\sigma\sqrt{\frac{2\pi}{N}}$. 
\end{restatable}
\begin{proof}
	We use \autoref{prop:TruncatedContToDisc} for each $\vA^{a}$ setting $\epsilon'\leftarrow\frac{\epsilon}{|A|}$ to prove the claim.
	First observe that $\nrm{f}_2=1$ and since $t_0=\bigO{\sigma}$ we have
	\begin{align}\label{eq:DiscGaussNorm}
		\sum_{\bomega\in S^{\lceil N \rfloor}_{t_0}}t_0\left|f(\bt)\right|^2=\bigTh{1}.
	\end{align}
	As $\hat{f}(w)= \sqrt{\sigma \sqrt{\frac{2}{\pi }}} e^{-\sigma^2 \omega^2}$, we also get $|\frac{\omega}{\mathrm{d}\omega}\hat{f}(w)|/\lVert\hat{f}\rVert_\infty\leq \sigma$, so $D\leq \sigma$. Since $\nrm{f}_\infty=\nrm{f_T}_\infty$ and $|\frac{\omega}{\mathrm{d}t}f(t)|/\lVert f\rVert_\infty\leq \frac{1}{2\sigma}$, we get $L\leq \frac{1}{2\sigma}$.
	
	Standard Gaussian tail bound tells us that 
	\begin{align}
		\nrm{f-f_T}_2^2\leq \frac{2}{T \sqrt{2\pi}}e^{-\frac{T^2}{2 \sigma^2}},
	\end{align}
	from which it follows that 
	\begin{align}
		\sum_{\bomega\in S^{\lceil N \rfloor}_{t_0}}t_0\left|f(\bt)-f_T(\bt)\right|^2
		\leq\left(\frac{2}{T \sqrt{2\pi}}+\frac{2t_0}{\sigma \sqrt{2\pi}}\right)e^{-\frac{T^2}{2 \sigma^2}},
	\end{align}
	implying that it suffices to choose $T=\bigTh{\sigma\sqrt{\log(1/\epsilon')}+1}$.

	By \eqref{eq:DiscGaussNorm} and \autoref{prop:Gaussian_tail} the Fourier-transformed tail satisfies
	\begin{align}
		\sqrt{\sum_{\underset{\labs{\bomega}\ge W}{\bomega\in S^{\lceil N \rfloor}_{\omega_0}} } \labs{\hat{\bar{f}}(\bomega)}^2} \le \CO\L(\frac{1}{\sqrt{N\omega_0\sigma}}\e^{ - N^2\omega_0^2\sigma^2/2} + \frac{1}{ \sqrt{Nt_0/\sigma}}\e^{ - N^2t_0^2/16\sigma^2} +\frac{1}{\sqrt{W\sigma}}\e^{ - W^2\sigma^2}\R),
	\end{align}
	so it suffices to choose $N=\bigOm{\left(\sqrt{\log(1/\epsilon')}+1\right)}$ and  $W=\bigTh{\frac{1}{\sigma}\sqrt{\log(1/\epsilon')}+1+2\nrm{\vH}}$.
	
	Therefore, it is easy to see that setting $N=\bigTh{\left(|A|\frac{(\nrm{\vH}+1)(\sigma+1/\sigma)(\beta+1)}{\epsilon}\right)^{\!\!\bigO1}}$ satisfies all requirements of \autoref{prop:TruncatedContToDisc}.
\end{proof}

\section{Implications for \Lword{}s from system-bath interaction}\label{sec:open_system}

Our algorithmic constructions are closely related to their physical origins, and the analytic framework conversely sheds light on the thermalization of open quantum systems. Under physical assumptions, one can microscopically derive a \Lword{} for a system coupled weakly to a bath (See, e.g.,~\cite{Rivas_2012_open_systems}). Among many candidates~\cite{davies74,davies76,Redfield1965, Nathan2020universal, Trushechkin_2021}, we mainly focus on the \textit{Coarsed Grained Master Equation}~\cite{Christian_2013_Coarse_graining,  Mozgunov2020completelypositive} that enjoys transparent nonasymptotic error bounds and nicely connects to our algorithmic construction. Recall
\begin{align}
\CL_{(CGME)}[\vrho] &:= - \ri[\vH_{LS}, \vrho] + \sum_{a \in A}\int_{-\infty}^{\infty}  \gamma(\omega) \left( \hat{\vA}^a(\omega)\vrho\hat{\vA^{a}}(\omega)^\dagg -\frac{1}{2}\{\hat{\vA^{a}}(\omega)^\dagg\hat{\vA}^a(\omega),\vrho \} \right)\mathrm{d}\omega\label{eq:total_Lind}\\
&=:\CL_{uni}[\vrho] + \CL_{diss}[\vrho]
\end{align}
and the correlation function $\gamma(\omega)$ satisfying the symmetry $\gamma(\omega)/\gamma(-\omega)=\e^{-\beta\omega}$. The Lamb-shift term~\footnote{Compared with~\cite[Eq.24]{Mozgunov2020completelypositive}, the factor of $\frac{1}{\sqrt{2\pi}}$ is due to our Fourier Transform convention. }
\begin{align}
\vH_{LS}:=\sum_{a \in A}\frac{\ri}{2\sqrt{2\pi} T} \int_{-T/2}^{T/2}\int_{-T/2}^{T/2}\ \textrm{sgn}(t_1-t_2)c(t_2-t_1)\vA^{a\dagger}(t_2)\vA^a(t_1)\mathrm{d}t_2\mathrm{d}t_1 \label{eq:Lamb}
\end{align}
depends on the inverse Fourier Transform $c(t)$ of $\gamma(\omega)$, i.e.,
\begin{align}
   c(t) = \frac{1}{\sqrt{2\pi}}\int_{-\infty}^{\infty} \gamma(\omega) \e^{\ri \omega t} \mathrm{d} \omega. 
\end{align}
Let us also impose the normalization convention
\begin{align}
    \norm{c}_1 \le \sqrt{2\pi} \label{eq:c_normalized}. 
\end{align}
which also control the frequency domain by $\norm{\gamma}_{\infty} \le \frac{\norm{c}_1}{\sqrt{2\pi}}\le 1$. Nicely, the strength of the Lamb-shift term is also suitably normalized:
\begin{lem}[Norm of Lamb-shift term] For $\vH_{LS}$ as in~\eqref{eq:Lamb},
    \begin{align}
    \norm{\vH_{LS}} \le \frac{\norm{c}_1 }{2\sqrt{2\pi}} \lnorm{ \sum_{a\in A} \vA^{a\dagger}\vA^a}. 
    \end{align}
\end{lem}
For our normalization conventions (Eq.~\eqref{eq:AAdagger}, Eq.~\eqref{eq:c_normalized}), the RHS would be $\frac{1}{2}$, which is comparable with the super-operator norm of the dissipative part~\eqref{eq:L_normalized}.

Note that if $c(t)$ is sufficiently smooth (for example, $l$-Lipschitz continuous with a not too large $l$), and we can efficiently prepare discretized states proportional to $\sqrt{|c(t)|}$ and $\frac{c(t)}{\sqrt{|c(t)|}}$, then we can get a block-encoding of a good approximation of $\frac{\vH_{LS}}{\frac{1}{\sqrt{2\pi}}\int_{-T}^T \labs{c(t)}\mathrm{d}t}$ by using Hamiltonian simulation time $\bigO{T}$. The key is to prepare states proportional to $\sqrt{|c(t_2-t_1)|}$ and $\frac{c(t_2-t_1)}{\sqrt{|c(t_2-t_1)|}}$ over the domain $[-T/2,T/2]\times[-T/2,T/2]$. This can be done by first preparing a uniform superposition over discretized values of $t_1$ on the interval $[-T/2,T/2]$, and also a state proportional to $\sqrt{|c(t_2)|}$ on the interval $[-T,T]$. This is a product state, but then we add the first variable to the second $\tilde{t}_2\leftarrow t_2+t_1$  (which we implement in superposition on the register containing the discretized values of $t_2$). The resulting new variables $t_1,\tilde{t}_2$ restricted to the domain $[-T/2,T/2]\times[-T/2,T/2]$ have the desired amplitudes. The case of $\frac{c(t_2-t_1)}{\sqrt{|c(t_2-t_1)|}}$ is completely analogous. Now the block-encoding is simple: prepare a state proportional $\frac{c(\bt_2-\bt_1)}{\sqrt{|c(\bt_2-\bt_1)|}}\ket{\bt_1}\ket{\bt_2}$, apply $\textrm{sgn}(\bt_2-\bt_1)\vA^a(\bt_2)\vA^a(\bt_1)$ and finally unprepare the state proportional to $\sqrt{|c(\bt_2-\bt_1)|}\ket{\bt_1}\ket{\bt_2}$ (run the preparation in reverse). This block-encoding ensures that we can accurately simulate the above Master Equation using \autoref{thm:LCUSim}.

Compared with our algorithmic construction, the CGME differs in the following ways. First, it contains a unitary part $\CL_{uni}$, especially the Lamb-shift term $\vH_{LS}$. This requires additional technical tools to handle, so we temporarily drop this term and postpone its discussion at~\autoref{sec:effect_of_Lamb}. The second difference is less essential: instead of discrete Fourier Transforms, the Kraus operators are labeled by continuous Bohr frequencies with appropriate normalizations.

\begin{align}
    \hat{\vA}(\omega) := \sqrt{\frac{1}{2\pi T}} \int_{-T/2}^{T/2} \e^{-\ri \omega t} \vA(t)\mathrm{d}t = \sqrt{\frac{1}{2\pi T}}\sum_{\nu\in B(\vH)} \vA_\nu \frac{\e^{-\ri (\omega-\nu) T/2} - \e^{\ri (\omega-\nu) T/2}}{\ri \omega}.
\end{align}
Note the Fourier Transform convention. The time scale $T$ depends on parameters of the open system, such as the bath correlation function and the coupling strength~\cite{Mozgunov2020completelypositive}; in our error bounds, we will keep $T$ as a tunable abstract parameter.

Still, using a similar argument for analyzing the algorithm, we control the fixed point error for the \Lword{}. In fact, the bounds and the presentation simplify as discretization errors vanish. The main idea is to introduce the secular approximation (using notations in~\autoref{sec:secular})
\begin{align}
    \CL_{sec} &:= \sum_{a \in A}\int_{-\infty}^{\infty}  \gamma(\omega) \left( \hat{\vS}^a(\omega)[\cdot]\hat{\vS^{a}}(\omega)^{\dagger} -\frac{1}{2}\{\hat{\vS^{a}}(\omega)^{\dagger}\hat{\vS}^a(\omega),\cdot \} \right)\mathrm{d}\omega.
    \label{eq:Lind_sec_open}
\end{align}

\begin{thm}[Fixed point of the dissipative part] \label{thm:CGME_fixedpoint}
The dissipative part of the CGME \Lword{}~\eqref{eq:total_Lind} (satisfying normalization and symmetry conditions~\eqref{eq:AAdagger},\eqref{eq:gamma_KMS},\eqref{eq:c_normalized}), has an approximate Gibbs fixed point
\begin{align}
\normp{\vrho_{fix}( \CL_{diss} ) - \vrho_{\beta}}{1} &\le \CO\L(  \sqrt{\frac{\beta}{T}}t_{mix}(\CL_{diss})\R).
\end{align}
\end{thm}
Since the proof structure is analogous, we present the altogether bounds and derive them in the following sections.
\begin{proof}
Telescope for the fixed points
\begin{align}
    \normp{\vrho_{fix}( \CL_{diss} ) - \vrho_{\beta}}{1}&\le \normp{\vrho_{fix}( \CL_{diss} ) - \vrho_{fix}( \CL_{sec})}{1} + \normp{\vrho_{fix}( \CL_{sec} ) - \vrho_{\beta}}{1}\\
    &\le \CO\L(\L(\lnormp{\CL-\CL_{sec}}{1-1} + \lnormp{\CD(\vrho, \CL_{sec}) - \CD(\vrho, \CL_{sec})^{\dagger}}{2-2}\R)t_{mix}(\CL_{diss})\R)\\
    &\le \CO\L( \L(\frac{1}{\sqrt{\mu T}} +\beta\sqrt{\frac{\mu}{T}}\R)t_{mix}(\CL_{diss}) \R).
\end{align}
The second inequality uses identical arguments as~\eqref{eq:L_fL_sec} from the proof of Theorem~\ref{thm:L_correctness}. The third inequality plugs in bounds for the secular approximation (the continuous case is the limit of the discrete case (\autoref{lem:secular})) and approximate detailed balance (\autoref{prop:ADB_open_sys}). Optimize the free parameter $\mu = \frac{1}{\beta}$ to conclude the proof.
\end{proof}

The mixing time can be bounded by the gap (\autoref{prop:gap_to_mixing}) if needed, but for conceptual simplicity, we stuck to the mixing time for the main presentation.
\subsection{Bounds for approximate detailed balance}\label{sec:bounds_approx_DB}
Here, we show that approximate detailed balance for the \Lword{} of interest $\CL_{sec}$.
\begin{prop}[Approximate detailed balance]\label{prop:ADB_open_sys}
Suppose the secular approximation for $\CL_{(CGME)}$~\eqref{eq:Lind_sec_open} is truncated at energy $\mu$. Then, 
\begin{align}\label{lem:diss_Approx_DB}
\frac{1}{2}\lnormp{\CD(\vrho, \CL_{sec}) - \CD(\vrho, \CL_{sec})^{\dagger}}{2-2} \le \CO\L(\beta\sqrt{\frac{\mu}{T}} \R).
\end{align}
\end{prop}
\begin{proof}
We simply telescope by inserting the algorithmically constructed discriminant $\CD_{sec}$ (\autoref{cor:improved_boltzmann})
 \begin{align}
\lnorm{\CD(\vrho, \CL_{sec}) - \CD(\vrho, \CL_{sec})^{\dagger}}_{2-2}
    &\le \norm{\CD(\vrho, \CL_{sec}) - \CD_{sec}}_{2-2} +\norm{ \CD_{sec} - \CD(\vrho, \CL_{sec})^{\dagger}}_{2-2}\\
     &\le 2 \norm{ \CD_{sec} - \CD(\vrho, \CL_{sec})}_{2-2}.
 \end{align}
The last inequality uses that $\CD_{sec}=\CD_{sec}^\dagg$ and that $\norm{\CA}_{2-2} = \norm{\CA^{\dagger}}_{2-2}$. Note that the algorithmic discussion considered discrete energy labels, so we have to take a continuum limit for the bilinear sum\footnote{Formally speaking, the correctness in this limit can be derived by using our discrete results and taking their limit as in \ref{thm:discToCont}.}
\begin{align}
    \sum_{\bomega\in S_{\omega_0}}\!\!\!\gamma(\bomega) \hat{\vS}^a(\bomega)^{\dagger} \hat{\vS}^a(\bomega) \rightarrow \int_{-\infty}^\infty  \gamma(\omega) \hat{\vS}^{a\dagger}(\omega) \hat{\vS}^{a}(\omega)\mathrm{d}\omega.
\end{align} 
Indeed, the bound does not depend on the discretization scale.
\end{proof}

\subsection{Effects of the Lamb-shift term}
\label{sec:effect_of_Lamb}
In this section, we include the unitary part of the CGME generator. The resulting bounds now depend on two mixing times, and we do not have a desirable conversion between the two mixing times. Still, one can upper bound both via the spectral gap of the Hermitian part of the dissipative part $\CH_{diss}$.
\begin{restatable}[Fixed point of CGME]{thm}{CGMEfull} \label{thm:CGME_total_fixedpoint}
For the full CGME generator~\eqref{eq:total_Lind}, which satisfy the symmetry and normalization conditions~\eqref{eq:AAdagger},\eqref{eq:fnormalized},\eqref{eq:gamma_KMS},\eqref{eq:c_normalized}, we have that
\begin{align}
\normp{\vrho_{fix}( \CL_{CGME} ) - \vrho_{\beta}}{1} &\le \CO\L( \sqrt{\frac{\beta}{T}}\L(t_{mix}(\CL_{CGME})+ t_{mix}(\CL_{diss})\R)\R).
\end{align}
\end{restatable}
We have already calculated the errors for the dissipative part; here we only study errors for the Lamb-shift term as in the following sections and combine them at \autoref{sec:full_CGME_altogether}. The strategy is similar; we discretize the Hamiltonian, truncate the operator via the secular approximation, and then argue that the resulting operator nearly commutes with the Gibbs state
\begin{align}
    \vH_{LS}\approx \vH_{LS,sec} \approx  \sqrt{\vrho} \vH_{LS,sec} \sqrt{\vrho}^{-1}.
\end{align}
Formally, we introduce the intermediate constructs
\begin{align}
\CL_{sec} &:=  \CL_{uni,sec} + \CL_{diss,sec} \quad \text{and}\quad \CL' :=  \sqrt{\vrho} (\CL_{diss,sec}-\CL_{uni,sec}) [\sqrt{\vrho}^{-1}\cdot \sqrt{\vrho}^{-1}]\sqrt{\vrho}\\
\quad \text{where}\quad \CL_{uni,sec} &:= -\ri[\vH+\vH_{LS,sec},\cdot].
\end{align}
Note that in $\CL'$ we had to manually flip the sign of the coherent part to ensure $\CL_{sec}\approx \CL'$. One may add any coherent term $-\ri [\vV,\cdot ]$ as long as $[\vV,\vrho]=0$ (most notably the Hamiltonian $\vH$); it would not contribute to the error bounds.

\subsubsection{Secular approximation for the Lamb-shift}
 Rewrite the integral by change-of-variable $s:=t_2-t_1$ and apply the secular approximation to the inner integral
\begin{align}
    \vH_{LS}&:= \frac{\ri}{2T}\int_{-T}^{T}\textrm{sgn}(-s)C(s) \left(\int_{\max(-T/2,-T/2-s)}^{\min(T/2-s,T/2)}\L(\sum_{a\in A} \vA^{a\dagger}(s)\vA^a\R)(t_1)\mathrm{d}t_1\right)\mathrm{d}s
\end{align}

The secular approximation in this context differs from the one we used for the bilinear expressions (\autoref{lem:secular}). We define and analyze the error as follows, inspired by~\cite{chen2023QThermalStatePrep}.
\begin{restatable}[Secular approximation for time average]{lem}{IntSecular}\label{lem:integral_secular}
Consider an operator $\vA$ and a Hermitian operator $\vH$. Then, for any unitarily invariant norm $\normp{\cdot}{*}$ and times $t_1,t_2$, there exists a secular approximated operator $\vS_{\mu}$ such that
\begin{align}
    \bra{\psi_j}\vS_{\mu}\ket{\psi_i} = 0 \quad \text{if}\quad \labs{ E_i-E_j } \le \mu
\end{align}
and 
\begin{align}
     \lnormp{ \int_{t_1}^{t_2} \e^{\ri \vH s}\vA \e^{-\ri \vH s}\mathrm{d}s - \vS_{\mu} }{*} \le \CO\L(  \normp{\vA}{*} \frac{1+ \log(\frac{1}{\mu \labs{t_2-t_1}})}{\mu} \R).
\end{align}
\end{restatable}
See~\autoref{sec:proof_sec_int} for the proof. Intuitively, the time average $ \int_{t_1}^{t_2} \vA(s)\mathrm{d}s$ weakens the off-diagonal entries (in the $\vH$ eigenbasis) with a large Bohr frequency. Dropping them incurs an error depending on the truncation value $\mu$. 

Applying the secular approximation for the Lamb-shift Hamiltonian yields the following bound.
\begin{cor}[Secular approximation for the Lamb-shift term]\label{cor:Lamb_sec_error} In the setting of~\autoref{thm:CGME_total_fixedpoint}, the exists a Hermitian operator $\vH_{LS,sec}$ such that
\begin{align}
    \bra{\psi_j}\vH_{LS,sec}\ket{\psi_i} = 0 \quad \text{if}\quad \labs{ E_i-E_j } \le \mu
\end{align}
and 
\begin{align}
     \norm{ \vH_{LS} - \vH_{LS,sec} } = \tCO\L( \frac{\norm{\sum_{a\in A} \vA^{a\dagger}\vA^a}}{\mu T}\int_{-T}^{T}\labs{C(s)}\mathrm{d}s \R) = \tCO\L( \frac{1}{\mu T}\R).
\end{align}
\end{cor}
\begin{proof}
Apply secular approximation to the inner integral (which depends on $s$) to obtain 
\begin{align}
    \vH_{LS,sec}&:= \frac{\ri}{2T}\int_{-T}^{T} \textrm{sgn}(-s)C(s) \vS^{(s)}_{\mu} \mathrm{d}s 
\end{align}
and calculate
\begin{align}
    \norm{ \vH_{LS} - \vH_{LS,sec} }
    &\le \frac{1}{2T}\int_{-T}^{T}\labs{C(s)} \lnorm{\int_{\max(-T/2,-T/2-s)}^{\min(T/2-s,T/2)}(\sum_{a\in A}\vA^{a\dagger}(s)\vA^a)(t_1)\mathrm{d}t_1 - \vS^{(s)}_{\mu} }\mathrm{d}s.
\end{align}
Use the secular approximation (\autoref{lem:integral_secular}) for the integral over $t_1$ to conclude the proof.
\end{proof}

\subsubsection{Approximate detailed balance}
Thirdly, we also control the error for approximate detailed balance.

\begin{lem}[Apprximate detailed balance for the unitary part]\label{lem:sec_open_error_from_Boltzmann} In the setting of~\autoref{thm:CGME_total_fixedpoint}, if $\beta \mu \le 1$,
\begin{align}
\lnormp{\CL^{\dagger}_{uni} + \sqrt{\vrho}\CL_{uni}[\sqrt{\vrho}^{-1}\cdot\sqrt{\vrho}^{-1} ]\sqrt{\vrho}}{2-2} \le \tCO\L(\beta \mu \int_{-T}^{T}\labs{C(s)} \,\mathrm{d}s\R) = \tCO\L(\beta \mu \R).
\end{align}
\end{lem}
We present the superoperator form to feed into our existing fixed-point analysis. Still, the calculation essentially reduces to the operator norm.
We will need the following proposition, whose proof is reminiscent of the arguments (\autoref{sec:weaker}) analyzing 
the $\delta \CR$ part of \autoref{lem:error_from_Boltzmann}.
\begin{prop}\label{prop:single_operator_error_boltzmann} Suppose an operator $\vA$ satisfies
\begin{align}
    \bra{E_i} \vA \ket{ E_j} = 0 \quad \text{whenever}  \quad \labs{E_i-E_j} \le \mu
\end{align}
and $\beta \mu \le 1$, then
    \begin{align}
        \norm{\vA - \sqrt{\vrho} \vA \sqrt{\vrho}^{-1}} =\CO( \norm{\vA}\beta \mu).
    \end{align}
\end{prop}
\begin{proof} Consider nearby energy projectors at energy resolution $\mu$.
\begin{align}
\vI = \sum_{a\in \BZ} \vP_{a\mu} \quad \text{where}\quad \vP_{a\mu} := \sum_{ (a+\frac{1}{2})\mu > E\ge (a-\frac{1}{2} )\mu } \vP_{E}.    
\end{align}
Then, the matrix $\vA$ is tri-block-diagonal $\vA = \vU +\vL+\vD$ with blocks labeled by integer multiples of $\mu$. For the lower-diagonal-blocks $\vL = \sum_a\vP_{(a+1)\mu} \vA\vP_{a\mu}$, we evaluate the commutator for each term
    \begin{align}
        \vP_{(a+1)\mu} \vA\vP_{a\mu} - \sqrt{\vrho}\vP_{(a+1)\mu} \vA\vP_{a\mu}\sqrt{\vrho}^{-1} = \vP_{(a+1)\mu} \vA\vP_{a\mu} - \e^{-\beta \vH'/2}\vP_{(a+1)\mu} \vA\vP_{a\mu}\e^{\beta \vH'/2} 
    \end{align}
    where 
    \begin{align}
        \vH' = \vH - (a+\frac{1}{2})\vI \quad \text{such that}\quad \norm{\vH'} \le \mu.
    \end{align}
    Therefore, 
    \begin{align}
        \norm{ \vL -  \sqrt{\vrho}\vL\sqrt{\vrho}^{-1}} &\le \max_a \norm{\vP_{(a+1)\mu} \vA\vP_{a\mu} - \sqrt{\vrho}\vP_{(a+1)\mu} \vA\vP_{a\mu}\sqrt{\vrho}^{-1}} \tag*{(By~\autoref{lem:bandNorm})}\\
        &= \CO(\beta \mu) \tag*{(By $\labs{\e^x - 1} \le 2\labs{x}$ for $\labs{x} \le 1$)}.
    \end{align}
    The bounds on $\vD$ and $\vU$ are analogous.
\end{proof}
\begin{proof}[Proof of \autoref{lem:sec_open_error_from_Boltzmann}]
Expand
\begin{align}
    \CL^{\dagger}_{uni}[\vA]+ \sqrt{\vrho}\CL_{uni}[\sqrt{\vrho}^{-1}\vA\sqrt{\vrho}^{-1} ]\sqrt{\vrho} &= \ri \vH_{LS,sec}\vA - \ri \vA\vH_{LS,sec} - \ri \sqrt{\vrho} \vH_{LS,sec} \sqrt{\vrho}^{-1}\vA + \ri \vA\sqrt{\vrho}^{-1} \vH_{LS,sec} \sqrt{\vrho}\\
    & = \ri \L(\vH_{LS,sec} - \sqrt{\vrho} \vH_{LS,sec} \sqrt{\vrho}^{-1}\R)\vA - \ri \vA \L(\vH_{LS,sec}- \sqrt{\vrho}^{-1} \vH_{LS,sec} \sqrt{\vrho}\R).
\end{align}
The Hamiltonian term $H$ disappears because it commutes with the Gibbs state. Now, Holder's inequality reduces the superoperator norm to the operator norm, which can be controlled by~\autoref{prop:single_operator_error_boltzmann}
\begin{align}
    \lnorm{\vH_{LS,sec} - \sqrt{\vrho} \vH_{LS,sec} \sqrt{\vrho}^{-1}}&\le  \CO(\norm{\vH_{LS,sec}} \beta \mu) = \CO\L( (\norm{\vH_{LS}}+\norm{\vH_{LS}-\vH_{LS,sec}} )\beta \mu\R)
\end{align}
Use~\autoref{cor:Lamb_sec_error} and~\autoref{lem:integral_secular} to conclude the proof.
\end{proof}

\subsection{Altogether: Proof of fixed point correctness ({\autoref{thm:CGME_total_fixedpoint})}}\label{sec:full_CGME_altogether}
We now put together the error bounds for the full CGME \Lword{}.
\begin{proof}[Proof of \autoref{thm:CGME_total_fixedpoint}]
Recall the bound on the fixed point error
\begin{align}
    &\normp{\vrho_{fix}(\CL) - \vrho}{1} \\
    &= \normp{\vrho_{fix}(\CL) - \vrho_{fix}(\CL_{sec})}{1}+ \normp{\vrho_{fix}(\CL_{sec}) - \vrho_{fix}(\CL')}{1} \\
    &\le  \CO\L( \normp{\CL - \CL_{sec}}{1-1}t_{mix}(\CL)+ \frac{\normp{\CD(\vrho, \CL_{sec}) -\CD(\vrho, \CL')}{2-2}}{\varsigma_{-2}( \CD(\vrho, \CL_{sec}))}\R)\\
    & \le \CO\bigg( \normp{\CL - \CL_{sec}}{1-1}t_{mix}(\CL) +\lnormp{\CL_{diss}  -\CL_{diss,sec}}{1-1}t_{mix}(\CL_{diss})\\    &+\L(\normp{\CD(\vrho,\CL_{uni,sec})+\CD(\vrho,\CL_{uni,sec})^{\dagger}}{2-2}+\normp{\CD(\vrho,\CL_{diss,sec}) -\CD(\vrho,\CL_{diss,sec})^{\dagger}}{2-2}\R) t_{mix}(\CL_{diss}) \bigg)\label{eq:tLtL_diss}
\end{align}
where we compare the fixed points of $\CL_{sec}$ and $\CL'$ by eigenvector perturbation (\autoref{prop:eigenvector_perturb_gen}, noting that $\CL_{sec}$ contains an eigenvalue zero as it generates a CPTP map; $\CL'$ has the Gibbs state as its fixed point, which has eigenvalue zero). The third inequality bound the singular value by the mixing time: apply Fan-Hoffman~\cite[Proposition III.5.1]{bhatia1997MatrixAnalysis} and use perturbation bounds for sorted singular values
\begin{align}
    2\varsigma_{-2}( \CD(\vrho, \CL_{sec})) &\ge \varsigma_{-2} \L(\CD(\vrho, \CL_{sec})+
    \CD(\vrho, \CL_{sec})^{\dagger} \R)\\
    & \ge \lambda_2 \L(\CD(\vrho, \CL_{diss,sec})+
    \CD(\vrho, \CL_{diss,sec})^{\dagger}\R) - \lnormp{\CD(\vrho, \CL_{uni,sec})+
    \CD(\vrho, \CL_{uni,sec})^{\dagger}}{2-2}\\
    &\ge \Omega(\frac{1}{t_{mix}(\CL_{diss,sec})}) - \lnormp{\CD(\vrho, \CL_{uni,sec})+
    \CD(\vrho, \CL_{uni,sec})^{\dagger}}{2-2}\\
    &\ge \Omega(\frac{1}{t_{mix}(\CL_{diss})}) - \lnormp{\CD(\vrho, \CL_{uni,sec})+
    \CD(\vrho, \CL_{uni,sec})^{\dagger}}{2-2}\label{eq:diss_sec_to_diss}.
    \end{align}
The rest are analogous to~\eqref{eq:L_fL_sec}, except that we have to manually include the term $\lnormp{\CL_{diss}-\CL_{diss,sec}}{1-1}t_{mix}(\CL_{diss})$ to ensure Eqn.~\eqref{eq:diss_sec_to_diss} holds.

By secular approximation with truncation energy $\mu$ and $\mu'$ (which will be set to different values to minimize the error bounds), 
\begin{align}
    \norm{\vH_{LS}-\vH_{LS,sec}} &= \tCO(\frac{1}{\mu T})\quad \text{and}\quad \normp{\CL_{diss} - \CL_{diss,sec}}{1-1} = \CO(\frac{1}{\sqrt{\mu' T}})
\end{align}
which combines to
\begin{align}
\normp{\CL - \CL_{sec}}{1-1} &= \lnormp{\ri[\vH_{LS},\cdot] + \CL_{diss} - \ri[\vH_{LS,sec},\cdot] -\CL_{diss,sec}}{1-1}\\
&= \tCO\L(\frac{1}{T\mu} + \frac{1}{\sqrt{\mu' T}}  \R)
\end{align}
The second inequality reduces the superoperator norm $\normp{\cdot}{1-1}$ to operator norm by $\normp{\vA \vrho}{1} \le \norm{\vA}\normp{\vrho}{1}$. Next, we combine the approximate detailed balance-type errors from the Lamb-shift term (\autoref{lem:sec_open_error_from_Boltzmann}) and the dissipative term (\autoref{lem:diss_Approx_DB})
\begin{align}
&\normp{\CD(\vrho,\CL_{uni,sec})+\CD(\vrho,\CL_{uni,sec})^{\dagger}}{2-2}+\normp{\CD(\vrho,\CL_{diss,sec}) -\CD(\vrho,\CL_{diss,sec})^{\dagger}}{2-2} \\
    &= \tCO\L(\beta\mu + \beta\sqrt{\frac{\mu'}{T}}\R).
\end{align}
Altogether, choose $\mu = \sqrt{\frac{1}{\beta T}} \le \frac{1}{\beta}$ and $\mu' = \frac{1}{\beta}$ so that
\begin{align}
    \eqref{eq:tLtL_diss} 
    &\le \tCO\L( \max\L( t_{mix}(\CL), t_{mix}(\CL_{diss})\R)\cdot \sqrt{ \frac{\beta}{T} }\R),
\end{align}
which concludes the proof.

\end{proof}

\subsection{Proof for secular approximation for time average (\autoref{lem:integral_secular})}
\label{sec:proof_sec_int}
Intuitively, we want to truncate the Bohr frequency far from zero. Unfortunately, the sharp truncation from~\autoref{sec:secular} does not seem to work here because the truncation error is related to the 1-norm $\norm{f}_1$ (instead of 2-norm $\norm{f}_2$). The 1-norm is more delicate to handle, forcing us to \textit{smoothly} truncate the tail and explicitly evaluate the Fourier Transform in the time domain. Pictorially, the time domain function becomes a smeared version of the sharp window function $\indicator(\labs{t}\le T)$ where the discontinuity is smoothed out due to convolution with a smooth bump function.

\begin{proof}[Proof of~\autoref{lem:integral_secular}]
Without loss of generality, we can conjugate with time evolution to shift the integral so that $t_2 = -t_1 = T/2$. Let $f(t) = \indicator(\labs{t} \le T/2)$, then 
    \begin{align}
        \frac{1}{\sqrt{2\pi}} \int_{-T/2}^{T/2} \vA(s)\mathrm{d}s= \frac{1}{\sqrt{2\pi}} \int_{-\infty}^{\infty}f(s)\vA(s)\mathrm{d}s&= \sum_{\nu} \vA_{\nu} \hat{f}(-\nu)
    \end{align}
    with $\hat{f}(\omega) = \frac{\e^{-\ri \omega T/2}-\e^{\ri \omega T/2}}{\sqrt{2\pi}\omega }$. Let us truncate the frequency domain function  
\begin{align}
            \frac{1}{\sqrt{2\pi}} \int_{-\infty}^{\infty} ( b*f )(s)\vA(s)\mathrm{d}s &= \sum_{\nu} \vA_{\nu} \hat{b}\cdot \hat{f}(-\nu) =: \vS_{\mu}
\end{align}
by multiplying with a carefully chosen smooth bump function
\begin{align}
    \hat{b}(x) := \begin{cases}
        0 & \text{if}\quad \labs{x} \ge \mu\\
        1 & \text{if}\quad x = 0\\
        \le 1 &\text{else}
    \end{cases}.
\end{align}
Then, for any unitarily invariant norm,
\begin{align}
    \lnormp{ \int_{-T/2}^{T/2} \vA(s)\mathrm{d}s - \vS_{\mu} }{*} &\le \frac{1}{\sqrt{2\pi}}  \norm{f- b*f}_1\\
    &= \frac{1}{\sqrt{2\pi}} \L( \int_{W^c_{\epsilon}}  (f-b*f)(t) \rd t + \int_{W_{\epsilon}}  (f-b*f)(t) \rd t\R)\label{eq:intW}\\
    &\le \CO\L( \frac{1}{\mu}+ \frac{1+\log(\frac{1}{\mu T})}{\mu} \R).
\end{align}
The second equality separately evaluates the integral around $\epsilon$-balls near $\pm T$
\begin{align}
    W_{\epsilon} := [\frac{T}{2}-\epsilon,\frac{T}{2}+\epsilon]\cup [-\frac{T}{2}-\epsilon,-\frac{T}{2}+\epsilon] \quad \text{for}\quad \epsilon = \frac{1}{\mu}.
\end{align}
For each $t\in W^c_{\epsilon}$, the convolution is point-wise close to the original value
\begin{align}
    \int_{-\infty}^{\infty} f(t-s)b(s) \rd s- f(t)& = f(t) \L(\int_{-\infty}^{\infty} b(s)\rd s - 1\R) + r(\mu \labs{t-T}) 
\end{align}
up to an error $r(x)$ falling super-polynomially with $\labs{x}$. Thus, the integral over $W^c_{\epsilon}$~\eqref{eq:intW} is then bounded by 
\begin{align}
    2\int_{\labs{t-T} \ge \mu} \labs{ r(\mu \labs{t-T})}\rd t = \frac{4}{\mu} \int_{1}^{\infty} \labs{r(x)} \rd x = \CO(\frac{1}{\mu}).
\end{align}
The main error arises from the sharp edge at $\pm T$; we invoke general norm bounds $\norm{f}_{\infty} = 1$ and 
\begin{align}
    \norm{ b*f}_{\infty} &\le \frac{1}{\sqrt{2\pi}} \norm{\hat{b}\cdot \hat{f}}_1  \\
    &\le \frac{1}{\sqrt{2\pi}} \L(\int_{-1/T}^{1/T} \labs{ \hat{f}(\omega)} \rd \omega + (\int_{1/T}^{\mu}+\int^{-1/T}_{-\mu}) \labs{\hat{f}(\omega)} \rd \omega\R)\tag*{(By $\norm{\hat{b}}\le 1$ )}\\
    &= \CO\L(1 + \log(\frac{1}{\mu T})\R) \quad \tag*{(By $\labs{\hat{f}(\omega)} \le \min(\frac{T}{\sqrt{2\pi}}, \frac{1}{\sqrt{2\pi}\omega})$ )}
\end{align}
and integrate over $W_{\epsilon}$ to obtain the bound.
\end{proof}

\section{Spectral bounds and mixing times}
\label{sec:spectral_bounds_mixing_times}
In this section, we present missing proofs for lemmas and propositions. While some arguments are standard and included merely for completeness, controlling the spectrum of nearly Hermitian matrices requires a substantial linear algebraic argument. We begin with eigenvalue and eigenvector perturbation theory (\autoref{sec:eigen_perturbation}), which is crucial for establishing mixing time bounds (\autoref{sec:gap_and_convergence}) and the correctness of fixed points (\autoref{sec:proof_for_approx_DB}).

\subsection{Perturbation bounds for eigenvalues and eigenvectors}\label{sec:eigen_perturbation}
In this section, we present some useful bounds for eigenvalue and eigenvector perturbation.

\begin{prop}[{Bauer-Fike Theorem with multiplicity, cf.\ \cite[Theorem VI.3.3 \& Problem VI.8.6]{bhatia1997MatrixAnalysis}}]\label{prop:Bauer-Fike}
    Perturb a normal matrix $\vN$ by an arbitrary matrix $\vA$. Then, the spectrum of $\vN$ and $\vN+\vA$ are $\nrm{\vA}$-close to each other:
\begin{align}
		\text{Spec}(\vN+\vA) &\subset \cup_{s\in \text{Spec}(\vN)} D(s,\norm{\vA}), \quad \kern4.5mm\text{ and }\\
		\text{Spec}(\vN) &\subset \cup_{s\in \text{Spec}(\vN+\vA)} D(s,\norm{\vA}), \quad \text{ where } \quad D(s,\epsilon)=\{z\in\mathbb{C}\colon |z-s|\leq \epsilon\}.
\end{align}
Moreover, the connected components of $\cup_{s\in \text{Spec}(\vN)} D(s,\norm{\vA})$ contain an equal number of eigeinvalues of $\vN$ and $\vN+\vA$ when counted with algebraic multiplicity.
\end{prop}
\begin{proof}
	The first half of the statement is the Bauer-Fike Theorem \cite[Theorem VI.3.3]{bhatia1997MatrixAnalysis}. To study the number of eigenvalues per connected component, we consider an interpolation path
\begin{align}
    \vX(t) = \vN + t \vA \quad \text{for}\quad 0\le t \le 1.
\end{align}
    The Bauer-Fike Theorem applied to $\vN + t \vA$ implies for every $0\le t \le 1$ that 
\begin{align}
    \text{Spec}(\vX(t)) \quad \subset \quad\cup_{s\in \text{Spec}(\vN)} D(s,t\norm{\vA}) \quad\subset\quad \cup_{s\in \text{Spec}(\vN)} D(s,\norm{\vA}).
\end{align}
Then, by continuity of eigenvalues along the path, cf.\ \cite[Corollary VI.1.6]{bhatia1997MatrixAnalysis}, no eigenvalues enter or exit the connected components of $\cup_{s\in \text{Spec}(\vN)} D(s,\norm{\vA})$, therefore their number (counted with algebraic multiplicity) is the same for $\vX(0)=\vN$ and $\vX(1)=\vN + \vA$.
\end{proof}

\begin{cor}[Eigenvalue perturbation for discriminants]\label{cor:eigenvalue_perturb}   
	If $\vec{\CD}$ is Hermitian and has norm bounded by $\nrm{\vec{\CD}}\leq 1$, and $\vec{\CD'}$ has a right eigenvector $\ket{\psi}$ with eigenvalue $1$, then the top eigenvalue of $\vec{\CD}$ satisfies
	\begin{align}
		\labs{\lambda_1(\vec{\CD}) - 1} \le \nrm{\vec{\CD}- \vec{\CD}' }.
	\end{align}
\end{cor}

\subsubsection{Eigenvector perturbation bounds}

Intuitively, perturbing a matrix yields small changes in eigenvectors with well-isolated eigenvalues -- we prove this below rigorously under suitable but quite general conditions using a simple linear algebraic argument.

\begin{prop}[Eigenvector perturbation]\label{prop:eigenvector_perturb_gen}
Perturb a matrix $\vM$ by another matrix $\vA$. 
Let $\ket{v}$ be a normalized right eigenvector $\vM\ket{v} = \lambda \ket{v}$, and $\lambda'$ an eigenvalue of $\vM+\vA$. Then the corresponding right eigenvector $(\vM+\vA)\ket{v'} = \lambda' \ket{v'}$ can be normalized such that
\begin{align}
    \braket{v|v} =\braket{v'|v'}=1 \quad \text{and}\quad \normp{\ket{v'} - \ket{v} }{}&\le \frac{2\sqrt{2}\left(\norm{\vA}+|\lambda'-\lambda|\right)}{\varsigma_{-2}\L(\vM - \lambda\vI\R)},
\end{align}
where $\varsigma_{-2}\L(\cdot\R)$ denotes the second-smallest singular value (with multiplicity). Due to Fan-Hoffman~\cite[Proposition III.5.1]{bhatia1997MatrixAnalysis} the singular value $\varsigma_{-2}\L(\vM - \lambda\vI\R) \ge -\lambda_2 \L(\frac{\vM+\vM^{\dagger}}{2}-\mathrm{Re}(\lambda)\vI\R)$ can be bounded in terms of the Hermitian part.
\end{prop}
\begin{proof}
We can assume without loss of generality that $\ket{v'} \propto \ket{v} + \epsilon \ket{v^{\perp}}$ for $\epsilon \in [0,\infty)$ and $\braket{v^{\perp}|v^{\perp}}=1$, yielding
\begin{align}
	(\vM+\vA )(\ket{v} + \epsilon \ket{v^{\perp}}) &\underset{\displaystyle\Updownarrow}{=}\lambda'\vI(\ket{v} + \epsilon \ket{v^{\perp}})\\ 
	\vA (\ket{v} + \epsilon \ket{v^{\perp}}) +(\vM-\lambda'\vI)\ket{v} &= \epsilon (\lambda'\vI-\vM)\ket{v^{\perp}}.
\end{align}
Taking the norms above on both sides and defining $\kappa:=|\lambda'-\lambda|$ we get
\begin{align}
	\epsilon = \frac{\norm{\vA (\ket{v} + \epsilon \ket{v^{\perp}}) +(\lambda-\lambda')\ket{v}} }{\norm{(\vM-\lambda'\vI)\ket{v^{\perp}}}}
	&\underset{\displaystyle\Downarrow}{\le} \frac{(1+\epsilon)\norm{\vA}+\kappa}{\varsigma_{-2}\L(\vM - \lambda\vI\R)-\kappa}\\
	\epsilon &\le \frac{\norm{\vA}+\kappa}{\varsigma_{-2}\L(\vM - \lambda\vI\R)-\nrm{\vA}-\kappa} = \frac{\nrm{\vA}+\kappa}{\varsigma_{-2}\L(\vM - \lambda\vI\R)}\cdot \frac{1}{1 - \frac{\nrm{\vA}+\kappa}{\varsigma_{-2}\L(\vM - \lambda\vI\R)}}.\label{eq:epsBoundFancy2}
\end{align}
The last inequality is a rearrangement. The first inequality uses the triangle inequality for the numerator, and for the denominator that $\norm{(\vM-\lambda'\vI)\ket{v^{\perp}}}\geq \nrm{(\vM-\lambda\vI)\ket{v^{\perp}}}-\nrm{(\lambda'-\lambda)\vI\ket{v^{\perp}}}=\nrm{(\vM-\lambda\vI)\ket{v^{\perp}}}-\kappa$ and 
\begin{align}
	\norm{(\vM - \lambda\vI)\ket{v^{\perp}}}_{} \ge \varsigma_{-2}\L(\vM - \lambda\vI\R). 
\end{align}
 We conclude by setting the appropriate normalization $\ket{v'} = \frac{1}{\sqrt{1+\epsilon^2}} (\ket{v} + \epsilon \ket{v^{\perp}})$ and utilizing the above bound~\eqref{eq:epsBoundFancy2}:
\begin{align}
	\nrm{\ket{v} - \ket{v'}} &= \frac{1}{\sqrt{1+\epsilon^2}}\nrm{\sqrt{1+\epsilon^2}\ket{v}- (\ket{v}+\epsilon\ket{v^{\perp}}) }\\
	&= \frac{\sqrt{(\sqrt{1+\epsilon^2}-1)^2+\epsilon^2}}{\sqrt{1+\epsilon^2}} \le \sqrt{2}\epsilon
	\le \min\L(\sqrt{2},\frac{2\sqrt{2}(\norm{\vA}+\kappa)}{\varsigma_{-2}\L(\vM - \lambda'\vI\R)}\R).
\end{align}
The first inequality uses $\sqrt{1+\epsilon^2}-1 \le \epsilon$. The last inequality uses that the bound is vacuous at $\nrm{\ket{v} - \ket{v'}}\le \sqrt{2}$ and combines \eqref{eq:epsBoundFancy2} with the elementary estimate $\forall x\in[0,1/2]\colon \sqrt{2}\frac{x}{1-x} \le 2\sqrt{2}x$ after substituting $x\leftarrow \frac{\norm{\vA}+\kappa}{\varsigma_{-2}\L(\vM - \lambda\vI\R)}$.  
\end{proof}
When we apply the above perturbation bound for a Hermitian $\vM$, we can set $\lambda$ to be the top eigenvalue such that $\varsigma_{-2}\L(\vM - \lambda\vI\R)$ will be the gap $\lambda_{gap}(\vM) = \lambda_1(\vM) - \lambda_2(\vM)$. In that case, \autoref{prop:Bauer-Fike} guarantees the existence of a nearby eigenvalue $\lambda'$ of $\vM+\vA$ such that $|\lambda'-\lambda| \leq \nrm{\vA}$. When we handle the Lamb-shift term, we need to consider a nonHermitian $\vM$. There, bounding eigenvalue perturbation is not generally obvious, so we will simply assume that there is a nearby eigenvalue $\lambda'$ of $\vM+\vA$.

\subsection{Approximate detailed balance implies approximately correct fixed point}
\label{sec:proof_for_approx_DB}
When detailed balance holds approximately for $\vrho$,
we still expect the fixed point to be \textit{approximately} $\vrho$; we provide a proof of this in this section, which relies on the matrix perturbation results (\autoref{sec:eigen_perturbation}). Recall that in \autoref{sec:L_ADB}, we defined the Hermitian and anti-Hermitian parts (under similarity transformation) as follows:
\begin{align} 
    \CD(\vrho, \CL) = \vrho^{-1/4}\CL[\vrho^{1/4}\cdot\vrho^{1/4}]\vrho^{-1/4} &= \CH + \CA, \\
    \CD(\vrho, \CL)^{\dagger} = \vrho^{1/4}\CL^{\dagger}[\vrho^{-1/4}\cdot \vrho^{-1/4}]\vrho^{1/4} & = \CH - \CA.
\end{align}
Observe that $ \CD(\vrho, \CL)^{\dagger}[\sqrt{\vrho}] = 0$, but it needs not be the case for $\CD(\vrho, \CL)$, which we care about.

We recall some facts: every \Lword{} satisfies that $\text{Spec}(\CL)\subseteq \{z\in\mathbb{C}\colon\mathrm{Re}(z)\leq0\}$~\cite[Proposition 6.16]{wolf2012quantum}. As $\CD(\vrho, \CL)^{\dagger}[\sqrt{\vrho}] = 0$ we also have that $0$ is an element of the spectrums $\text{Spec}(\CD(\vrho, \CL)^{\dagger})$,  $\text{Spec}(\CD(\vrho, \CL))$. Since $\CD(\vrho, \CL)$ is defined by a similarity transformation we have that $\text{Spec}(\CL)=\text{Spec}(\CD(\vrho, \CL))$ and due to \autoref{prop:Bauer-Fike} this implies
\begin{align}\label{eq:topEigenvalueBound}
 	|\lambda_1(\CH)|\leq \norm{\CA}_{2-2}.
\end{align}
\begin{restatable}[Fixed point accuracy]{prop}{Lfixedpointerror}\label{prop:fixed_point_error}
	Suppose a \Lword{} $\CL$ satisfies the $\epsilon$-approximate $\vrho$-detailed balance condition. If $\lambda_{gap}(\CH)>2\epsilon$, then there is a unique state as its fixed point $\vrho_{fix}(\CL)\succ 0$ and its deviation from $\vrho$ is bounded by
	\begin{align}
		\lnormp{\vrho_{fix}(\CL) - \vrho}{1} \le \frac{14\epsilon}{\lambda_{gap}(\CH)}.
	\end{align}
\end{restatable}
The RHS indicates that the fixed point accuracy may deteriorate if the map has a large anti-Hermitian component or if the gap closes.
\begin{proof}  
Every CPTP map has at least a stationary state~\cite[Theorem 6.11]{wolf2012quantum}, and thus there is a fixed point $\vrho_{fix}(\CL)\succeq 0$ of unit trace.
The condition $\lambda_{gap}(\CH)>2\norm{\CA}_{2-2}$ translates to $\lambda_2(\CH)<-\norm{\CA}_{2-2}$ implying that $0$ has algebraic multiplicity $1$ in $\text{Spec}(\CD(\vrho, \CL))$ due to \autoref{prop:Bauer-Fike}, which then proves the uniqueness of the fixed point. 
	
By our eigenvector perturbation bound (\autoref{prop:eigenvector_perturb_gen}), we get that there is a matrix $\vR$ of unit Frobenius norm in the kernel of $\CD(\vrho, \CL)$ such that
\begin{align}
\nrm{\sqrt{\vrho} - \vR}_2 &\le \frac{4\sqrt{2} \norm{\CA}_{2-2}}{-\lambda_2(\CH)},
\end{align}
where we used that the Frobenius norm of a matrix is equal to the Euclidean norm of its vectorization $\normp{\vA}{2} = \nrm{\ket{\vA}}_{}$.
This in turn means that $\vrho^{1/4}\vR\vrho^{1/4}$ is in the kernel of $\CL$, thus $\vrho_{fix}(\CL)=\vrho^{1/4}\vR\vrho^{1/4}/\tr(\vrho^{1/4}\vR\vrho^{1/4})$, moreover
\begin{align}
	\nrm{\vrho^{1/4}\vR\vrho^{1/4}-\vrho_{fix}(\CL)}_1
	&=\left|\tr(\vrho^{1/4}\vR\vrho^{1/4})-1\right|\nrm{\vrho_{fix}(\CL)}_1 
	= \left|\tr(\vrho^{1/4}\vR\vrho^{1/4}-\vrho)\right| 
	\leq \nrm{\vrho^{1/4}\vR\vrho^{1/4}-\vrho}_1,
\end{align}
where in the last step we used the trace-norm inequality $|\tr(\vA)|\leq \nrm{\vA}_1$. We can further bound
\begin{align}
\nrm{\vrho^{1/4}\vR\vrho^{1/4} - \vrho}_1
&= \normp{\vrho^{1/4} (\vR-\sqrt{\vrho})\vrho^{1/4}}{1}
\le \normp{\vrho^{1/4}}{4}^2\cdot \lnormp{\vR-\sqrt{\vrho}}{2}
=\nrm{\vR-\sqrt{\vrho}}_2,\label{eq:FrobToTrace}
\end{align}
where we used Hölder's inequality $\normp{\vB\vA\vB}{1} \le \normp{\vB}{4}^2 \normp{\vA}{2}$. 
Combining the above three inequalities we get 
\begin{align}
	\nrm{\vrho_{fix}(\CL)-\vrho}_1&\leq \nrm{\vrho_{fix}(\CL)-\vrho^{1/4}\vR\vrho^{1/4}}_1 + \nrm{\vrho^{1/4}\vR\vrho^{1/4} - \vrho}_1
	\leq 2 \nrm{\vR-\sqrt{\vrho}}_2
	\leq\frac{8\sqrt{2} \norm{\CA}_{2-2}}{-\lambda_2(\CH)}.
\end{align}
Finally, we convert $\lambda_2$ to $\lambda_{gap}$ in the above bound. 
Due to \eqref{eq:topEigenvalueBound} we have $-\lambda_2(\CH)=\lambda_{gap}(\CH)-\lambda_1(\CH)\geq \lambda_{gap}(\CH)-\norm{\CA}_{2-2}$ so we can further bound the above by $8\sqrt{2} \norm{\CA}_{2-2}/(\lambda_{gap}(\CH)-\norm{\CA}_{2-2})$. But this bound is vacuous at $\nrm{\vrho_{fix}(\CL)-\vrho}_1\le 2$, i.e., when $\lambda_{gap}(\CH)< (4\sqrt{2}+1)\norm{\CA}_{2-2}$. 
If $\lambda_{gap}(\CH)\geq (4\sqrt{2}+1)\norm{\CA}_{2-2}$, then $\lambda_{gap}(\CH)-\norm{\CA}_{2-2}\geq \frac{4\sqrt{2}}{4\sqrt{2}+1}\lambda_{gap}(\CH)$ yielding 
\begin{align}
	\nrm{\vrho_{fix}(\CL)-\vrho}_1&	\leq\frac{(8\sqrt{2}+2) \norm{\CA}_{2-2}}{\lambda_{gap}(\CH)}. \tag*{\qedhere}
\end{align}
\end{proof}
\Lfixedpointerrormix*
\begin{proof}
	In the proof of \autoref{prop:fixed_point_error} we got the ultimate bound by converting $\lambda_2$ to $\lambda_{gap}$ by the observation $-\lambda_2(\CH)\geq \lambda_{gap}(\CH)-\norm{\CA}_{2-2}$. Due to \autoref{prop:Bauer-Fike} the same bound $-\lambda_2(\CH)\geq \lambda_{\mathrm{Re}(gap)}(\CL)-\norm{\CA}_{2-2}$ also holds for $\lambda_{\mathrm{Re}(gap)}(\CL)$ (c.f.\ \autoref{prop:mixing_to_gap}) since $0$ is an eigenvalue of $\CL$. Hence the conversion combined with \autoref{prop:mixing_to_gap} shows
	\begin{align}
		\nrm{\vrho_{fix}(\CL)-\vrho}_1&	
		\leq\frac{(8\sqrt{2}+2) \norm{\CA}_{2-2}}{\lambda_{\mathrm{Re}(gap)}(\CL)}
		\leq\frac{(8\sqrt{2}+2) \norm{\CA}_{2-2}}{\ln(2)}t_{mix}(\CL). \tag*{\qedhere}
	\end{align}
\end{proof}

\subsection{Perturbation bounds for \Lword{}s regarding gaps and mixing times}\label{sec:gap_and_convergence}
This section provides proof for scattered statements circling spectral gaps and mixing time. Most results are standard, except maybe the most technical result (\autoref{lem:H+A_decay}). 

\mixingtimetofixedpoint*
\begin{proof}
We begin by recalling Duhamel's identity. We use its integral form derived in, e.g., \cite[Eq. (40)]{haber2021MatExpAndLogLectureNotes}:
\begin{equation}\label{eq:Duhamel}
	\e^{\vA}-\e^{\vB}=\int_{0}^{1} \e^{s \vA}(\vA-\vB)\e^{(1-s)\vB}\mathrm{d}s .
\end{equation}
We apply the above identity with $\vA\leftarrow t\CL_1$, $\vB\leftarrow t\CL_2$ and take the $1-1$ operator norm on both sides
\begin{align}
\nrm{\e^{t\CL_1}-\e^{t\CL_2}}_{1-1}
	&\leq \int_{0}^{1} \nrm{\e^{s t\CL_1}(t\CL_1-t\CL_2)\e^{(1-s)t\CL_2} }_{1-1}\mathrm{d}s\\
	&\leq t\int_{0}^{1} \nrm{\e^{s t\CL_1}}_{1-1}\nrm{t\CL_1-t\CL_2}_{1-1}\nrm{\e^{(1-s)t\CL_2} }_{1-1}\mathrm{d}s = t\nrm{\CL_1-\CL_2}_{1-1}.\label{eq:DuhamelLind}
\end{align}
The last equality uses that $\e^{x\CL}$ is a CPTP map and so $\lVert \e^{x\CL}\rVert_{1-1}=1$~\cite[Theorem 8.16]{wolf2012quantum}.\footnote{Here, by the $1-1$ norm, we mean the $1-1$ norm of the operators restricted to the subspace of Hermitian matrices.}

Now, let $\vrho_1:=\vrho_{fix}(\CL_1)$, $\vrho_2:=\vrho_{fix}(\CL_2)$, and  $t_{mix} := t_{mix}(\CL^{\dagger}_1)$. Then
\begin{align}
    \lnormp{ \vrho_1 - \vrho_2 }{1} &= \lnormp{\e^{\CL_1\cdot t_{mix}}[ \vrho_1] -  \e^{\CL_2 t_{mix}} [\vrho_2]}{1}\\
    & \le \lnormp{\e^{\CL_1\cdot t_{mix}}[ \vrho_1] - \e^{\CL_1t_{mix}}[ \vrho_2]}{1}+\lnormp{\e^{\CL_1t_{mix}}[ \vrho_2] - \e^{\CL_2 t_{mix}} [\vrho_2]}{1}\\
    &\le \frac{1}{2} \cdot \lnormp{ \vrho_1 - \vrho_2 }{1} + 2t_{mix}\normp{ \CL_1 - \CL_2 }{1-1}.\label{eq:tMix}    
\end{align}
The second inequality follows from \eqref{eq:DuhamelLind} and that $\lnormp{ \vrho_1 - \vrho_2 }{1} \le 2$. Rearrange \eqref{eq:tMix} to conclude the proof. 
\end{proof}
\begin{prop}[Mixing time difference]\label{prop:mixingtime_diff}
The mixing times of two \Lword{}s $\CL_1$, $\CL_2$ are related by
\begin{align}
t_{mix}(\CL_2) \le t_{mix}(\CL_1)\left\lceil\frac{\ln(1/2)}{\ln(1/2+t_{mix}(\CL_1) \normp{\CL_1-\CL_2}{1-1})}\right\rceil \quad \text{ if } \quad  t_{mix}(\CL_1) \normp{\CL_1-\CL_2}{1-1} < \frac{1}{2}.
\end{align}    
\end{prop}
\begin{proof}
Let $t_{mix} := t_{mix}(\CL_1)$ and $\vR$ be a traceless Hermitian matrix with $\nrm{\vR}_1=1$ maximizing $\normp{\e^{\CL_2t_{mix}}[\vR]}{1}$, then
\begin{align}
    \normp{\e^{\CL_2t_{mix}}[\vR]}{1} &\le \normp{\e^{\CL_1t_{mix}}[\vR]}{1}+\normp{\e^{\CL_2t_{mix}}[\vR]-\e^{\CL_1t_{mix}}[\vR]}{1}\\
    &\le \frac{1}{2}\normp{\vR}{1}+t_{mix}\normp{\CL_1-\CL_2}{1-1} \normp{\vR}{1} = \L(\frac{1}{2} +  t_{mix}\normp{\CL_1-\CL_2}{1-1}\R) \normp{\vR}{1}.
\end{align}
The second inequality follows from \eqref{eq:DuhamelLind}. Set $\vR\ =\vrho_1-\vrho_2$, rearrange, and take the logarithm to conclude.
\end{proof}

\subsubsection{Relating the mixing time to the spectral gap using exact detailed balance}
The mixing time of a general \Lword{} may be difficult to analyze. Fortunately, many handy bounds exist, especially circling the spectral gap when detailed balance holds.

\mixingtimegapDB*
\begin{proof}
Write $\vR = \vrho_1-\vrho_2$, then
    \begin{align}
        \lnormp{ \e^{\CL t}[\vR]}{1} 
        &= \lnormp{\vrho^{1/4}\e^{\CD t}[\vrho^{-1/4}\vR\vrho^{-1/4}]\vrho^{1/4}}{1} 	\\
        &\le \lnormp{\vrho^{1/4}}{4}\cdot \lnormp{\e^{\CD t}[\vrho^{-1/4}\vR\vrho^{-1/4}]}{2}\cdot\lnormp{\vrho^{1/4}}{4}\\
        &\le \e^{-\lambda_{gap}(\CH) t} \lnormp{\vrho^{-1/4}\vR\vrho^{-1/4}}{2}\\
        &\le \e^{-\lambda_{gap}(\CH)t}\norm{\vrho^{-1/4}}^2 \normp{\vR}{2} \\
        &\le \e^{-\lambda_{gap}(\CH)t}\norm{\vrho^{-1/4}}^2 \normp{\vR}{1} \\
        &= \e^{-\lambda_{gap}(\CL^{\dagger})t}\norm{\vrho^{-1/2}} \normp{\vR}{1}. 
    \end{align}
    The first inequality uses Hölder's inequality. The second inequality uses the orthogonality to the leading eigenvector such that $\tr[\sqrt{\vrho}\cdot\vrho^{-1/4}\vR\vrho^{-1/4}]=\tr[\vR]=0$. Take the logarithm to conclude the proof.
\end{proof}

\subsubsection{Relating the mixing time to the Hermitian gap (and approximate detailed balance)}

\begin{prop}[Spectral gap from mixing time]\label{prop:mixing_to_gap}
    For any \Lword{} $\CL$, let $-\lambda_{\mathrm{Re}(gap)}(\CL)$ be the second largest real part in its spectrum (counted by algebraic multiplicity), then
    \begin{align}\label{eq:gapBound}
        \lambda_{gap}(\CH)+2\nrm{\CA}_{2-2}\ge\nrm{\CA}_{2-2}-\lambda_{2}(\CH)\ge\lambda_{\mathrm{Re}(gap)}(\CL)\ge \frac{\ln(2)}{t_{mix}(\CL)}.
    \end{align}
    Moreover, if $\lambda_{\mathrm{Re}(gap)}(\CL)\geq 2\nrm{\CA}_{2-2}$, then there is unique eigenvalue $\lambda_{1}(\CH)\geq -\nrm{\CA}_{2-2}$ and
    \begin{align}\label{eq:gapRelations}
        \lambda_{gap}(\CH)+2\nrm{\CA}_{2-2}\ge\nrm{\CA}_{2-2}-\lambda_{2}(\CH)\ge\lambda_{\mathrm{Re}(gap)}(\CL).
    \end{align}    
\end{prop}
\begin{proof}
    We know any \Lword{} has at least a stationary state of eigenvalue $0$, and each eigenvalue which has no real part has a trivial Jordan block~\cite[Theorem 6.11 \& Proposition 6.2]{wolf2012quantum}. Therefore, if $\lambda_{\mathrm{Re}(gap)}(\CL)=0$, then $t_{mix}(\CL)=\infty$. If $\lambda_{\mathrm{Re}(gap)}(\CL)>0$, take any eigenvalue $\lambda$ such that $-\lambda_{\mathrm{Re}(gap)}(\CL)=\mathrm{Re}(\lambda)$, and let $\vR$ be a corresponding right eigenvector of $\CL$, which is then necessarily traceless. We can assume without loss of generality that the Hermitian part $\vR_H:=(\vR+\vR^\dagg)/2$ is nonzero (otherwise, we can just take $\vR\leftarrow \ri \vR$). Since $\CL$ is Hermiticity preserving, we get that $\vR_H$ is also a right eigenvector with eigenvalue $\lambda$. For $t<\frac{\ln(2)}{\lambda_{\mathrm{Re}(gap)}(\CL)}$ we have
    \begin{align}
         \normp{\e^{\CL t}[\vR_H]}{1} = \labs{\e^{\lambda t}} \cdot \normp{\vR_H}{1} > \frac{1}{2}\normp{\vR_H}{1}
    \end{align}
    implying that
    \begin{align}
        t_{mix}(\CL) \ge \frac{\ln(2)}{\lambda_{\mathrm{Re}(gap)}(\CL)}.
    \end{align}
    We conclude using Bauer-Fike \autoref{prop:Bauer-Fike} to show $\lambda_{\mathrm{Re}(gap)}(\CL)\leq-\lambda_{2}(\CH)+\nrm{\CA}_{2-2}$, combined with \eqref{eq:topEigenvalueBound}. 
\end{proof}

What if the detailed balance condition is violated? In the worst case, the conversion from spectral gap to mixing time can be really poor. However, the \Lword{}s we consider are ``sufficiently'' detailed balanced so that essentially the same consequences hold. Our main use cases are covered by the following two scenarios
\begin{align}
    \CD(\vrho, \CL) = \CH + \CA \quad &\text{where}\quad &\normp{\CA}{2-2} &\ll \lambda_{gap}(\CH), \tag*{(approximate detailed balance)}\\
     &\text{or}\quad &\lambda_1(\CH) &\ll \lambda_{gap}(\CH) \tag*{(nonperturbative $\CA$)}.
\end{align}
The first case should be understood as nonHermitian eigenvalue perturbation, consistent with the framework of approximate detailed balance. The second case is less intuitive, as the anti-Hermitian part can be arbitrarily large. Intriguingly, the spectral properties of the Hermitian part $\CH$ suffice to control convergence even in the presence of a large perturbation $\CH+\CA$. In fact, the second case is strictly more general since $\lambda_1(\CH) \le \normp{\CA}{2-2}$~\eqref{eq:topEigenvalueBound}, therefore we will only analyze the second scenario. Intuitively, when the \Lword{} is exactly detailed balanced $\CA=0$, we have that $\lambda_1(\CH) =0$; the $\lambda_1(\CH) \ll \lambda_{gap}(\CH)$ condition is essentially the requirement that $\lambda_1(\CH) \approx 0$ in spite of a large anti-hermitian component $\norm{\CA}_{2-2} \gg 0$. Before proving our result, we need a few lemmas. 

\begin{lem}[Norm of matrix exponential]\label{lem:decay_hermitian}
For a Hermitian matrix $\vH$ and an anti-Hermitian matrix $\vB$, we have
\begin{align}
\norm{\e^{(\vH+\vB)t}} \le \e^{\lambda_{1}(\vH) t} \quad \text{for each} \quad t\geq 0.
\end{align}
If $\vB$ is an arbitrary matrix, we alternatively get $\norm{\e^{(\vH+\vB)t}} \le \e^{(\lambda_{1}(\vH)+\nrm{\vB}) t}$.
\end{lem}
\begin{proof}
    The claim follows from the Trotter representation of exponential and triangle inequality
    \begin{align}
        \!\norm{\e^{(\vH + \vB )t}} &=\nrm{ \lim_{r \rightarrow \infty}\!\L(\e^{\vH /r} \e^{\vB /r}\R)^{\!\!rt}}\!
        = \!\lim_{r \rightarrow \infty}\nrm{\L(\e^{\vH /r} \e^{\vB /r}\R)^{\!\!rt}}
        \leq \liminf_{r \rightarrow \infty}\L(\nrm{\e^{\vH /r}}\nrm{\e^{\vB /r}}\R)^{\!\!rt}\!\!
        = \liminf_{r \rightarrow \infty}\L(\nrm{\e^{\vH /r}}\R)^{\!\!rt}\!\!
        = \e^{\lambda_{1}(\vH) t},
    \end{align}
    where we used that $\e^{\vB/r}$ is unitary so that $\norm{\e^{\vB/r}} =1$. The second claim follows from isolating the Hermitian part $\frac{1}{2}(\vB+\vB^\dagg)$, whose norm is bounded by $\nrm{\vB}$.
\end{proof}

\begin{lem}[Hermitian gap controls decay]
\label{lem:H+A_decay}
Consider a Hermitian matrix $\vH$ and an anti-Hermitian $\vA$.\linebreak If 
\begin{align}
    r := \frac{\lambda_1(\vH)}{\lambda_{gap}(\vH)} \le \frac{1}{100},
\end{align}
and $\vH+\vA$ has an eigenvalue $0$ with left and right eigenvectors 
\begin{align}
	\bra{L} (\vH+\vA) = 0\quad \text{and}\quad(\vH+\vA) \ket{R}=0,
\end{align}
then 
\begin{align}
	\lnorm{\e^{(\vH+\vA)t}\L(\vI-\vP_0\R)} \le \frac32\exp\L(\frac{\lambda_2(\vH)t}{2}\R) \quad \text{where}\quad\vP_0 :=\frac{1}{\braket{L|R}}\ketbra{R}{L}.
\end{align}
\end{lem}
Intuitively speaking, the conditions above ensure that even if the anti-Hermitian $\vA$ is large, it mainly introduces ``rotations'' and the spectral properties of $\vH$ still guarantee fast convergence to $\vP_0$.
\begin{proof}
	Let $\ket{\psi_1}$ be the eigenvector of $\vH$ corresponding to its top eigenvalue $\lambda_1(\vH)$. We can assume without loss of generality that $\bra{L}$, $\ket{R}$, and $\ket{\psi_1}$ have unit norm
	\begin{align}
		\norm{\ket{\psi_1}}=\norm{\bra{L}}=\norm{\ket{R}}=1,
	\end{align}
	and we are free to choose the phase of these vectors as the projector $\vP_0$ is invariant under changing the phase of $\bra{L}$, $\ket{R}$.
	The guiding intuition behind the proof is that if $r$ is small, then (up to a phase)
\begin{align}
     \ket{L}\approx \ket{\psi_1}\approx \ket{R}.
\end{align}
To show this, take the real part of $\bra{R} (\vH+\vA)\ket{R}=0$ to obtain\footnote{Note that $\Re(\bra{R} (\vH+\vA)\ket{R})=0$ implies $\lambda_1(\vH)\geq 0$, and therefore $r\geq 0$. This then further implies $\lambda_2(\vH)\le 0$ due to $r\leq 0.01$.}
\begin{align}
    \bra{R}\vH\ket{R} = 0 &= \lambda_1(\vH) \labs{\braket{R|\psi_1}}^2 + \sum_{i=2}\lambda_i(\vH) \labs{\braket{R|\psi_i}}^2\\
    &\le  \lambda_1 \labs{\braket{R|\psi_1}}^2 + \lambda_2 (1-\labs{\braket{R|\psi_1}}^2 ),\\
     \text{implying} \quad 1-\labs{\braket{R|\psi_1}}^2 &\le  \frac{\lambda_1(\vH)}{\lambda_{gap}(\vH)}=r,\\
     \text{and thus} \quad \nrm{\ket{R}- \ket{\psi_1}} &\le 1.01\sqrt{r},
\end{align}
where in the last step, we assumed without loss of generality that the phase of $\ket{R}$ is such that $\braket{R|\psi_1}$ is nonnegative real,
so we get $\braket{R|\psi_1}=\cos(\theta)$ for some $\theta\in [0,\pi/2]$. Then $\sin(\theta)\le\sqrt{r}$, and since for every $\theta\in [0,\arcsin(1/10)]$ we have $1-\cos(\theta)\leq \sin(\theta)/10$, we get the norm bound $\frac{101}{100}\sin(\theta)\le\frac{101}{100}\sqrt{r}$.
Similarly, 
\begin{align}
    \bra{L} (\vH+\vA)\ket{L}=0 \quad \text{implies}\quad \nrm{\ket{L}- \ket{\psi_1}} &\le 1.01\sqrt{r}.
\end{align}
As a direct consequence, using that $\sqrt{r}\le 1/10$, we get 
\begin{align}
|\braket{L|R}-1|=
|\braket{L|R}-\braket{\psi_1|\psi_1}|
&\leq|\braket{L|R}-\braket{L|\psi_1}|+|\braket{L|\psi_1}-\braket{\psi_1|\psi_1}|\\
&\leq \nrm{\ket{R}- \ket{\psi_1}} + \nrm{\ket{L}- \ket{\psi_1}}
\leq 2.02\sqrt{r}\le \frac{1}{4},
\end{align}
and similarly that
\begin{align}
\left|\frac{1}{\braket{L|R}}-1\right|
=\left|\frac{1-\braket{L|R}}{\braket{L|R}}\right|
\leq \frac{4}{3} |1-\braket{L|R}| 
\leq 2.73\sqrt{r}. 
\end{align}
Therefore, 
\begin{align}
\nrm{\vP_0-\ketbra{\psi_1}{\psi_1}}
&\leq \nrm{\frac{\ketbra{R}{L}}{\braket{L|R}} -\ketbra{R}{L}}
+ \nrm{\ketbra{R}{L}-\ketbra{R}{\psi_1}}
+ \nrm{\ketbra{R}{\psi_1}-\ketbra{\psi_1}{\psi_1}}\\&
=\left|\frac{1}{\braket{L|R}}- 1 \right| + \nrm{\ket{L}- \ket{\psi_1}} + \nrm{\ket{R}- \ket{\psi_1}}
\leq 4.75 \sqrt{r}. \label{eq:P0bound}
\end{align} 
We will use the following properties of the projector $\vP_0$; one can easily see that it commutes with $(\vH + \vA )$:
\begin{align}
    \vP_0 (\vH+\vA) = \frac{1}{\braket{L|R}} \ketbra{R}{L} (\vH+\vA) = 0 =   (\vH+\vA)\vP_0, \label{eq:P0commutes}
\end{align}
and since it is a projector, it satisfies the algebraic identity
\begin{align}
\vP_0(\vI-\vP_0) = 0\quad \text{so that}\quad \e^{x \vP_0}(\vI - \vP_0) =\vI - \vP_0 \quad \text{for any $x\in \BC$}.\label{eq:eP_0}
\end{align} 
The above properties streamline the rest of the proof of our bound
\begin{align}
\nrm{\e^{(\vH + \vA )t} (\vI - \vP_0)} 
&= \nrm{\e^{(\vH + \vA )t}\cdot \e^{-\lambda_{gap}(\vH) \vP_0 t} (\vI - \vP_0)}   \tag*{(by~\eqref{eq:eP_0})}\\& 
= \nrm{\e^{(\vH + \vA -\lambda_{gap}(\vH) \vP_0)t} (\vI - \vP_0)} \tag*{(by~\eqref{eq:P0commutes})}\\&
\leq \nrm{\e^{(\vH + \vA -\lambda_{gap}(\vH) \vP_0)t}}\cdot\nrm{\vI - \vP_0} \\&
= \Big\lVert \exp\L(\overset{\vH':=}{\overbrace{(\vH -\lambda_{gap}(\vH)\ketbra{\psi_1}{\psi_1})}}t + \vA t + \overset{\vB:=}{\overbrace{\lambda_{gap}(\vH)(\ketbra{\psi_1}{\psi_1}-\vP_0)}}t\R)\Big\rVert\cdot \nrm{\vI - \vP_0} \\&
\leq \e^{(\lambda_{1}(\vH') + \nrm{\vB})t}\left(\nrm{\vI- \ketbra{\psi_1}{\psi_1}} + \nrm{\ketbra{\psi_1}{\psi_1} - \vP_0}\right)\tag*{(by \autoref{lem:decay_hermitian})}\\&
\leq \exp\L( \lambda_2(\vH)( 1- \frac{4.75 \sqrt{r}}{1-r}) t\R)(1+4.75 \sqrt{r}). \tag*{(by Eq. \eqref{eq:P0bound})}\\
&\leq \frac32\cdot \exp\L( \frac{1}{2} \lambda_2(\vH) t\R). \tag*{(by $r\le 1/100$)\qedhere}
\end{align}
\end{proof}

Now, we can specialize the above to the case of \Lword{}s to obtain mixing times.

\gaptomixing*
\begin{proof}
	We proceed as in the proof of \autoref{prop:mixDetail}.
	For any traceless Hermitian $\vR$, we have 
	\begin{align}
		\lnormp{\e^{\CL t}[\vR]}{1} &= \lnormp{\vrho^{1/4}\e^{t(\CH + \CA )}[\vrho^{-1/4}(\vR)\vrho^{-1/4}]\vrho^{1/4}}{1} \tag*{(by definition)}	\\
		&\le\lnormp{\vrho^{1/4}}{4}\cdot\lnormp{\e^{t(\CH + \CA )}[\vrho^{-1/4}(\vR)\vrho^{-1/4}]}{2}\cdot\lnormp{\vrho^{1/4}}{4} \tag*{(by Hölder's inequality)}\\
		&\le \frac32\e^{\lambda_2(\vH)t/2}\lnormp{\vrho^{-1/4}\vR\vrho^{-1/4}}{2}\tag*{(by \autoref{lem:H+A_decay})}\\
		&\le \frac32\e^{\lambda_2(\vH)t/2} \norm{\vrho^{-1/4}}^2 \normp{\vR}{2}\tag*{(by Hölder's inequality)}\\
		&\le \frac32\e^{\lambda_2(\vH)t/2} \norm{\vrho^{-1/2}} \normp{\vR}{1}.   \tag*{(since $\vrho\succ 0$ and $\nrm{\cdot}_2\leq \nrm{\cdot}_1$)}
	\end{align}
	The second inequality uses~\autoref{lem:H+A_decay} since $\lim_{t\rightarrow\infty}\e^{\CL t}[\vR] = 0$ (or in the notation of \autoref{lem:H+A_decay}, $\vP_0\ket{\vR}=0$). By assumption,
	$\lambda_2(\CH) = \lambda_1(\CH)-\lambda_{gap}(\CH) \le  - 0.99\lambda_{gap}(\CH)$; take the logarithm to conclude the proof.
\end{proof}

\section{Improved incoherent \Lword{} simulation}\label{apx:improvLind}
The improved algorithm builds on the circuit in \autoref{fig:postWeakMeasCircuit}, which is similar to that of our weak measurement scheme in \autoref{fig:weakMeasCircuit}, with a key technical difference: here, a specific ancilla state $\vY_{\delta}\ket{0}=\sqrt{1-\delta}\ket{0}+\sqrt{\delta}\ket{1}$ ``triggers'' the appropriate weak measurement. Thus, the circuit is ``idle'' if we remove the $\vY_{\delta}$ gate, which is essential for our ``compression'' argument. However, this also makes the circuit post-selective, which could exponentially decrease the success probability if we were to run for longer times naively. Thus, for $t>2$ we decompose the simulation to $\lceil t/2\rceil$ equal segments, each of which can be amplified with constant (coherent) repetitions by oblivious amplitude amplification. 

\begin{figure}[!ht]
	\begin{quantikz}[wire types={q,q,b,b,b},classical	gap=1mm]
	 	\lstick{\ket{0}}		&\gate{\vY^{\phantom{\dagg}}_{{_{\kern-1.5mm\phantom{f}}}\delta}}			&\ctrl{1}	&\gate{\vY^\dagg_{\frac{\delta}{4}}}				&\meter{}\rstick[wires=4]{accept the all-zero outcome \\ \kern8mm and when the second qubit is $1$ \kern8mm}\\		
	 	\lstick{\ket{0}}		&\qw			&\targ{}	&\octrl{1}							&|[meter]| \qw\\
	 	\lstick{\ket{0^b}}		&\gate[3]{\vU}	&\octrl{-1} 	&\gate[3]{\vU^\dagg}	&|[meter]| \qw \\
	 	\lstick{\ket{0^{a-b}}}	& 				&\qw		&\qw 			&|[meter]| \qw \\
	 	\lstick{$\vrho$}		&				&\qw		&\qw			& \qw \rstick{$= (1-\frac{\delta}{4}) \e^{\delta \CL}\![\vrho]+\bigO{\delta^2}$}
	 \end{quantikz}
	\caption{Alternative quantum circuit implementation of an approximate $\delta$-time step via a postselective weak measurement scheme. Let $\vC'$ be the circuit that we get by removing the two single qubit rotation gates $\vY_{\delta},\vY^{\dagger}_{\delta/4}$ from the first qubit. For our compression argument it is of paramount importance that $\vC'\cdot\ketbra{0^{a+2}}{0^{a+2}}\otimes \vI =\ketbra{0^{a+2}}{0^{a+2}}\otimes \vI$.}\label{fig:postWeakMeasCircuit}
\end{figure}
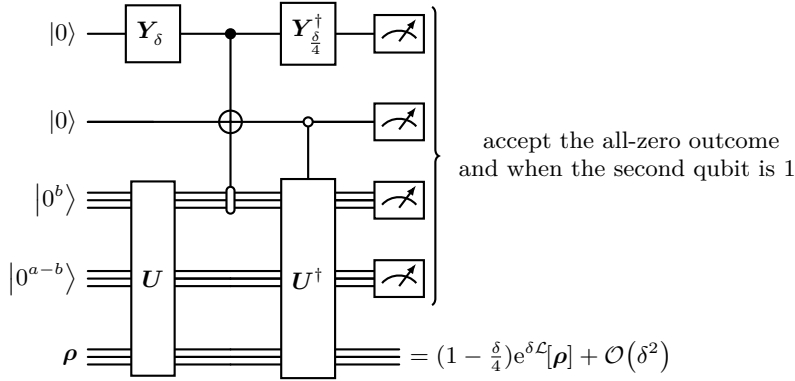

The main conceptual ingredient leading to the substantial improvement is \textit{compression}~\cite{cleve2016EffLindbladianSim}: instead of naively running $r$-repetitions of the circuit from \autoref{fig:postWeakMeasCircuit} as outlined in \autoref{fig:postWeakMeasCircuitRep}, we \textit{compress} the circuit such that it uses only $h\sim t\log(t/\epsilon) \ll r = \Theta(t^2/\epsilon)$ repetitions. The compression technique relies on understanding the joint initial state $\left(\ket{0^{c+1}}\vY_\delta\ket{0}\right)^{\otimes r}$ of the circuit in \autoref{fig:postWeakMeasCircuitRep} after the single-qubit gates. The upshot is that $\vY_\delta^{\otimes r}\ket{0^r}$ is \textit{concentrated} on strings with Hamming weight $\le h$, thus $h$ repetitions will suffice for the mass of the amplitudes. 

	\begin{figure}[!ht]
	 \yquantdefinebox{dots}[inner sep=0pt]{$\dots$}
	 \begin{tikzpicture}
	 	\begin{yquant}
	 		qubit {$\ket{0}$} a1;
	 		qubits {$\ket{0^{c+1}}$} a2;
	 		nobit etc;
	 		qubit {$\ket{0}$} b1;
	 		qubits {$\ket{0^{c+1}}$} b2;			
	 		qubit {$\ket{0}$} c1;
	 		qubits {$\ket{0^{c+1}}$} c2;			
	 		qubits {$\vrho$} r;
			
	 		box {$\vY_\delta$} a1;
	 		text {$\vdots$} etc;
	 		text {} etc;
	 		text {} etc;
	 		text {} etc;
	 		text {$\iddots$\kern-5mm} etc;
	 		box {$\vY_\delta$} b1;
	 		box {$\vY_\delta$} c1;	
			
	 		box {$\vC'$} (c1,c2,r);					
	 		box {$\vC'$} (b1,b2,r);
	 		hspace {7mm} -;
	 		box {$\vC'$} (a1,a2,r);
			
	 		align -;
	 		box {$\vY^\dagg_{\delta/4}$} a1;
	 		text {$\vdots$} etc;
	 		box {$\vY^\dagg_{\delta/4}$} b1;
	 		box {$\vY^\dagg_{\delta/4}$} c1;	
			
	 		measure a1,a2,b1,b2,c1,c2;
			
	 		output {\parbox{0.25\textwidth}{accept the all-zero outcome and when the second qubit is $1$}} (a1,a2);
	 		output {\parbox{0.25\textwidth}{accept the all-zero outcome and when the second qubit is $1$}} (b1,b2);
	 		output {\parbox{0.25\textwidth}{accept the all-zero outcome and when the second qubit is $1$}} (c1,c2);
	 		output {$= (1-\frac{\delta}{4})^r \e^{r\delta \CL}\![\vrho]+\bigO{r\delta^2}$} r;
	 	\end{yquant}
	\end{tikzpicture}
	\caption{$r$ subsequent repetitions of the circuit $\vC'$ from \autoref{fig:postWeakMeasCircuit}. The circuits $\vC'$ act on potentially nonadjacent qubits, which is indicated by the vertical curly connection between the visually split ``halves'' of the affected $\vC'$  circuits.}\label{fig:postWeakMeasCircuitRep}
\end{figure}
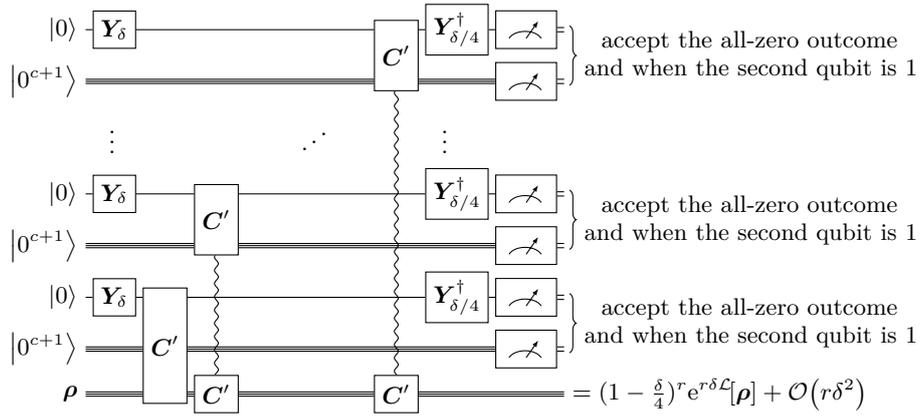

Before we dive into the proof, we explain the intuitive compression strategy in more detail; the actual proof closely follows this, but with technical changes. Let $t \le 2$ and let $\vC'$ be the circuit that we get from $\vC$ by removing the two single qubit gates $\vY_{\delta},\vY^{\dagger}_{\delta/4}$ from the first qubit in \autoref{fig:postWeakMeasCircuit}. Let $X_i$ be the random variable representing the measurement outcome of a computational basis measurement of the $i$-th qubit in $\vY_\delta^{\otimes r}\ket{0^r}$. Then, the Chernoff bound tells us that $\Pr\left(\sum_{i=1}^r X_i>(1+y)t\right)\leq \left(\frac{\e^{y}}{(1+y)^{1+y}}\right)^{\!t}$ so that the probability that the Hamming weight of $\left(\sqrt{1-\delta}\ket{0}+\sqrt{\delta}\ket{1}\right)^{\otimes r}$ is greater than $h:=(1+y)t$ is at most $\frac{\e^{h-t}t^h}{h^h}\leq (\frac{et}{h})^h\leq (\frac{2e}{h})^h$. In particular, choosing 
\begin{align}
	h=\Theta\left(\frac{\log(1/\epsilon)}{\log\log(1/\epsilon)}\right) 
\end{align}
ensures that this probability is at most $\bigO{\epsilon^2}$. Therefore, the initial state can be replaced by its (normalized) projection $\ket{\phi_0}$ to the subspace of Hamming-weight $\leq h$ states while inflicting an error that is bounded by 
\begin{align}
	\norm{\vY_\delta^{\otimes r}\ket{0^r} - \ket{\phi_0}} \le \epsilon.
\end{align}
This bound on Hamming-weights translates into a reduction of applications of the circuit $\vC'$. Since $\ket{\phi_0}$ is a superposition of bitstrings of Hamming-weights at most $h$, in all branches of the superposition all but $h$ applications of $\vC'$ can be neglected, crucially because it acts trivially when the ancilla register is in state $\ket{0^{a+2}}$:
\begin{align}
	\vC'\cdot\ketbra{0^{a+2}}{0^{a+2}}\otimes \vI =\ketbra{0^{a+2}}{0^{a+2}}\otimes \vI. \label{eq:C_idle}   
\end{align}
Now we define a compression scheme for the $r$ ancilla registers, each containing $(a+2)$ qubits\footnote{Note that the first qubit is redundant in this encoding, but we add it here for clarity of the presentation.} in \autoref{fig:postWeakMeasCircuitRep}, inspired by \cite{berry2014GateEfficientDiscreteSim,cleve2016EffLindbladianSim}. The compression scheme can represent the $\leq h$ Hamming-weight states of the ancilla registers on just $h\cdot(\log(r+1)+a+2)$ qubits (with respect to the $r$ registers to be compressed, by Hamming weight, we mean the number of registers that do not contain the state $\ket{0^{a+2}}$). Marking the register state $\ket{0^{a+2}}$ by $\pmzerodot$ and the content of the $i$-th nonzero register by $d_i$ the encoding works as follows:
\begin{align}\label{eq:encoding}
	\left(\{0,1\}^{a+2}\right)^r\ni \pmzerodot^{s_1}d_1\pmzerodot^{s_2}d_2\ldots \rightarrow (s_1,s_2,\ldots) \times (d_1,d_2,\ldots) \in\{0,1,2,\ldots,r\}^h \times \left(\{0,1\}^{a+2}\right)^h.
\end{align}
The compressed representation's first ``compression'' register contains $h$ blocks of $\log(r+1)$ qubits, designated to store a sequence $s\in\{0,1,2,\ldots,r\}^h$, where $s_i$ is the number of consecutive ancilla registers containing $\ket{0^{a+2}}$ before the $i$-th ancilla register that is not in state $\pmzerodot$; if $i$ exceeds the Hamming weight, then we set $s_i=r$. The second ``data'' register consists of $h$ blocks of $(a+2)$ qubits, where the $i$-th block represents the qubits of the $i$-th nonzero register of the uncompressed state; if the Hamming weight is less than $i$, then the block is set to $\pmzerodot$. The property~\eqref{eq:C_idle} means that the $\vC'$ gates can be applied ``transversally'' on the second ``data'' register of the encoded scheme because they do not change the location of the nonzero registers of the uncompressed state. 
\LCUSim*
\begin{proof}[Proof of \autoref{thm:LCUSim}]
	We begin with analyzing the (modified) weak-measurement scheme using similar calculations to~\autoref{thm:weakMeasSim}, and then compress it. We focus on the purely irreversible scenario, and at the end, handle the general case.
	
	\textbf{(Postselected weak-measurement.)} The circuit $\vC$ from \autoref{fig:postWeakMeasCircuit} on a pure input state $\ket{\psi}$ acts as:
	\begin{align}
	\ket{0^{a+2}}\ket{\psi}&
	\stackrel{(1)}{\rightarrow} (\sqrt{1-\delta}\ket{0}+\sqrt{\delta}\ket{1})\ket{0}\vU\ket{0^c}\ket{\psi}\nonumber\\&
	\stackrel{(2)}{\rightarrow}\sqrt{1-\delta}\ket{00}\vU\ket{0^c}\ket{\psi} + \sqrt{\delta} \ket{11}\L(\ketbra{0^b}{0^b}\otimes \vI\R)\vU\ket{0^c}\ket{\psi} + \sqrt{\delta}\ket{10}(\vI-\ketbra{0^b}{0^b}\otimes \vI) \vU\ket{0^c}\ket{\psi}\nonumber\\&
	=(\sqrt{1-\delta}\ket{00}+\sqrt{\delta}\ket{10})\vU\ket{0^c}\ket{\psi} + \sqrt{\delta} \ket{11}\ket{0^b}\underset{\ket{\psi'_0} :=}{\underbrace{(\bra{0^b}\otimes \vI)\vU\ket{0^c}\ket{\psi}}}  - \sqrt{\delta}\ket{10}\L(\ketbra{0^b}{0^b}\otimes \vI\R) \vU\ket{0^c}\ket{\psi}\nonumber\\&
	\stackrel{(3)}{\rightarrow}\left(\vY_{\frac{\delta}{4}}^\dagg\!\!\otimes \!\vI\right)\left( (\sqrt{1-\delta}\ket{00}+\sqrt{\delta}\ket{10})\ket{0^c}\ket{\psi} + \sqrt{\delta} \ket{11}\ket{0^b}\ket{\psi'_0} 
	- \sqrt{\delta}\ket{10}\vU^\dagg\L(\ketbra{0^b}{0^b}\otimes \vI\R)\vU\ket{0^c}\ket{\psi}\!\right)\label{eq:weakMeasPsi2}
	\end{align}
	Now, let us compute the part of $\vC\ket{0^{a+2}}\ket{\psi}$ starting with $\ket{0^{a+2}}$:
	\begin{align}
	(\bra{0^{a+2}}\otimes \vI) \vC\ket{0^{a+2}}\ket{\psi}
	&=\underset{=1-\frac{\delta}{8}+\bigO{\delta^2}=\sqrt{1-\frac{\delta}{4}}+\bigO{\delta^2}}{\underbrace{\Big(\sqrt{\Big(1-\frac{\delta}{4}\Big)(1-\delta)}+\frac{\delta}{2}\Big)}}\ket{\psi}-\frac{\delta}{2}\underset{\sum_{j\in J}\vL_j^\dagg\vL_j\ket{\psi}}{\underbrace{(\bra{0^{c}}\otimes \vI)\vU^\dagg\L(\ketbra{0^b}{0^b}\otimes \vI\R)\vU\ket{0^c}\ket{\psi}}}\\&
	=\sqrt{1-\frac{\delta}{4}}\left(\left(\vI-\frac{\delta}{2}\sum_{j\in J}\vL_j^\dagg\vL_j\right)\ket{\psi}+\bigO{\delta^2}\right),
	\end{align}
	and the part where second qubit is $\ket{1}$:
	\begin{align}
	(\vI\otimes \bra{1}\otimes \vI) \vC\ket{0^{a+2}}\ket{\psi}
	&=\sqrt{\delta}\vY_{\frac{\delta}{4}}^\dagg\ket{1} \ket{0^b}\ket{\psi'_0} 
	= \sqrt{\delta}\vY_{\frac{\delta}{4}}^\dagg\ket{1}\ket{0^b}\sum_{j\in J} \ket{j}\vL_j\ket{\psi}.
	\end{align}
	Let $\widetilde{\vPi}:=\frac{1}{\sqrt{1-\frac{\delta}{4}}}\left(\ketbra{0^{a+2}}{0^{a+2}}\otimes \vI+\vI\otimes\ketbra{1}{1}\otimes \vI\right)$, the above implies that 
	\begin{align}
	\tr_{a+2}\left(\widetilde{\vPi} C \left(\ketbra{0^{a+2}}{0^{a+2}}\otimes\ketbra{\psi}{\psi}\right)C^\dagg \widetilde{\vPi}\right)
	&=\left(\vI-\frac{\delta}{2}\sum_{j\in J}\vL_j^\dagg\vL_j\right)\ketbra{\psi}{\psi}\left(\vI-\frac{\delta}{2}\sum_{j\in J}\vL_j^\dagg\vL_j\right)+\delta \sum_{j\in J}\vL_j\ketbra{\psi}{\psi}\vL_j^\dagg+\bigO{\delta^2}\\&
	=\e^{\delta \CL}\![\ketbra{\psi}{\psi}]+\bigO{\delta^2}.
	\end{align}
	Similarly to the proof of \autoref{thm:weakMeasSim} it is easy to see that this implies 
	\begin{align}
	\nrm{\tr_{a+2}\left(\widetilde{\vPi} C \left(\ketbra{0^{a+2}}{0^{a+2}}\otimes[\cdot]\right)C^\dagg \widetilde{\vPi}\right)-\e^{\delta \CL}[\cdot]}_\Diamond = \bigO{\delta^2}.
	\end{align}
	
	Like in \autoref{thm:weakMeasSim}, choosing $\delta=\Theta(\epsilon/t)$ and repeating the process $r:=t/\delta$ times (every time using $a+2$ fresh ancillas) yields an $\epsilon$-accurate simulation (Figure~\ref{fig:postWeakMeasCircuitRep}). This gives similar circuit complexity as \autoref{thm:weakMeasSim}, except that the resulting postselective protocol has success probability about $(1-\frac{\epsilon}{t})^{\frac{t^2}{\epsilon}}=\exp(-\Theta(t))$.
	More precisely, the square of the subnormalization factor is $(1-\frac{\delta}{4})^{\frac{t}{\delta}}$, which is at least $\frac{1}{2}$ for $t\leq 2$ and $\delta\in(0,1]$ since $(1-\frac{\delta}{4})^{\frac{t}{\delta}}\geq (1-\frac{\delta}{4})^{\frac{2}{\delta}}\geq (\frac{3}{4})^2>\frac{1}{2}$.
	
	\textbf{(Compression.)} While the above discussion is largely similar to~\autoref{thm:weakMeasSim}, we now further compress to obtain the desired complexity. We focus on analyzing the case when $t\leq 2$, and later show how to bootstrap the results for arbitrary large $t$.

	Instead of running the verbose circuit of \autoref{fig:postWeakMeasCircuitRep}, we directly prepare the compressed version of the state $\ket{\phi_0}$ using the techniques of \cite{berry2014GateEfficientDiscreteSim}. Using the algorithm of \cite{berry2014GateEfficientDiscreteSim}, we can directly prepare the ``compression'' register corresponding to $\ket{\phi_0}$, then we can initialize the ``data'' register by looping through all blocks and applying an $\vX$ gate on the first qubit conditioned on the corresponding ``compression'' block containing a number less than~$r$.
	Then we apply the $\vC'$ gates ``transversally'', looping through each of the $h$ blocks of the ``data'' register of the compressed state in increasing order starting from the first block. 
	
	The last and technically most challenging difficulty that we face is that we need to evaluate the post-selection criterion in \autoref{fig:postWeakMeasCircuitRep}. We could, of course, uncompress the state, apply the single qubit gates $(\vY^\dagg_{\delta/4})^{\otimes r}$ and perform the measurement literally as depicted on \autoref{fig:postWeakMeasCircuitRep}, however it is possible to evaluate this criterion while keeping the compressed representation.\footnote{
		The compression and compact verification procedure for checking whether all segments were applied successfully in the encoded scheme (and the resulting reflection operator) does not appear to work as described in \cite{cleve2016EffLindbladianSim}. 
		The issue is manifest in their description of the reflection operator about accepted outcomes: ``Therefore, the corresponding operation in the encoded representation is first applying $\e^\dagg$, then applying the reflection about the encoded state $\ket{0^c}|0^b\rangle$ on the first two registers, and last applying $E$.'' 
		
		However, this approach does not seem to work, as noted in~\cite{berry2014GateEfficientDiscreteSim} (here $U_m$ is the analog of $E$ in the above quote from \cite{cleve2016EffLindbladianSim}): ``At first glance, one might imagine that applying $U_m$ in place of $R^{\otimes m}$ would yield a succinct representation of the final outcome state, so measuring in the computational basis would provide the correct result. Unfortunately, this does not accurately simulate the final measurement except in the case where the all-zero string is obtained.''
		
		For completeness, we include and analyze our modified compression / verification scheme.
		There may be a simpler fix for the above issue~\cite{wang2023email}, or alternatively, the techniques of the recent work~\cite{li2022SimMarkOpen} -- circumventing compression -- could also be compatible with our improvements~\cite{wang2023email}, which would in turn probably also simplify our circuits.
	} 
	
	We use a slightly modified variant of the compressed measurement scheme outlined in \cite[Section 5]{berry2014GateEfficientDiscreteSim} that enables us to perform a measurement of the form $\left(\vP_0\otimes\ketbra{0^{c+1}}{0^{c+1}},\vI-\vP_0\otimes\ketbra{0^{c+1}}{0^{c+1}}\right)^{\otimes r}$, where $\vP_0=\vY_{\delta/4}\ketbra{0}{0}\vY_{\delta/4}^\dagg$. This compressed measurement procedure reports the result also in a compressed form by listing the (uncompressed) indices $(i_1,i_2,\ldots,i_\ell)$ where the measurement outcome is $\vI-\vP_0\otimes\ketbra{0^{c+1}}{0^{c+1}}$.
	
	The initial observation of \cite{berry2014GateEfficientDiscreteSim} is that the measurement $\left(\left(\vP_0\otimes\ketbra{0^{c+1}}{0^{c+1}}\right)^{\otimes r},\vI-\left(\vP_0\otimes\ketbra{0^{c+1}}{0^{c+1}}\right)^{\otimes r}\right)$ can be approximately performed by using the compressed state preparation circuit approximately preparing the compressed version of the state $\left(\ket{0^{c+1}}\vY_{\delta/4}\ket{0}\right)^{\otimes r}$. Indeed, we can just run this compressed state preparation in reverse and verify that we get the all-zero state. 
	
	A similar procedure can be devised for performing the compressed measurement 
	\begin{align}\label{eq:submeasurement}
	\left(\left(\vP_0\otimes\ketbra{0^{c+1}}{0^{c+1}}\right)^{\otimes \ell},\vI-\left(\vP_0\otimes\ketbra{0^{c+1}}{0^{c+1}}\right)^{\otimes \ell}\right)
	\end{align}
	for any given consecutive $\ell$ (uncompressed) registers $d_i,d_{i+1},\ldots,d_{i+\ell-1}$. Towards this, observe that we can efficiently convert an encoded string $e_{1;r}$ of the from \eqref{eq:encoding} to a new form $(e_{1;i-1},e_{i;i+\ell-1},e_{i+\ell;r})$ where $e_{j;k}$ is an encoding of the uncompressed block-string $u_j,u_{j+1},\ldots,u_{k}$. The measurement is then performed by applying this conversion $e_{1;r}\rightarrow (e_{1;i-1},e_{i;i+\ell-1},e_{i+\ell;r})$ in superposition, then applying the ``full'' measurement \eqref{eq:submeasurement} on $\ket{e_{i;i+\ell-1}}$, and finally reversing the conversion $(e_{1;i-1},e_{i;i+\ell-1},e_{i+\ell;r})\rightarrow e_{1;r}$. This way, using binary search we can, e.g., locate the first (uncompressed) index where the measurement result is $\vI-\vP_0\otimes\ketbra{0^{c+1}}{0^{c+1}}$ using $\log(r)$ such compressed measurements. With very high probability at most $\bigO{h}$ indices will result in outcome $\vI-\vP_0\otimes\ketbra{0^{c+1}}{0^{c+1}}$, so this binary-search-based compressed measurement scheme will terminate after $\bigO{h\log(r)}$ steps with very high probability.
	We refer the reader to \cite[Section 5]{berry2014GateEfficientDiscreteSim} for further details about the precise error and run-time bounds.

	Once we obtained the list of indices $(i_1,i_2,\ldots,i_\ell)$ where the uncompressed measurement would have resulted in $\vI-\vP_0\otimes\ketbra{0^{c+1}}{0^{c+1}}$, we need to check the alternative acceptance condition, i.e., whether the second qubit is in state $1$ or not, completing the verification whether all circuit segments were applied successfully. Given such an index $i$ we first use the first ``compression'' register of the compressed encoding to identify how many nonzero (uncompressed) registers are before $i$, and then look up the corresponding block in the ``data'' register to check whether the second qubit is in state $1$; if the $i$-th uncompressed register contains $\ket{0^{a+2}}$ according to the ``compression'' register then we conclude that the second qubit is in state $0$ without looking at the ``data'' register. All of these operations can be performed in time that is polynomial in the size of the ``compression'' register, which is $\bigO{\mathrm{polylog}(1/\epsilon)}$.
	
	This completes the description of how to simulate $\e^{t \CL}\![\cdot]$ to precision $\bigO{1/\epsilon}$ with success probability $\geq 1/2$ when $t\leq 2$ using
	\begin{align}
\bigO{h\left(a+\log((r+1))\right)}&=\bigO{\log(1/\epsilon)\left(a+\log(1/\epsilon)\right)}\quad \text{ancilla qubits,}\\
	h&=\bigO{\frac{\log(1/\epsilon)}{\log\log(1/\epsilon)}}\quad \text{(controlled) uses of $\vU$ and $\vU^\dagg$,}\\
	&\bigO{(a+1)\mathrm{polylog}(1/\epsilon)}\quad \text{other two-qubit gates}.
	\end{align}

	The success probability can be improved to $1$ while keeping the precision $\bigO{1/\epsilon}$, using 3-steps of oblivious amplitude amplification, cf.~\cite{cleve2016EffLindbladianSim,gilyen2018QSingValTransf}. For $t>2$ we divide up the evolution to $\lceil t/2\rceil$ equal segments, and repeat the process $\lceil t/2\rceil$ times, setting the precision to $\epsilon/t$ in each segment. This gives the stated final complexity.
	
	With a slight modification, we can make the above algorithm work for general Lidbladians as well, assuming that we have an (at most) $a$-qubit block-encoding $\vV$ of the driving Hamiltonian term $\vH$. One just needs to modify the circuit of \autoref{fig:postWeakMeasCircuit}, sketched in \autoref{fig:postWeakMeasCircuitGen}.
	
	\begin{figure}[!ht]	
	\begin{quantikz}[wire types={q,q,q,b,b,b},classical	gap=1mm]
		 	\lstick{\ket{0}} &\gate{\vY^{\phantom{\dagg}}_{{_{\kern-1.5mm\phantom{f}}}2\delta}} &\ctrl{1}&\qw&\qw	&\ctrl{1}&\ctrl{1}	&\gate{\vY^\dagg_{\frac{5}{4}\delta}}				&\meter{}\rstick[wires=4]{accept the all-zero outcome \\\kern8mm and when the third qubit is $1$\kern8mm}\\	
		 	\lstick{\ket{0}}		&\qw&\gate{\vY_{\frac{1}{2}}}&\ctrl{2}&\qw			&\octrl{1}	&\gate{\e^{-\ri\arcsin(\sqrt{\frac{4}{5}})\vX}}&\qw						&|[meter]| \qw\\			
		 	\lstick{\ket{0}}		&\qw&\qw&\qw&\qw			&\targ{}	&\octrl{1}	&\qw						&|[meter]| \qw\\
		 	\lstick{\ket{0^b}}		&\qw&\qw&\gate[3]{\vV}&\gate[3]{\vU}	&\octrl{-1} 	&\gate[3]{\vU^\dagg}& \qw	&|[meter]| \qw \\
		 	\lstick{\ket{0^{a-b}}}	&\qw&\qw&\qw& 					&\qw	&\qw 	& \qw		&|[meter]| \qw \\
		 	\lstick{$\vrho$}		&\qw&\qw&\qw&					&\qw		&\qw			& \qw	& \qw \rstick{$= (1-\frac{5}{4}\delta) \e^{\delta \CL}\![\vrho]+\bigO{\delta^2}$}
		\end{quantikz}
		\caption{Alternative quantum circuit implementation of an approximate $\delta$-time step via a postselective weak measurement scheme including the coherence term $-i[\vH,\vrho]$ for the block-encoded Hamiltonian $\vH=\left(\bra{0^c}\otimes \vI\right)\vV\left(\ket{0^c}\otimes \vI\right)$. }\label{fig:postWeakMeasCircuitGen}
	\end{figure}
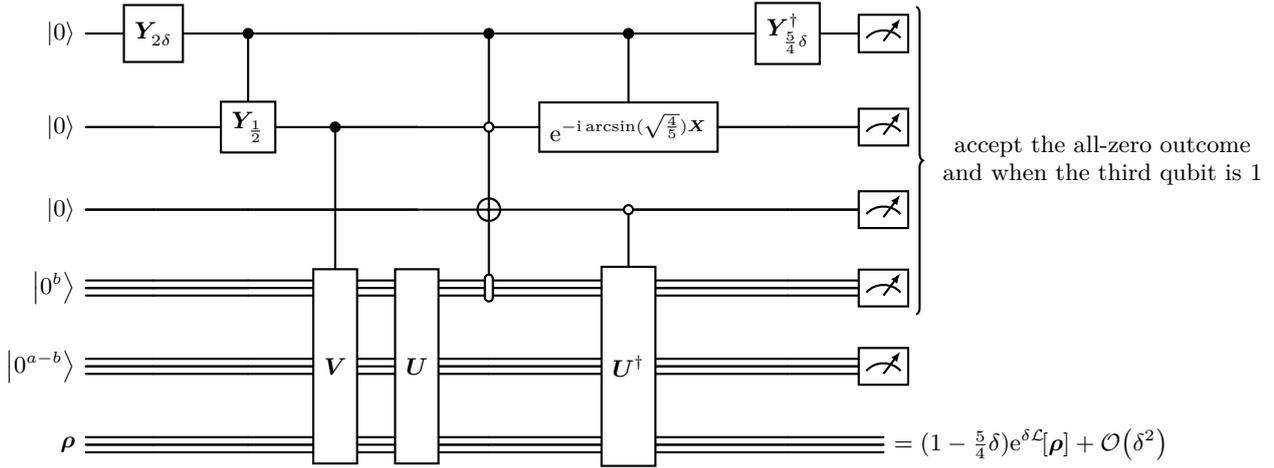
	For completeness, we include the analysis of the circuit $\vC_{\vH}$ from \autoref{fig:postWeakMeasCircuitGen}. Let us define
	\begin{align}
	\vR:=(\vY^\dagg_{\frac{5}{4}\delta}\otimes \vI)\cdot\left(\ketbra{0}{0}\otimes \vI+ \ketbra{1}{1}\otimes \e^{-\ri\arcsin(\sqrt{\frac{4}{5}})\vX}\right).
	\end{align}
	Similarly to \eqref{eq:weakMeasPsi2} we analyze the action of $\vC_{\vH}$ on a pure state $\ket{\psi}$:
	\begin{align}
	\ket{0^{a+3}}\ket{\psi}&
	\stackrel{(1)}{\rightarrow} (\sqrt{1-2\delta}\ket{0}+\sqrt{2\delta}\ket{1})\ket{00}\ket{0^c}\ket{\psi}\nonumber\\&
	\stackrel{(2)}{\rightarrow} (\sqrt{1-2\delta}\ket{00}+\sqrt{\delta}\ket{10}+\sqrt{\delta}\ket{11})\ket{0}\ket{0^c}\ket{\psi}\nonumber\\&
	\stackrel{(3)}{\rightarrow} (\sqrt{1-2\delta}\ket{000}+\sqrt{\delta}\ket{100})\ket{0^c}\ket{\psi}+\sqrt{\delta}\ket{110}\vV\ket{0^c}\ket{\psi}\nonumber\\&
	\stackrel{(4)}{\rightarrow} (\sqrt{1-2\delta}\ket{000}+\sqrt{\delta}\ket{100})\vU\ket{0^c}\ket{\psi}+\sqrt{\delta}\ket{110}\vU\vV\ket{0^c}\ket{\psi}\nonumber\\&
	\stackrel{(5)}{\rightarrow}\sqrt{1-2\delta}\ket{000}\vU\ket{0^c}\ket{\psi} + \sqrt{\delta} \ket{101}\L(\ketbra{0^b}{0^b}\otimes \vI\R)\vU\ket{0^c}\ket{\psi} + \sqrt{\delta}\ket{100}(\vI-\ketbra{0^b}{0^b}\otimes \vI) \vU\ket{0^c}\ket{\psi}\\&\kern120mm + \sqrt{\delta}\ket{110}\vU\vV\ket{0^c}\ket{\psi}\nonumber\\&
	=(\sqrt{1-2\delta}\ket{000}+\sqrt{\delta}\ket{100})\vU\ket{0^c}\ket{\psi} + \sqrt{\delta} \ket{101}\ket{0^b}\underset{\ket{\psi'_0} :=}{\underbrace{(\bra{0^b}\otimes \vI)\vU\ket{0^c}\ket{\psi}}}  - \sqrt{\delta}\ket{100}\L(\ketbra{0^b}{0^b}\otimes \vI\R) \vU\ket{0^c}\ket{\psi}\nonumber\\&\kern120mm + \sqrt{\delta}\ket{110}\vU\vV\ket{0^c}\ket{\psi}\nonumber\\&
	\stackrel{(6-7)}{\rightarrow}\left(\vR\otimes\vI\right)\left( (\sqrt{1-\delta}\ket{000}+\sqrt{\delta}\ket{100})\ket{0^c}\ket{\psi} + \sqrt{\delta} \ket{101}\ket{0^b}\ket{\psi'_0} 
	- \sqrt{\delta}\ket{100}\vU^\dagg\L(\ketbra{0^b}{0^b}\otimes \vI\R)\vU\ket{0^c}\ket{\psi}\!\right)\\&\kern110mm + \left(\vR\otimes\vI\right)\sqrt{\delta}\ket{110}\vV\ket{0^c}\ket{\psi}
	\label{eq:weakMeasPsi3}
	\end{align}
	Considering that 
	\begin{align}
	\bra{000}(\vR\otimes I)=\left(\vR^\dagg\ket{00}\right)^{\!\dagg}\bra{0}=\left(\sqrt{1-\frac{5}{4}\delta}\bra{000}+\sqrt{\frac{\delta}{4}}\bra{100}-i\sqrt{\delta}\bra{110}\right),
	\end{align}
	let us compute the part of $\vC_{\vH}\ket{0^{a+3}}\ket{\psi}$ starting with $\ket{0^{a+3}}$:
	\begin{align}
	(\bra{0^{a+3}}\otimes \vI) \vC_{\vH}\ket{0^{a+3}}\ket{\psi}
	&=\underset{=1-\frac{5}{8}\delta+\bigO{\delta^2}=\sqrt{1-\frac{5}{4}\delta}+\bigO{\delta^2}}{\underbrace{\Big(\sqrt{\Big(1-\frac{5}{4}\delta\Big)(1-\delta)}+\frac{\delta}{2}\Big)}}\ket{\psi}-\frac{\delta}{2}\underset{\sum_{j\in J}\vL_j^\dagg\vL_j\ket{\psi}}{\underbrace{(\bra{0^{c}}\otimes \vI)\vU^\dagg\L(\ketbra{0^b}{0^b}\otimes \vI\R)\vU\ket{0^c}\ket{\psi}}}\\&
	\\[-5mm]&\kern50mm - i\delta(\bra{0^{c}}\otimes \vI)\vV\ket{0^c}\ket{\psi}\nonumber\\&
	=\sqrt{1-\frac{5}{4}\delta}\left(\left(\vI-i\delta\vH-\frac{\delta}{2}\sum_{j\in J}\vL_j^\dagg\vL_j\right)\ket{\psi}+\bigO{\delta^2}\right),
	\end{align}
	and the part where third qubit is $\ket{1}$:
	\begin{align}
	(\vI\otimes \bra{1}\otimes \vI) \vC_{\vH}\ket{0^{a+3}}\ket{\psi}
	&=\sqrt{\delta}\vR\ket{10} \ket{0^b}\ket{\psi'_0} 
	= \sqrt{\delta}\vR\ket{10}\ket{0^b}\sum_{j\in J} \ket{j}\vL_j\ket{\psi}.
	\end{align}
	Let $\widetilde{\vPi}:=\frac{1}{\sqrt{1-\frac{5}{4}\delta}}\left(\ketbra{0^{a+3}}{0^{a+3}}\otimes \vI+\vI\otimes\ketbra{1}{1}\otimes \vI\right)$, the above implies that 
	\begin{align}
	\tr_{a+3}\left(\widetilde{\vPi} C \left(\ketbra{0^{a+3}}{0^{a+3}}\otimes\ketbra{\psi}{\psi}\right)C^\dagg \widetilde{\vPi}\right)
	&=\left(\vI-i\delta\vH-\frac{\delta}{2}\sum_{j\in J}\vL_j^\dagg\vL_j\right)\ketbra{\psi}{\psi}\left(\vI+i\delta\vH-\frac{\delta}{2}\sum_{j\in J}\vL_j^\dagg\vL_j\right)\\&\kern50mm+\delta \sum_{j\in J}\vL_j\ketbra{\psi}{\psi}\vL_j^\dagg+\bigO{\delta^2}\\&
	=\e^{\delta \CL}\![\ketbra{\psi}{\psi}]+\bigO{\delta^2}.
	\end{align}
	Similarly to the proof of \autoref{thm:weakMeasSim} it is easy to see that this implies
	\begin{align}
	\nrm{\tr_{a+2}\left(\widetilde{\vPi} C \left(\ketbra{0^{a+3}}{0^{a+3}}\otimes[\cdot]\right)C^\dagg \widetilde{\vPi}\right)-\e^{\delta \CL}[\cdot]}_\Diamond = \bigO{\delta^2}.\tag*{\qedhere}
	\end{align}	
\end{proof}

\section{Quantum simulated annealing}\label{sec:simulated_annealing}
A subroutine for a coherent Gibbs sampler is to prepare the top eigenvector of the discriminant. In semi-group settings, one simply iterates the map for an arbitrary initial state to find its fixed point; given coherent access to some discriminant $\CD$ (which is not quite a CPTP map), the standard approach is \emph{quantum simulated annealing}~\cite{Wocjan_2008_quantum_sampling, yung2010QuantumQuantumMetropolis,boixo2010QAlgTraversingEigStatePaths}. 
To keep this section self-contained, in the following, we assume coherent access to some discriminants $\vec{\CD}_{\beta_j}$.
First, we use QSVT to boost the gap; this is the origin of the quadratic speedup.
\begin{prop}[Quadratic speedup~\cite{low2017HamSimUnifAmp,gilyen2018QSingValTransf}]
Given a block-encoding $\vU_{\CD}$ of a Hermitian matrix $\vI+\vec{\CD}$ with eigenvalue gap $\lambda_{gap}(\CD)$ and $\lambda_1(\CD)\geq -\lambda_{gap}(\CD)$, we can construct a unitary $\vU'$ block-encoding a matrix $p(\vec{\vI+\CD})$ that  has the same top eigenvector as $\vec{\CD}$ but with $\CO(1)$ eigenvalue gap, with $\CO(\frac{1}{\sqrt{\lambda_{gap}(\CD)}})$ uses of $\vU_{\CD}$ and $\vU_{\CD}^\dagg$.
\end{prop}
Second, following~\cite{boixo2010QAlgTraversingEigStatePaths}, we consider the discretized adiabatic paths through temperatures
\begin{align}
&\ket{\lambda_1(\CD_{\beta_0})}     \rightarrow \cdots  \ket{\lambda_1(\CD_{\beta_j})}     \rightarrow \cdots  \ket{\lambda_1(\CD_{\beta_{k}})} \\
&\vU_{\CD_{\beta_0}}     \rightarrow \cdots  \vU_{\CD_{\beta_j}}     \rightarrow \cdots  \vU_{\CD_{\beta_k}} \quad \text{where}\quad \beta_j = \frac{j}{k} \beta\label{eq:adiabatic_path}.
\end{align}
While a more refined annealing schedule is possible, we consider the above linear schedule for simplicity. In particular, the initial state is the maximally entangled state $\beta = 0$. 
In the following, we show that choosing $k=\Theta(\beta\nrm{\vH})$ ensures that the consecutive overlaps remain constant large, allowing us to jump between consecutive states using a few steps of (fixed-point) amplitude amplification. 

\begin{prop}[Consecutive overlaps]\label{prop:consecutive_overlaps}
Suppose the discriminants have a top eigenvector close to the purified Gibbs states
\begin{align}\label{eq:topEVDiscrepancy}
    \norm{\ket{\lambda_1(\vec{\CD}_{\beta_j}} - \ket{\sqrt{\vrho_{\beta_j}}}} \le \frac{1}{10}.
\end{align}
Let $\delta \beta:=\beta_{j+1}-\beta_{j}$, then the consecutive overlaps are large
\begin{align}\label{eq:largeOverlap}
    \labs{\braket{\lambda_1(\vec{\CD}_{\beta_j})|\lambda_1(\vec{\CD}_{\beta_{j+1}})}}^2 \ge \frac{7}{10}-\CO((\delta \beta)^2 \norm{\vH^2\e^{-\delta \beta\vH}}) .
\end{align}
\end{prop}
\begin{proof}
 Let us evaluate the overlap between the ideal Gibbs states and rewrite using the Hilbert-Schmidt inner product
\begin{align}
\labs{\braket{\sqrt{\vrho_{\beta_j}}|\sqrt{\vrho_{\beta_{j+1}}}}}^2 = \frac{\tr\L[\e^{-\beta_j \vH/2} \e^{-\beta_j \vH/2} \e^{-\delta\beta \vH/2}\R]^2 }{\tr[\e^{-\beta_j \vH}]\tr[\e^{-\beta_j \vH} \e^{-\delta\beta \vH}]} = \frac{\braket{\e^{-\delta\beta \vH/2}}_{\beta_j}^2}{\braket{\vI}_{\beta_j}\braket{\e^{-\delta\beta \vH}}_{\beta_j}} &= 1 - \CO((\delta \beta)^2 \braket{\vH^2\e^{-\delta \beta\vH}}_{\beta_j})\\
    &\ge 1 - \CO((\delta \beta)^2 \norm{\vH^2\e^{-\delta \beta\vH}}),
\end{align}
where we denote the thermal expectation by $\braket{\vA}_{\beta} = \tr[\vrho_{\beta}\vA]$. The last equality expands the exponential
\begin{align}
    \e^{-\delta \beta \vH/2}  =  \vI - \delta \beta \vH/2 + (\e^{-\delta \beta \vH/2} - \vI + \delta \beta \vH/2).
\end{align}

The small discrepancy \eqref{eq:topEVDiscrepancy} between $\ket{\lambda_1(\CD_{\beta_j}}$ and $\ket{\sqrt{\vrho_{\beta_j}}}$ only induces a minor change in the overlaps thus proving \eqref{eq:largeOverlap}. 
\end{proof}

\begin{prop}[Simulated annealing] In the setting of \autoref{prop:consecutive_overlaps},
following the discretized adiabatic path~\eqref{eq:adiabatic_path} with $k = \lceil c \beta \norm{\vH}\rceil$ prepares a state $\ket{\psi}$ such that $\nrm{\ket{\psi}-\ket{\lambda_1(\vec{\CD}_{\beta_{j+1}}}}\leq \delta$ using
\begin{align}
    \bigO{\frac{\beta \norm{\vH}}{\min_{j}\sqrt{\lambda_{gap}(\vec{\CD}_{\beta_j})}}\log^2(\beta \norm{\vH}/\delta)}.
\end{align}
total calls for the oracle for discriminants $\vec{\CD}_{\beta_j}$.
\end{prop}
\begin{proof}
Use fixed-point amplitude amplification~\cite{yoder2014FixedPointSearch} to ``jump'' between the eigenvectors $\ket{\lambda_1(\vec{\CD}_{\beta_{j}}}$. We can implement a $\frac{\sqrt{\delta}}{k}$-approximate projector to each eigenvector with $\CO(\frac{\log(k/\delta)}{\sqrt{\lambda_{gap}(\vec{\CD}_{\beta_j})}})$ calls to the block-encoded discriminants $\vec{\CD}_{\beta_j}$. To ensure that all the $k$ jumps are all approximated to error $\leq \frac{\delta}{k}$, all fixed-point amplitude amplification consists of $\log(k/\delta)$ rounds, each calling the block-encoded (approximate) projectors.
\end{proof}
\subsection{A simple lower bound on $\beta$ dependence}\label{sec:beta_lowerbound}
In this section, we prove a simple lower-bound for the temperature dependence in the sense of implementing a reflection about the purified Gibbs state.

\begin{restatable}[Lower-bound on simulation time]{prop}{betaLowerBound}\label{prop:beta_lower_bound}
	A circuit implementing the reflection operator
	\begin{align}
		\vR_{\beta, \vH } : =\vI - \ket{\sqrt{\vrho_{\beta,\vH}}}\!\bra{\sqrt{\vrho_{\beta,\vH}}}
	\end{align}
	using Hamiltonian simulation for $\vH$ as a black-box must use Hamiltonian simulation time $T = \Omega(\beta)$.
\end{restatable}
\begin{proof}
The idea is to argue that the reflection operator is sensitive to $\beta$ and the Hamiltonian $\vH$, so the Hamiltonian simulation time $T$ cannot be too short.
First, we control the norm change of the reflection operator. Let $\vrho_{\beta,\vH}=: \vrho$ and $\vrho_{\beta,\vH'}=: \vrho'$, then 
\begin{align}
    \norm{\vR_{\beta, \vH }-\vR_{\beta, \vH'}} &= \norm{ \ketbra{\sqrt{\vrho}}{\sqrt{\vrho}} - \ketbra{\sqrt{\vrho'}}{\sqrt{\vrho'}} }\\
    & \ge  \labs{\bra{\sqrt{\vrho}} \L(\ketbra{\sqrt{\vrho}}{\sqrt{\vrho}} - \ketbra{\sqrt{\vrho'}}{\sqrt{\vrho'}}\R) \ket{\sqrt{\vrho}^{\perp}}}\\
    & = \labs{\braket{\sqrt{\vrho}| \sqrt{\vrho'}}\braket{\sqrt{\vrho'}|\sqrt{\vrho}^{\perp}}}\\
    & = \sqrt{ 1 - \theta^2 } \cdot \theta
\end{align}
where $\theta:= \sqrt{1- \labs{\braket{\sqrt{\vrho}|\sqrt{\vrho'}}}^2}$ and
\begin{align}
 \ket{\sqrt{\vrho}^{\perp}} \in Span\{ \ket{\sqrt{\vrho}}, \ket{\sqrt{\vrho'}}\} \quad \text{such that}\quad \braket{\sqrt{\vrho}^{\perp}|\sqrt{\vrho}} = 0.
\end{align}
Now, for infinitesimal $\epsilon \rightarrow 0$, let
\begin{align}
    \vH &= \ketbra{0}{0} + \ketbra{1}{1}\\
    \vH'&= \ketbra{0}{0} + (1+\epsilon)\ketbra{1}{1}
\end{align}
such that
\begin{align}
    \ket{\sqrt{\vrho}} &\propto \ket{00} + \ket{11}\\
    \ket{\sqrt{\vrho'}} &\propto \ket{00} + (1+ \beta \epsilon/2)\ket{11} + \CO(\epsilon^2).
\end{align}
Direct calculation gives $\theta = \Omega(\beta \epsilon)$ for small $\epsilon\rightarrow 0$.
To conclude the proof, suppose the block-box circuit uses only Hamiltonian simulation time $T$. Then, the resulting circuits for $\vH$ and $\vH'$ can only differ by $\norm{\vH-\vH'}T = \epsilon T$. Therefore,  for small $\epsilon\rightarrow 0$,
\begin{align}
    T \epsilon\geq \sqrt{ 1 - \theta^2 }\cdot \theta = \Omega(\beta \epsilon ),
\end{align}
proving the advertised result.
\end{proof}
The above sensitivity argument similarly applies to algorithms preparing the Gibbs state using black-box Hamiltonian simulation.
\begin{restatable}[Lower-bound on simulation time]{prop}{betaLowerBoundmixed}\label{prop:beta_lower_bound_mixed}
	A circuit preparing the Gibbs state $\vrho_{\beta}\propto \e^{-\beta \vH}$ using Hamiltonian simulation for $\vH$ as a black-box must use Hamiltonian simulation time $T = \Omega(\beta)$.
\end{restatable}
\begin{proof}
Again, consider $\vH$ and $\vH'$ as above and their Gibbs states
\begin{align}
    \vrho &\propto \ket{0}\bra{0} + \ket{1}\bra{1}\\
    \vrho' &\propto \ket{0}\bra{0} + (1+ \beta \epsilon)\ket{1}\bra{1} + \CO(\epsilon^2).
\end{align}
Then, for infinitesimal $\epsilon\rightarrow 0$, 
\begin{align}
    \norm{\vrho - \vrho'}_1 \ge \Omega(\beta\epsilon),
\end{align}
which implies $T=\Omega(\beta)$ as advertised. 
\end{proof}

\section{Impossibility of boosted shift-invariant in-place phase estimation}\label{sec:impossible}

In this section, we include the proof that certain ``boosted shift-invariant in place phase estimation'' utilized in~\cite{temme2009QuantumMetropolis} is impossible. The impossibility result was developed in parallel with this work; we reproduce the main argument here with the permission of András Gilyén and Dávid Matolcsi until their manuscript becomes publicly available.

We begin by reviewing the phase estimation assumptions made by~\cite{temme2009QuantumMetropolis}. First, they assume~\cite[Eqn.(11), Supplemental Information]{temme2009QuantumMetropolis} the phase estimation map is shift-invariant in the sense that
\begin{align}
    \Phi := \sum_{\bnu} \sum_{\bmu}\vM^{\bmu}_{\bnu} \otimes \ket{\bnu} \bra{\bmu}\quad \text{where} \quad \vM^{\bmu}_{\bnu} &:= \sum_i \alpha(E_i,\bnu-\bmu) \ket{\psi_i}\bra{\psi_i}\quad \text{and}\quad \bnu, \bmu \in \{\BZ \bomega\} \label{eq:M_nu_mu}.
\end{align} 
In~\cite{temme2009QuantumMetropolis}, the shift-invariance was shown to hold for an unboosted phase estimation unitary~\cite[Eqn.(12), Supplemental Information]{temme2009QuantumMetropolis}, which has a slowly decaying tail when $E_i$ deviates substantially from $\bnu - \bmu$. 

Second, to prove the correctness of the fixed point, they impose~\cite[Eqn.(10), Supplemental Information]{temme2009QuantumMetropolis} that the profile $\alpha(E_i,\bnu-\bmu)$ can be boosted:\footnote{This is explicitly stated in the last paragraph of page 13~\cite[Supplemental Information]{temme2009QuantumMetropolis}: ``According to (10) we can replace the function $f(E_j, k - p)$ by its enhanced counterpart $\alpha_{E_j}(k - p)$, which acts as a binary amplitude for the two closest $r$-bit integers to the actual energy $E_j$.''} the only nonzero matrix elements are such that
\begin{align}\label{eq:boosted}
 |\bnu - \bmu - E_i| < \bomega,
\end{align}
where $\bomega$ is the energy resolution for the phase estimation readout registers. Unfortunately, these two assumptions are not compatible with each other, as argued by the following.

\begin{prop}[Impossibility for shift-invariant boosting]
There exist no continuous family of ``boosted shift-invariant in place phase estimation'' unitaries. More precisely, for every constant $k$ for large enough $N$ there exists no profile $\alpha\colon \mathbb{R}\times \frac{2\pi}{N}\cdot\{0,1,2,\ldots, N-1\}\rightarrow \mathbb{C}$ that simultaneously satisfies:
\begin{itemize}
	\item (almost) unitarity: $\vM_E := \sum_{i,j=0}^{N-1} \alpha(E,\frac{2\pi}{N}(i-j\mod N))  \ketbra{i}{j}$ is close to some unitary $\vU_E$ for all $E\in \mathbb{R}\colon$ $\nrm{\vM_E - \vU_E}\leq \frac{1}{2}$
	\item boosting: $\alpha(E,\bnu-\bmu)=0$ if $|\bnu - \bmu - E \mod 2\pi| > k\frac{2\pi}{N}$
	\item continuity: $\alpha(E,\bnu)$ depends continuously on $E$
\end{itemize}
\end{prop}
\begin{proof}
We prove the statement by contradiction. Let us assume that such a profile exists for $N\gg k^3$.
We will track how the profile changes as we increase the energy from $E_0=k\frac{2\pi}{N}$ to $E_1=(3k+1)\frac{2\pi}{N}$. 

Now, consider the polynomial $p_E(z)=\sum_{j=0}^{4k+1}\alpha(E,j\frac{2\pi}{N}) z^{j}$ (whose physical meaning will become clear). Due to boosting the polynomial $p_{E_0}(z)$ has a degree at most $2k$, so it has at most $2k$ roots, and in particular, at most $2k$ roots are situated within the complex unit circle. Due to continuity, the polynomial $p_E(z)$ is continuously transformed to $p_{E_1}(z)$ whose smallest nonzero coefficient comes with a power of $z$ at least $2k+1$ due to boosting. This implies that $0$ is a root with multiplicity at least $2k+1$, and in particular, we have at least $2k+1$ roots within the unit circle. Since the polynomial changes continuously, its multi-set of roots also changes continuously, which means that at some point, a root must enter the unit circle\footnote{The idea of tracking the roots of this polynomial is due to Dávid Matolcsi.} (here we acknowledge that some roots enter from infinity when the degree of the polynomial increases but that does not affect our argument -- one can make this precise by tracking roots on the surface of the Riemann sphere). Thus, there is some energy $E'$ for which a complex unit number $z=\e^{2\pi i\varphi}$ is a root of the corresponding polynomial $p_{E'}(z)$. 

We show that this implies that the plane wave with quasi-momentum $\varphi$ is (almost) in the kernel of the shift-invariant matrix $\vM_{E'}$: for every $i\leq N-4k-2$ we have
\begin{align}
	\bra{i}\vM_{E'}\sum_{j=0}^{N-1}\ket{j} z^{-j}=\sum_{j=0}^{N-1}\alpha(E',\frac{2\pi}{N}(i-j\mod N))z^{-j}=z^{-i}p_{E'}(z)=0,
\end{align}
implying
\begin{align}
\nrm{\vM_{E'}\sum_{j=0}^{N-1}\ket{j} z^{-j}}^2=\bigO{k^3} \ll N =\nrm{\vU_E\sum_{j=0}^{N-1}\ket{j} z^{-j}}^2
\end{align}
a contradiction.
\end{proof}

\end{document}